\documentclass[11pt]{article}

\usepackage{wright}

\newcommand{\codec}{\mathrm{Code}}
\newcommand{\codee}{\mathrm{Enc}}
\newcommand{\coded}{\mathrm{Dec}}
\newcommand{\codes}{\mathrm{Sub}}

\newcommand{\ktimes}[2]{\underbrace{#1, \dots, #1}_{#2}}
\newcommand{\oktimes}[2]{\underbrace{#1 \ot \dots \ot #1}_{#2}}

\usepackage{mathrsfs}
\usepackage{array}
\usepackage{tabulary}
\newcolumntype{K}[1]{>{\centering\arraybackslash}p{#1}}

\usepackage{float}
\usepackage{tocloft}

\usepackage{phoenician}

\usepackage{multirow}

\newcommand{\comstrat}[1]{\mathrm{ComEPR}(#1)}

\newcommand{\consistency}{\simeq}

\newcommand{\polymeas}[3]{\mathrm{PolyMeas}(#1,#2,#3)}
\newcommand{\simulpolymeas}[4]{\mathrm{PolyMeas}(#1,#2,#3, #4)}

\newcommand{\hideq}{H}
\newcommand{\noop}{\bot}

\newcommand{\arith}[2]{\mathrm{arith}_{#1}(#2)}
\newcommand{\shorten}[2]{#1|_{#2}}
\newcommand{\register}{\lambda}

\newcommand{\sat}{\mathsf{3Sat}}
\newcommand{\succinct}{\mathsf{Succinct}\text{-}\mathsf{3Sat}}
\newcommand{\succinctsquared}{\mathsf{Succinct}\text{-}\mathsf{Succinct}\text{-}\mathsf{3Sat}}

\newcommand{\params}{\mathsf{params}}

\newcommand{\game}{\mathscr{G}}
\newcommand{\val}[1]{\mathrm{val}(#1)}
\newcommand{\valstrat}[2]{\mathrm{val}_{#1}(#2)}

\newcommand{\valreg}[2]{\mathrm{val}_{#1}(#2)}
\newcommand{\valsemi}[2]{\mathrm{val}_{#1}^{\mathrm{semi}}(#2)}

\newcommand{\indicator}[3]{\mathrm{ind}_{#1,#2}(#3)}

\newcommand{\proj}[1]{\ket{#1}\!\bra{#1}}
\newcommand{\twirl}{\mathscr{T}}
\newcommand{\epr}{\mathrm{EPR}}
\newcommand{\rmaux}{\mathrm{aux}}

\newcommand{\id}{\mathrm{id}}
\newcommand{\reg}[1]{{\mathsf{#1}}}

\newcommand{\hide}[2]{\mathrm{hide}_{#1}(#2)}

\newcommand{\surfaces}[1]{\mathsf{Surfaces}_{#1}}

\newcommand{\qtime}[1]{\mathsf{Q}\textnormal{-}\mathsf{time}(#1)}
\newcommand{\atime}[1]{\mathsf{A}\textnormal{-}\mathsf{time}(#1)}
\newcommand{\qlength}[1]{\mathsf{Q}\textnormal{-}\mathsf{length}(#1)}
\newcommand{\alength}[1]{\mathsf{A}\textnormal{-}\mathsf{length}(#1)}

\newcommand{\paulistrat}[2]{\mathrm{Pauli}(#1,#2)}

\begin{document}

\title{$\neexp \subseteq \mip^*$\vspace*{10pt}}

\author{Anand Natarajan\thanks{anandn@caltech.edu}\\
 \small{\sl California Institute of Technology} \and John Wright\thanks{jswright@mit.edu}\\ \small{\sl Massachusetts Institute of Technology}\vspace*{10pt}
}

\date{\today}

\maketitle

\begin{abstract}
We study multiprover interactive proof systems.
The power of classical multiprover interactive proof systems,
in which the provers do not share entanglement,
was characterized in a famous work by Babai, Fortnow, and Lund
(Computational Complexity 1991),
whose main result was the equality $\MIP = \NEXP$.
The power of quantum multiprover interactive proof systems,
in which the provers are allowed to share entanglement,
has proven to be much more difficult to characterize.
The best known lower-bound on $\mip^*$ is $\nexp \subseteq \mip^*$ due
to Ito and Vidick~(FOCS 2012).
As for upper bounds, $\mip^*$ could be as large as $\mathsf{RE}$, the class of recursively enumerable languages.

The main result of this work is the inclusion $\neexp = \ntime[2^{2^{\poly(n)}}] \subseteq \mip^*$.
This is an exponential improvement over the prior lower bound
and shows that proof systems with entangled provers are at least exponentially more
powerful than classical provers. In our protocol the verifier
delegates a classical, exponentially large $\MIP$ protocol for
$\neexp$ to two entangled provers: the provers obtain their
exponentially large questions by measuring their shared state, and use
a classical PCP to certify the correctness of their exponentially-long
answers. For the soundness of our protocol, it is
crucial that each player should not only sample its own question correctly
but also avoid performing measurements that would reveal the
\emph{other} player's sampled question. We ensure this by commanding
the players to perform a complementary measurement, relying on the Heisenberg uncertainty
principle to prevent the forbidden measurements from being performed.

\end{abstract}


\newpage

\tableofcontents
\vfill
\thispagestyle{empty}
\newpage

\part{Introduction}


\section{Introduction}  \label{sec:intro}

This paper is about the complexity class $\MIP^*$ of multiprover interactive proof systems with
entangled quantum provers---the quantum version of the classical class
$\MIP$. Classically, the study of $\MIP$ has had far-reaching
implications in theoretical computer science. In complexity theory, the proof by
Babai, Fortnow, and Lund~\cite{BFL91} that $\MIP = \NEXP$ was the
direct antecedent of the PCP theorem~\cite{ALM+98,AS98}, a seminal
result which is the foundation of the modern theory of hardness of approximation. In cryptography, the $\MIP$ model was introduced
to allow for information-theoretic zero-knowledge proofs~\cite{BGKW88}, and more recently $\MIP$ protocols have become essential building blocks
in designing delegated computation schemes (see e.g.~\cite{KRR14}). These implications alone
would be a sufficient motivation for considering the quantum class
$\MIP^*$, but remarkably, the study of $\MIP^*$ is also deeply related to
long-standing questions in the foundations of quantum
mechanics regarding the nature of quantum entanglement. Indeed, the
$\MIP^*$ model itself was anticipated by the \emph{nonlocal games} or
\emph{Bell tests} introduced in the
work of John Bell~\cite{Bel64}, who was in turn inspired by the thought
experiment proposed by Einstein, Podolsky, and Rosen~\cite{EPR35}.
These nonlocal games have had applications to quantum
cryptography~\cite{Eke91, MY98, Col06}, delegated quantum
computation~\cite{RUV13}, and more.

Even though the class $\MIP$ is now well-understood, it has proven difficult to
determine the computational power of $\MIP^*$. A priori, it is not
even clear that  $\MIP^*$ contains $\MIP$, since adding entanglement could increase or decrease the power of
the proof system. This is because the added resource of entanglement
can make it easier for dishonest provers to cheat the verifier. Indeed, Cleve et al.~\cite{CHTW04} showed that for proof systems
based on so-called XOR games (where the verifier's decision can only
depend on the XOR of the provers' answer bits), the quantum class
$\oplus\MIP^* \subseteq \EXP$, whereas classically $\oplus\MIP =
\NEXP$. In particular, this result implied that the classical
$\oplus\MIP$ protocol for $\NEXP$ of H{\aa}stad~\cite{Has97} could not be sound against entangled
provers.
In spite of this, Ito and Vidick~\cite{IV12,Vid16} were able to show that $\NEXP \subseteq
\MIP^*$, by proving that a different classical protocol \emph{is}
sound against entanglement. Note that the protocol of~\cite{Vid16} is
\emph{identical} to a protocol shown to be unsound by Cleve et al.,
except in that it uses 3 provers rather than 2 (the protocol is played
by choosing a random subset of 2 provers from the 3). This illustrates
the subtleties of dealing with entangled provers.

With the lower bound $\NEXP \subseteq \MIP^*$ established, a natural
follow-up question is whether $\MIP^*$ is \emph{strictly} more powerful
than $\MIP$. Indeed, it was long known that some $\MIP^*$ protocols
possess a uniquely quantum property called \emph{self-testing}, which
has no direct analog in the classical setting. Roughly speaking, an $\MIP^*$ protocol is
a self-test for a particular entangled state $\ket{\psi}$ if only
provers using states close to $\ket{\psi}$ can achieve close to optimal
success in the protocol. In such a protocol, observing that the
provers succeed with nearly optimal probability \emph{certifies} that
they share a state close to the target state $\ket{\psi}$. The germ of
this idea came from the work of Bell~\cite{Bel64}, who studied the types of
bipartite correlations that could be obtained from measuring an
entangled state called the EPR state, which had been introduced by
Einstein, Podolsky, and Rosen~\cite{EPR35}. Bell gave a protocol where
provers using the EPR state could succeed with a greater probability
than purely classical provers, and subsequent works of Tsirelson~\cite{Tsi80}, and Summers and
Werner~\cite{SW88} showed that (a variant of) Bell's protocol certifies
the EPR state in the sense of self-testing.

In order to prove stronger lower bonds on $\MIP^*$, the post-Ito-Vidick phase of
$\MIP^*$ research aimed to use this self-testing property to design protocols
for problems in Hamiltonian complexity, the quantum analog of the
theory of $\NP$-completeness. In Hamiltonian complexity, the complexity class
$\QMA$ plays the role of $\NP$; it is the set of problems for which
there exists a quantum witness state that can be efficiently checked
by a polynomial-time quantum verifier. Problems in $\QMA$ seemed like a natural match for
the powers of $\MIP^*$ as one could potentially construct a protocol for $\QMA$ by
designing a self-test for accepting witness states of some $\QMA$-complete
problem. The connection between $\MIP^*$ and $\QMA$ was also well
motivated from the point of view of the ``quantum PCP'' research
program, which strives to find quantum analogues of the classical PCP
theorem. In the classical setting, the PCP theorem can be viewed as a scaled-down
version of $\MIP^* = \NEXP$, showing that there exists an $\MIP^*$
protocol for 3SAT (and thus for all of $\NP$) with $O(\log(n))$-sized
messages. Drawing inspiration from this, Fitzsimons and Vidick~\cite{FV15}
stated a ``quantum games PCP conjecture'': that there should exist an
$\MIP^*$ protocol with $\log(n)$-sized messages for the local
Hamiltonian problem, and thus for the class $\QMA$. This
was proved by Natarajan and Vidick~\cite{NV18a} in 2018 with a 7-prover protocol. Along the way to achieving
this goal,~\cite{NV18a} developed a highly efficient self-test for
high-dimensional entangled states: their ``quantum low-degree test''
is a self-test for $n$ EPR pairs with only $O(\log(n))$ communication.

Already, the result of~\cite{NV18a} is strong evidence that $\MIP^*
\neq \MIP$, since it is believed that $\QMA \neq \NP$. But, at the same
time, several other works showed that even larger separations were
possible in the regime of subconstant soundness gaps. Here there are
results in two settings. For $\MIP^*$ with a soundness gap scaling
inverse-exponentially (i.e. $1/\exp(n)$) in the instance size,
Ji~\cite{Ji17} showed a protocol for $\NEEXP$: nondeterministic
\emph{doubly}-exponential time, and a subsequent work by Fitzsimons,
Ji, Vidick, and Yuen~\cite{FJVY19} showed protocols for
non-deterministic iterated exponential time (e.g. $\NTIME(2^{2^n})$)
with a correspondingly small soundness gap (e.g. $2^{-C \cdot
  2^n}$). In the ``gapless'' case, Slofstra~\cite{Slo16,Slo19} showed that given
a description of an $\MIP^*$ protocol, determining whether there exists an entangled
strategy that succeeds with probability exactly $1$ is undecidable by any Turing machine.

These results hint at the full power of $\MIP^*$ but are not
conclusive, as it is not unusual for quantum complexity classes to
increase significantly in power when a numerical precision parameter
is allowed to shrink. For instance, $\QIP$ (quantum interactive proofs
with a single prover) with an exponentially small
gap is equal to $\EXP$~\cite{IKW12}, while $\QIP$ with a polynomial gap is equal to
$\IP = \PSPACE$. Likewise, $\QMA$ with exponentially small
gap (known as $\mathsf{PreciseQMA}$) is known to be equal to
$\PSPACE$~\cite{FL18}, while $\QMA$ is contained $\PP$, and $\QMA(k)$
($\QMA$ with multiple unentangled Merlins) with exponentially small gap is
equal to $\NEXP$~\cite{Per12}, whereas in the constant-gap
regime the best known lower bound is that
$\QMA(k) \supseteq \QMA$. Moreover, even the $\QMA$ lower bound for
$\MIP^*_{\log}$ obtained by~\cite{NV18b} holds for $7$ provers only;
with $2$ provers, the best known lower bound for $\MIP^*_{\log}$ is
$\NP = \MIP_{\log}$~\cite{NV18a}. Could it be that 2-prover $\MIP^*$ is
equal to $\MIP$, with entanglement providing no advantage at all?

This paper conclusively answers this question in the negative. Our main result (\Cref{thm:main-informal}) is to show that $\MIP^*$ contains $\NEEXP$, with only
two provers and with a constant completeness-soundness gap. This is
establishes the first known unconditional separation between $\MIP^*$
and $\MIP$ in the constant-gap regime: previously, such a separation
was known only assuming $\QMA \neq \NP$, and only in the scaled-down
setting of logarithmic-sized messages.

\begin{theorem}[\Cref{thm:main} in the body]
  \label{thm:main-informal}
  There is a two-prover, one-round $\MIP^*$ protocol for the $\NEEXP$-complete
  problem $\succinctsquared$ with completeness $1$, soundness $1/2$, and
  question and answer length $\poly(n)$.
\end{theorem}

As a corollary of \Cref{thm:main-informal}, we obtain a lower bound on the hardness
of approximation for the entangled value $\omega^*$ of a nonlocal game.
\begin{corollary}\label{cor:main-games-informal}
  There exists a constant $c < 1$ such that given a two-prover
  nonlocal game $\game$ of size $N$, the problem of deciding
  whether $\omega^*(\game) = 1$ or $\omega^*(\game) \leq 1/2$,
  promised one of the two holds, is $\NTIME(2^{N^{\log^{-c} N}})$-hard.
\end{corollary}
For two-player games, the best prior lower bound was
$\NP$~\cite{NV18b}. The lower bound achieved in
\Cref{cor:main-games-informal} is
stronger as for any $c < 1$, the function $2^{N^{\log^{-c} N}}$ is
superpolynomial.

\paragraph{Techniques.}

Our construction, inspired by~\cite{Ji17}
and~\cite{FJVY19}, involves ``compression'': we show how to take an $\MIP$ protocol for $\NEEXP$
with exponentially-long questions and answers (the ``big'' protocol),
and simulate it by an $\MIP^*$ protocol with polynomial-sized messages (the ``small''
protocol). However, the techniques we use to achieve our compression
are quite different. We
eschew the Hamiltonian-complexity ideas that were used in previous
works, and in particular the use of history states. In our protocol,
honest provers need only share a quantum resource state of (exponentially
many) EPR pairs, together with a \emph{classical} assignment to the
$\NEEXP$ instance being tested. The use of history states was the main
barrier preventing previous works from applying to the case of
two provers. 

We
divide compression into two steps: \emph{question compression} and
\emph{answer compression}. We
achieve question compression by a technique which we call
\emph{introspection}, in which we command the provers to perform
measurements on their shared EPR pairs whose outcomes are pairs of questions
from the ``big'' protocol. To force the provers to sample their
questions honestly, we use a variant of the quantum low-degree test
from~\cite{NV18a}, which certifies Pauli measurements on exponentially
many EPR pairs using messages of only polynomial size. A crucial
challenge is to prevent each prover from learning the other prover's
sampled question, since this would destroy the soundness of the
``big'' protocol. To achieve this, we use the \emph{``data-hiding''} properties of quantum measurements
in incompatible bases: if a set of qubits is measured in the Pauli $X$-basis,
this ``erases'' all information about $Z$-basis
measurements. This means that if Alice samples her question by
measuring her half of a block of EPR pairs in the $Z$-basis, then her question can
be hidden from Bob by forcing him (via self-testing) to measure his
half of the EPR pairs in the $X$-basis. Interestingly, our data-hiding scheme
does \emph{not} operate in a black-box way on the ``big'' protocol, but
rather makes essential use of its structure. In particular, we start
with a ``big'' protocol based on a scaled-up version of a PCP
construction using the low-degree test, where the question
distribution consists of pairs of random points in a vector space and
affine subspaces containing them. The linear structure of the vector
space is essential for our data-hiding procedure to work.

Our approach to answer compression is more standard, essentially using composition with a
classical PCP of proximity. Here, the verifier asks the provers
to compute a PCP proof that their ``big'' answers satisfy the success
conditions of the protocol, and verifies this PCP proof by reading an
exponentially smaller number of bits. Care is needed to deal
with entanglement between the provers. The first, fundamental
challenge we face is that the success condition of the ``big''
protocol is a function of \emph{both} provers' answers. Thus, to
compute a PCP proof that the condition is satisfied, one of the
provers must have access to both provers' answers. Classically, this
is achieved using the technique of oracularization, in which one prover
receives \emph{both} provers' questions and is checked for consistency
against the other prover, which only receives a single question. In
the entangled setting, this oracularization procedure is sound, but not
necessarily complete. This is because oracularization requires that each prover, if given the
\emph{other} prover's question, could predict its answer with
certainty, even though this answer is obtained from a nondeterministic
quantum measurement. In our protocol, we are able to use
oracularization because honest provers always use
a maximally entangled state, which they measure with projective
measurements that pairwise commute for every pair of questions asked
in the game. While this commutation requirement is restrictive, it
still permits non-trivial quantum behavior; indeed, the linear system games used by Slofstra~\cite{Slo19}
involve similar commutation conditions.

The second challenge is to ensure that the PCP of proximity we use for
composition is itself sound against entanglement. We achieve this by
performing a further step of composition: we ask the provers to encode
their PCP proof in the low-degree code and verify it with the
low-degree test, which is known to be sound even against entangled
provers~\cite{NV18b}. This technique was introduced in the $\QMA$
protocol of~\cite{NV18a} in
order to perform energy measurements on the provers' state.

\paragraph{Implications and future work}
We believe that our work opens up several exciting directions for
further progress. For the complexity theorist, the most obvious future
direction is to obtain even stronger lower bounds on $\MIP^*$ by
iterating our protocol, as in~\cite{FJVY19}. At the most basic level,
we could imagine taking our $\MIP^*$ protocol for $\NEEXP$ and
performing a further layer of question compression and answer compression
on it, thus obtaining an $\MIP^*$ protocol with
logarithmic message size for $\NEEXP$, or, scaling up, an $\MIP^*$
protocol with polynomial message size for
$\ntime(2^{2^{2^{\poly(n)}}})$. By further iterating question
reduction and answer reduction $k$ times, we could obtain potentially
obtain lower bounds of
$\ntime(\underbrace{2^{\cdot^{\cdot^{n}}}}_{k})$ on $\MIP^*$ while
retaining a constant completeness-soundness gap. The main obstacle to
achieving such results is that the question compression procedure
developed in this paper is tailored to a special distribution of
questions (that of the $\MIP_{\exp}$ protocol for $\NEEXP$), whereas our
answer compression procedure produces protocols whose question
distribution is not of this form.

Assuming that this obstacle can be surmounted, we could aspire to a
more ambitious goal: a general ``gap-preserving compression procedure'' for some
subclass of $\MIP^*$ protocols, which we may label ``compressible'' protocols. Such a procedure would consist of a
Turing  machine that takes as input any compressible $\MIP^*$
protocol $\game$, and generates a new compressible protocol $\game'$ with exponentially
smaller message size, but approximately the same entangled value. It
was shown by~\cite{FJVY19} that the existence of such a compression
procedure for the set of \emph{all} $\MIP^*$ protocols would imply
that $\MIP^*$ contains the set of all computable languages,
and moreover that there exists an undecidable language in
$\MIP^*$. These consequences would continue to hold as long as the set
of compressible protocols contains a family of protocols solving
problems in $\NTIME(f(n))$, where $f(n)$ is a growing function of
$n$.

Showing that $\MIP^*$ contains undecidable languages would be
significant not just for complexity theory but also for the
foundations of quantum mechanics, as it would resolve a long-standing
open problem known as \emph{Tsirelson's problem}. Tsirelson's problem
asks whether two notions of quantum nonlocality are equivalent: the
\emph{tensor-product model}, in which two parties Alice and Bob each
act on their respective factor of a tensor-product Hilbert space
$\calH_{\reg{Alice}} \ot \calH_{\reg{Bob}}$, and the
\emph{commuting-operator model}, in which both parties act on a common
Hilbert space $\calH$, but the algebra of Alice's measurement
operators must commute with Bob's, and vice versa. It was shown by
Slofstra~\cite{Slo16} that in the ``zero-error'' setting, these two
models differ: there are quantum correlations which can be
\emph{exactly} achieved in the commuting-operator model but not in the
tensor product model. Surprisingly, showing that $\MIP^*$ contains
undecidable languages would imply that the two models are separated
even in the bounded-error setting: it would imply that there exist
correlations that can be achieved in the commuting-operator model that cannot
even be approximated (up to constant precision) in the tensor-product
model. The reason for this implication is that if the two models are
indistinguishable up to bounded error, then there exists a Turing machine
that can decide any language in $\MIP^*$ and is guaranteed to halt. 
This observation, which is folklore in the community, follows from the completeness of the non-commutative
sum of squares hierarchy for the commuting-operator model, as
documented in~\cite{FJVY19}. Showing a separation between the two
models would have significant mathematical consequences as well, as it
would yield a negative answer to the long-standing Connes' embedding problem.

In addition to these connections to complexity and mathematical
physics, we hope that our results will have applications in other
areas such as to delegated computation or quantum cryptography. In particular,
our use of introspection is reminiscent of ideas used in quantum
randomness expansion, where randomness generated by measuring EPR
pairs is used to generate questions for a nonlocal game. Could our
results improve on the infinite randomness expansion protocol of
Coudron and Yuen~\cite{CY14}?
\paragraph{Acknowledgements}
We thank Henry Yuen for many useful conversations about the idea of
``introspecting'' interactive proof protocols, which inspired us to
start this project. AN is also grateful to the Simons Institute for
the hospitable environment of the Summer Cluster on Challenges in
Quantum Computation during which these conversations where held. We thank Thomas Vidick for his guidance and advice.  We thank Ryan O'Donnell and
Ryan Williams for a succinct review of the literature on the
complexity of succinct (succinct) 3Sat and $\mathsf{NE(E)XP}$. We are
also grateful to Zhengfeng Ji for several useful discussions, especially regarding the
consequences of recursively composing our protocol with itself.

AN was partially supported by NSF grant CCF-1452616.  JW was partially
supported by ARO contract W911NF-17-1-0433. Both authors acknowledge funding provided by the Institute for Quantum Information and Matter, an NSF Physics Frontiers Center (NSF Grant PHY-1733907).



\section{Overview of our proof}  \label{sec:overview}

In this section we give a more detailed overview of the technical
parts of the paper.

\subsection{Basic quantum notation and qudits}
While the main body of the paper contains a more complete set of
quantum preliminaries in~\Cref{sec:prelims}, for the purposes of this
introduction we define some basic notation, aimed at the reader who is
familiar with the standard quantum computing formalism over qubits but
is less familiar with \emph{qudits}: quantum systems of dimension not equal to $2$.
In this paper, we make extensive use of such qudits: in particular, for a finite field $\F_Q$, we will
consider qudits of dimension $Q$, with a basis state $\ket{i}$ for
every element $i \in \F_Q$. Under tensor product, we obtain a basis
for the space of $M$ qudits of dimension $Q$ where each basis state
$\ket{x}$ corresponds to a vector $x \in \F_Q^M$.

The basic resource state used in our protocols will be the EPR state
over $2M$ qudits of dimension $Q$. The qudits are split into two
registers of $M$ qudits each, held by the two provers Alice and Bob,
respectively.
\[ \ket{\epr_Q^M} = \frac{1}{\sqrt{Q^M}} \sum_{x \in \F_Q^M} \ket{x}_{\reg{Alice}}
  \ot \ket{x}_{\reg{Bob}}. \]
This state is a \emph{maximally entangled} state between Alice and
Bob.

Acting on this state, we will ask the provers to perform measurements
from a special class called \emph{Pauli basis measurements}. To define
these over a general field $\F_Q$ requires the introduction of some
finite field technology, in particular the finite field trace
function. For simplicity, in this overview we will imagine that $Q$ is
prime, allowing the addition in $\F_Q$ to be identified with the
additive group $\mathbb{Z}_Q$, and simplifying the definition of the
Paulis; in the main body of the paper, we will work with $Q$ a power
of $2$. For a single qudit of dimension $Q$, the Pauli $X$ and $Z$
bases are the sets $\{\ket{\tau^X_u}\}_{u \in \F_Q}$ and $\{\ket{\tau^Z_u}\}_{u
  \in \F_Q}$ of vectors
\[ \ket{\tau^X_u} = \frac{1}{\sqrt{Q}}\sum_{x \in \F_Q}  \omega^{x u} \ket{x}, \qquad
  \ket{\tau^Z_u} = \ket{u}, \]
where $\omega = \exp(2\pi i/Q)$ is the $Q$-th root of unity. We denote
the projectors onto these basis states by $\tau^X_u$ and $\tau^Z_u$,
respectively. For a system of $M$ qudits, the Pauli
$X$ and $Z$ observables are a set of \emph{generalized observables}
indexed by elements of $\F_Q^M$: a generalized observable is a
Hermitian matrix with eigenvalues that are $Q$-th roots of unity.
They are given by
\[ X(v) =\sum_{u \in \F_Q^M} \omega^{u \cdot v} \tau^X_{u_1} \ot \dots
  \ot \tau^Z_{u_M}, 
  \qquad Z(v) = \sum_{u \in \F_Q^M} \omega^{u \cdot v} \tau^Z_{u_1}
  \ot \dots \ot \tau^Z_{u_M},\]
where $u_1, \dots, u_M$ are the components of the vector $u$, and
$u\cdot v$ is the dot product $\sum_{i=1}^{M} u_i \cdot
v_i$. Measuring a generalized observable means performing a projective
measurement onto the eigenvectors of the observable, with the
outcome $a$ corresponding to the eigenvector with eigenvalue
$\omega^{a}$.

\subsection{Our starting point: a classical interactive proof for
  $\NEEXP$}

We start with a classical multiprover interactive proof protocol for
$\NEEXP$. The equality $\MIP = \NEXP$ was originally shown by Babai,
Fortnow, and Lund~\cite{BFL91} using a protocol based on the \emph{multilinearity
  test}: the idea is that an exponentially-long witness for a problem
in $\NEXP$ is encoded in the truth-table of a multivariate polynomial
function over a finite field, which is linear in each of the variables
individually. The verifier is able to verify the witness
by evaluating the multilinear polynomial over appropriately chosen
points and subspaces. To scale up to $\NEEXP$, we use a much more
efficient version of the same idea, replacing the multilinearity test
with the \emph{low-degree test}, which works with multivariate
polynomials of low total degree. This more efficient construction
comes from the PCP literature. We give a relatively self-contained
presentation of the protocol in~\Cref{sec:classical-pcp}. For the
purposes of this overview, it is sufficient to know the following: any
problem in $\NEEXP$ can be reduced to satisfiability for a doubly exponentially long 3Sat
formula, succinctly encoded by a polynomial-sized circuit. (We refer
to this problem as $\succinctsquared$). Given a 3Sat formula
$\psi$, we would like the provers to prove to us that they have a
satisfying assignment $a$ to this formula. Instead of reading the
assignment directly, we will ask the provers to
encode their assignment as a multivariate polynomial $g_a: \F_Q^M \to
\F_Q$, where the number of variables $M$ and the finite field size $Q$
are appropriately chosen parameters, and return evaluations of this polynomial. To check that $a$
satisfies $\psi$, the verifier first uses a technique called
arithmetization to convert the
formula $\psi$ into a multivariate polynomial $g_{\psi}: \F_Q^{3M + k}
\to \F_Q$. The polynomial $g_{\psi}$ is chosen such that the assignment $a$
satisfies $\psi$ if and only if the expression
\[ \mathrm{sat}_{\psi,a}(x,b, w) := g_{\psi}(x,b,w) \cdot (g_a(x_1)
  -b_1)(g_a(x_2) - b_2)(g_a(x_3) - b_3) \]
is equal to $0$ at every point in a particular subset $H \subseteq \F_Q^{3M
  +k}$. Our classical protocol for $\NEEXP$ checks this condition:
\begin{infthm}[\Cref{sec:classical-pcp} in the body]
There exists a protocol $\game_0$ for $\mathsf{Succinct}$-$\mathsf{Succinct}$-$\mathsf{3Sat}$ (and hence
$\NEEXP$), where the verifier's questions to the provers are
constant-dimension subspaces of $\F_Q^{M}$, and the provers' responses
are evaluations of degree-$D$ $M$-variate polynomials on these
subspaces. The parameters $M, Q, D$ are all chosen to be $\exp(n)$,
and hence the question and answer lengths as well as the runtime of
the verifier in this protocol are
$\exp(n)$. 
\end{infthm}
The distribution over subspaces sent to the provers in $\game_0$ is
relatively simple, and in fact is independent of the instance of
$\succinctsquared$ being tested. For the purposes of this overview, the
reader can take the distribution over pairs of questions to be the
\emph{plane-point distribution} $\calD$. A pair $(\bs, \bu) \sim
\calD$ consists of a uniformly random affine plane $\bs \subseteq
\F_Q^M$, which is sent to Alice, and a uniformly random point $\bu
\in \bs$ which is sent to Bob. The full distribution over
questions in $\game_0$ is more complicated than this but the essential
ideas of our protocol will be illustrated by restricting to this case.
\subsection{Restricting the strategies: registers and compilers}
One of the main challenges in working with entangled provers is
showing soundness against general entangled strategies. An important
technique in this area is to force the provers to use a particular
state and class of measurements by playing a type of game known as a
\emph{self-test}.
\begin{infdef}
  A game $\game_{\mathrm{test}}$ is a \emph{self-test} for a state
  $\ket{\psi}$ and measurements $M^x$ if any strategy that succeeds in
  $\game_{\mathrm{test}}$ with probability $1- \eps$ must use a state
  $\ket{\psi'}$ and measurements $(M')^x$ that are
  $\delta(\eps)$-close, in the appropriate metric, to $\ket{\psi}$ and $M^x$.
\end{infdef}
Some of the earliest self-tests include the famous CHSH game, which
self-tests the Pauli $X$ and $Z$ operators on a single EPR pair (of
qubits). Self-testing technology has greatly advanced over the years,
and in this paper we design a highly efficient self-test based on the
low-degree test of~\cite{NV18a}.
\begin{infthm}[\Cref{thm:basis-test} in the body]
  The Pauli basis test $\paulistrat{n}{q}$ is a self-test for the
  state $\ket{\epr_q^n}$ and the Pauli $X$ and $Z$ basis
  measurements. This test sends the players questions of length
  $O(\log(n))$ and receives answers of length $O(\poly(n))$.
\end{infthm}
The Pauli $X$ and $Z$ measurements are ``complete'' measurements, and
as a consequence, there is no nontrivial measurement on a set $n$
qudits that can be measured jointly with both the Pauli $X$ and $Z$
measurements on those qudits. Using this property, we design a game
called the \emph{data-hiding game}, which certifies that a prover's
measurements act trivially on a specified set of qudits.
\begin{infthm}[\Cref{thm:data-hiding-game} in the body]
  The data-hiding game $\game_{\mathrm{hide}}$ is a self test for
  states $\ket{\psi} = \ket{\epr_q^n} \ot \ket{\rmaux}$ and
  measurements $M^x$ of the form $M^x = I \ot (M')^x_{\reg{aux}}$. It
  has questions of length $O(\log(n))$ and answers of length $O(\poly(n))$.
\end{infthm}
Together, the Pauli basis test and the data-hiding game allow us to restrict our analysis of our protocols to a
class of strategies we call \emph{register strategies}: strategies for
which the shared state is a collection of $\ell$ registers, each in an
EPR state, together with some auxiliary register:
\[ \ket{\psi} = \ket{\epr_{q_1}^{n_1}} \ot \dots \ot
  \ket{\epr_{q_\ell}^{n_\ell}} \ot \ket{\rmaux}, \]
and where the provers can be commanded to perform either (1) Pauli basis measurements on specified subsets of the
registers, or (2) measurements that do \emph{not} act on specified
subset of the EPR registers (but act on the auxiliary register or the
remaining EPR registers). We formalize this by designing a \emph{compiler},
which takes in a protocol $\game$ that is complete and sound for register
strategies, and produces a new protocol $\game'$ which is complete and
sound over all strategies.
\begin{infthm}[\Cref{theorem:one-to-zero-compiler}
  and~\Cref{theorem:two-to-one-compiler} in the body]\label{thm:compileroonie}
  Suppose $\game$ is a protocol for a computation problem for which completeness and soundness
  hold for register strategies, with $O(1)$ many registers of size
  $n$. (That is, for YES instances of the problem, there exists a
  register strategy achieving value $1$, and for NO instances, no
  register strategy achieves value greater than $1/2$). Let the
  questions in $\game$ be of length $Q$ and the answers be of length
  $A$. Then there exists a protocol $\game'$ which is complete and
  sound for general strategies, and for which the question length is $Q
  + \log(n)$ and the answer length is $A + \poly(n)$.
\end{infthm}
The compiled protocol $\game'$ either runs the original protocol
$\game$, or, with some probability, runs the Pauli basis test, the
data-hiding game, or a consistency test.

\subsection{Question reduction through introspection}

With our compiler in place, we have now given the verifier the power
to command the provers to perform Pauli basis measurements on a set of
EPR pairs.
We would like to use this to reduce the question size of the classical protocol $\game_0$ for $\NEEXP$ described above
from $\exp(n)$ to $\poly(n)$.
We will do so by forcing the provers, rather than the verifier, to sample the protocol's $\exp(n)$-length questions,
a technique we call ``introspection".
That is, we would like to force the provers to
sample pairs $(\bs, \bu)$ from the plane-vs-point distribution $\calD$, where $\bs$ is a uniformly random
affine plane in $\F_Q^M$, and $\bu$ a uniformly random point on
$\bs$.

To design a scheme to sample from this distribution, let us first fix
a representation of affine planes. We will represent an affine plane
by an \emph{intercept} $u \in \F_Q^M$ and two \emph{slopes} $v_1, v_2
\in \F_Q^M$. The plane given by $u, v_1, v_2$ is the set $s_u^v = \{ u +
\lambda_1 v_1 + \lambda_2 v_2 : \lambda_1, \lambda_2 \in \F_Q\}$. As a
first attempt, we may try the following scheme:
\begin{enumerate}
  \item Alice and Bob share three registers, each of which contains an
    EPR state, so their shared state is
    \[ \ket{\psi_0} = \ket{\epr_Q^M}_{R_0} \ot \ket{\epr_Q^M}_{R_1} \ot
      \ket{\epr_Q^M}_{R_2}. \]
  \item Alice first measures her half of registers $R_1$ and $R_2$ in
    the Pauli $Z$-basis, to obtain uniformly random outcomes $\bv_1,
    \bv_2$. The shared state is now
    \[ \ket{\psi_1} = \ket{\epr_Q^M}_{R_0} \ot
      (\ket{\bv_1}_{\reg{Alice}} \ot \ket{\bv_1}_{\reg{Bob}})_{R_1} \ot
      (\ket{\bv_2}_{\reg{Alice}} \ot \ket{\bv_2}_{\reg{Bob}})_{R_2}. \]
  \item Now, Alice and Bob both measure register $R_0$ in the Pauli
    $Z$-basis, both obtaining the same outcome $\bu$. The shared state
    is now
    \[ \ket{\psi_2} = (\ket{\bu}_{\reg{Alice}} \ot
      \ket{\bu}_{\reg{Bob}})_{R_0}    \ot  (\ket{\bv_1}_{\reg{Alice}} \ot \ket{\bv_1}_{\reg{Bob}})_{R_1} \ot
      (\ket{\bv_2}_{\reg{Alice}} \ot \ket{\bv_2}_{\reg{Bob}})_{R_2}. \]
    Alice sets her plane $\bs$ to be $s_{\bu}^{\bv}$ and Bob sets his point to be $\bu$.
  \end{enumerate}
  Indeed, the pair $(\bs, \bu)$ generated by this procedure is
  distributed according to $\calD$. However, there is a problem:
  through her measurement, Alice obtains additional side information,
  specifically the value of Bob's point $\bu$. Can we command Alice to
  erase the side information? In fact, we can, using the
  \emph{Heisenberg uncertainty principle}: if two observables
  anticommute, then measuring one completely destroys information
  about the other. Using this idea, we modify our protocol as follows:
  \begin{enumerate}
  \item As above.
  \item As above. At this point,  applying the definition of
    $\ket{\epr_Q^M}$, we can write the shared state as
    \[ \ket{\psi_1} \propto \sum_{u \in \F_Q^M} (\ket{u}_{\reg{Alice}}
      \ot \ket{u}_{\reg{Bob}})_{R_0} \ot
      (\ket{\bv_1}_{\reg{Alice}} \ot \ket{\bv_1}_{\reg{Bob}})_{R_1} \ot
      (\ket{\bv_2}_{\reg{Alice}} \ot \ket{\bv_2}_{\reg{Bob}})_{R_2}. \]
  \item {\bf New:} Intuitively, we would like Alice to be
    \emph{prevented} from measuring the component of the intercept
    along the directions $\bv_1, \bv_2$. This information would be
    obtained by measuring the observables\footnote{Strictly speaking,
      this is only true when $\bv_1 \cdot \bv_1 \neq 0$ and $\bv_2
      \cdot \bv_2 \neq 0$. A more rigorous treatment of this is given
      in~\Cref{sec:rotated-data-hiding}.} $Z(\bv_1), Z(\bv_2)$. To
    destroy it, we will ask Alice to measure the \emph{complementary}
    Pauli observables $X(\bv_1),
    X(\bv_2)$ on register $R_0$, obtaining outcomes $\balpha_1, \balpha_2
    \in \F_Q$. The shared state is now
    \begin{align*}
      \ket{\psi_2'} &\propto \sum_{u} \sum_{\lambda, \mu} \left(
                      \omega^{\balpha_1 \lambda + \balpha_2 \mu}
          \ket{\underbrace{u + \lambda \bv_1 + \mu \bv_2}_{u'}}_{\reg{Alice}}
        \ket{u}_{\reg{Bob}}\right)_{R_0} 
            (\ket{\bv_1}_{\reg{Alice}} \ot
                      \ket{\bv_1}_{\reg{Bob}})_{R_1} \\
      &\qquad\qquad\qquad\qquad\ot
        (\ket{\bv_2}_{\reg{Alice}} \ot \ket{\bv_2}_{\reg{Bob}})_{R_2}.
    \end{align*}
    where, as above, $\omega = \exp(2\pi i / Q)$ is a $Q$-th root of
    unity. Alice and Bob's state on $R_0$ is now a uniform 
    superposition over pairs $u, u'$ of points lying on the same
    affine subspace with slopes $\bv_1, \bv_2$.
  \item Alice and Bob both measure register $R_0$ in the $Z$ basis,
    obtaining outcomes $\bu$ and $\bu'$, respectively. The shared
    state is now
        \[ \ket{\psi_3'} = (\ket{\bu}_{\reg{Alice}} \ot
      \ket{\bu'}_{\reg{Bob}})_{R_0}    \ot  (\ket{\bv_1}_{\reg{Alice}} \ot \ket{\bv_1}_{\reg{Bob}})_{R_1} \ot
      (\ket{\bv_2}_{\reg{Alice}} \ot \ket{\bv_2}_{\reg{Bob}})_{R_2}. \]
Alice sets her
    plane to be $s^{\bv}_{\bu}$ and Bob sets his point to be
    $\bu'$. 
  \end{enumerate}
  Now, from the calculation performed above, it's clear that Bob's
  point $\bu'$ is uncorrelated with Alice's intercept $\bu$, apart
  from lying in the plane $s^{\bv}_{\bu}$, and hence there is no
  further information about Bob's point that Alice can learn by
  measuring her portion of the final state $\ket{\psi_3'}$. But Alice still obtains
  some additional information from her measurements along the way, in particular the
  outcomes $\alpha_1, \alpha_2$ of the $X$ measurements. And moreover,
  how can we certify that the $X$ measurements were performed
  correctly, since they are not Pauli basis measurements as given to
  us by the compiler? To answer these questions, we define a new game
  called the \emph{partial data-hiding game} (\Cref{thm:partial-data-hiding-game}), which certifies that
  Alice and Bob perform the steps described above and that no
  extra information is leaked. Building on this game, we can now
  design a protocol for $\NEEXP$ with small question size:

 \begin{infthm}[\Cref{thm:main-result-of-this-part} in the body]\label{thm:small-quest}
There is an $\MIP^*$ protocol $\game_{1}$ for $\NEEXP$ with questions
of length $\poly(n)$, and answers of length $\exp(n)$. The verifier can
generate the questions in $\poly(n)$ time but needs $\exp(n)$ time to
verify the answers.
\end{infthm}

\subsection{Answer reduction through PCP composition}
We have succeeded in obtaining a game with short questions, but the
answers are now exponentially long. In the last step, we will use
composition with a classical probabilistically checkable proof (PCP)
to delegate verification of the answers to the provers.

Schematically, the protocol $\game_{1}$ consists of the following
steps:
\begin{enumerate}
  \item The verifier sends Alice a question $\bx$ and Bob a question
    $\by$.
  \item Alice returns an (exponentially-long) answer $\bA$ and Bob an
    exponetially-long answer $\bB$.
  \item The verifier computes a verification predicate $V(\bx, \by,
    \bA, \bB)$ in exponential time.
  \end{enumerate}
  We would like to delegate the last step to the provers by asking
  them to compute a PCP proof that $V(\bx, \by,\bA,\bB) = 1$, which
  the verifier can check by communicating only polynomially many bits
  with the provers. However, we face
  an obstacle: Alice cannot know $\by$ and $\bB$, and neither can Bob
  know $\bx$ and $\bA$, and distributed PCPs (where neither party
  knows the entire assignment) are known to be
  impossible~\cite{ARW17}. To proceed, we will first have to modify
  $\game_{1}$ by \emph{oracularizing} it:
  \begin{enumerate}
    \item The verifier sends Alice the questions $\bx, \by$, and Bob
      either $\bx$ or $\by$, chosen uniformly at random.
    \item Alice returns exponentially-long answers $\bA, \bB$, and
      Bob returns an answer $\bC$.
    \item The verifier computes a verification predicate $V(\bx,
      \by, \bA, \bB)$ on Alice's questions and answers, and further
      checks that $\bA = \bC$, if Bob received $\bx$, or that $\bB =
      \bC$, if Bob received $\by$.
    \end{enumerate}
    The idea is that the new Alice simulates both Alice and Bob from
    the original protocol, and the new Bob certifies that the new
    Alice does not take advantage of her access to both questions to cheat.
    It is well-known that oracularization does not harm the soundness
    of interactive protocols, be they classical or quantum. However,
    in the quantum world, it is not necessarily the case that the
    oracularized protocol retains \emph{completeness}. This is because
    Alice and Bob may have been asked to perform non-compatible
    measurements in the original protocol, rendering it impossible for
    the new Alice to simulate both the original Alice and
    Bob. Fortunately for us, the honest strategy for protocol
    $\game_1$ is such that completeness under oracularization.
    
    Now that a single prover is in possession of all inputs to the
    verification predicate $V$, we can implement our idea of using a
    PCP proof. Classically, this idea is known as PCP
    \emph{composition}, and is extensively used in the PCP
    literature. In the quantum case, the requirement to maintain soundness against
    entanglement makes composition technically difficult, and we defer
    the details to \Cref{part:answer} of the paper. Once the
    composition is performed, we reach our main result.
    \begin{infthm}[\Cref{thm:main} in the body]\label{thm:whatevewre}
      There is an $\MIP^*$ protocol $\game_{2}$ for $\mathsf{Succinct}$-$\mathsf{Succinct}$-$\mathsf{3Sat}$
      (and hence for $\NEEXP$) with
      question size, answer, and verifier runtime $\poly(n)$.
    \end{infthm}
    
\subsection{Organization}

The paper is organized into five parts.  The first part is the introduction and this overview. The remaining parts are organized as follows.
\begin{itemize}
\item \Cref{part:prelims} contains two sections of preliminaries, one containing the classical background and another the quantum background.
\item \Cref{part:stack} contains the register compiler, i.e.\ the proof of \Cref{thm:compileroonie}.
	This involves designing the Pauli basis test (\Cref{sec:self-test-pauli}) and the data hiding test (\Cref{sec:data-hiding-layer}).
	\Cref{sec:register-overview} serves as an introduction to this part and contains more details on the organization.
\item \Cref{part:neexp} contains the ``introspection" question reduction step, i.e.\ the proof of \Cref{thm:small-quest}.
	To begin, we sketch the classical $\mip$ protocol for $\succinct$ in \Cref{sec:classical-pcp}.
	Then we give the introspected, i.e.\ ``big", low-degree test in \Cref{sec:big-degree},
	and finish by giving the entire small-question $\neexp$ protocol in \Cref{sec:big-neexp-protocol}.
	\Cref{sec:intersection} contains a test necessary for the protocol called the ``intersecting lines test".
	It allows us carry over the results of the low-degree test from one register to another.
\item \Cref{part:answer} contains the answer reduction, i.e.\ the proof of \Cref{thm:whatevewre}.
	The construction involves composing PCP protocols with error-correcting codes, and so \Cref{sec:error} surveys the properties we need of an error-correcting code.
	Finally, \Cref{sec:answer-reduction} contains the actual proof of the answer reduction step.
\end{itemize}


\part{Preliminaries}

\label{part:prelims}

\newcommand{\tstep}{\mathrm{twostep}}
\newcommand{\unif}{\mathrm{uniform}}
\newcommand{\tv}{\mathrm{TV}}
\renewcommand{\ol}{\overline}

\section{Classical preliminaries}  \label{sec:prelims}

\subsection{Finite fields and polynomials}
In this section we review some basic facts about finite fields and
polynomials over them. These facts can be found in a standard
reference such as~\cite{MBG+13}. Let $p$ be a prime and $q = p^t$ be a prime power. We denote by $\F_p$
and $\F_q$ the finite fields with $p$ and $q$ elements,
respectively. The field $\F_p$ is called the \emph{base field} or
\emph{prime subfield} of $\F_q$. The larger field $\F_q$ can be viewed as a
$t$-dimensional vector space $\F_p^{t}$ over the smaller field.
We define the trace $\tr: \F_q \to \F_p$ by
\[ \tr[a]  = \sum_{\ell=0}^{t-1} a^{p^\ell}. \]
The trace is a linear map under linear combinations with coefficients
drawn from $\F_p$. 

A \emph{basis} for
$\F_q$ over $\F_p$ consists of $k$ elements $\{\alpha_1, \dots,
\alpha_k\}$, such that any element $u \in \F_q$ can be written as a
linear combination
\[ u = \sum_{i=1}^k c_i \alpha_i, \]
where the coefficients $c_i \in \F_p$. Two bases $\{\alpha_i\}$ and
$\{\beta_i\}$ are \emph{dual} bases if $\tr[\alpha_i \beta_j]
= \delta_{ij}$, where $\delta_{ij}$ is the Kronecker delta function.

\paragraph{Fields of characteristic $2$}:
When $p =2$ (i.e. $q$ is even), several useful properties hold. Most
importantly for us, the field $\F_q$ has a \emph{self-dual} basis: that is, there exists a basis $\{\alpha_i\}$ such
that $\tr[\alpha_i \alpha_j] =
\delta_{ij}$~\cite[Theorem 1.9]{MBG+13}. This means that
given a field element $u = \sum_i c_i \alpha_i$, we can recover the
coefficient $c_j$ by the expression $c_j = \tr[u \alpha_j]$.

\paragraph{Fourier identities}
Below, we give two useful identities for simplifying Fourier sums over
finite fields. We set $\omega = e^{2\pi i / p}$ to be a $p$-th root of unity.

\begin{fact}\label{fact:averages-to-zero}
$\E_{\bu \in \F_q} \omega^{\tr[\bu \cdot a]} = 0$ if $a$ is nonzero.
\end{fact}
\begin{proof}
  If $a \neq 0$, then there must exist some nonzero $y \in \F_q$ such that
  $\tr[ay] \neq 0$. Let the value of the expectation we want to compute be
  denoted by $\sigma$. Then we have
  \begin{align*}
    \sigma &= \E_{\bu \in \F_q} \omega^{\tr[a \bu]} \\
           &= \E_{\bu \in \F_q} \omega^{\tr[a(\bu + y)]} \\
           &= \omega^{\tr[ay]} \E_{\bu \in \F_q}\omega^{\tr[a\bu]} \\
           &= \omega^{\tr[ay]} \sigma,
  \end{align*}
    and thus $\sigma = 0$.
  \end{proof}

  \begin{fact}\label{fact:averages-to-zero-subspace}
    Let $V$ be a subspace of $\F_q^n$. Then $\E_{\bu \sim V}
    \omega^{\tr[\langle \bu, a \rangle]} = 0$ if $a \notin V^\perp$.
  \end{fact}
  \begin{proof}
    The idea is the same as the proof of the previous fact. Suppose $a
    \notin V^\perp$. Then there exists some nonzero $y \in \F_q^n$
    such that $\tr[\langle a, y \rangle] \neq 0$. Letting the value of
    the expectation we wish to compute be denoted $\sigma$, we have
    \begin{align*}
      \sigma &= \E_{\bu \sim V} \omega^{\tr[\langle \bu, a \rangle]}
      \\
             &= \E_{\bu \sim V} \omega^{\tr[ \langle \bu + y,
               a\rangle]} \\
             &= \omega^{\tr[\langle a, y\rangle]} \sigma,
    \end{align*}
    and hence $\sigma $ must vanish.
  \end{proof}

\subsection{Two-player one-round games and $\MIP$}
A two-prover nonlocal game is an interaction between a verifier and
two noncommunicating provers, in which the verifier samples a pair of random
questions and sends them to the provers, receives a pair of answers,
and decides whether to accept or reject based on the questions and answers. In the literature, a game is usually
taken to be described by the verifier's distribution over question
pairs, together with a table describing the verifier's behavior
for all possible choices of questions and answers. For our purposes,
it will be more convenient to work with \emph{uniformly generated} families of
games, which are specified by Turing machines that sample the
questions and decide whether to accept or reject given the questions
and answers.
\begin{definition}[Two-player one-round uniform game family]\label{def:two-player-one-round}
A \emph{two-prover one-round game uniform game family} $\game$
is an interaction between a verifier and two provers, Alice and Bob.
The verifier $V = (\mathrm{Alg_{\mathrm{Q}}}, \mathrm{Alg_{\mathrm{A}}})$
consists of a ``question" randomized Turing machine $\mathrm{Alg_{\mathrm{Q}}}$
and an ``answer" deterministic Turing machine $\mathrm{Alg_{\mathrm{A}}}$.
Given an input string $\mathsf{input}$, the verifier samples two questions
$(\bx_0, \bx_1) \sim \mathrm{Alg_{\mathrm{Q}}}(\mathsf{input})$ and distributes~$\bx_0$ to Alice and~$\bx_1$ to Bob.
They reply with answers~$\ba_0$ and~$\ba_1$, respectively, and the verifier accepts if
$\mathrm{Alg_{\mathrm{A}}}(\mathsf{input}, \bx_0, \bx_1, \ba_0, \ba_1) = 1$.
A strategy for Alice and Bob is said to be \emph{classical} if they are allowed shared randomness but no shared quantum resources.
The \emph{value} of Alice and Bob's strategy is simply the probability that the verifier accepts,
and the \emph{classical value} of the game is the maximum value of any classical strategy.
We write $\qlength{\game}$ for the maximum bit length
of the questions as a function of the input $\mathsf{input}$, and similarly
$\alength{\game}$ for the maximum bit length of the answers,
$\qtime{\game}$ for the maximum running time of $\mathrm{Alg}_{\mathrm{Q}}$,
and $\atime{\game}$ for the maximum running time of
$\mathrm{Alg}_{\mathrm{A}}$. Often we will not explicitly write the
dependence of these quantities on $\mathsf{input}$.
\end{definition}

\begin{definition}[Multiprover interactive proofs]\label{def:mip-protocols}
A \emph{2-player 1-round multiprover interactive proof protocol} is a uniform
game family~$\game$ as in \Cref{def:two-player-one-round}. For
parameters $0 < s < c \leq 1$, we say that
the protocol $\game$ decides the language $L$ with completeness
$c$ and soundness $s$ if the following three conditions are true.
\begin{itemize}
\item[$\circ$] (Completeness) Suppose $\mathsf{input} \in L$.  Then there is a classical strategy for~$\game$ with value at least~$c$.
\item[$\circ$] (Soundness) Suppose $\mathsf{input} \notin L$.  Then every classical strategy for~$\game$ has value at most~$s$.
\item[$\circ$] All of $\qlength{\game}$, $\alength{\game}$,
  $\qtime{\game}$, and $\atime{\game}$ are $\poly(n)$ where $n$ is
  the bit length of $\mathsf{input}$.
\end{itemize}
The class $\MIP_{c, s}$ is the set of all languages that can be
decided by multiprover interactive proof protocols with the parameters $c,s$.

If $c - s$ is a constant, then we will suppress the dependence on them when writing $\MIP$ and just say that $L \in \MIP$.
Here, ``$c$" is referred to as the \emph{completeness} and ``$s$" is referred to as the \emph{soundness}.
We will typically deal with the case when $c = 1$ and $s = 1-\eps$, where $\eps > 0$ is a small constant.

\end{definition}
In this definition of $\MIP$, the parameters $\qlength{\game}, \alength{\game},
\qtime{\game}, \atime{\game}$ are required to be polynomial in the
input length $n$. However, in this paper, several of the intermediate
results we achieve are protocols where these parameters scale
superpolynomially (indeed, even exponentially or worse) in $n$. In
these cases, we will explicitly indicate the dependence of these
parameters on $n$.

\subsection{Low-degree code}

Let $q$ be a prime power and $h\leq q$ be an integer.  
Let $H$ be a subset of $\F_q$ of size~$h$.
For $n \geq 0$, let $x \in H^n$.
The \emph{indicator function of~$x$ over $H^n$}
is the polynomial with inputs $y \in \F_q^m$ defined as
\begin{equation*}
\indicator{H}{x}{y} := \frac{\prod_{i=1}^m \prod_{b \in H, b \neq x_i}
  (y_i - b)}{\prod_{i=1}^m \prod_{b \in H, b \neq x_i} (x_i - b)}.
\end{equation*}
There are two properties of this polynomial that we will need:
\begin{enumerate}
\itemsep -.5pt
\item[(i)] that it is low-degree, i.e. a degree-$m(h-1)$ polynomial,
\item[(ii)] that for any $x, y \in H^m$, $\indicator{H}{x}{y} = 1$ if and only if $x = y$,
and otherwise $\indicator{H}{x}{y}= 0$.
\end{enumerate}
Using this, we can define the low-degree code.

\begin{definition}[Low-degree encoding]
Let $|\calS| \leq h^m$, and let $\pi:\calS \rightarrow H^m$ be an injection.
Then the \emph{low-degree encoding} (sometimes also called the \emph{Reed-Muller encoding})
of a string $a \in \F_q^{\calS}$ is the polynomial $g_a :\F_q^m \rightarrow \F_q$ defined as
\begin{equation*}
g_a(x) := \sum_{i\in\calS} a_i \cdot \indicator{H}{\pi(i)}{x}.
\end{equation*}
\end{definition}
By the properties of the indicator function above,
(i) $g_a$ is a degree-$m(h-1)$ polynomial,
and (ii)~$g_a(\pi(i)) = a_i$ for all $i \in \calS$.
We will typically, though not always, take $\calS = [n]$.
Given an error-correcting code, there are two key properties we care about:
the rate and the distance.
The rate of the low-degree code is $n / q^m$.
As for the distance, we can estimate it with the following lemma.

\begin{lemma}[Schwartz-Zippel lemma~\cite{Sch80,Zip79}]\label{lem:schwartz-zippel}
Let $f, g$ be two unequal $m$-variate degree-$d$ polynomials over $\F_q$. Then
\begin{equation*}
\Pr_{\bx \sim \F_q^m}[f(\bx) = g(\bx)] \leq d/q.
\end{equation*}
\end{lemma}
As a result, the low-degree encoding has relative distance $m(h-1) / q$.
In a typical application, we would like a code with large rate and distance.
To achieve this, we will often use the following ``rule of thumb" setting of parameters:
\begin{equation}\label{eq:parameters}
h = \Theta(\log(n)),
\qquad
m = \Theta\left(\frac{\log(n)}{\log\log(n)}\right),
\qquad
q = \mathrm{polylog}(n).
\end{equation}
This gives a code with rate $1/\mathrm{poly}(n)$ and distance $o(1)$.
The polynomials involved are degree $d = \Theta(\log(n)^2/\log\log(n))$.

\subsection{A canonical low-degree encoding}

The low-degree encoding affords us some flexibility when choosing the parameters and the injection;
however, for our application we will have to choose these with care,
because each of our uses of the low-degree code requires that the injection~$\pi$ be efficiently computable.
In this section, we give a simple, canonical choice for the subset~$H$ and the injection~$\pi$ so that this is true.

\begin{definition}\label{def:admissible-prelims}
We say that
$n$,
$h= 2^{t_1}$,
$q = 2^{t_2}$,
and $m$
are \emph{admissible parameters} if $t_1 \leq t_2$ and $h^m\geq n$.
\end{definition}

The following definition gives the canonical encoding.

\begin{definition}[Canonical low-degree encoding]\label{def:canonical-low-degree}
Let $n$, $h = 2^{t_1}$, $q= 2^{t_2}$, and $m$ be admissible parameters.
Set $\ell = t_1 \cdot m$.
The \emph{canonical low-degree code} is defined as follows.
\begin{itemize}
\item[(i)] Let $e_1, \ldots, e_{t_2}$ be a self-dual basis for~$\F_q$ over~$\F_2$. Then we set $H$ to be the subset
	\begin{equation*}
		H:= H_{t_1, t_2} = \{b_1 \cdot e_1 + \cdots + b_{t_1} \cdot e_{t_1} \mid b_1, \ldots, b_{t_1} \in \F_2\}.
	\end{equation*}
	As desired, $|H| = h$.
\item[(ii)] Let $\sigma:= \sigma_{t_1, t_2}: \{0, 1\}^{t_1} \rightarrow H_{t_1, t_2}$ be the bijection $\sigma(b_1, \ldots, b_{t_1}) = b_1 \cdot e_1 + \cdots + b_{t_1} \cdot e_{t_1}$.
	From this, we can construct a bijection $\sigma_{\ell, t_1, t_2} : \{0, 1\}^\ell \rightarrow H^m$ by setting
	\begin{equation*}
		\sigma_{\ell, t_1, t_2}(b_1, \ldots, b_\ell) = (\sigma(b_1, \ldots, b_{t_1}), \sigma(b_{t_1+1}, \ldots, b_{2t}), \ldots, \sigma(b_{\ell-t_1+1}, \ldots, b_{\ell})).
	\end{equation*}
\item[(iii)] Given an index $i \in [n]$, write $\mathrm{bin}_\ell(i)$ for its $\ell$-digit binary encoding. 
	Then we define the injection $\pi:=\pi_{\ell, t_1, t_2}:[n] \rightarrow H^m$ as $\pi(i) = \sigma_{\ell, t_1, t_2}(\mathrm{bin}_\ell(i))$.
\end{itemize}
\end{definition}

The following proposition gives the time complexity of the canonical low-degree encoding.

\begin{proposition}\label{prop:canonical-time}
The bijection~$\sigma_{\ell, t_1, t_2}$ and the injection~$\pi:=\pi_{\ell, t_1, t_2}$ are both computable in time $m \cdot \mathrm{poly log}(q)$.
As a result, given a string $a \in \F_q^n$ and a point $x \in \F_q^m$, the value $g_a(x)$ takes time $\poly(n, m, q)$ to compute.
\end{proposition}

\subsection{Low-degree testing}

\begin{definition}[Surface-versus-point test]
The \emph{surface-versus-point low-degree test with parameters} $m$, $d$, $q$ (a prime power), and~$k$, denoted $\game_{\mathrm{Surface}}(m, d, q, k)$, is defined as follows.
Let $\bv_1, \ldots,\bv_k$ be~$k$ uniformly random vectors in $\F_q^m$, and
let $\bs$ be a uniformly random affine subspace parallel to
$\mathrm{span}\{\bv_1, \ldots, \bv_k\}$ (that is, $\bs$ is the set $\{\bw + \lambda_1 \bv_1 +\cdots
+\lambda_k \bv_k : \lambda_1, \ldots, \lambda_k \in \F_q\}$ for a uniformly
random $\bw$),
and let $\bu$ be a uniformly random point on~$\bs$.
Given these, the test is performed as follows.
\begin{itemize}
\itemsep -.5pt
\item[$\circ$] The vectors $\bv_1, \ldots, \bv_k$ and the surface $\bs$ are given to Alice, who responds with a degree-$d$ polynomial $\boldf: \bs \rightarrow \F_q$.
\item[$\circ$] The point $\bu$ is given to Bob, who responds with a number $\bb \in \F_q$.
\end{itemize}
Alice and Bob pass the test if $\boldf(\bu) = \bb$.
\end{definition}

\begin{remark}
Let us remark briefly on the encodings used in this test.
A surface~$s$ with directions $v_1, \ldots, v_k$
is encoded by the string $(u, w_1, \ldots, w_k) \in \F_q^{(k+1) n}$.
Here, $u$ is the lexicographically minimum point in~$s$,
and $w_1, \ldots, w_k$ are the rows of the matrix produced by taking the matrix with rows $v_1, \ldots, v_k$
and transforming it to reduced row echelon form.
We note that given $v_1, \ldots, v_k$ and a point $u \in s$,
this encoding can be produced in time $\poly(n, k, \log(q))$.

A function $f:s\rightarrow \F_q$ is \emph{a degree-$d$ polynomial on~$s$}
if there exists a degree-$d$ $k$-variate polynomial $f':\F_q^k \rightarrow \F_q$
such that $f'(\lambda_1, \ldots, \lambda_k) = f(u + \lambda_1 w_1 + \cdots + \lambda_k w_k)$.
When~$s$ is already known, we can encode~$f$ by specifying~$f'$,
which involves writing out its $d[k] := \binom{d+k}{k}$ coefficients in some arbitrary but fixed order.
\end{remark}

We note that this definition of the surface-versus-point test differs slightly from the standard
definition of the surface-versus-point test in two respects. First, we
do not require that the vectors $\bv_1, \ldots, \bv_k$ be linearly independent
or even nonzero, which implies that there is a
\begin{equation*}
1-
\left(1 -\frac{1}{q^m}\right)
\left(1 -\frac{q}{q^m}\right)
\cdots
\left(1 -\frac{q^{k-1}}{q^m}\right)
\leq \frac{q^k}{q^m}
\end{equation*}
 chance that $\bs$
is less than $k$-dimensional. Second, we send the vectors $\bv_1,
\ldots, \bv_k$ to Alice in addition to the description of the surface $\bs$. It
is not hard to see that these two modifications do not asymptotically harm the
soundness guarantee obtained for the standard plane-versus-point test shown by Raz and
Safra~\cite{RS97}, which we restate here.

\begin{theorem}[\cite{RS97}]\label{thm:raz-safra}
There exist absolute constants $c, c' > 0$ such that the following holds.
Suppose Alice and Bob pass $\game_{\mathrm{Surface}}(m,d,q,2)$ with probability at least $\mu$.
Then there exists a degree-$d$ polynomial $g:\F_q^m \rightarrow \F_q$ such that
\begin{equation*}
\Pr_{(\bs, \bu)}[g(\bu) = \bb] \geq \mu - c\cdot m (d/q)^{c'}.
\end{equation*}
\ignore{where $\bb$ is Bob's random answer given the point query~$\bu$.}
\end{theorem}

Explicit values for $c, c'$ have been derived by Moshkovitz and Raz~\cite{MR08},
albeit for the weaker guarantee that~$g$ be a degree-$m d$ polynomial,
which is still sufficient for most applications.

\paragraph{Communication cost.}
We can compute the communication cost of this test as follows.

\begin{itemize}
\itemsep -.5pt
\item[$\circ$] \textbf{Question length:} We encode a plane in $\F_q^m$ with a string $(u, v_1, v_2) \in \F_q^{3m}$. This requires $3 m \log(q)$ bits to communicate.
\item[$\circ$] \textbf{Answer length:} A degree-$d$ bivariate polynomial on~$\F_q$ can be described with $\binom{d+2}{2} \leq (d+1)^2$ coefficients in~$\F_q$. These require $(d+1)^2 \log(q)$ bits to communicate.
\end{itemize}
\noindent
Recalling \Cref{eq:parameters}, a typical setting of parameters
gives questions of length $\Theta(\log(n))$,
and  answers of length $\Theta(\log(n)^4/\log\log(n))$.

\subsection{Simultaneous low-degree testing}\label{sec:simultaneous-classical}

\begin{definition}[Simultaneous surface-versus-point test]\label{def:simultaneous-plane-v-point}
The \emph{simultaneous surface-versus-point low-degree test with parameters} $m$, $d$, $q$ (a prime power), $k$, and $\ell$, denoted $\game_{\mathrm{Surface}}^\ell(m, d, q, k)$,
 is defined as follows.
A draw $(\bs, \bu)$ is sampled as in $\game_{\mathrm{Surface}}(m,d,q,k)$. Given this, the test is performed as follows.
\begin{itemize}
\itemsep -.5pt
\item[$\circ$] The surface $\bs$ is given to Alice, who responds with $\ell$ degree-$d$ polynomials $\boldf_1, \ldots, \boldf_{\ell}: \bs \rightarrow \F_q$.
\item[$\circ$] The point $\bu$ is given to Bob, who responds with $\ell$ numbers $\bb_1, \ldots, \bb_{\ell} \in \F_q$.
\end{itemize}
Alice and Bob pass the test if $\boldf_1(\bu) = \bb_1, \ldots,
\boldf_{\ell}(\bu) = \bb_{\ell}$.
\label{def:simul}
\end{definition}
Classically, the $k=2$ case of this test can be reduced to a slight generalization of \Cref{thm:raz-safra} using a
simple and standard union-bound argument.
Quantumly, however, a corresponding entanglement-sound analogue of this generalization is not known to hold.
Instead, we use a slightly more involved reduction in which the~$\ell$ outputs of Alice and Bob are ``combined" to create a strategy for the ``standard" plane-versus-point test.
(This technique is standard and was also used in the proof of Lemma~4.6 in~\cite{NV18a} for the case $\ell =2$.)
In this section, we will introduce the notation needed for this reduction
and carry out the proof of the classical soundness of the simultaneous low-degree test
as a warm-up for our proof of quantum soundness later.
We begin by showing how to combine~$\ell$ functions by introducing~$\ell$ ``indexing" variables.

\begin{notation}
Let $g_1, \ldots, g_\ell:s \rightarrow \F_q$ be functions, where~$s$ is a subset of $\F_q^m$.
Then we define the new function $\mathrm{combine}_g(x, y): \F_q^\ell \otimes s \rightarrow \F_q$ as follows:
\begin{equation*}
\mathrm{combine}_g(x, y) = x_1\cdot g_1(y) + \cdots + x_\ell \cdot g_\ell(y).
\end{equation*}
We will typically apply this with $s = \F_q^m$ or~$s$ a dimension-$k$ subspace of~$\F_q^m$.
\end{notation}

If the $g_i$'s are degree-$d$ polynomials on $\F_q^m$, this produces a degree-$(d+1)$ polynomial on $\F_q^{\ell + m}$.
First, we show that given a surface-versus-point query from this $(\ell + m)$-dimensional space,
we can produce a surface-versus-point query from the $m$-dimensional space.

\begin{proposition}\label{prop:really-the-same-dist}
Given a subset $s \subseteq \F_q^{\ell + m}$, let $s_{\mathrm{proj}} = \{y \mid (x, y) \in \F_q^{\ell} \otimes \F_q^m\}$.
\begin{itemize}
\item[$\circ$] If~$s$ is a dimension-$k$ subspace of $\F_q^{\ell+m}$, then $s_{\mathrm{proj}}$ is a dimension-$k'$ subspace of $\F_q^{m}$, for $k' \leq k$.
\end{itemize}
Define $\bs' \sim_k s_{\mathrm{proj}}$ to be a uniformly random dimension-$k$ subspace of $\F_q^{m}$ containing $s_{\mathrm{proj}}$.
\begin{itemize}
\item[$\circ$] If $\bs$ and $(\bx, \by)$ are distributed as $\mathscr{D}_{\mathrm{Surface}}(\ell + m,q,k)$,
		then $\bs' \sim \bs_{\mathrm{proj}}$ and $\by$ are distributed as $\mathscr{D}_{\mathrm{Surface}}(m, q, k)$.
\end{itemize}
\end{proposition}
\begin{proof}
The first bullet follows because if $\{(x_1, y_1), \ldots, (x_k, y_k)\}$ is a set of~$k$ linearly independent vectors which span~$s$,
then $\{y_1, \ldots, y_k\}$ is a set of~$k$ vectors which span~$s_{\mathrm{span}}$, though they may no longer be linearly independent.
The second bullet follows by symmetry.
\end{proof}

Next, we show that answers to the queries on the $m$-dimensional space
can be used to produce answers to the queries on the $(\ell+m)$-dimensional space.

\begin{proposition}\label{prop:sub-subspace}
Let $s$ be a dimension-$k$ subspace of $\F_q^{\ell + m}$, and let $s' \subseteq \F_q^{m}$ be a subspace which contains $s_{\mathrm{proj}}$.
Then $\F_q^\ell \otimes s'$ is a subspace, and it contains~$s$.
In particular, if $f_1, \ldots, f_\ell$ are degree-$d$ functions on $s'$, then $\mathrm{combine}_f$ is a degree-$(d+1)$ function on $\F_q^\ell \otimes s'$,
and it can be restricted to a degree-$(d+1)$ function on~$s$.
\end{proposition}
\begin{proof}
Consider a point $(x, y) \in s$. Then $y \in s_{\mathrm{proj}} \subseteq s'$, and so $(x, y) \in \F_q^{\ell} \otimes s$.
The statement about $\mathrm{combine}_f$ follows immediately.
\end{proof}

Finally, we need a technical result: that nonlinear low-degree polynomials rarely become linear after restricting variables.

\begin{definition}
Let $n \geq 0$.
A function $f:\F_q^{\ell + n} \rightarrow \F_q$ is \emph{exactly linear in~$x$}
if it can be written as
\begin{equation*}
f(x, y) = x_1 \cdot f_1(y) + \cdot + x_\ell \cdot f_{\ell}(y).
\end{equation*}
(We do not allow constant terms.)
Note that when $n = 0$, such a function can be written as $c_1 \cdot x_1 + \cdots + c_\ell \cdot x_\ell$, where each $c_i \in \F_q$,
in which case we simply call it ``exactly linear".
Given a function $f(x, y)$ and a string $y \in \F_q^m$,
we will also write $f|_y$ for the function defined as $f_y(x) = f(x, y)$.
\end{definition}

\begin{proposition}\label{prop:exactly-linear-prop}
Suppose $f(x, y): \F_q^{\ell + m} \rightarrow \F_q$
is a degree-$d$ polynomial which is not exactly linear in~$x$.
Then the probability that $f|_{\by}$ is exactly linear, over a uniformly random $\by \sim \F_q^m$, is at most $d/q$.
\end{proposition}
\begin{proof}
Because~$f$ is not exactly linear in~$x$, it contains some non-linear $x$-monomial $x^i = x_1^{i_1} \cdots x_\ell^{i_\ell}$
in which $i_1 + \cdots + i_\ell$ is either zero or at least two.
Thus, $f$ can be written as
$f(x, y) = x^i \cdot g_i(y) + f'(x, y)$, where $g_i(y)$ is degree-$d$ and $f'$ contains no~$x^i$ terms.
For $f|_{\by}$ to be exactly linear, this term must vanish, which means $g_i(\by) = 0$.
But by Schwartz-Zippel (\Cref{lem:schwartz-zippel}), this happens with probability at most $d/q$.
\end{proof}

We are now ready to prove soundness of the simultaneous low-degree test in the $k = 2$ case.

\begin{theorem}\label{thm:simultaneous-raz-safra}
There exists absolute constants $c, c' >0$ such that the following holds.
Suppose Alice and Bob pass $\game_{\mathrm{Surface}}^\ell(m, d, q,2)$ with probability at least $\mu$.
Then there exist degree-$d$ polynomials $g_1, \ldots, g_{\ell} : \F_q^m \rightarrow \F_q$ such that
\begin{equation*}
\Pr_{(\bs, \bu)} [g_1(\bu) = \bb_1, \ldots, g_{\ell}(\bu) = \bb_\ell] \geq \mu - c \cdot (m + \ell) (d/q)^{c'}.
\end{equation*}
\end{theorem}

\begin{proof}
Let $c, c' > 0$ be as in \Cref{thm:raz-safra}.
We pick the constants in this theorem, say $\hat{c}, \hat{c}'$  so that
\begin{equation*}
\mu - c \cdot (m + \ell) ((d+1)/q)^{c'} - 2(d+1)/q \geq \mu - \hat{c} \cdot (m + \ell) (d/q)^{\hat{c}'}.
\end{equation*}
Note that this means that the theorem is trivial when $2(d+1)/q \geq  \mu - c \cdot (m + \ell) ((d+1)/q)^{c'}$.
As such, we will assume below that
\begin{equation}\label{eq:really-technical}
2(d+1)/q <  \mu - c \cdot (m + \ell) ((d+1)/q)^{c'}.
\end{equation}

Suppose Alice and Bob pass $\game_{\mathrm{Surface}}^\ell(m, d, q,2)$ with probability at least $\mu$.
We will use them to simulate two provers, ``Combined Alice" and ``Combined Bob", who pass
the single-function low-degree test $\game_{\mathrm{Surface}}(\ell + m, d+1, q,2)$ with probability at least $\mu$.
They are specified as follows:
\begin{itemize}
\item[$\circ$] \textbf{Combined Alice:} Given $\bs \subseteq \F_q^{\ell + m}$, draw $\bs' \sim_2 \bs_{\mathrm{proj}}$.  Give it to Alice, who responds with $\boldf_1, \ldots, \boldf_{\ell}: \bs' \rightarrow \F_q$.  Output the function $\mathrm{combine}_{\boldf}|_{\bs}$.
\item[$\circ$] \textbf{Combined Bob:} Given $(\bx, \by) \in \F_q^{\ell + m}$, compute $\by \in \F_q^m$. Give it to Bob, who responds with $\bb_1, \ldots, \bb_\ell \in \F_q$.
			Return $\mathrm{combine}_{\bb}(\bx) \in \F_q$.
\end{itemize}
By \Cref{prop:really-the-same-dist}, $\bs'$ and $\by$ are distributed as the questions in $\game_{\mathrm{Surface}}^\ell(m, d, q,2)$.
Using our assumption on Alice and Bob, this means that $\boldf_1(\by) = \bb_1$, \ldots, $\boldf_\ell(\by) = \bb_\ell$ with probability at least $\mu$.
As a result, $(\mathrm{combine}_{\boldf}|_{\bs})(\bx, \by) = \mathrm{combine}_{\bb}(\by)$ with probability at least~$\mu$.
By \Cref{prop:sub-subspace}, $\mathrm{combine}_{\boldf}|_{\bs}$ is a degree-$(d+1)$ function on~$\bs$,
and so it is a valid response to subspace queries.
This means Combined Alice and Bob pass~$\game_{\mathrm{Surface}}(\ell + m, d+1, q,2)$ with probability at least~$\mu$.

Thus, we can apply \Cref{thm:raz-safra}. It gives a degree-$(d+1)$ function $g: \F_q^{\ell + m} \rightarrow \F_q$ such that
\begin{equation}\label{eq:g-equals-combine}
\Pr_{\bx, \by}[g(\bx, \by) = \mathrm{combine}_{\bb}(\bx)] \geq \mu - c \cdot (\ell + m)\cdot ((d+1)/q)^{c'}.
\end{equation}
We would like to show that~$g$ is exactly linear in~$x$.
Assume for the sake of contradiction that this is not the case.
Because $\bb$ depends only on~$\by$ (and Bob's internal randomness),
we can consider varying these two variables independently of~$\bx$.
By \Cref{prop:exactly-linear-prop}, the probability that $g_{\by}$ is not exactly linear is at least $1-(d+1)/q$.
In this case, because $\mathrm{combine}_b(\bx)$ is always exactly linear,
the probability that $g|_{\by}(\bx) = \mathrm{combine}_{\bb}(\bx)$ is at most $(d+1)/q$ by Schwartz-Zippel (\Cref{lem:schwartz-zippel}).
As a result, the probability that $g(\bx, \by) = \mathrm{combine}_{\bb}(\bx)$ is at most $(d+1)/q + (d+1)/q$, which contradicts \Cref{eq:really-technical,eq:g-equals-combine}.
Thus, we may conclude that~$g$ is exactly linear in~$x$.

This implies that we can write $g(x, y) = \sum_i x_i \cdot g_i(y)$, where each $g_i$ is a degree-$d$ polynomial.
Now, for any fixed~$b$ and~$y$, if it is not the case that $g_1(y) = b_1$, \ldots, $g_\ell(y) = b_\ell$,
then the probability that $g(\bx, y) = \mathrm{combine}_{b}(\bx)$ over a random~$\bx$ is at most $1/q$ by Schwartz-Zippel since both are exactly linear functions.
Thus, if $\eta$ is the probability that $g_1(\by) = \bb_1$, \ldots, $g_\ell(\by) = \bb_\ell$,
then the probability that $g(\bx, \by) = \mathrm{combine}_{\bb}(\bx)$ is at most $\eta + (1-\eta)/q \leq \eta + 1/q$.
Combined with \Cref{eq:g-equals-combine}, this implies the theorem.
\end{proof}

\subsection{$\NEXP$, $\NEEXP$, and complete problems for them}

\begin{definition}
  \label{def:neexp}
  The class $\neexp$ (respectively, $\NEEXP$) is the class of all problems that can be solved
  in exponential (respectively, doubly-exponential) time by a nondeterministic Turing
  machine. Formally,
\begin{equation*}
\NEXP = \bigcup_{c \in \mathbb{N}}  \NTIME(2^{n^c}),\qquad
\NEEXP = \bigcup_{c \in \mathbb{N}}  \NTIME(2^{2^{n^c}}).
\end{equation*}
\end{definition}

A standard way of generating $\NEXP$-complete problems 
is by considering ``succinct" versions of $\NP$-complete problems,
in which an exponential-sized input is encoded by a polynomial-sized circuit.
The canonical complete problem is a succinct version of $\mathsf{3Sat}$,
but there is considerable freedom in choosing the succinct encoding used.
We choose the following encoding.

\begin{definition}
$\succinct$ is the following problem.
\begin{enumerate}
\item[$\circ$] \textbf{Input:} a circuit~$\calC$ with $3n+3$ input bits and size $\mathrm{poly}(n)$.
It encodes the 3-Sat instance $\psi_{\calC}$ with variable set $x_u$ for $u \in \{0, 1\}^n$
which includes the constraint $(x_{u_1}^{b_1} \lor x_{u_2}^{b_2}\lor x_{u_3}^{b_3})$ whenever
\begin{equation*}
\calC(u_1, u_2, u_3, b_1, b_2, b_3) = 1.
\end{equation*}
(Here, $x_i^1$ refers to the literal $x_i$ and $x_i^0$ refers to the negated literal $\overline{x_i}$.)
\item[$\circ$] \textbf{Output:} accept if $\psi_{\calC}$ is satisfiable and reject otherwise.
\end{enumerate}
\end{definition}

A proof that $\succinct$ is $\NEXP$ complete can be found in~\cite[Chapter 20]{Pap94},
albeit with a different encoding.
Below, we show this implies $\NEXP$-completeness for our encoding as well.

\begin{proposition}
$\succinct$ is $\NEXP$-complete.
\end{proposition}
\begin{proof}
Papadimitriou~\cite{Pap94}
considers circuits~$\calC_{\mathrm{Pap}}$ which encode $\mathsf{3Sat}$ formulas~$\phi$ with $n$ variables and~$m$ clauses
as follows:
$\calC_{\mathrm{Pap}}$ takes as input a string $(b, u, k)$, where $b, k \in \{0, 1\}^2$ are interpreted as integers in $\{0, 1, 2, 3\}$
and $u \in \{0, 1\}^{\log(m)}$ is interpreted either as a vertex $1 \leq u \leq n$ or a clause $1 \leq u \leq m$.
If $1 \leq u \leq n$ and $0 \leq k \leq 2$, 
then on input $(0, u, k)$, $\calC_{\mathrm{Pap}}$ outputs the index of the clause where $\overline{x_u}$ appears for the $k$-th time,
and on input $(1, u, k)$,  it outputs the index of the clause where $x_u$ appears for the $k$-th time.
(In addition, if $1 \leq u \leq m$ and $0 \leq k \leq 3$, then on input $(2, u, k)$,
$\calC_{\mathrm{Pap}}$ outputs the $k$-th literal of the $u$-th clause in $\phi$.
We state this for completeness, though we will not need it for the proof.)
Such a $\mathsf{3Sat}$ formula~$\psi$ has $2n$ literals, each occurring~$3$ times, and so $m = 2n$.
By~\cite[Chapter 20]{Pap94}, this succinct encoding of $\mathsf{3Sat}$ is $\neexp$-complete.
Using this, we can generate an instance of the $\succinct$ problem~$\calC$ such that $\phi_\calC = \phi$ as follows:
given input $(u_1, u_2, u_3, b_1, b_2, b_3)$, we simply evaluate $\calC_{\mathrm{Pap}}$ on $(b_i, u_i, k)$, for each $1 \leq i \leq 3$ and $0 \leq k \leq 2$, and output~$1$ if there is any clause containing all three literals.
\end{proof}

The complete problem for $\NEEXP$ is, appropriately enough, a succinct
version of $\succinct$. To define it precisely, it helps to fix a
notion of a Boolean circuit. Following Section~4.3 of~\cite{Pap94}, we consider Boolean circuits in which each gate can
be one of six types: $\mathsf{input}, \mathsf{true}$, $\mathsf{false}$, $\wedge$,
$\vee$, or $\lnot$. These gates have 0, 0, 0, 2, 2, or 1 inputs,
respectively. A succinct representation of a circuit $\calC_1$ is a
circuit $\calC_2$ that, given an index $i$, outputs the type of gate
$i$ as well as the indices $j_1, j_2$ of its inputs (one or both of
these indices may be the null index $\varnothing$ depending on the
type of the gate $i$).
\begin{definition}
  $\succinctsquared$ is the following problem.
  \begin{itemize}
    \item[$\circ$] {\bf Input:} a circuit $\calC$ with size
      $\poly(n)$, which is a succinct representation of
      a circuit $\calC'$, which is itself an instance of $\succinct$
      with instance size $N = 2^{\poly(n)}$.
      \item[$\circ$] {\bf Output:} accept if $\psi_{\calC'}$ (the $\sat$
        formula on $2^N = 2^{2^{\poly(n)}}$ variables generated by the
        circuit $\calC'$) is
        satisfiable and reject otherwise.
      \end{itemize}
\end{definition}
\begin{fact}\label{fact:p-to-ppoly}
  Let $M$ be a deterministic Turing machine which takes two inputs $x_1, x_2$. Then
  for any input $x_1$ of size $n_1$ and for any size parameter $n_2$
  and time $T > n_1 + n_2$, there exists a
  circuit $C_{M,T,x_1}$ of size $N = O(T^2)$ which, on an input $x_2$ of
  size $n_2$, computes $M$ run for $T$
  steps on the input pair $x_1, x_2$. Moreover, there exists a Turing
  machine $M'$ that given $x_1$, $n_2$, and an
  index $i \in  \{1, \dots, N\}$ in binary, outputs in polynomial time
  the type of the $i$th gate of
  $C_{M,T, x_1}$ and the indices $j_1, j_2'$ of the inputs to this gate.
\end{fact}
\begin{proof}
  The construction in the proof of Theorem~8.1 of~\cite{Pap94} yields
  a circuit of the desired size. This circuit consists of $O(T^2)$
  copies of a constant-sized circuit $C_{M}$ that depends only on
  $M$. 
\end{proof}
\begin{theorem}[Cook-Levin]\label{thm:cook-levin}
  Let $L$ be a language in $\NTIME(T(n))$. Then the following
  properties hold:
  \begin{enumerate}
  \item For every string $x$ of length $n$, there exists a $\sat$
    formula $\Phi_x$ on $n' = \poly(T(n))$ variables $z_1, \dots, z_{n'}$, such that $x \in L$ iff
    $\Phi_x$ is satisfiable.
    \item There exists a Turing
      machine $R$ that given an input $x$ of
      length $n$, three indices $u_1, u_2, u_3 \in \{1, \dots, n'\}$ in
      binary, and three bits $b_1, b_2, b_3 \in \{0,1\}$, runs in
      $\poly\log(n') = \poly\log(T(n))$ time and outputs $1$ iff the clause
      $(z_{u_1}^{b_1}, z_{u_2}^{b_2}, z_{u_3}^{b_3})$ is included in $\Phi_x$.
    \end{enumerate}
  \end{theorem}
  \begin{proof}
    We follow the proof of Theorem~8.2 of~\cite{Pap94} to obtain the
    $\sat$ instance $\Phi_x$. 
  \end{proof}

\begin{theorem}
  $\succinctsquared$ is complete for $\NEEXP$ under polynomial time
  mapping reductions. That is, for any language $L$ in $\NEEXP$,
  there exists a Turing machine $R$ which takes as input a string $x
  \in \{0,1\}^n$, and in time $\poly(n)$ outputs an instance $\calC_{x}$ of
  $\succinctsquared$, such that $\calC_{x}$ is satisfiable iff $x \in L$.
\end{theorem}
\begin{proof}
  Suppose we start with a language $L \in \NEEXP$. This means there is
  a nondeterministic Turing machine $M$ which decides $L$ in time
  $T_0 = 2^{2^{n^c}}$.
  By \Cref{thm:cook-levin}, there exists a Turing machine $R_1$
  which runs in time $T_1 = \poly\log(T_0) = 2^{O(n^c)}$ such that
  given input $x$ and clause indices $i = (u_1, u_2, u_3, b_1, b_2, b_3)$,
  represented as a binary string of length $|i| = \log(\poly(T_0)) =
  \log(2^{O(2^{n^c})}) = O(2^{n^c})$, runs in time polynomial in
  $\ell$ and outputs $1$ iff the corresponding clause exists in a
  $\sat$ formula  $\Phi_x$ such that $x \in L$ iff $\Phi_x$ is satisfiable. 

  Now, if we apply~\Cref{fact:p-to-ppoly} to $R_1$, with $x$ playing
  the role of the first input $x_1$ and $i$ the role of the second
  input, we obtain that for every $x$ there exists a circuit $C_{R_1, T_1,
    x}$ of size $O(T_1^2) = 2^{O(n^c)}$ which takes as input a tuple
  of indices $i$, and runs $R_1$ for time $T_1$ on this time to output
  whether clause $i$ is present in the formula $\Phi_x$. Moreover, there exists a
  Turing machine $R_2$ that, given $x$, the size parameter $|i| =
  O(2^{n^c})$, represented in binary as a string of $O(n^c)$ bits, and an index $j$, represented
  as a string of $O(n^c)$ bits, outputs
  the $j$th gate of $C_{R_1, T_1, x}$ in time $T_2 = \poly(n)$. Note that
  $R_2$ is a Turing machine which takes in input of size $\poly(n)$
  and runs in time $\poly(n)$.

  We are now almost where we need to be. In the final step, we once
  again apply \Cref{fact:p-to-ppoly} to $R_2$, obtaining a third Turing machine $R_3$
  that takes as input $x$ and the size parameters, and an index $k$,
  and generates the $k$th gate of the circuit $C_{R_2, T_2, x}$
  corresponding to running $R_2$ for $T_2$ steps. Finally, by
  fixing the dependence of the size parameters on the size of $x$, and
  iterating through all possible values of the index parameter, we
  obtain a Turing machine $R_3'$ that takes as input $x$ and runs in
  time $\poly(n)$, and outputs the complete description of a $\succinctsquared$ instance $\calC_x$
  with the desired properties.
\end{proof}

\subsection{The Tseitin transformation}

In this section, we introduce the Tseitin transformation, which is a simple method of converting a Boolean circuit into a Boolean formula.

\begin{definition}[Tseitin transformation]
Let~$\calC$ be a Boolean circuit with~$n$ input variables $x_1, \ldots, x_n$ and~$s$ gates.
Then the \emph{Tseitin transformation of $\calC$}, denoted $\calF := \mathrm{Tseitin}(\calC)$, is the Boolean formula defined as follows.
\begin{itemize}
\item[(i)] Introduce new variables $w_1, \ldots, w_s$ corresponding to the output wires of the gates in~$\calC$.
		Then the input variables to $\calF$ consist of  $x_1, \ldots, x_n$ along with $w_1, \ldots, w_s$.
\item[(ii)] Each gate in~$\calC$ operates on one or two variables in $\{x_1, \ldots, x_n, w_1, \ldots, w_s\}$.
		Write $g_i(x, w)$ for the function computed by the $i$-th gate.
		Then $\calF$ computes the intermediate expression
		\begin{equation*}
			z_i := (g_i(x, w) \land w_i) \lor (\overline{g_i(x, w)} \land \overline{w_i}).
		\end{equation*}
		The final output of $\calF$ is $z_1 \land (z_2 \land ( \cdots \land z_s))$.
\end{itemize}
By construction, $\calC(x) = 1$ if and only if there exists a~$w$ such that $\calF(x, w) = 1$
(in particular, $w$ is taken to be the wire values of~$\calC$ on input~$x$).
In addition, $\calF$ contains exactly $7s + (s-1)$ gates, meaning that it has size~$O(s)$.
\end{definition}

Next, we show how to convert Boolean formulas into functions over $\F_q$.

\begin{definition}[Arithmetization]\label{def:arithmetization}
Let $\calF$ be a Boolean formula of~$n$ variables and size~$s$.
The \emph{arithmetization of $\calF$ over $\F_q$}, denoted $\arith{q}{\calF}$, is the formula produced by the following two-step process.
\begin{itemize}
\item[(i)] Transform $\calF$ by replacing all~$\lor$ gates with appropriate~$\land$ and~$\neg$ gates.
\item[(ii)] Transform each Boolean gate into an $\F_q$ gate as follows:
		Replace each~$\land$ gate in $\calF$ with a $\times$ gate.
		Replace each~$\neg$ gate with a $\times -1$ gate followed by a $+1$ gate (enacting the transformation $b \in \F_q \mapsto 1-b$).
		Call the resulting formula $\arith{q}{\calF}$.
\end{itemize}
Set $\calF_{\mathrm{arith}} := \arith{q}{\calF}$.
On inputs $x \in \{0, 1\}^n$, $\calF_{\mathrm{arith}}(x) = \calF(x)$.
On general inputs $x \in \F_q^n$, $\calF_{\mathrm{arith}}(x)$ is computable in time $\poly(s, q)$.
\end{definition}

The following proposition shows that small Boolean formulas have low-degree arithmetizations.

\begin{proposition}[Low-degree arithmetization]\label{prop:low-degree-arithmetization}
Let $\calF$ be a Boolean formula of~$n$ variables, size~$s$, and~$m$ gates.
Then $\arith{q}{\calF}$ is a degree-$s$ polynomial over~$\F_q$.
\end{proposition}
\begin{proof}
By induction on the number of gates, the base case $(m = 0)$ being trivial.
For the induction hypothesis, assume the proposition holds for Boolean formulas which have fewer than~$m$ gates.
Either the gate at the root of $\calF$ is a $\neg$ gate or an $\{\lor, \land\}$-gate.
In the former case, $\calF = \neg \calF'$ for some Boolean formula with $m-1$ gates,
and so $\arith{q}{\calF} = 1 - \arith{q}{\calF'}$ by construction.
But these have the same degree, and so $\arith{q}{\calF}$ is degree~$s$ by the induction hypothesis.
In the latter case, assume without loss of generality that it is an $\land$-gate. 
Then $\calF= \calF_{\mathrm{left}} \land \calF_{\mathrm{right}}$ for two formulas
of size $s_{\mathrm{left}} + s_{\mathrm{right}} = s$ and fewer than~$m$ gates.
By the induction hypothesis, $\arith{q}{\calF_{\mathrm{left}}}$ has degree-$s_{\mathrm{left}}$ and
$\arith{q}{\calF_{\mathrm{right}}}$ has degree-$s_{\mathrm{right}}$,
and so $\arith{q}{\calF} = \arith{q}{\calF_{\mathrm{left}}} \times \arith{q}{\calF_{\mathrm{right}}}$ has degree~$s$.
\end{proof}

The arithmetization procedure describe in \Cref{def:arithmetization} can also be applied to general Boolean circuits~$\calC$, not just Boolean formulas.
But \Cref{prop:low-degree-arithmetization} does not apply to general circuits;
in fact, the arithmetization of a Boolean circuit can have very high degree, even if that circuit is small.
This motivates using the Tseitin transformation: it allows us to convert a small circuit into a small formula, which has a low-degree arithmetization.



\section{Quantum preliminaries}  \label{sec:prelims}
\subsection{Quantum measurements}
The most general notion of a quantum measurement is a POVM
measurement, which consists of a set of Hermitian operators $\{M_a\}_{a \in S}$
indexed by outcomes $a$ from a set $S$. These satisfy the
conditions
\[ \forall a,\; M_a \succeq 0, \qquad \sum_a M_a = I. \]
To refer to the measurement as a whole
we will use the letter $M$, without the subscript indicating the
outcome. For a state $\ket{\psi}$, the
probability that the measurement $M$ returns outcome $a$ is
\[ \Pr[\ba] = \bra{\psi} M_{\ba} \ket{\psi}. \]
A POVM is said to be \emph{projective} if each element $M_a$ is an
orthogonal projector, i.e. $M_a^2 = M_a$. Note that this implies that
$M_a M_b = 0$ for any $a \neq b$, i.e. that the projectors are
pairwise orthogonal. Naimark's theorem says that any POVM measurement
can be simulated by a projective measurement on an enlarged space.
\begin{theorem}[Naimark]\label{thm:naimark}
  Suppose $\{M_a\}$ is a POVM acting on a Hilbert space $\calH$. Then
  there exists a projective measurement $\{M'_a\}$ acting on the space
  $\calH \ot \calH_{\rmaux}$ together with a state $\ket{\rmaux}$ such
  that for all states $\ket{\psi} \in \calH$ and all outcomes $a$, the
  post-measurement state after applying $M$ and $M'$ is the same:
  \begin{equation} \sqrt{M_a} \ket{\psi}\bra{\psi} \sqrt{M_a} = \tr_{\reg{aux}}(M'_a (\ket{\psi}\bra{\psi}
    \ot \ket{\rmaux}\bra{\rmaux})M'_a). \label{eq:naimark}\end{equation}
  As a consequence, $M$ and $N$ induce the same distribution over
  outcome probabilities:
  \[ \bra{\psi} M_a \ket{\psi} = (\bra{\psi} \ot \bra{\rmaux}) N_a
    (\ket{\psi} \ot \ket{\rmaux}). \]
  Moreover, given any upper-bound $n$ on the number of outcomes of
  $M_a$, there is a universal choice of the state $\ket{\rmaux}$ that
  works for all POVMs $M_a$ with at most $n$ outcomes. The projective measurement
  $M'_a$ and state $\ket{\rmaux}$ together constitute a \emph{Naimark
    dilation} of the POVM $M_a$.
\end{theorem}

  \begin{theorem}[Partial Naimark]\label{thm:partial-naimark}
    Suppose $\{M_{a_1, a_2}\}$ is a POVM acting on a tensor product
  Hilbert space $\calH = \calH_1 \ot \calH_2$ of the form $M_{a_1,
    a_2} = \Pi_{a_1} \ot A^{a_1}_{a_2}$, where the operators $\{\Pi_{a_1}\}_{a_1}$ is a
  projective measurement. Then there is a Naimark dilation $M'_{a_1,
    a_2}$ and a state $\ket{\rmaux}$ as above, with the property that
  $M'_{a_1, a_2} = \Pi_{a_1} \ot A'^{a_1}_{a_2}$ for projectors
  $A'^{a_1}_{a_2}$ acting on $\calH_2 \ot \calH_{\rmaux}$.
\end{theorem}
\begin{proof}
  The condition that $\{\Pi_{a_1}\}$ forms a projective measurement
  implies that for each $a_1$, $\{A^{a_1}_{a_2}\}$ is a
  POVM. By~\Cref{thm:naimark}, for each $a_1$ there exists a POVM
  $A'^{a_1}_{a_2}$ dilating $A^{a_1}_{a_2}$, and all of these POVMs
  act on the same universal auxiliary state $\ket{\rmaux}$. Now,
  define $M'_{a_1, a_1} = \Pi_{a_1} \ot A'^{a_1}_{a_2}$. For every
  state $\ket{\psi}$, we have that
  \begin{align*}
    \tr_{\reg{aux}}(M'_{a_1, a_2} (\ket{\psi}\bra{\psi}) M'_{a_1,
    a_2}) &= \tr_{\reg{aux}}( A'^{a_1}_{a_2} (\Pi_{a_1} \ket{\psi}
            \bra{\psi} \Pi_{a_1} \ot \ket{\rmaux} \bra{\rmaux}) A'^{a_1}_{a_2}) \\
          &= \sqrt{A^{a_1}_{a_2}} (\Pi_{a_1} \ket{\psi} \bra{\psi}
            \Pi_{a_1}) \sqrt{A^{a_1}_{a_2}} \\
          &= \sqrt{M_{a_1, a_2}} \ket{\psi} \bra{\psi} \sqrt{M_{a_1,
            a_2}},
  \end{align*}
  where in going from the first to the second line we used~\Cref{thm:naimark}.
\end{proof}

For the purposes of this paper, we will need to specialize the POVM
notation introduced above in several ways. First, we will often work
with families of POVM measurements indexed by questions in an
interactive proof protocol. These will be denoted $\{M^q_a\}$, where
$q$ indexes the question and $a$ the outcome.
(We note that the reverse convention ``$\{M_q^a\}$", which we will \emph{not} use, is also common in the literature.)
In many cases, the outcomes will consist of
tuples of elements, some of which we may wish to discard. We use the
convention that if an outcome element is not written, it
is understood to be \emph{summed} over. Thus, if $\{M^{x}_{a,b}\}$
is a family of POVMs, we would have
\begin{equation}\label{eq:spider-man}
M^{x}_{a} := \sum_b M^{x}_{a,b}.
\end{equation}

\begin{notation}\label{not:marginalize}
We will also often consider situations where some of the information
in a measurement outcome is discarded. In particular, given a
POVM $\{M_f\}$ whose outcomes are functions $f:U \to V$ over some
domain $U$, and given a point $x \in U$, we will denote by $\{M_{f(x)=
  y}\}_{y \in V}$ the measurement corresponding to applying $M$ to
obtain a function $f$, and returning the value of $f$ at
$x$. Formally, the POVM elements of this measurement are given by
\[ M_{[f(x) = a]} = \sum_{f: f(x) = a} M_f. \]
(We note that \Cref{eq:spider-man} can be viewed as a special case
in which the ``discarding function"~$f$ simply removes the second coordinate.
For this case, it is simpler to use the convenient notation in \Cref{eq:spider-man}
than the more cumbersome bracket notation given here.)
\end{notation}

The following lemma contains a useful fact about marginalized
projective measurements.
\begin{lemma}\label{lem:marginal-tensor-product}
  Let $M_{a,b}$ be a projective measurement on a tensor product
  Hilbert space $\calH_1 \ot \calH_2$, and suppose that for all $a$, $M_{a} = A_a
  \ot B_a$ where $A_a$ is a rank-one matrix on $\calH_1$. Then for all
  $a, b$, $M_{a,b} = A_a \ot C_{a,b}$ with $C_{a,b}$ projectors.
\end{lemma}
\begin{proof}
  By the Schmidt decomposition, we can write $M_{a,b}$ as
  \begin{align*}
    M_{a,b} &= \sum_j \ket{\psi_{a,b,j}}\bra{\psi_{a,b,j}} \\
    &= \sum_{j,k}  \sigma_{jk}              \ket{u_{a,b,j,k}}\bra{u_{a,b,j,k}}
      \ot\ket{v_{a,b,j,k}}\bra{v_{a,b,jk}}.
  \end{align*}
  Write the rank-one matrix $A_a$ as an outer product
  $\ket{\psi_a}\bra{\psi_a}$. Then we have
  \[ \sum_b M_{a,b} = \sum_{j,k, b} \sigma_{jk}              \ket{u_{a,b,j,k}}\bra{u_{a,b,j,k}}
      \ot\ket{v_{a,b,j,k}}\bra{v_{a,b,jk}} = 
      \ket{\psi_a}\bra{\psi_a} \ot B_a. \]
    Taking the partial trace on the $B$ system, we have
    \[  \tr_B(\sum_b M_{a,b}) = \sum_{j,k,b} \sigma_{jk}
      \ket{u_{a,b,j,k}}\bra{u_{a,b,j,k}} =
      \ket{\psi_a}\bra{\psi_a}. \]
    Suppose we multiply on the left by $\bra{v}$ and on the right by
    $\ket{v}$, for $\ket{v}$ orthogonal to $\ket{\psi_a}$. Then the
    RHS is 0 while the LHS is a sum $\sum_{j,k,b} \sigma_{jk}
    |\braket{v |u_{a,b,j,k}}|^2$ of nonnegative terms. Hence, each of
    these terms must be zero. Thus, the equation can only hold if all
    the vectors $\ket{u_{a,b,j,k}}$ are multiples of
    $\ket{\psi_a}$. This implies that
    \[ M_{a,b} = \ket{\psi_a} \bra{\psi_a} \ot C_{a,b}\]
    for some $C_{a,b}$ as desired.
\end{proof}

\subsection{Nonlocal games and $\MIP^*$}

Now, we augment \Cref{def:two-player-one-round,def:mip-protocols}
to allow for provers to share quantum resources.

\begin{definition}
Given a game~$\game$, a \emph{quantum strategy}
is one in which Alice and Bob are allowed to share entanglement but not to communicate.
We can model their behavior with the \emph{strategy} $\calS= (\rho, A, B)$.
Here,
\begin{itemize}
\itemsep -.5pt
\item[$\circ$] Write $\mathcal{H}_A$ for Alice's local Hilbert space and $\mathcal{H}_B$ for Bob's.
		Then $\rho$ is a (possibly entangled) state in $\mathcal{L}(\mathcal{H}_A \otimes \mathcal{H}_B)$.
\item[$\circ$] The set $A$ contains a matrix $A_a^x$ for each question~$x$ and answer~$a$,
		with the guarantee that for each question~$x$, $A^x := \{A_a^x\}_a$ is a POVM. (Likewise for~$B$.)
\end{itemize} 
Alice and Bob perform their strategy as follows:
given question~$x$, Alice performs the POVM $\{A_a^x\}_a$ and returns her measurement outcome to the verifier.
Bob plays similarly.
The \emph{value} of their strategy, denoted $\valstrat{\game}{\calS}$, is the probability that they pass the test,
over the randomness in~$\game$ and in their measurement outcomes.
\begin{align*}
\valstrat{\game}{\calS}
&= \E_{(\bx_0, \bx_1) \sim \mathrm{Alg}_{\mathrm{Q}}} \Pr_{\ba_0, \ba_1}[\mathrm{Alg}_{\mathrm{A}}(\bx_0, \bx_1, \ba_0, \ba_1) = 1]\\
&= \E_{(\bx_0, \bx_1) \sim \mathrm{Alg}_{\mathrm{Q}}} \sum_{\substack{a_0, a_1,\\\mathrm{Alg}_{\mathrm{A}}(\bx_0, \bx_1, a_0, a_1) = 1}}
	\tr(A_{a_0}^{\bx_0} \otimes A_{a_1}^{\bx_1} \cdot \rho),
\end{align*}
where in the first line, $(\ba_0, \ba_1)$ is the distribution on answers given questions~$\bx_0, \bx_1$.
We write $\val{\game}$ for the infimum of $\valstrat{\game}{\calS}$ over all strategies~$\calS$.
We define value analogously for interactive proofs.

We say that $L \in \MIP_{c, s}^*$ if there is an quantum interactive proof~$\game$ that decides it.
This means that the following three conditions are true.
\begin{itemize}
\item[$\circ$] (Completeness) Suppose $\mathsf{input} \in L$.  Then there is a quantum strategy for~$\game$ with value at least~$c$.
\item[$\circ$] (Soundness) Suppose $\mathsf{input} \notin L$.  Then every quantum strategy for~$\game$ has value at most~$s$.
\item[$\circ$] All of $\qlength{\game}$, $\alength{\game}$, $\qtime{\game}$, and $\atime{\game}$ are $\poly(n)$.
\end{itemize}
If $c - s$ is a constant, then we will suppress the dependence on them and just say that $L \in \MIP^*$.
\end{definition}

\begin{remark}
A game~$\game$ is \emph{symmetric} if its distribution on questions treats Alice and Bob symmetrically.
In this case, we may assume without loss of generality that Alice and Bob's strategies are also symmetric,
i.e.\ that $A_a^x = B_a^x$ for all questions~$x$ and answers~$a$.
This allows us to represent their measurements by a single set of matrices $M$ (for which $M_a^x = A_a^x = B_a^x$).
As a further simplification, by applying Naimark's dilation theorem to Alice and Bob's strategy 
we can assume that their shared state~$\psi$ is pure and their measurements are projectors.
\end{remark}

Occasionally, it will be useful to speak of the distribution over
measurement outcomes induced by a strategy independently of any
particular game. For this, we introduce the notion of a bipartite correlation
\begin{definition}
  Given a strategy $\calS = (\rho, A, B)$, the bipartite correlation produced by it is
  the function $P(a, b | x, y) = \tr(A^x_a \ot B^x_y \cdot \rho)$.
\end{definition}
If two strategies produce the same bipartite correlation, they have
the same value for any game they are used for. Naimark's theorem
(\Cref{thm:naimark}) implies for any strategy $\calS$, there exists a
strategy $\calS'$ using only projective measurements that produces the
same correlation:

\begin{corollary}[Naimark's theorem for strategies]\label{cor:bipartite-naimark}
  Suppose $\{M_a\}$ and $\{N_b\}$ are two POVMs acting on the $A$ and
  $B$ factors of a tensor Hilbert space $\calH_{A} \ot \calH_{B}$,
  respectively. Then for any Naimark dilation of $\{M_a\}$ given by
  projectors $M'_a$ and an auxiliary state $\ket{\rmaux_A} \in
  \calH_{\rmaux_A}$, and any Naimark dilation of $\{N_b\}$ given by
  projectors $\N'_b$ and an auxiliary state
  $\ket{\rmaux_B} \in \calH_{\rmaux_B}$, it holds that for any bipartite state $\ket{\psi} \in
  \calH_{A} \ot \calH_{b}$, the post-measurement state after applying
  $M \ot N$ to $\ket{\psi}$ and $M' \ot N'$ to $\ket{\psi} \ot
  \ket{\rmaux_A} \ot \ket{\rmaux_B}$ is the same:
  \[ \sqrt{M_a} \ot \sqrt{N_b}  \ket{\psi} \bra{\psi}  \sqrt{M_a} \ot
    \sqrt{N_b} = \tr_{\rmaux}[ (M'_a \ot N'_b) (\ket{\psi}\bra{\psi}
    \ot \ket{\rmaux_{A}}\bra{\rmaux_{A}} \ot
    \ket{\rmaux_B}\bra{\rmaux_B} )(M'_a \ot N'_b)]. \]
  Moreover, such dilations exist by \Cref{thm:naimark}.
  As a consequence, $M, N$ and $M', N'$ induce the same joint
  distribution over outcome probabilities:
  \[ \bra{\psi} M_a \ot N_b \ket{\psi} = (\bra{\psi} \ot
    \ket{\rmaux_A} \ot \ket{\rmaux_B}) M'_a \ot N'_b (\ket{\psi} \ot
    \ket{\rmaux_A} \ot \ket{\rmaux_B}). \]
\end{corollary}
\begin{proof}
  The existence of such dilations follows immediately
  from~\Cref{thm:naimark}. To deduce the equality of post-measurement
  states, we apply~\Cref{eq:naimark} twice, and use the fact that the partial
  trace composes, i.e. that $\tr_{\rmaux}[\cdot] = \tr_{\rmaux_B}[\tr_{\rmaux_A}[\cdot]]$.
  \begin{align*}
    &\tr_{\rmaux}[ (M'_a \ot N'_b) (\ket{\psi}\bra{\psi}
    \ot \ket{\rmaux_{A}}\bra{\rmaux_{A}} \ot
      \ket{\rmaux_B}\bra{\rmaux_B}) (M'_a \ot N'_b)]\\
    &\qquad= \tr_{\rmaux_B}[
                                                    \sqrt{M_a \ot I}
                                                    ((I \ot N'_b) \ket{\psi}
                                                    \ket{\rmaux_A}
                                                    \bra{\rmaux_A}
                                                    \bra{\psi} (I \ot
                                                    N'_b)) \sqrt{M_a
                                                    \ot I}] \\
    &\qquad= \sqrt{I \ot N_b} \sqrt{M_a \ot I} \ket{\psi} \bra{\psi}
      \sqrt{M_a \ot I} \sqrt{I \ot N_b}.\qedhere 
  \end{align*}
\end{proof}

\subsection{Pauli matrices and the EPR state}
\label{sec:paulis}
Over a finite field $\F_q$ with order $q = p^t$ for prime $p$, the single-qudit Pauli matrices are a set
of unitary matrices acting on $\mathbb{C}^q$. Every Pauli matrix can
be uniquely written as a product $\omega^{a} X(x) Z(z)$, where $\omega$ is the $p$-th root
$\omega = e^{2\pi/p}$, and $X(x)$ and $Z(z)$ are the matrices
\begin{equation}
  X(x) = \sum_{j \in \F_q} \ket{j+x}\bra{j}, \qquad Z(z) =
  \sum_{j \in \F_q} \omega^{\tr[z j]}\ket{j}
  \bra{j}, \label{eq:pauli-definition} \end{equation}
where the arguments $x, z$ are in $\F_q$, $\tr: \F_q \to \F_p$ is the finite field
trace. The set of all Pauli matrices form a group, known as the Pauli
group or the Weyl-Heisenberg group. For the most part,
in this paper, it will suffice to consider only the group elements of
the form
$\omega^a X(x)$ (``$X$-type'' Paulis) and $Z(z)$ (``$Z$-type''
Paulis). Elements of the form $\omega^a X(x)Z(z)$ for $x, z \neq 0$
are sometimes called ``$Y$-type''.

The eigenvalues of $X(x)$ and $Z(z)$ for
all $x, z$ are powers
of $\omega$, as can be seen from the facts $X(x)^p = Z(z)^p =
I$. Any unitary with this property is known as a \emph{(generalized)
  observable}. Every generalized observable $U$ induces a projective
measurement with $p$ outcomes, corresponding to the $p$ possible eigenvalues of
$U$. As a convenient shorthand, we will refer to performing this projective
measurement as ``measuring $U$.''  In the case of of the $X$ and $Z$
operators, the
eigenvectors $\ket{\tau^X_u}$ and $\ket{\tau^Z_u}$ of $X(1)$ and $Z(1)$ are indexed by elements $u$ of $\F_q$,
with eigenvalue $\tr[u]$; thus, each eigenvalue occurs with
multiplicity $q/p$. Explicitly, they are given by
\begin{equation}\label{eq:pauli-eigenstates}
  \ket{\tau^X_u} = \frac{1}{\sqrt{q}} \sum_{v \in \F_q}
  \omega^{-\tr[uv]} \ket{v},\qquad \ket{\tau^Z_u} = \ket{u}.
  \end{equation}
We denote the projectors onto these eigenvectors by $\tau^X_u$ and
$\tau^Z_u$, respectively. These eigenvectors are also the eigenvectors
of the remaining $X(x), Z(z)$ observables, as shown by the following fact.
\begin{fact}
  For $W \in \{X, Z\}$, the observables $W(v)$ are related to the
  projectors $\tau^W_u$ by
  \begin{align}
    W(v) &= \sum_{u} \omega^{\tr[u \cdot v]}
           \tau^W_u \label{eq:pauli-obs-to-proj} \\
    \tau^W_u &= \E_{\bv} \omega^{-\tr[u \cdot \bv]} W(\bv). \label{eq:pauli-proj-to-obs}
  \end{align}
\end{fact}
\begin{proof}
  We start with~\Cref{eq:pauli-obs-to-proj}. For $W = Z$, the relation
  follows immediately from the definitions. For $W = X$, by
  calculation we have:
  \begin{align*}
    \sum_{u} \omega^{\tr[u \cdot x]} \tau^X_u &= \frac{1}{q} \sum_{u,
                                                v, v'} \omega^{\tr[u
                                                \cdot x]}
                                                \omega^{\tr[u(v -
                                                v')]} \ket{v'}\bra{v}
    \\
    &= \sum_{v, v'} \E_{\bu} \omega^{\tr[\bu \cdot (x + v - v')]}\ket{v'}\bra{v}\\
                                              &= \sum_{v} \ket{v+x}
                                                \bra{v} \\
                                              &= X(x),
  \end{align*}
  where we have applied \Cref{fact:averages-to-zero} in passing from
  the second to the third line.
  Now we show \Cref{eq:pauli-proj-to-obs}:
  \begin{equation*}
    \E_{\bv} \omega^{-\tr[u\cdot \bv]} W(\bv) = \E_{\bv} \sum_{a}
                                         \omega^{-\tr[u \cdot \bv]}
                                         \omega^{\tr[\bv \cdot a]}
                                         \tau^W_{a}
                                         = \E_{\bv} \sum_{a}
                                         \omega^{\tr[(a - u) \cdot \bv]}
                                         \tau^W_{a}
                                       = \tau^W_{u},
  \end{equation*}
  where we first applied \Cref{eq:pauli-obs-to-proj} in the first
  equality, and then used \Cref{fact:averages-to-zero} to perform the
  expectation over $\bv$.
\end{proof}
The Pauli matrices obey the commutation
relation
\begin{equation}\label{eq:commutation-relations}
 X(x)Z(z) = \omega^{-\tr[xz]} Z(z)X(x).
 \end{equation}
This follows directly from \Cref{eq:pauli-definition}. It follows from
this that all of the Pauli matrices (including the $Y$-type matrices)
are generalized observables.

The maximally entangled state, or EPR state, over qudits of dimension
$q$ is the state
\[ \ket{\epr_q} = \frac{1}{\sqrt{q}} \sum_{u \in \F_q} \ket{u} \ot \ket{u}. \]
We will write $\ket{\epr_q^n}$ for $\ket{\epr_q}^{\otimes n}$.
This state obeys the stabilizer relations
\begin{align}
  X(x) \ot X(x) \ket{\epr_q} &= Z(z) \ot Z(-z) \ket{\epr_q} =
                             \ket{\epr_q} \label{eq:epr-stab} \\
  \tau^X_u \ot I \ket{\epr_q} &= I \ot \tau^X_{-u}
                              \ket{\epr_q} \label{eq:epr-stab-x} \\
  \tau^Z_u \ot I \ket{\epr_q} &= I \ot \tau^Z_{u} \ket{\epr_q} \label{eq:epr-stab-z}.
\end{align}
Relations~\eqref{eq:epr-stab-x} and~\eqref{eq:epr-stab-z} imply that
measuring $X(x)$ on both halves of an EPR state will yield two
outcomes $a,b$ satisfying $a = -b$, and measuring $Z(z)$ on both
halves will yield two outcomes $a,b$ that are equal.

In the important special case of finite fields with characteristic $2$
(i.e. $\F_q$ for even $q$), $u = -u$ for all $u \in \F_q$, and thus
measuring any of the $X$ and $Z$ operators on both sides of the state
will always yield the same outcome.

\subsection{State dependent distances}

In this section, we introduce two state-dependent distances.
To motivate them, we first define the consistency game, perhaps the simplest nontrivial two-player game.

\begin{definition}
The \emph{consistency game with question~$x$}, denoted $\game_{\mathrm{con}}(x)$ is defined as follows.
The question~$x$ is given to Alice and Bob, who respond with answers~$\ba$ and~$\ba'$, respectively.
The verifier accepts if~$\ba = \ba'$.
\end{definition}

We will typically play the consistency game when~$\bx$, rather than being a fixed question, is drawn from some distribution.
Our first state-dependent distance quantifies the players` success probability in this case.

\begin{definition}
Let $\{A^x_a\}$ and $\{B^x_a\}$ be sets of matrices in $\mathcal{L}(\mathcal{H}_A)$ and $\mathcal{L}(\mathcal{H}_B)$, respectively.
Let~$\calD$ be a distribution on questions~$x$ and $\ket{\psi}$ be a state in $\mathcal{H}_A \otimes \mathcal{H}_B$.
Consider the game in which the verifier selects $\bx \sim \calD$ and then plays $\game_{\mathrm{con}}(\bx)$.
We say that
\begin{equation*}
A^x_a \otimes I_{\reg{Bob}} \consistency_{\delta} I_{\reg{Alice}} \otimes B^x_a
\end{equation*}
\emph{on state $\ket{\psi}$ and distribution $\calD$} if
Alice and Bob win with probability~$1-O(\delta)$ using the measurements~$A$ and~$B$, respectively.
\end{definition}

We will sometimes leave the state or distribution unspecified, as they are often clear from context.
This distance has a clear operational interpretation.
Our second state-dependent distance,
defined next,
is more analytic.

\begin{definition}
Let $\{Q^x_a\}$ and $\{R^x_a\}$ be sets of matrices in $\mathcal{L}(\mathcal{H})$.
Let $\calD$ be a distribution on the variables~$x$ and $\ket{\psi}$ be a state in $\mathcal{H}$.
Then we say that \emph{$Q^x_a \approx_{\delta} R^x_a$ on state $\ket{\psi}$ and distribution $\calD$} if
\begin{equation*}
\E_{\bx \sim \calD} \sum_a \Vert (Q_{a}^{\bx} - R_{a}^{\bx} ) \ket{\psi}\Vert^2 = O(\delta).
\end{equation*}
\end{definition}

As above, we will sometimes leave the state or distribution unspecified when clear from context.
This is sometimes referred to as \emph{the} state-dependence distance,
whereas our first distance measure is often referred to as the ``consistency".
A typical setting of parameters is $\mathcal{H} = \mathcal{H}_A \otimes \mathcal{H}_B$,
$Q^x_a := A^x_a \otimes I_{\reg{Bob}}$,
and $R^x_a := I_{\reg{Alice}} \otimes B^x_a$.
In this case, we have the following relationship between the two state-dependent distances.

\begin{fact}\label{fact:agreement}
Let $\{A^x_a\}$ and $\{B^x_a\}$ be POVM measurements.
The following two facts hold.
\begin{enumerate}
\item If $A^x_a \otimes I_{\reg{Bob}} \consistency_{\delta} I_{\reg{Alice}} \otimes B^x_a$
	then $A^x_a \otimes I_{\reg{Bob}} \approx_{\delta} I_{\reg{Alice}} \otimes B^x_a$.
\item If $A^x_a \otimes I_{\reg{Bob}} \approx_{\delta} I_{\reg{Alice}} \otimes B^x_a$ \emph{and} $\{A^x_a\}$ and $\{B^x_a\}$ are projective measurements,
	then $A^x_a \otimes I_{\reg{Bob}} \consistency_{\delta} I_{\reg{Alice}} \otimes B^x_a$.\label{item:partial-converse}
\end{enumerate}
\end{fact}
\begin{proof}
Suppose that $A^x_a \otimes I_{\reg{Bob}} \consistency_{\delta} I_{\reg{Alice}} \otimes B^x_a$.
This is equivalent to the statement
\begin{equation}\label{eq:i-derived-a-formula}
\E_{\bx} \sum_a \bra{\psi} A^{\bx}_a \otimes B^{\bx}_a \ket{\psi} \geq 1- O(\delta).
\end{equation}
As a result, using the fact that~$A$ and~$B$ are POVMs,
\begin{align}
\E_{\bx} \sum_a \Vert (A_{a}^{\bx}\otimes I - I \otimes B_{a}^{\bx} ) \ket{\psi}\Vert^2 
& = \E_{\bx} \sum_a \bra{\psi} ((A_{a}^{\bx})^2\otimes I + I \otimes (B_{a}^{\bx})^2 - 2 A_a^{\bx} \otimes B_{a}^{\bx}) \ket{\psi}\nonumber\\
& \leq \E_{\bx} \sum_a \bra{\psi} (A_{a}^{\bx}\otimes I + I \otimes B_{a}^{\bx} - 2 A_a^{\bx} \otimes B_{a}^{\bx}) \ket{\psi}\label{eq:sometimes-tight}\\
& = 2 - 2 \E_{\bx} \sum_a \bra{\psi} A_a^{\bx} \otimes B_{a}^{\bx} \ket{\psi}\nonumber.
\end{align}
By \Cref{eq:i-derived-a-formula}, this is $O(\delta)$.
As a result, $A^x_a \otimes I_{\reg{Bob}} \approx_{\delta} I_{\reg{Alice}} \otimes B^x_a$.
The reverse statement holds when~$A$ and~$B$ are projective measurements
because \Cref{eq:sometimes-tight} is an equality in this case.
\end{proof}

The following fact shows that we can derive a weaker converse in the case when only one of~$A$ or~$B$ is projective.

\begin{fact}\label{fact:almost-agreement}
Suppose $\{A^x_a\}$ and $\{B^x_a\}$ are two measurements such that $A^x_a \otimes I_{\reg{Bob}} \approx_{\delta} I_{\reg{Alice}} \otimes B^x_a$.
Suppose further that either~$A$ or~$B$ is a projective measurement (and the other is a POVM measurement).
Then $A^x_a \otimes I_{\reg{Bob}} \consistency_{\delta^{1/2}} I_{\reg{Alice}} \otimes B^x_a$.
\end{fact}
\begin{proof}
Our goal is to upper bound the expression
\begin{equation}\label{eq:anand-is-a-stable-genius-yet-again}
1 - \E_\bx \sum_a \bra{\psi} A^{\bx}_a \otimes B^{\bx}_a \ket{\psi}.
\end{equation}
We begin by rewriting the number~$1$.
Here we use the fact that because~$A$ is projective, $(A^x_a)^2 = A^x_a$,
and so $\sum_a (A^x_a)^2 = I$.
This gives us:
\begin{align}
\eqref{eq:anand-is-a-stable-genius-yet-again}
& = \E_\bx \sum_a \bra{\psi} (A^{\bx}_a)^2 \otimes I_{\reg{Bob}} \ket{\psi} - \E_\bx \sum_a \bra{\psi} A^{\bx}_a \otimes B^{\bx}_a \ket{\psi}\nonumber\\
& = \E_\bx \sum_a \bra{\psi} (A^{\bx}_a \otimes I_{\reg{Bob}}) \cdot (A^{\bx}_a \otimes I_{\reg{Bob}} - I_{\reg{Alice}} \otimes B^{\bx}_a) \ket{\psi}\nonumber\\
& \leq \E_\bx \sum_a \Vert A^{\bx}_a \otimes I_{\reg{Bob}} \ket{\psi}\Vert \cdot \Vert(A^{\bx}_a \otimes I_{\reg{Bob}} - I_{\reg{Alice}} \otimes B^{\bx}_a) \ket{\psi}\Vert\label{eq:anand-continues-to-be-a-stable-genius}
\end{align}
Now we apply Cauchy-Schwarz and then Jensen's inequality:
\begin{equation*}
\eqref{eq:anand-continues-to-be-a-stable-genius}
\leq  \sqrt{\E_\bx \sum_a \Vert A^{\bx}_a \otimes I_{\reg{Bob}} \ket{\psi}\Vert^2
	\cdot \sum_a \Vert(A^{\bx}_a \otimes I_{\reg{Bob}} - I_{\reg{Alice}} \otimes B^{\bx}_a) \ket{\psi}\Vert^2}
\end{equation*}
The first of these terms we bound by~$1$, and the second is~$O(\delta)$ by assumption.
\end{proof}

\begin{remark}\label{remark:fraction-of-identity}
We note that the requirement in \Cref{fact:almost-agreement}
that one of the two measurements be projective is necessary.
Consider the measurements $\{A^x_a\}$ and $\{B^x_a\}$ with~$m$ separate outcomes~$a$
in which $A^x_a = B^x_a = I/m$ for all~$x$ and~$a$. Then $A^x_a \otimes I_{\reg{Bob}} \approx_{0} I_{\reg{Alice}} \otimes B^x_a$, but as $m \rightarrow \infty$,
\begin{equation*}
1- \E_\bx \sum_a \bra{\psi} A^{\bx}_a \otimes B^{\bx}_a \ket{\psi} \rightarrow 1.
\end{equation*}
\end{remark}

Thus, when one of the measurements is projective,
the ``$\approx_\delta$" distance is roughly equivalent to the ``$\consistency_\delta$" distance,
up to a polynomial factor (which we can tolerate losing in our proofs).
More generally, however, the ``$\approx_\delta$" distance
can be viewed as a weakening of the ``$\consistency_\delta$" distance.
In spite of this, we will spend much of the paper dealing with the ``$\approx_\delta$" distance,
as it is easier to manipulate but still strong enough to reach our desired consequences.
(See \cite[Section~$2.3.1$]{Vid11} for a further defense of this distance.)
We note that even when~$\ket{\psi}$ is a bipartite state,
the ``$\approx_\delta$" distance is defined for matrices~$Q$ and~$R$
which are not necessarily tensor products over the bipartition.
Such matrices will often be useful to pass through during intermediate steps of our proofs.

A common use case for these distances is when the verifier (i) samples a pair of questions~$\bx = (\bx_0, \bx_1)$,
(ii) hands~$\bx_0$ to Alice and~$\bx_1$ second to Bob,
(iii) receives their answers $\ba_0$ and~$\ba_1$
and (iv) accepts if $f(\bx, \ba_0) = g(\bx, \ba_1)$ for some functions~$f$ and~$g$.
Write $\{A^{x_0}_{a_0}\}_{a_0}$ and $\{B^{x_1}_{a_1}\}_{a_1}$ for Alice and Bob's measurements, respectively.
We can view these as measurements which receive the pair~$(\bx_0, \bx_1)$ and simply ignore one coordinate.
Suppose the verifier accepts with probability $1-\delta$.
Using \Cref{not:marginalize}, we can view this as performing the consistency game
between the measurements $A^{x_0}_{[f(x, a_0) = b]}$ and  $B^{x_1}_{[f(x, a_1) = b]}$.
Hence, we can derive the following two facts:
\begin{equation*}
A^{x_0}_{[f(x, a_0) = b]} \otimes I_{\reg{Bob}}
\consistency_{\delta} I_{\reg{Alice}} \otimes B^{x_1}_{[f(x, a_1) = b]},
\quad
\text{and}
\quad
A^{x_0}_{[f(x, a_0) = b]} \otimes I_{\reg{Bob}}
\approx_{\delta} I_{\reg{Alice}} \otimes B^{x_1}_{[f(x, a_1) = b]}.
\end{equation*}
This generic format will be the most common use of these notations.

\Cref{remark:fraction-of-identity} highlights the importance of \emph{projective} measurements when dealing with the~``$\approx_\delta$" distance.
This, and several other key facts about the ``$\approx_\delta$" distance, are true only for projective measurements.
As a result, we will sometimes apply Naimark's theorem (\Cref{thm:naimark})
during our proofs to ``round" POVM measurements into projective measurements.
However, there is a subtlety in doing so, namely that 
because Naimark's theorem preserves measurement outcomes,
any ``$\consistency_\delta$" statements we have derived about our measurement operators will remain true,
but Naimark's theorem is \emph{not} guaranteed to preserve ``$\approx_\delta$" statements.
In this work, we will be able to dispense with this subtlety and assume all ``$\approx_\delta$" statements \emph{are} preserved,
because \emph{all ``$\approx_\delta$" statements in our proofs will be derived from ``$\consistency_\delta$" statements},
and so they will remain true after performing Naimark's theorem, since we could simply rederive them.

\subsection{Miscellaneous properties of the state-dependent distances}

In this section, we record some facts about the ``$\approx_\delta$" notation which we will use repeatedly throughout the paper.
A good rule of thumb is that everything one expects to be true about the ``$\approx_\delta$" notation actually \emph{is}  true,
except for those things which are not. 
As a result, we will be overly pedantic in this section in order to call attention to these cases.

\subsubsection{Simple state-dependent distance facts}

\begin{fact}\label{fact:obvious-vector-fact}
For two vectors $\ket{\psi_1}, \ket{\psi_2}$, 
$\displaystyle
\Vert \ket{\psi_1} + \ket{\psi_2} \Vert^2 \leq 2 \Vert \ket{\psi_1} \Vert^2 + 2 \Vert \ket{\psi_2} \Vert^2.
$
\end{fact}

\begin{fact}\label{fact:measurement-sub-measurement-switcheroo}
Let $\{A_a\}$ be a measurement. Then
\begin{equation*}
\sum_a \Vert A_{a} \ket{\psi}\Vert^2 \leq\Vert \ket{\psi} \Vert^2.
\end{equation*}
\end{fact}
\begin{proof}
If $\{A_a\}$ is a measurement, then
\begin{equation*}
\sum_a \Vert A_{a} \ket{\psi}\Vert^2
= \sum_a \tr( A_a A_{a} \ket{\psi}\bra{\psi})
\leq \tr(I \cdot \ket{\psi}\bra{\psi}) = \Vert \ket{\psi}\Vert^2.\qedhere
\end{equation*}
\end{proof}

\begin{fact}\label{fact:easy-bound}
Let $\{A_a\}$, $\{B_a\}$ be measurements. Then for any state~$\ket{\psi}$,
\begin{equation*}
\sum_a \Vert (A_{a} - B_{a} ) \ket{\psi}\Vert^2 \leq 4 \cdot \Vert \ket{\psi} \Vert^2.
\end{equation*}
\end{fact}
\begin{proof}
By \Cref{fact:obvious-vector-fact},
\begin{equation*}
\sum_a \Vert (A_{a} - B_{a} ) \ket{\psi}\Vert^2 \leq 2 \sum_a \Vert A_{a} \ket{\psi}\Vert^2 + 2 \sum_a \Vert  B_{a} \ket{\psi}\Vert^2.
\end{equation*}
The fact now follows from \Cref{fact:measurement-sub-measurement-switcheroo}.
\end{proof}

\begin{fact}\label{fact:trivial-upper-bound-approx-delta}
Let $\{A^x_a\}$ and $\{B^x_a\}$ be POVM measurements.
Then $A^x_a \approx_1 B^x_a$.
\end{fact}

\begin{fact}\label{fact:add-a-proj}
Let $\{A_a^{x}\}$ and $\{B_a^{x}\}$ be matrices.
Let $\{C_b^{y}\}$ be matrices such that $\sum_b (C_b^y)^\dagger C_b^y \leq I$ for all~$y$.
(This includes the case when $\{C_b^y\}$ form projective or POVM measurements.) Then
\begin{equation*}
\text{$A_a^{x} \approx_{\delta} B_a^x$ implies $C_b^y A_a^x \approx_{\delta} C_b^y B_a^x$.}
\end{equation*}
\end{fact}
\begin{proof}
Fix questions~$x, y$ and answers~$a$. Because of our property on $\{C_b^y\}_b$,
\begin{align*}
\sum_{b} \Vert (C_b^{y} A_a^{x} - C_b^{y} B_a^{x})\ket{\psi}\Vert^2
&= \sum_{b} \bra{\psi} (A_a^x - B_a^x)^\dagger (C_b^y)^\dagger (C_b^y) (A_a^{x} -  B_a^{x})\ket{\psi}\\
&\leq \bra{\psi} (A_a^x - B_a^x)^\dagger (A_a^x - B_a^x) \ket{\psi}
 = \Vert (A_a^{x} -B_a^{x})\ket{\psi}\Vert^2
\end{align*}
We can therefore derive our desired conclusion:
\begin{equation*}
\E_{\bx, \by} \sum_{a, b} \Vert (C_b^{\by} A_a^{\bx} - C_b^{\by} B_a^{\bx}) \ket{\psi} \Vert^2
\leq \E_{\bx, \by} \sum_{a} \Vert (A_a^{\bx} - B_a^{\bx}) \ket{\psi} \Vert^2
= \delta.\qedhere
\end{equation*}
\end{proof}

\begin{fact}\label{fact:swap-dists}
Let $\calD, \calD'$ be two distributions such that $d_{\mathrm{TV}}(\calD, \calD') \leq \epsilon$.
Let $\{A_a^x\}$ and $\{B_a^x\}$ be measurements,
and suppose $A_a^x \approx_\delta B_a^x$ with respect to~$\calD$. Then $A_a^x \approx_{\delta + \epsilon} B_a^x$ with respect to~$\calD'$.
\end{fact}
\begin{proof}
By the definition of total variation distance,
for any set of numbers $\{\nu_x\}$ satisfying $0 \leq \nu_x \leq c$,
the expectations under the two distributions are similar:
\begin{equation*}
|\E_{\bx \sim \calD} [\nu_\bx] - \E_{\bx \sim \calD'}[\nu_{\bx}]| \leq c \cdot \epsilon.
\end{equation*}
We will take for our numbers $\nu_x = \sum_a \Vert (A_{a}^{x} - B_{a}^{x} ) \ket{\psi}\Vert^2$,
which is always less than~$4$ by \Cref{fact:easy-bound}.
As a result,
\begin{equation*}
\E_{\bx \sim \calD'} \nu_{\bx} \leq \E_{\bx \sim \calD} \nu_{\bx} + 4 \eps \leq \delta + 4 \eps.
\end{equation*}
This is $O(\delta + \epsilon)$, which proves the fact.
\end{proof}

\begin{fact}\label{fact:the-ol-state-y-swaperoonie}
Suppose $A_a^x \approx_{\delta} B_a^x$ on state $\ket{\psi}$, and suppose $\Vert \ket{\psi} - \ket{\overline{\psi}} \Vert^2 \leq \epsilon$.
Then $A_a^x \approx_{\delta+\epsilon} B_a^x$ on state $\ket{\overline{\psi}}$.
\end{fact}
\begin{proof}
Applying \Cref{fact:obvious-vector-fact} to $(A_{a}^{x} - B_{a}^{x} ) \ket{\psi}$ and $(A_{a}^{x} - B_{a}^{x} ) (\ket{\overline{\psi}} - \ket{\psi})$,
\begin{equation*}
\E_{\bx \sim \calD} \sum_a \Vert (A_{a}^{\bx} - B_{a}^{\bx} ) \ket{\overline{\psi}}\Vert^2
\leq \E_{\bx \sim \calD} \sum_a \Vert (A_{a}^{\bx} - B_{a}^{\bx} ) \ket{\psi}\Vert^2 + \E_{\bx \sim \calD} \sum_a \Vert (A_{a}^{\bx} - B_{a}^{\bx} ) (\ket{\overline{\psi}}-\ket{\psi})\Vert^2.
\end{equation*}
The first of these is bounded by $O(\delta)$ by the assumption, and the second of these is bounded by $4\epsilon$ by the assumption and \Cref{fact:easy-bound}.
\end{proof}

\begin{fact}\label{fact:close-on-marginal}
Suppose $A_a^{x_1} \approx_{\delta} B_a^{x_1}$ with respect to a distribution $\calD_{\mathrm{margin}}$ on $x_1$.
Let $\calD$ be a distribution on $(x_1, x_2)$ such that the marginal distribution on~$x_1$ is $\calD_{\mathrm{margin}}$.
Then $A_a^{x_1} \approx_{\delta} B_a^{x_1}$ with respect to~$\calD$.
\end{fact}
\begin{proof}
This is a simple calculation involving \Cref{not:marginalize}.
\end{proof}

\begin{fact}\label{fact:average-over-dists}
Let $k$ be a constant, and consider distributions over questions $\calD_1, \ldots, \calD_k$.
Let $\calD$ be a mixture of these distributions,
meaning that there is a probability distribution $p = (p_1, \ldots, p_k)$
such that a draw from $\calD$ can be simulated as follows: draw $\bi \sim p$ and output $\bx$ sampled from $\calD_{\bi}$.
Suppose $A^x_a \approx_{\delta} B^{x}_a$ with respect to $\calD_i$, for all $i \in [k]$.
Then $A^x_a \approx_{\delta} B^{x}_a$ with respect to $\calD$.
\end{fact}
\begin{proof}
By definition, for each $i \in [k]$ there is some constant $C_i$ such that 
\begin{equation*}
\E_{\bx \sim \calD_i} \sum_{a} \Vert (A^{\bx}_a - B^{\bx}_b) \ket{\psi} \Vert^2 \leq C_k \cdot \delta.
\end{equation*}
Then we can bound the mixture with
\begin{equation*}
\E_{\bx \sim \calD} \sum_{a} \Vert (A^{\bx}_a - B^{\bx}_b) \ket{\psi} \Vert^2
= \E_{\bi \sim p} \E_{\bx \sim \calD_{\bi}} \sum_{a} \Vert (A^{\bx}_a - B^{\bx}_b) \ket{\psi} \Vert^2
\leq \E_{\bi \sim p} C_{\bi} \cdot \delta \leq \max_{i}\{C_i\} \cdot \delta = O(\delta).\qedhere
\end{equation*}
\end{proof}

\begin{fact}\label{fact:sub-zero}
Suppose $\{A^x_a\}$ is a projective measurement and $\{B^x_a\}$ is a set of matrices such that
each~$B^x_a$ is positive semidefinite
and $\sum_{a}B^x_a \preceq I$.
Define $C^x_a$ such that for each~$x$, there exists an~$a$ such that $C^x_a := B^x_a + (I - \sum_{a'} B^x_{a'})$
and for all other~$a' \neq a$, $C^x_{a'} := B^x_{a'}$.
Thus, $C^x$ is a POVM for each~$x$.
If $A^x_a \approx_{\eps} B^x_a$ then $A^x_a \approx_{\eps^{1/2}} C^x_a$.
\end{fact}
\begin{proof}
By \Cref{fact:obvious-vector-fact},
\begin{equation*}
\E_{\bx} \sum_a \Vert (A^{\bx}_a - C^{\bx}_a) \ket{\psi} \Vert^2
\leq 2\E_{\bx} \sum_a \Vert (A^{\bx}_a - B^{\bx}_a) \ket{\psi} \Vert^2 + 2 \E_{\bx}\Vert (I - \sum_a B^{\bx}_a) \ket{\psi}\Vert^2.
\end{equation*}
The first of these terms we can bound by $O(\eps)$.  As for the second, 
\begin{multline*}
\E_{\bx}\Vert (I - \sum_a B^{\bx}_a) \ket{\psi}\Vert^2
= \E_{\bx} \bra{\psi} (I-\sum_a B^{\bx}_a)^2 \ket{\psi}
\leq \E_{\bx} \bra{\psi} (I-\sum_a B^{\bx}_a) \ket{\psi}\\
= 1- \E_{\bx} \sum_a\bra{\psi}  B^{\bx}_a \ket{\psi}
\leq1- \E_{\bx} \sum_a\bra{\psi}  (B^{\bx}_a)^2 \ket{\psi}.
\end{multline*}
Now, we write $1 = \E_{\bx} \sum_a \bra{\psi} (A^{\bx}_a)^2 \ket{\psi}$,
which holds because~$A$ is a projective measurement.
We bound the result as follows.
\begin{multline*}
\E_{\bx} \sum_a \bra{\psi} ((A^{\bx}_a)^2  - (B^{\bx}_a)^2)\ket{\psi}
= \Re\left(\E_{\bx} \sum_a \bra{\psi} (A^{\bx}_a  + B^{\bx}_a)(A^{\bx}_a - B^{\bx}_a)\ket{\psi}\right)\\
\leq \E_{\bx} \sqrt{\sum_a \Vert (A^{\bx}_a + B^{\bx}_a) \ket{\psi}\Vert^2}
	\cdot \sqrt{\sum_a \Vert (A^{\bx}_a - B^{\bx}_a)\ket{\psi} \Vert^2}.
\end{multline*}
For each~$\bx$,
we can bound the first square root by $O(1)$
due to \Cref{fact:obvious-vector-fact} and \Cref{fact:measurement-sub-measurement-switcheroo}.
Having done so, we can move the expectation into the second square root by Jensen's inequality.
The result is $O(\eps^{1/2})$ by assumption. This proves the fact.
\end{proof}

\subsubsection{Data processing}

In this section, we show a simple data processing inequality for the ``$\consistency_\delta$" distance.
We also observe that one does \emph{not} hold for the ``$\approx_\delta$" distance.

\begin{fact}\label{fact:specialize-the-simeq}
Suppose that $A^x_a \otimes I_{\reg{Bob}} \consistency_\delta I_{\reg{Alice}} \otimes B^x_a$.
Then $A^x_{[f(a) = b]} \otimes I_{\reg{Bob}} \consistency_\delta I_{\reg{Alice}} \otimes B^x_{[f(a) = b]}$.
\end{fact}
\begin{proof}
Given question~$\bx$, if Alice and Bob return~$\ba$ and~$\ba'$ in which $\ba = \ba'$, then $f(\ba) = f(\ba')$.
As a result, applying~$f$ to their answers cannot decrease the probability they agree.
\end{proof}

\begin{remark}
We note that the same fact is \emph{not} true for the ``$\approx_\delta$" distance.
Consider answers of the form $a = (b, i)$, where $b \in \{0, 1\}$ and $i \in [m]$.
Suppose $A^x_{b, i} = I/(2m)$ for all~$a$, whereas $B^x_{0, i} = I/m$ and $B^x_{1, i} = 0$ for all~$i$.
Consider the function $f(b, i) = b$.
It can be checked that in this case, $A^x_a \otimes I_{\reg{Bob}} \approx_{1/2m} I_{\reg{Alice}} \otimes B^x_a$
but $A^x_{[f(a) = b]} \otimes I_{\reg{Bob}} \approx_{1/2} I_{\reg{Alice}} \otimes B^x_{[f(a) = b]}$.
\end{remark}

\subsubsection{Triangle inequalities}

In this section, we give two triangle inequalities. Our first is for the state-dependent distance.

\begin{fact}[Triangle inequality]\label{fact:triangle}
Suppose $A_a^x \approx_\delta B_a^x$ and $B_a^x \approx_\epsilon C_a^x$. Then $A_a^x \approx_{\delta + \epsilon} C_a^x$.
\end{fact}
\begin{proof}
Applying \Cref{fact:obvious-vector-fact} to $(A_{a}^{x} - B_{a}^{x} ) \ket{\psi}$ and $(B_{a}^{x} - C_{a}^{x} ) \ket{\psi}$,
\begin{align*}
\E_{\bx \sim \calD} \sum_a \Vert (A_{a}^{\bx} - C_{a}^{\bx} ) \ket{\psi}\Vert^2
&\leq 2 \E_{\bx \sim \calD} \sum_a \Vert (A_{a}^{\bx} - B_{a}^{\bx} ) \ket{\psi}\Vert^2 + 2 \E_{\bx \sim \calD} \sum_a \Vert (B_{a}^{\bx} - C_{a}^{\bx} ) \ket{\psi}\Vert^2\\
&\leq 2 (\delta + \epsilon).\qedhere
\end{align*}
\end{proof}

Note that this does \emph{not} show that if
\begin{equation*}
A^x_a \otimes I_{\reg{Bob}} \approx_\delta I_{\reg{Alice}} \otimes B^x_a
\quad\text{and}\quad
B^x_a \otimes I_{\reg{Bob}} \approx_\delta I_{\reg{Alice}} \otimes C^x_a
\end{equation*}
then $A^x_a \otimes I_{\reg{Bob}} \approx_\delta I_{\reg{Alice}} \otimes C^x_a$.
This would only follow if, for example, we also knew that $D^x_a \otimes I_{\reg{Bob}} \approx_\delta I_{\reg{Alice}} \otimes D^x_a$,
for $D$ equal to one of~$A$, $B$, or $C$.
We do, however, always have the following triangle-like inequalities.

\begin{fact}[Triangle-like inequalities]\label{fact:triangle-like}
The following two facts are true.
\begin{enumerate}
\item Suppose $A^x_a \otimes I_{\reg{Bob}} \consistency_\delta I_{\reg{Alice}} \otimes B^x_a$,
	$B^x_a \otimes I_{\reg{Bob}} \consistency_\delta I_{\reg{Alice}} \otimes C^x_a$,
	and $C^x_a \otimes I_{\reg{Bob}} \consistency_\delta I_{\reg{Alice}} \otimes D^x_a$.
	Then $A^x_a \otimes I_{\reg{Bob}} \consistency_\delta I_{\reg{Alice}} \otimes D^x_a$.
\item Suppose $A^x_a \otimes I_{\reg{Bob}} \approx_\delta I_{\reg{Alice}} \otimes B^x_a$,
	$B^x_a \otimes I_{\reg{Bob}} \approx_\delta I_{\reg{Alice}} \otimes C^x_a$,
	and $C^x_a \otimes I_{\reg{Bob}} \approx_\delta I_{\reg{Alice}} \otimes D^x_a$.
	Then $A^x_a \otimes I_{\reg{Bob}} \approx_\delta I_{\reg{Alice}} \otimes D^x_a$.
\end{enumerate}
\end{fact}

Before proving this, we need the following fact from linear algebra.

\begin{fact}\label{fact:fun-matrix-fact}
Suppose $0 \preceq A, B, C, D \preceq I$.  Then
\begin{equation*}
1- \bra{\psi} A \otimes D \ket{\psi}
\leq (1 - \bra{\psi} A \otimes B \ket{\psi}) + (1 - \bra{\psi} B \otimes C \ket{\psi}) + (1 - \bra{\psi} C \otimes D \ket{\psi}).
\end{equation*}
\end{fact}
\begin{proof}
Rearranging, we want to show that
\begin{equation*}
\bra{\psi} A \otimes B \ket{\psi} + \bra{\psi} B \otimes C \ket{\psi} + \bra{\psi} C \otimes D \ket{\psi} - \bra{\psi} A \otimes D \ket{\psi}
\leq 2.
\end{equation*}
Or, equivalently
\begin{equation*}
\tr(\ket{\psi}\bra{\psi} \cdot (A \otimes B + B \otimes C + C \otimes D - A \otimes D)) \leq 2.
\end{equation*}
The left-hand side is at most the maximum eigenvalue of $A \otimes B + B \otimes C + C \otimes D - A \otimes D$.
To bound this maximum eigenvalue, we note that $A \otimes B \preceq A \otimes I$, $B \otimes C \preceq I \otimes I$, and $C \otimes D \preceq I \otimes D$.
As a result, 
\begin{equation*}
A \otimes B + B \otimes C + C \otimes D - A \otimes D
\preceq A \otimes I + I \otimes I + I \otimes D - A \otimes D.
\end{equation*}
Next, $I \otimes D - A \otimes D = (I - A) \otimes D \preceq (I-A) \otimes I$ because $A \preceq I$.
Thus,
\begin{equation*}
A \otimes I + I \otimes I + I \otimes D - A \otimes D
\preceq A \otimes I + I \otimes I  + (I-A) \otimes I = 2 \cdot I \otimes I.
\end{equation*}
But the maximum eigenvalue of this is~$2$.
\end{proof}

Now we prove \Cref{fact:triangle-like}.
\begin{proof}[Proof of \Cref{fact:triangle-like}]
The second fact follows from several applications of \Cref{fact:triangle}.
As for the first fact, we can write the consistency as
\begin{equation*}
\E_\bx \sum_a ( 1- \bra{\psi} A^{\bx}_a \otimes D^{\bx}_a \ket{\psi})
\end{equation*}
Applying \Cref{fact:fun-matrix-fact}, this is at most
\begin{equation*}
\E_\bx \sum_a (1 - \bra{\psi} A^\bx_a \otimes B^\bx_a \ket{\psi}) + (1 - \bra{\psi} B^\bx_a \otimes C^\bx_a \ket{\psi}) + (1 - \bra{\psi} C^\bx_a \otimes D^\bx_a \ket{\psi})
\end{equation*}
Averaging over questions and summing over answers, each of these terms is at most~$\delta$, by assumption.
\end{proof}

\subsubsection{Close strategies have close game values}

In this section, we will show that two strategies which are close in state-dependent distance are also close in value for any game~$\game$.
We note crucially that one of the two strategies must be \emph{projective} to apply this fact.

\begin{fact}\label{fact:approx-delta-generalized-game-value}
Let $\calD$ be a distribution on questions~$x$, and for each~$x$ let $\mathrm{acc}(x)$ be a set of ``accepting" answers.
Given a state~$\psi$ and a strategy $\{A^x_a\}$ define
\begin{equation*}
\mathrm{val}(A) = \E_{\bx \sim \calD} \sum_{a \in \mathrm{acc}(x)} \bra{\psi} A^{\bx}_a \ket{\psi}.
\end{equation*}
Suppose $\{A^x_a\}$ and $\{B^x_a\}$ are two strategies such that $A^x_a \approx_\delta B^x_a$ on state~$\psi$ and distribution~$\calD$.
Suppose further that either~$A$ or~$B$ is a projective measurement (and the other is a POVM measurement).
Then
\begin{equation*}
\mathrm{val}(A) - O(\delta^{1/2}) \leq \mathrm{val}(B) \leq \mathrm{val}(A) + O(\delta^{1/2}).
\end{equation*}
\end{fact}
\begin{proof}
Assume without loss of generality that~$A$ is a projective measurement and~$B$ is a POVM measurement.
We will prove the fact by showing the following stronger statement:
for each~$x$, let~$S(x)$ be any set of answers~$a$,
and define
\begin{equation*}
\mathrm{val}(A,S) := \E_\bx \sum_{a \in S(x)} \bra{\psi} A^{\bx}_a\ket{\psi}.
\end{equation*}
Then $\mathrm{val}(A,S) \leq \mathrm{val}(B,S) + O(\delta^{1/2})$.
By taking $S(x) := \mathrm{acc}(x)$ this implies the lower bound $\mathrm{val}(A) - O(\delta^{1/2}) \leq \mathrm{val}(B)$,
and by taking $S(x) := \mathrm{rej}(x)$, defined to be the set of answers \emph{not} in $\mathrm{acc}(x)$,
then this implies the upper bound $\mathrm{val}(B) \leq \mathrm{val}(A) + O(\delta^{1/2})$.

If we write $\ket{u_a^x} = A_{a}^{x}  \ket{\psi}$
	and $\ket{w_a^x} = (B_{a}^{x} - A_{a}^{x}) \ket{\psi}$, then
\begin{equation*}
\Vert B_{a}^{x} \ket{\psi} \Vert^2
= \Vert \ket{u_a^x} + \ket{v_a^x} \Vert^2
 = \Vert\ket{u_a^x}\Vert^2 + \Vert\ket{w_a^x}\Vert^2
 + \braket{u_a^x \mid w_a^x} + \braket{w_a^x \mid u_a^x}.
\end{equation*}
By definition,
\begin{align*}
\mathrm{val}(B)
& = \E_{\bx} \sum_{a\in S(\bx)} \bra{\psi} B_{a}^{\bx} \ket{\psi}\\
& \geq \E_{\bx} \sum_{a\in S(\bx)} \bra{\psi} (B_{a}^{\bx})^2 \ket{\psi}\tag{because~$B$ is a POVM}\\
& = \E_{\bx} \sum_{a\in S(\bx)}\Vert B_{a}^{\bx} \ket{\psi} \Vert^2\\
& = \E_{\bx} \sum_{a\in S(\bx)} \Vert\ket{u_a^{\bx}}\Vert^2 + \Vert\ket{w_a^{\bx}}\Vert^2
		 + \braket{u_a^{\bx} \mid w_a^{\bx}} + \braket{w_a^{\bx} \mid u_a^{\bx}}.
\end{align*}
Averaging over questions and summing over answers, the first term is exactly $\mathrm{val}(A)$ because~$A$ is projective.
The second term is always nonnegative, so we lower bound it by zero.
As for the last two terms,
\begin{equation}\label{eq:our-powers-combined--wind!}
\braket{u_a^x \mid w_a^x} + \braket{w_a^x \mid u_a^x}
\geq - 2 \cdot |\braket{u_a^x \mid w_a^x}| 
\geq -2 \cdot \Vert u_a^x \Vert \cdot \Vert w_a^x \Vert.
\end{equation}
Applying Cauchy-Schwarz,
Jensen's inequality,
and \Cref{fact:measurement-sub-measurement-switcheroo},
\begin{equation}\label{eq:our-powers-combined--heart!}
\E_{\bx} \sum_{a \in S(\bx)} \Vert u_a^{\bx} \Vert \cdot \Vert w_a^{\bx} \Vert
\leq \E_{\bx} \sqrt{\sum_{a \in S(\bx)} \Vert u_a^{\bx} \Vert^2 \cdot  \sum_{a \in S(\bx)}\Vert w_a^{\bx}\Vert^2 }
\leq \sqrt{\E_{\bx} \sum_{a \in S(\bx)}\Vert w_a^{\bx}\Vert^2 }.
\end{equation}
But the expectation inside the root is at most $O(\delta)$ because $A^{x}_a \approx_{\delta} B^{x}_a$.
Combining \Cref{eq:our-powers-combined--wind!,eq:our-powers-combined--heart!} completes the proof.
\end{proof}

We will typically, though not always,
apply \Cref{fact:approx-delta-generalized-game-value} in the following special case.

\begin{fact}\label{fact:approx-delta-game-value}
Let $\game$ be a game whose questions $(\bx_1, \bx_2) \sim \game$ have marginal distribution $\bx_1 \sim \calD$.
Suppose $\{A_a^x\}$
and $\{B_a^x\}$ are measurements such that
$
A_a^x  \otimes I \approx_\delta B_a^x \otimes I
$
on state $\psi$ and distribution~$\calD$.
Consider the strategies $\calS_A = \{\psi, A\}$ and $\calS_B = \{\psi, B\}$.
If either~$A$ or~$B$ is a projective measurement (and the other is a POVM measurement), then
\begin{equation*}
\valstrat{\game}{\calS_A} - O(\delta^{1/2}) \leq \valstrat{\game}{\calS_B} \leq \valstrat{\game}{\calS_A} + O(\delta^{1/2}).
\end{equation*}
\end{fact}
\begin{proof}
First, we observe that
\begin{equation*}
A^{x_1}_{a_1} \otimes A^{x_2}_{a_2}
\approx_{\delta} A^{x_1}_{a_1} \otimes B^{x_2}_{a_2}
\approx_{\delta} B^{x_1}_{a_1} \otimes B^{x_2}_{a_2}
\end{equation*}
by \Cref{fact:add-a-proj}.
The result follows by applying \Cref{fact:approx-delta-game-value} with
``$A$" set to $A^{x_1}_{a_1} \otimes A^{x_2}_{a_2}$,
``$B$" set to $B^{x_1}_{a_1} \otimes B^{x_2}_{a_2}$,
and ``$\calD$" set to the distribution on $(\bx_1, \bx_2)$.
We note that ``$\mathrm{val}(A)$" there is equal to $\valstrat{\game}{\calS_A}$ here and
``$\mathrm{val}(B)$" there is equal to $\valstrat{\game}{\calS_B}$ here.
\end{proof}

\subsubsection{Generating new measurements}

In this section, we show how to combine multiple measurements into a single measurement by ``sandwiching" them together.

\begin{fact}\label{fact:sandwich}
Let $k \geq 0$ be a constant.
Let $\{A^x_{a_1, \ldots, a_k}\}$ be a projective measurement.
For each $1 \leq i \leq k$, let $\{(B_i)^x_{a_i}\}$ be a projective measurement, and suppose that
\begin{equation}\label{eq:no-take-only-throw}
(A^x_{a_i})_{\reg{Alice}} \otimes I_{\reg{Bob}} \simeq_{\delta} I_{\reg{Alice}} \otimes ((B_i)^x_{a_i})_{\reg{Bob}}.
\end{equation}
Define the POVM measurement $\{J^x_{g_1, \ldots, g_k}\}$ as
\begin{equation*}
J^x_{a_1, \ldots, a_k} := (B_k)^x_{a_k} \cdots (B_2)^x_{a_2} \cdot (B_1)^x_{a_1} \cdot (B_2)^x_{a_2} \cdots (B_k)^x_{a_k}.
\end{equation*}
Then
\begin{equation*}
(A^x_{a_1, \ldots, a_k})_{\reg{Alice}} \otimes I_{\reg{Bob}} \simeq_{\delta^{1/2}} I_{\reg{Alice}} \otimes (J^x_{a_1, \ldots, a_k})_{\reg{Bob}}.
\end{equation*}
\end{fact}
\begin{proof}
For each $1 \leq i \leq k$, 
\Cref{eq:no-take-only-throw} implies that
\begin{equation*}
(A^x_{a_i})_{\reg{Alice}} \otimes I_{\reg{Bob}} \approx_{\delta} I_{\reg{Alice}} \otimes ((B_i)^x_{a_i})_{\reg{Bob}}.
\end{equation*}
Now, we repeatedly apply this using \Cref{fact:add-a-proj}:
\begin{align*}
(A^x_{a_1, \ldots, a_k})_{\reg{Alice}} \otimes I_{\reg{Bob}}
& = (A^x_{a_k} \cdots A^x_{a_2} \cdot A^x_{a_1} \cdot A^x_{a_2} \cdots A^x_{a_{k}})_{\reg{Alice}} \otimes I_{\reg{Bob}}\\
& \approx_{\delta} (A^x_{a_k} \cdots A^x_{a_2} \cdot A^x_{a_1} \cdot A^x_{a_2} \cdots A^x_{a_{k-1}})_{\reg{Alice}} \otimes ((B_k)^x_{a_k})_{\reg{Bob}}\\
&\qquad\qquad\qquad \qquad \qquad \qquad \cdots\\
& \approx_{\delta} I_{\reg{Alice}} \otimes ((B_k)^x_{a_k} \cdots (B_2)^x_{a_2} \cdot (B_1)^x_{a_1} \cdot (B_2)^x_{a_2} \cdots (B_k)^x_{a_k})_{\reg{Bob}}\\
& = I_{\reg{Alice}} \otimes (J^x_{a_1, \ldots, a_k})_{\reg{Bob}}.
\end{align*}
The fact now follows from \Cref{fact:almost-agreement} and the fact that~$A$ is a projective measurement.
\end{proof}

Next, we extend \Cref{fact:sandwich} to the case of polynomial measurements (see \Cref{sec:classical-q-low-deg} below).
These are structured measurements in which the prover returns the evaluation of a function sampled independently from their input.
The goal is to retain this structure even after ``sandwiching" them together.

\begin{fact}\label{fact:low-degree-sandwich}
Let $k \geq 0$ be a constant.
Let $\calD$ be a distribution on questions $x \in \calX$.
For each $1 \leq i \leq k$, let $\calG_i$ be a set of functions $g_i : \calX \rightarrow \calR_i$.
and let $\{G^i_g\}$ be a projective measurement with outcomes from this set.
Suppose that the set $\calG_i$ has the following distance property:
for any two nonequal $g_i, g_i' \in \calG_i$, the probability that $g_i(\bx) = g_i'(\bx)$, over a random $\bx \sim \calD$, is at most $\eps$.

Let $\{A_{g_1, \ldots, g_k}\}$ be a projective measurement with outcomes $g_i \in \calF_i$.
For each $1 \leq i \leq k$, suppose that
\begin{equation}\label{eq:treat-yo-self}
(A_{[g_i(x) = a_i]})_{\reg{Alice}} \otimes I_{\reg{Bob}} \simeq_{\delta} I_{\reg{Alice}} \otimes (G^i_{[g_i(x)=a_i]})_{\reg{Bob}}.
\end{equation}
Define the POVM measurement $\{J_{g_1, \ldots, g_k}\}$ as
\begin{equation*}
J_{g_1, \ldots, g_k} := G^k_{g_k} \cdots G^2_{g_2} \cdot G^1_{g_1} \cdot G^2_{g_2} \cdots G^k_{g_k}.
\end{equation*}
Then
\begin{equation*}
(A_{[g_1(x), \ldots, g_k(x) = a_1, \ldots, a_k]})_{\reg{Alice}} \otimes I_{\reg{Bob}} \simeq_{(\delta+\eps)^{1/2}} I_{\reg{Alice}} \otimes (J_{[g_1(x), \ldots, g_k(x) = a_1, \ldots, a_k]})_{\reg{Bob}}.
\end{equation*}
\end{fact}
\begin{proof}
Let $1 \leq i \leq k$.
By~\Cref{eq:treat-yo-self},
if Alice measures with~$A$, producing~$\bg_i$, and Bob measures with~$G^i$, producing~$\bg_i'$,
then the probability that $\bg_i(\bx) \neq \bg_i'(\bx)$ is $O(\delta)$.
Write $\eta$ for the probability that $\bg_i \neq \bg_i'$.
Then we have the expression $\eta \cdot (1-\eps) \leq O(\delta)$ or, equivalently, $\eta \leq O(\delta/(1-\eps))$.
When $\eps < 1/2$, this gives the bound $\eta \leq O(\delta)$, and when $\eps \geq 1/2$, we have the trivial bound $\eta \leq O(\eps)$.
As a result, $\eta = O(\delta + \eps)$.

In conclusion,
\begin{equation*}
(A_{g_i})_{\reg{Alice}} \otimes I_{\reg{Bob}} \simeq_{\delta + \eps} I_{\reg{Alice}} \otimes (G^i_{g_i})_{\reg{Bob}}.
\end{equation*}
We can now apply \Cref{fact:sandwich} to $A_{g_1, \ldots, g_k}$ and the $G^i_{g_i}$ measurements.
It implies that
\begin{equation*}
(A_{g_1, \ldots, g_k})_{\reg{Alice}} \otimes I_{\reg{Bob}} \simeq_{(\delta+\eps)^{1/2}} I_{\reg{Alice}} \otimes (J_{g_1, \ldots, g_k})_{\reg{Bob}}.
\end{equation*}
The fact now follows from the data processing inequality \Cref{fact:specialize-the-simeq}.
\end{proof}

In our next fact, we show that \Cref{fact:low-degree-sandwich} holds even when we drop the structured assumption on the~$A$ matrix.
The tradeoff is that we must now assume that the~$k$ different measurements act on different parts of the input string.
In this case, the distance condition becomes slightly more cumbersome to state.

\begin{fact}\label{fact:low-degree-sandwich-on-steroids}
Let $k \geq 0$ be a constant.
Let $\calD$ be a distribution on questions $(x, y_1, \ldots, y_k)$, where each $y_i \in \calY_i$.
For each $1 \leq i \leq k$, let $\calG_i$ be a set of functions $g_i : \calY_i \rightarrow \calR_i$.
and let $\{(G_i)^x_g\}$ be a projective measurement with outcomes from this set.
(For the $i=1$ case, we also allow this measurement to be a POVM.)
Suppose that the set $\calG_i$ has the following distance property:
fix a question $z = (x, y_1, \ldots, y_{i-1}, y_{i+1}, \ldots, y_k)$, and let $\calD_z$ be the distribution on~$y_i$ conditioned on the other outcomes~$z$.
Then for any two nonequal $g_i, g_i' \in \calG_i$, the probability that $g_i(\by_i) = g_i'(\by_i)$, over a random $\by_i \sim \calD_z$, is at most $\eps$.

Let $\{A^{x, y_1, \ldots, y_k}_{a_1, \ldots, a_k} \}$ be a projective measurement with outcomes $g_i \in \calF_i$.
For each $1 \leq i \leq k$, suppose that
\begin{equation}\label{eq:gonna-compare-like-a-pear}
(A^{x, y_1, \ldots, y_k}_{a_i})_{\reg{Alice}} \otimes I_{\reg{Bob}} \simeq_{\delta} I_{\reg{Alice}} \otimes ((G_i)^x_{[g_i(y_i)=a_i]})_{\reg{Bob}}.
\end{equation}
Suppose also that
\begin{equation}\label{eq:spend-some-time-reflecting}
(A^{x, y_1, \ldots, y_k}_{a_1, \ldots, a_k})_{\reg{Alice}} \otimes I_{\reg{Bob}}
	\simeq_{\delta} I_{\reg{Alice}} \otimes (A^{x, y_1, \ldots, y_k}_{a_1, \ldots, a_k})_{\reg{Bob}}.
\end{equation}
Define the POVM measurement $\{J^x_{g_1, \ldots, g_k}\}$ as
\begin{equation*}
J^x_{g_1, \ldots, g_k} := (G_k)^x_{g_k} \cdots (G_2)^x_{g_2} \cdot (G_1)^x_{g_1} \cdot (G_2)^x_{g_2} \cdots (G_k)^x_{g_k}.
\end{equation*}
Then
\begin{equation*}
(A^{x, y_1, \ldots, y_k}_{a_1, \ldots, a_k})_{\reg{Alice}} \otimes I_{\reg{Bob}}
	\simeq_{\poly(\delta, \eps)} I_{\reg{Alice}} \otimes (J^x_{[g_1(y_1), \ldots, g_k(y_k) = a_1, \ldots, a_k]})_{\reg{Bob}}.
\end{equation*}
\end{fact}
\begin{proof}
First, we show how to reduce this to the $k = 2$ case. Then we prove it for that case.
Assume the fact holds when $k = 2$.  Define the POVM measurement $\{(J_i)^x_{g_1, \ldots, g_i}\}$ as
\begin{equation*}
(J_i)^x_{g_1, \ldots, g_i} := (G_i)^x_{g_i} \cdots (G_2)^x_{g_2} \cdot (G_1)^x_{g_1} \cdot (G_2)^x_{g_2} \cdots (G_i)^x_{g_i}.
\end{equation*}
We will show by induction that 
\begin{equation}\label{eq:hi-i-work-in-graphics-but-not-that-graphics}
(A^{x, y_1, \ldots, y_k}_{a_1, \ldots, a_i})_{\reg{Alice}} \otimes I_{\reg{Bob}}
	\simeq_{\poly(\delta, \eps)} I_{\reg{Alice}} \otimes ((J_i)^x_{[g_1(y_1), \ldots, g_i(y_i) = a_1, \ldots, a_i]})_{\reg{Bob}},
\end{equation}
the base case being trivial. Assume this holds for~$i$.
We apply the $k = 2$ case as follows:
consider the question tuple $(y_1, \ldots, y_i)$ as a single question
and consider functions of the form $(y_1, \ldots, y_i) \mapsto (g_1(y_1), \ldots, g_i(y_i))$.
Then the first POVM measurement is $J_i$, which satisfies \Cref{eq:gonna-compare-like-a-pear} due to \Cref{eq:hi-i-work-in-graphics-but-not-that-graphics}.
The second measurement is the projector $G_{i+1}$.
Then the $k = 2$ case immediately implies the $i+1$ case of \Cref{eq:hi-i-work-in-graphics-but-not-that-graphics}.

Now we prove the $k = 2$ case.
Our goal is to show that
\begin{equation}\label{eq:food-rakes}
\E_{\bx, \by_1, \by_2} \sum_{a_1, g_2} \bra{\psi} (A^{\bx, \by_1, \by_2}_{a_1, g_2(\by_2)})_{\reg{Alice}}
\otimes ((G_2)^{\bx}_{g_2} \cdot (G_1)^{\bx}_{[g_1(\by_1) = a_1]} \cdot (G_2)^{\bx}_{g_2})_{\reg{Bob}} \ket{\psi}
\end{equation}
is at least $1- \poly(\delta, \eps)$.
We will do this by showing that
\begin{equation}\label{eq:what-a-johnnie-wants}
((G_1)^x_{[g_1(y_1)=a_1]} \cdot (G_2)^x_{g_2})_{\reg{Alice}} \otimes I_{\reg{Bob}}
\approx_{\poly(\delta, \eps)} ((G_2)^x_{g_2} \cdot (G_1)^x_{[g_1(y_1)=a_1]})_{\reg{Alice}} \otimes I_{\reg{Bob}}.
\end{equation}
is at most $\poly(\delta, \eps)$.
To see that this is sufficient, note that the related expression
\begin{equation*}
\E_{\bx, \by_1, \by_2} \sum_{a_1, g_2} \bra{\psi} (A^{\bx, \by_1, \by_2}_{a_1, g_2(\by_2)})_{\reg{Alice}}
\otimes ((G_2)^{\bx}_{g_2} \cdot (G_2)^{\bx}_{g_2} \cdot (G_1)^{\bx}_{[g_1(\by_1) = a_1]})_{\reg{Bob}} \ket{\psi}
\end{equation*}
is exactly equal to~$1$ because $G_2$ is a projector.
Taking the difference between this and \Cref{eq:food-rakes}, we get
\begin{equation*}
\E_{\bx, \by_1, \by_2} \sum_{a_1, g_2} \bra{\psi} (A^{\bx, \by_1, \by_2}_{a_1, g_2(\by_2)})_{\reg{Alice}}
\otimes ((G_2)^{\bx}_{g_2} \cdot ((G_1)^{\bx}_{[g_1(\by_1) = a_1]} \cdot (G_2)^{\bx}_{g_2}- (G_2)^{\bx}_{g_2} \cdot (G_1)^{\bx}_{[g_1(\by_1) = a_1]}))_{\reg{Bob}} \ket{\psi}.
\end{equation*}
Cauchy-Schwarz allows us to bound this by
\begin{multline*}
\leq \E_{\bx, \by_1, \by_2} \sqrt{ \sum_{a_1, g_2}\Vert (A^{\bx, \by_1, \by_2}_{a_1, g_2(\by_2)})_{\reg{Alice}} \otimes ((G_2)^{\bx}_{g_2})_{\reg{Bob}} \ket{\psi} \Vert^2}\\
	\cdot \sqrt{ \sum_{a_1, g_2}\Vert I_{\reg{Alice}} \otimes ((G_1)^{\bx}_{[g_1(\by_1) = a_1]} \cdot (G_2)^{\bx}_{g_2}- (G_2)^{\bx}_{g_2} \cdot (G_1)^{\bx}_{[g_1(\by_1) = a_1]})_{\reg{Bob}} \ket{\psi} \Vert^2}.
\end{multline*}
The expression inside the first square root is always at most~$1$.
This allows us to bring the expectation into the second square root by Jensen's inequality,
and the resulting expression we can bound due to \Cref{eq:what-a-johnnie-wants}.

Now we bound \Cref{eq:what-a-johnnie-wants}.
Showing this is small is equivalent to showing
\begin{equation*}
\E_{\bx, \by_1, \by_2}
\sum_{a_1, g_2}\Vert ((G_1)^{\bx}_{[g_1(\by_1) = a_1]}
	\cdot (G_2)^{\bx}_{g_2}- (G_2)^{\bx}_{g_2} \cdot (G_1)^{\bx}_{[g_1(\by_1) = a_1]})_{\reg{Alice}} \otimes I_{\reg{Bob}} \ket{\psi} \Vert^2
\end{equation*}
is small.  Expanding this, we get
\begin{align}
\E_{\bx, \by_1, \by_2} \sum_{a_1, g_2} \bra{\psi}\big(
	&(G_2)^{\bx}_{g_2} \cdot (G_1)^{\bx}_{[g_1(\by_1) = a_1]} \cdot (G_1)^{\bx}_{[g_1(\by_1) = a_1]} \cdot (G_2)^{\bx}_{g_2} \otimes I_{\reg{Bob}}\nonumber\\
	+ &(G_1)^{\bx}_{[g_1(\by_1) = a_1]} \cdot (G_2)^{\bx}_{g_2} \cdot (G_2)^{\bx}_{g_2} \cdot (G_1)^{\bx}_{[g_1(\by_1) = a_1]} \otimes I_{\reg{Bob}}\nonumber\\
	- &(G_2)^{\bx}_{g_2} \cdot (G_1)^{\bx}_{[g_1(\by_1) = a_1]} \cdot (G_2)^{\bx}_{g_2} \cdot (G_1)^{\bx}_{[g_1(\by_1) = a_1]} \otimes I_{\reg{Bob}}\nonumber\\
	- &(G_1)^{\bx}_{[g_1(\by_1) = a_1]} \cdot (G_2)^{\bx}_{g_2} \cdot (G_1)^{\bx}_{[g_1(\by_1) = a_1]} \cdot (G_2)^{\bx}_{g_2} \otimes I_{\reg{Bob}} \big)\ket{\psi}.\label{eq:burger-king}
\end{align}
We \emph{do} know that $G_1$ and $G_2$ satisfy some form of commutation.
Because they satisfy \Cref{eq:gonna-compare-like-a-pear} and~$A$ is a projector, 
we know that
\begin{equation*}
((G_1)^x_{[g_1(y_1)=a_1]} \cdot (G_2)^x_{[g_2(y_2) = a_2]})_{\reg{Alice}} \otimes I_{\reg{Bob}}
\approx_{\delta} ((G_2)^x_{[g_2(y_2) = a_2]} \cdot (G_1)^x_{[g_1(y_1)=a_1]})_{\reg{Alice}} \otimes I_{\reg{Bob}}.
\end{equation*}
Expanding this as above, we can bound the following expression by~$\delta$:
\begin{align}
\E_{\bx, \by_1, \by_2} \sum_{a_1, a_2} \bra{\psi}\big(
	&(G_2)^\bx_{[g_2(\by_2) = a_2]} \cdot (G_1)^{\bx}_{[g_1(\by_1) = a_1]} \cdot (G_1)^{\bx}_{[g_1(\by_1) = a_1]} \cdot (G_2)^\bx_{[g_2(\by_2) = a_2]} \otimes I_{\reg{Bob}}\nonumber\\
	+ &(G_1)^{\bx}_{[g_1(\by_1) = a_1]} \cdot (G_2)^\bx_{[g_2(\by_2) = a_2]} \cdot (G_2)^\bx_{[g_2(\by_2) = a_2]} \cdot (G_1)^{\bx}_{[g_1(\by_1) = a_1]} \otimes I_{\reg{Bob}}\nonumber\\
	- &(G_2)^\bx_{[g_2(\by_2) = a_2]} \cdot (G_1)^{\bx}_{[g_1(\by_1) = a_1]} \cdot (G_2)^\bx_{[g_2(\by_2) = a_2]} \cdot (G_1)^{\bx}_{[g_1(\by_1) = a_1]} \otimes I_{\reg{Bob}}\nonumber\\
	- &(G_1)^{\bx}_{[g_1(\by_1) = a_1]} \cdot (G_2)^\bx_{[g_2(\by_2) = a_2]} \cdot (G_1)^{\bx}_{[g_1(\by_1) = a_1]} \cdot (G_2)^\bx_{[g_2(\by_2) = a_2]}\otimes I_{\reg{Bob}} \big)\ket{\psi}.\label{eq:what-to-do-with-this}
\end{align}
We can therefore show \Cref{eq:burger-king} is small by
upper-bounding (\Cref{eq:burger-king}$ - $\Cref{eq:what-to-do-with-this}).
There are four terms in this difference;
write $\Delta_i$ for the $i$-th term in \Cref{eq:burger-king} minus the $i$-th term in \Cref{eq:what-to-do-with-this}.
We will bound each $\Delta_i$ one-by-one.

The first term in the difference, $\Delta_1$, is
\begin{multline*}
\E_{\bx, \by_1, \by_2} \sum_{a_1}\sum_{g_2} \bra{\psi}
	(G_2)^{\bx}_{g_2}
	\cdot (G_1)^{\bx}_{[g_1(\by_1) = a_1]} \cdot (G_1)^{\bx}_{[g_1(\by_1) = a_1]} \cdot (G_2)^{\bx}_{g_2} \otimes I_{\reg{Bob}} \ket{\psi}\\
-\E_{\bx, \by_1, \by_2} \sum_{a_1, a_2} \bra{\psi}
	(G_2)^{\bx}_{[g_2(\by_2) = a_2]}
	\cdot (G_1)^{\bx}_{[g_1(\by_1) = a_1]} \cdot (G_1)^{\bx}_{[g_1(\by_1) = a_1]} \cdot (G_2)^{\bx}_{[g_2(\by_2) = a_2]} \otimes I_{\reg{Bob}} \ket{\psi}.
\end{multline*}
The first of these terms is at most~$1$, and so we just have to show that the second term is close to~$1$ as well.
Note that by repeated applications of \Cref{eq:gonna-compare-like-a-pear}, we have that 
\begin{equation*}
(G_2)^{\bx}_{[g_2(\by_2) = a_2]}
	\cdot (G_1)^{\bx}_{[g_1(\by_1) = a_1]} \cdot (G_1)^{\bx}_{[g_1(\by_1) = a_1]} \cdot (G_2)^{\bx}_{[g_2(\by_2) = a_2]} \otimes I_{\reg{Bob}}
	\approx_{\delta} I_{\reg{Alice}} \otimes A^{\bx, \by_1, \by_2}_{a_1, a_2}.
\end{equation*}
But then by \Cref{fact:approx-delta-generalized-game-value}, the expression we want to lower-bound is $O(\delta^{1/2})$-close to
\begin{equation*}
\E_{\bx, \by_1, \by_2} \sum_{a_1, a_2} \bra{\psi} I_{\reg{Alice}} \otimes A^{\bx, \by_1, \by_2}_{a_1, a_2} \ket{\psi},
\end{equation*}
which is exactly~$1$. As a result, $\Delta_1$ is at most $O(\delta^{1/2})$.

The second term in the difference, $\Delta_2$, can be written as
\begin{equation*}
-\E_{\bx, \by_1, \by_2} \sum_{a_1}\sum_{\substack{g_2\neq g_2',\\g_2(\by_2) = g_2'(\by_2)}} \bra{\psi}(
	(G_1)^{\bx}_{[g_1(\by_1) = a_1]} \cdot (G_2)^{\bx}_{g_2} \cdot (G_2)^{\bx}_{g_2'} \cdot (G_1)^{\bx}_{[g_1(\by_1) = a_1]} \otimes I_{\reg{Bob}})\ket{\psi}.
\end{equation*}
This is zero because $G_2$ is a projector.

The third and fourth terms in \Cref{eq:burger-king} are complex conjugates of each other,
as are the third and fourth terms in \Cref{eq:what-to-do-with-this}.
As a result, it suffices to bound the magnitude of $\Delta_4$, and this will serve to bound $\Delta_3$ as well.
We begin by manipulating the fourth term in \Cref{eq:burger-king};
specifically, we will show that it is close to
\begin{equation}\label{eq:cloud}
-\E_{\bx, \by_1, \by_2} \sum_{a_1, g_2} \bra{\psi}
	 (G_1)^{\bx}_{[g_1(\by_1) = a_1]} \cdot (G_2)^{\bx}_{g_2} \cdot (G_1)^{\bx}_{[g_1(\by_1) = a_1]} \otimes (G_2)^{\bx}_{g_2} \ket{\psi}.
\end{equation}
To do so, we take their difference:
\begin{multline*}
\E_{\bx, \by_1, \by_2} \sum_{a_1, g_2} \bra{\psi}
	 ((G_1)^{\bx}_{[g_1(\by_1) = a_1]} \cdot (G_2)^{\bx}_{g_2} \cdot (G_1)^{\bx}_{[g_1(\by_1) = a_1]} \otimes I_{\reg{Bob}})\\
	\cdot (I_{\reg{Alice}} \otimes  (G_2)^{\bx}_{g_2} - (G_2)^{\bx}_{g_2} \otimes I_{\reg{Bob}}) \ket{\psi}.
\end{multline*}
To bound the magnitude, we apply Cauchy-Schwarz:
\begin{multline*}
\E_{\bx, \by_1, \by_2}
	\sqrt{\sum_{a_1, g_1} \Vert  ((G_1)^{\bx}_{[g_1(\by_1) = a_1]} \cdot (G_2)^{\bx}_{g_2} \cdot (G_1)^{\bx}_{[g_1(\by_1) = a_1]} \otimes I_{\reg{Bob}}) \ket{\psi} \Vert^2}\\
	\cdot \sqrt{\sum_{a_1, g_1} \Vert ((G_1)^{\bx}_{[g_1(\by_1) = a_1]} \otimes I_{\reg{Bob}})
	\cdot (I_{\reg{Alice}} \otimes  (G_2)^{\bx}_{g_2} - (G_2)^{\bx}_{g_2} \otimes I_{\reg{Bob}}) \ket{\psi} \Vert^2}.
\end{multline*}
The expression inside the first square root is always at most~$1$.
This allows us to bring the expectation into the second square root by Jensen's inequality.
Because $G_1$ is a POVM, we can bound the resulting expectation by
\begin{equation}\label{eq:majora}
\E_{\bx, \by_1, \by_2} \sum_{g_1} \Vert (I_{\reg{Alice}} \otimes  (G_2)^{\bx}_{g_2} - (G_2)^{\bx}_{g_2} \otimes I_{\reg{Bob}}) \ket{\psi} \Vert^2.
\end{equation}
To bound this, we note that \Cref{eq:gonna-compare-like-a-pear,eq:spend-some-time-reflecting} along with \Cref{fact:triangle-like} imply that
\begin{equation*}
(G_2)^{x}_{[g_2(y_2) = a_2]} \otimes I_{\reg{Bob}}
\consistency_{\delta} I_{\reg{Bob}} \otimes (G_2)^{x}_{[g_2(y_2) = a_2]}.
\end{equation*}
Using the distance properties of $\calG_2$, this implies that
\begin{equation*}
(G_2)^{x}_{g_2} \otimes I_{\reg{Bob}}
\consistency_{\delta + \eps} I_{\reg{Bob}} \otimes (G_2)^{x}_{g_2}.
\end{equation*}
Hence, \Cref{eq:majora} is at most $O((\delta+\eps)^{1/2})$.
A similar argument shows that the fourth term in \Cref{eq:what-to-do-with-this}
is $O(\delta^{1/2})$-close to
\begin{equation}\label{eq:strife}
-\E_{\bx, \by_1, \by_2} \sum_{a_1, a_2} \bra{\psi}
(G_1)^{\bx}_{[g_1(\by_1) = a_1]} \cdot (G_2)^\bx_{[g_2(\by_2) = a_2]} \cdot (G_1)^{\bx}_{[g_1(\by_1) = a_1]} \otimes (G_2)^\bx_{[g_2(\by_2) = a_2]} \ket{\psi}.
\end{equation}

Now, we compute \Cref{eq:cloud} minus \Cref{eq:strife}:
\begin{align*}
&\E_{\bx, \by_1, \by_2} \sum_{a_1}\sum_{\substack{g_2, g_2'\\g_2(\by_2) \neq g_2'(\by_2)}} \bra{\psi}
	 (G_1)^{\bx}_{[g_1(\by_1) = a_1]} \cdot (G_2)^{\bx}_{g_2} \cdot (G_1)^{\bx}_{[g_1(\by_1) = a_1]} \otimes (G_2)^{\bx}_{g_2'} \ket{\psi}\\
=&\E_{\bx, \by_1, \by_2} \sum_{a_1}\sum_{g_2, g_2'} \bra{\psi}
	 (G_1)^{\bx}_{[g_1(\by_1) = a_1]} \cdot (G_2)^{\bx}_{g_2} \cdot (G_1)^{\bx}_{[g_1(\by_1) = a_1]} \otimes (G_2)^{\bx}_{g_2'} \ket{\psi}
	 	\cdot \bone(g_2, g_2', \by_2),
\end{align*}
where $\bone(g_2, g_2', \by_2)$ is the indicator that $g_2 \neq g_2'$ but $g_2(\by_2) = g_2'(\by_2)$.
This is the only part of the expression that depends on $\by_2$, and by our distance assumption it is at most $\eps$ in expectation.
Since the rest of the expression is guaranteed to be positive, we can upper-bound this by
\begin{equation*}
\E_{\bx, \by_1} \sum_{a_1}\sum_{g_2, g_2'} \bra{\psi}
	 (G_1)^{\bx}_{[g_1(\by_1) = a_1]} \cdot (G_2)^{\bx}_{g_2} \cdot (G_1)^{\bx}_{[g_1(\by_1) = a_1]} \otimes (G_2)^{\bx}_{g_2'} \ket{\psi} \cdot \eps.
\end{equation*}
But the remaining part of the expression is at most~$1$, and so in total we can upper-bound it by~$\eps$.
This completes the proof.
\end{proof}

\subsection{Commuting EPR strategies}

In this section, we introduce a class of strategies important for our proof.

\begin{definition}
A strategy $\calS = (\psi, M)$ is called an \emph{EPR strategy} if it satisfies the following two properties.
First, there is an integer~$k$ and powers of two $q_1, \ldots, q_k$ such that
\begin{equation*}
\ket{\psi} = \ket{\mathrm{EPR}_{q_1}} \otimes \ket{\mathrm{EPR}_{q_2}} \otimes \cdots \otimes \ket{\mathrm{EPR}_{q_k}}.
\end{equation*}
Second, for each question~$x$, $M^x$ is a projective measurement. If
for all questions $x$ and answers $a$, $M^x_a$ is a real-valued
matrix, we say that the strategy is \emph{real}

In addition, given a game~$\game$, we say that a real EPR strategy $\calS$ is a \emph{real
  commuting EPR strategy (with respect to~$\game$)}
if for every~$(x_1, x_2)$ in the support of~$\calS$ and every~$a_1, a_2$, $M^{x_1}_{a_1}$ commutes with $M^{x_2}_{a_2}$.
We denote the set of real commuting EPR strategies by $\comstrat{\game}$.
\end{definition}

Real commuting EPR strategies are motivated by the \emph{completeness} cases that arise in this work.
We give a series of transformations which modify games to make them sound against increasingly broader sets of strategies.
Unfortunately, these transformations are not complete for all strategies, in the sense that value-$1$ strategies may be mapped to value-less-than-1 strategies,
but we will be careful to ensure that they \emph{are} complete for all commuting EPR strategies.
For the majority of the paper, the one property of commuting EPR strategies that we will use, not shared by all value-$1$ strategies, is the following.

\begin{fact}\label{fact:heh-heh-heh-gonna-make-anand-prove-this-so-i-can-take-the-day-off}
Let $(\psi, M)$ be a real EPR strategy.
Then $M^x_a \otimes I_{\reg{Bob}} \consistency_0 I_{\reg{Alice}} \otimes M^x_a$ for every distribution on~$x$.
\end{fact}
\begin{proof}
  From the definition of EPR strategies, we know that $\ket{\psi} =
  \ket{\epr_{q_1}} \ot \dots \ot \ket{\epr_{q_k}} \in (\C^{q_1 \cdot
    q_2 \cdot \dots \cdot q_k})^{\ot 2}$. We may choose a basis
  $\{\ket{i} : 1 \leq i \leq q_1 \cdot q_2 \cdot \dots \cdot q_k\}$
  for $\C^{q_1 \cdot \dots \cdot q_k}$, so that
  \[ \ket{\psi} = \sum_{i} \ket{i}_{\reg{Alice}} \ot
    \ket{i}_{\reg{Bob}} . \]
  
Let us denote the components of $M^x_a$ by the notation
$(M^x_a)_{ij}$, so that $M^x_a = \sum_{ij} (M^x_a)_{ij}
\ket{i}\bra{j}$.  Now, for an arbitrary pair $x, a$, we can compute
the post-measurement states from applying $M^x_a$ on Alice's and Bob's systems.
\begin{align*}
  M^x_a \ot I_{\reg{Bob}} \ket{\psi} &= (M^x_a \ot I_{\reg{Bob}}) \sum_{i} \ket{i}_{\reg{Alice}}
                                       \ot \ket{i}_{\reg{Bob}} \\
                                     &= \sum_{ij} (M^x_a)_{ij} \ket{i}
                                       \ot \ket{j} \\
                                     &= \sum_{ij} (M^x_a)_{ji} \ket{i}
                                       \ot \ket{j} \\
                                     &= (I_{\reg{Alice}} \ot M^x_a) \ket{\psi},
\end{align*}
where in going from the second to the third line, we have used the
fact that $M^x_a$ is real and Hermitian, and thus symmetric.
\end{proof}

The following fact is a useful special case.

\begin{fact}\label{fact:the-ol-pauli-swaperoonie}
Let $n > 0$, $q = 2^t$, and $W \in \{X, Z\}$.
Then $\tau^W_u \otimes I \approx_0 I \otimes \tau^W_u$ on the state $\ket{\mathrm{EPR}_q^n}$.
\end{fact}
\begin{proof}
By \Cref{eq:pauli-eigenstates}, we can write 
\begin{equation*}
  \ket{\tau^X_u} = \frac{1}{\sqrt{q}} \sum_{v \in \F_q}
  \omega^{-\tr[uv]} \ket{v},\qquad \ket{\tau^Z_u} = \ket{u}.
\end{equation*}
The second of these self-evidently has real-valued coefficients.
As for the first, $q = 2^t$ implies that $p = 2$.
This means that $\omega = -1$ and $\tr[uv] \in \{0, 1\}$ for all $u, v$.
As a result, it too has real-valued coefficients.
The fact then follows from \Cref{fact:heh-heh-heh-gonna-make-anand-prove-this-so-i-can-take-the-day-off}.
\end{proof}

This property of real commuting strategies is useful for answer
reduction because it allows us to perform oracularization, giving one
prover both questions~$x_1$ and~$x_2$ so that they may simulate the
action of both provers by simultaneously measuring $M^{x_1}$ and $M^{x_2}$.
For more details, see \Cref{part:answer}.

\subsection{Quantum soundness of the classical low-degree test}\label{sec:classical-q-low-deg}

An important tool for quantum protocols
is a version of the Raz-Safra theorem (\Cref{thm:raz-safra})
in which the soundness of the low-degree test is extended to hold even in the case when the provers are allowed to share entanglement.
For the plane-versus-point test, this was first developed by Vidick in~\cite{Vid16}, but for technical reasons he could only show it for the case of three or more quantum provers.
In~\cite{NV18b}, this was improved to hold for the two-prover case, and this is the result we use in this work.
We begin by defining the class of polynomial measurements.

\begin{definition}
Define $\mathrm{PolyMeas}(m, d, q)$ to be the set of POVM measurements whose outcomes correspond to degree-$d$, $\F_q$-valued polynomials.
In other words, $G \in \polymeas{m}{d}{q}$ if $G = \{G_g\}_g$ with outcomes degree-$d$ polynomials $g:\F_q^m \rightarrow \F_q$.
More generally, we let $\simulpolymeas{m}{d}{q}{\ell}$ be the set of measurements $G = \{G_{g_1, \ldots, g_\ell}\}$ outputting~$\ell$ degree-$d$ polynomials $g_i:\F_q^m \rightarrow \F_q$.
\end{definition}

The following theorem establishes the quantum soundness of the classical low-degree test in the $k = 2$ case.

\begin{theorem}[Quantum soundness of the classical low-degree test~{\cite[Theorem 2]{NV18b}}]\label{thm:anand-thomas-classical-low-degree}
There exists a constant $c > 0$ and a function $\delta(\eps) = \mathrm{poly}(\epsilon, dm/q^c)$ such that the following holds.
Suppose Alice and Bob are entangled provers who pass $\game_{\mathrm{Surface}}(m,d,q,2)$ with probability at least $1-\eps$
using the strategy $(\psi, M)$, where~$M$ consists of projective measurements.
Then there exists a POVM measurement $G \in \polymeas{m}{d}{q}$ such that
\begin{equation*}
M^w_b \otimes I_{\reg{Bob}} \consistency_{\delta(\eps)} I_{\reg{Alice}} \otimes G_{[g(w)=b]},
\qquad
G_g \otimes I_{\reg{Bob}} \consistency_{\delta(\eps)} I_{\reg{Alice}} \otimes G_g,
\end{equation*}
where the first is on the uniform distribution over $\F_q^m$.
\end{theorem}

\begin{remark}
The statement of \Cref{thm:anand-thomas-classical-low-degree} is modified from how it appears in \cite[Theorem 2]{NV18b} to better suit our needs.
In this remark, we show how to derive our version from theirs,
which is stated as follows.
\begin{itemize}
\item[$\circ$] There exists a constant $c >0$ and a function $\delta(\eps) = \mathrm{poly}(\eps)$ such that the following holds.
		Suppose $q \geq (dm/\eps)^c$. Then if Alice and Bob pass the surface-versus-point test with probability $1-\eps$, there is a measurement $G \in \polymeas{m}{d}{q}$
		such that
		\begin{equation}\label{eq:what-anand-and-thomas-proved}
		\E_{\bs} \sum_g \sum_{f \neq g|_s} \bra{\psi} M^{\bs}_f \otimes G_g \ket{\psi} \leq \delta(\eps),\qquad
		\sum_g \bra{\psi} G_g \otimes (I - G_g) \ket{\psi} \leq \delta(\eps).
		\end{equation}
\end{itemize}
These are equivalent to the statements
\begin{equation*}
M^s_f \otimes I_{\reg{Bob}} \consistency_{\delta(\eps)} I_{\reg{Alice}} \otimes G_{[g|_s=f]},
\qquad
G_g \otimes I_{\reg{Bob}} \consistency_{\delta(\eps)} I_{\reg{Alice}} \otimes G_g,
\end{equation*}
where the first is on the uniform distribution over surface in $\F_q^m$.
The second of these matches the corresponding statement above.
Next, by \Cref{fact:specialize-the-simeq} we derive
\begin{equation*}
M^s_{[f(w) = b]} \otimes I_{\reg{Bob}} \consistency_{\delta(\eps)} I_{\reg{Alice}} \otimes G_{[g(w)=b]}
\qquad
\text{and}
\qquad
G_{[g(w) = b]} \otimes I_{\reg{Bob}} \consistency_{\delta(\eps)} I_{\reg{Alice}} \otimes G_{[g(w) = b]}.
\end{equation*}
with respect to the distribution $(\bs, \bw)$ from $\game_{\mathrm{Surface}}(m,d,q,2)$.
On top of that, since the strategy passes the test with probability $1-\eps$,
\begin{equation*}
M^w_b \otimes I_{\reg{Bob}} \consistency_{\eps} I_{\reg{Alice}} \otimes M^s_{[f(w) = b]}
\end{equation*}
As a result, if we use \Cref{fact:agreement} to switch these to ``$\approx_\delta$" statements, then
\begin{equation*}
M^w_b \otimes I_{\reg{Bob}}
\approx_{\eps} I_{\reg{Alice}} \otimes M^s_{[f(w) = b]}
\approx_{\delta(\eps)} G_{[g(w)=b]} \otimes  I_{\reg{Bob}} 
\approx_{\delta(\eps)} I_{\reg{Alice}} \otimes G_{[g(w) = b]}.
\end{equation*}
The result now follows from the triangle inequality (\Cref{fact:triangle}) followed by \Cref{fact:almost-agreement} and the fact that~$M$ was assumed to be projective.

Finally, we remove the condition on~$q$ using a trick from~\cite{NV18a}. If $q < (dm/\eps)^c$, then we select $\eps' > \eps$ such that $q = (dm/\eps')^c$.
Alice and Bob also pass the plane-versus-point test with probability $1-\eps'$ because $1-\eps' < 1-\eps$, and so we can apply the theorem with these parameters,
giving a robustness of $\delta(\eps') = \delta(dm/q^{1/c})$.
(In the case when $\eps' > 1$, which is not allowed, this bound trivially still holds because $dm/q^{1/c} > 1$.)
In general, then, we can remove the condition on~$q$ so long as we replace the robustness of $\poly(\eps)$ with $\poly(\eps, dm/q^{1/c})$, which holds in both cases.
\end{remark}

We will use the following proposition about polynomial measurements several times.

\begin{proposition}\label{prop:same-on-point-same-on-subspace}
Let $d > 0$ be an integer.
Consider a distribution $\calD$ on pairs $(\bs, \bu)$, where~$\bs$ is a subspace in $\F_q^m$ and~$\bu$ is a uniformly random point in~$\bs$.
Let $\{M^s_f\}$ be a measurement whose outcomes are degree-$d$ polynomials $f:s \rightarrow \F_q$,
and let $G \in \polymeas{m}{d}{q}$.
Suppose that
\begin{equation*}
M^s_{[f(u) = b]} \otimes I_{\reg{Bob}} \consistency_\delta I_{\reg{Alice}} \otimes G_{[g(u) = b]}
\end{equation*}
with respect to~$\calD$. Then
\begin{equation*}
M^s_{f} \otimes I_{\reg{Bob}} \consistency_{\delta+d/q} I_{\reg{Alice}} \otimes G_{[g|_s = f]}
\end{equation*}
with respect to~$\calD$.
\end{proposition}
\begin{proof}
Suppose the verifier (i)  samples $(\bs, \bu) \sim \calD$,
(ii) gives Alice~$\bs$, who measures with $M^{\bs}$ and returns her outcome~$\boldf : \bs \rightarrow \F_q$,
(iii) receives~$\bg : \F_q^m \rightarrow \F_q$ from Bob, sampled via~$G$,
and (iv) accepts if $\boldf(\bu) = \bg(\bu)$.
Then by assumption, the verifier accepts with probability at least $1-O(\delta)$.

We can use this to bound the probability that $\boldf$ and $\bg$ disagree on the subspace~$\bs$.
By Schwartz-Zippel (\Cref{lem:schwartz-zippel}), conditioned on $\boldf$ and $\bg$ disagreeing,
the probability they disagree on a random point $\bu \sim \bs$  is at least $1-d/q$.
This gives us the inequality $\Pr[\boldf \neq \bg|_{\bs}] \cdot (1-d/q) \leq O(\delta)$.
Now, assume first that $q \geq 2d$. Then this bound implies, via \Cref{fact:agreement}, that
\begin{equation}\label{eq:almost-there-almost-there}
M^s_f \otimes I_{\reg{Bob}} \approx_{\delta + d/q} I_{\reg{Alice}} \otimes G_{[g|_{s} = f]}.
\end{equation}
On the other hand, when $q \geq 2d$, then this bound is also true for trivial reasons.
This is because we can pick $\delta(\cdot)$ such that $\delta(\eps) \geq 1$ in this case.
\end{proof}

\subsection{Quantum soundness of the classical simultaneous low-degree test}\label{sec:q-simultaneous}

We would now like to use \Cref{thm:anand-thomas-classical-low-degree}
to show quantum soundness for the simultaneous classical low-degree test.
This will be done using the same reduction presented in~\Cref{sec:simultaneous-classical}.
The main result is the following.

\begin{theorem}[Quantum soundness of the simultaneous classical
  low-degree test]
  \label{thm:simultaneous-ldt}
There exists a constant $c > 0$ and a function $\delta(\eps) = \mathrm{poly}(\epsilon, d (m + \ell)/q^c)$ such that the following holds.
Suppose Alice and Bob are entangled provers who pass $\game_{\mathrm{Surface}}^\ell(m,d,q,2)$ with probability at least $1-\eps$
using the strategy $(\psi, M)$, where~$M$ consists of projective measurements.
Then there exists a measurement $G \in \simulpolymeas{m}{d}{q}{\ell}$ such that
\begin{equation*}
M^w_{b_1, \ldots, b_\ell} \otimes I_{\reg{Bob}} \consistency_{\delta(\eps)} I_{\reg{Alice}} \otimes G_{[g_1(w), \ldots, g_\ell(w) = b_1, \ldots, b_\ell]},
\qquad
G_{g_1, \ldots, g_\ell} \otimes I_{\reg{Bob}}
\consistency_{\delta(\eps)}
I_{\reg{Alice}} \otimes G_{g_1, \ldots, g_\ell},
\end{equation*}
where the first is on the uniform distribution over $\F_q^m$.
\end{theorem}

\begin{proof}
Suppose Alice and Bob pass $\game_{\mathrm{Surface}}^\ell(m, d, q,2)$ with probability at least $1-\eps$.
We will use them to simulate two provers, ``Combined Alice" and ``Combined Bob", who pass
the single-function low-degree test $\game_{\mathrm{Surface}}(\ell + m, d+1, q,2)$ with probability at least $1-\eps$.
They are specified as follows:
\begin{itemize}
\item[$\circ$] \textbf{Combined Alice:} Given $\bs \subseteq \F_q^{\ell + m}$, draw $\bs' \sim_2 \bs_{\mathrm{proj}}$.  Give it to Alice, who responds with $\boldf_1, \ldots, \boldf_{\ell}: \bs' \rightarrow \F_q$.  Output the function $\mathrm{combine}_{\boldf}|_{\bs}$.
\item[$\circ$] \textbf{Combined Bob:} Given $(\bx, \by) \in \F_q^{\ell + m}$, compute $\by \in \F_q^m$. Give it to Bob, who responds with $\bb_1, \ldots, \bb_\ell \in \F_q$.
			Return $\mathrm{combine}_{\bb}(\bx) \in \F_q$.
\end{itemize}
By \Cref{prop:really-the-same-dist}, $\bs'$ and $\by$ are distributed as the questions in $\game_{\mathrm{Surface}}^\ell(m, d, q,2)$.
Using our assumption on Alice and Bob, this means that $\boldf_1(\by) = \bb_1$, \ldots, $\boldf_\ell(\by) = \bb_\ell$ with probability at least $1-\eps$.
As a result, $(\mathrm{combine}_{\boldf}|_{\bs})(\bx, \by) = \mathrm{combine}_{\bb}(\by)$ with probability at least~$1-\eps$.
By \Cref{prop:sub-subspace}, $\mathrm{combine}_{\boldf}|_{\bs}$ is a degree-$(d+1)$ function on~$\bs$,
and so it is a valid response to subspace queries.
This means Combined Alice and Bob pass~$\game_{\mathrm{Surface}}(\ell + m, d+1, q, 2)$ with probability at least~$1-\eps$.

Thus, we can apply \Cref{thm:anand-thomas-classical-low-degree}.
It gives a measurement $G \in \polymeas{\ell+m}{d+1}{q}$ such that
\begin{equation}\label{eq:g-meas-equals-combine-meas}
M^y_{[\mathrm{combine}_b(x)=\nu]} \otimes I_{\reg{Bob}}
\consistency_{\delta(\eps)} I_{\reg{Alice}} \otimes G_{[g(x, y) = \nu]},
\qquad
G_g \otimes I_{\reg{Bob}} \consistency_{\delta(\eps)} I_{\reg{Alice}} \otimes G_g,
\end{equation}
where $\delta(\eps) = \poly(\eps, (d+1)(\ell+m)/q^c)$.
This means that if we give Alice~$\by$ and she returns~$\bb$, and Bob simply returns~$\bg$,
then $\mathrm{combine}_{\bb}(\bx) = g(\bx, \by)$ with probability at least~$1-\delta(\eps)$.

We would like to show that~$\bg$ is exactly linear in~$x$ with high probability, over the randomness in the measurement~$G$.
Let us condition on a $\bg$ which is not exactly linear.
By \Cref{prop:exactly-linear-prop}, the probability that $\bg|_{\by}$ is not exactly linear is at least $1-(d+1)/q$.
On the other hand, because $\mathrm{combine}_{\bb}(\bx)$ is always exactly linear by construction,
the probability that $\bg|_{\by}(\bx) = \mathrm{combine}_{\bb}(\bx)$ is at most $(d+1)/q$ by Schwartz-Zippel (\Cref{lem:schwartz-zippel}).
As a result, the probability that $\bg(\bx, \by) = \mathrm{combine}_{\bb}(\bx)$ is at most $(d+1)/q + (d+1)/q$.
Thus, if we write $\mu_{\mathrm{linear}}$ for the probability that $\bg$ is exactly linear, we have equality at most $\mu_{\mathrm{linear}} + 2(d+1)/q$ fraction of the time.
Rearranging, $\bg$ is exactly linear with probability
\begin{equation*}
\mu_{\mathrm{linear}} \geq 1 - 2\delta(\eps) - 2(d+1)/q.
\end{equation*}
Define a new measurement $\{H_{g_1, \ldots, g_\ell}\} \in \simulpolymeas{m}{d}{q}{\ell}$ operationally as follows: first, measure~$G$ and receive~$\bg$.
If it is exactly linear, it can be written as $\sum_i x_i \cdot \bg_i(y)$, and so output $\bg_1, \ldots, \bg_\ell$.
If~$\bg$ is \emph{not} exactly linear, output any arbitrary degree-$d$ polynomials instead.
When~$\bg$ is exactly linear, we have $\mathrm{combine}_{\bg_1, \ldots, \bg_\ell}(x, y) = \bg(x, y)$.
Since this happens with probability at least $1-\delta(\eps)$,
we can replace~$G$ with~$H$ in \Cref{eq:g-meas-equals-combine-meas}, yielding
\begin{equation}\label{eq:no-name-is-a-bad-name}
M^y_{[\mathrm{combine}_b(x)=\nu]} \otimes I_{\reg{Bob}}
\consistency_{\delta(\eps)} I_{\reg{Alice}} \otimes H_{[\mathrm{combine}_g(x, y) = \nu]},
\end{equation}
On the other hand, because~$H$ is just~$G$ with data processing applied to its output,
we can apply \Cref{fact:specialize-the-simeq} to \Cref{eq:g-meas-equals-combine-meas}. This produces the equation
\begin{equation*}
H_{g_1, \ldots, g_\ell} \otimes I_{\reg{Bob}} \consistency_{\delta(\eps)} I_{\reg{Alice}} \otimes H_{g_1, \ldots, g_\ell}.
\end{equation*}

Consider $\bg_1, \ldots, \bg_\ell$ drawn by Bob using~$H$.
For any fixed~$b$ and~$y$, if it is not the case that $g_1(y) = b_1$, \ldots, $g_\ell(y) = b_\ell$,
then the probability that $\mathrm{combine}_{\bg}(\bx, y) = \mathrm{combine}_{b}(\bx)$
over a random~$\bx$ is at most $1/q$ by Schwartz-Zippel, since both are exactly linear functions.
Thus, if Alice draws~$\bb_1, \ldots, \bb_\ell$ given~$\by$,
and we write~$\eta$ for the probability that $\bg_1(\by) = \bb_1$, \ldots, $\bg_\ell(\by) = \bb_\ell$,
then the probability that $\mathrm{combine}_{\bg}(\bx, \by) = \mathrm{combine}_{\bb}(\bx)$ is at most
$\eta +  (1-\eta) \cdot 1/q$.
Combined with \Cref{eq:no-name-is-a-bad-name}, this implies that
\begin{equation*}
\Pr[\bg_1(\by) = \bb_1, \ldots, \bg_\ell(\by) = \bb_\ell] \geq 1- \delta(\eps) - 1/q.
\end{equation*}
Or, equivalently,
\begin{equation*}
M^y_{b_1, \ldots, b_\ell} \otimes I_{\reg{Bob}}
\consistency_{\delta(\eps)} I_{\reg{Alice}} \otimes G_{[g_1(y), \ldots, g_\ell(y) = b_1, \ldots, b_\ell]}.\qedhere
\end{equation*}
\end{proof}

\subsection{Self-testing}

The games presented in \Cref{sec:classical-q-low-deg,sec:q-simultaneous} might be referred to as ``measurement testers":
if a strategy passes them with high probability, then we can extract some property on its measurements.
In this section, we will introduce a significantly stronger notion of testing called \emph{self-testing}.
A self-tester is a game in which if a prover passes with high probability,
then not only do we do exactly which measurements the prover must be performing,
we also know which exactly state it must be performing them on (up to local isometry).
(We note that some works use ``self-testing" to refer both to ``measurement testing" and what we refer to as ``self-testing"~\cite{NV18a}.
In this work, we will reserve the term exclusively for the latter.) We begin with a definition.

\begin{definition}
We say that $\mathcal{S} = (\psi, M)$ is a \emph{partial strategy} for $\game$
if $M$ contains the POVM $M^x$ for only a subset of the questions in~$\game$.
We call this set of questions $\mathcal{S}$'s \emph{question set}.
A strategy $\mathcal{S}' = (\psi, M')$ \emph{extends} $\mathcal{S}$ if $(M')^x = M^x$ for every~$x$ in $\mathcal{S}$'s question set.
\end{definition}

Next, we define self-testing.

\begin{definition}[Self-testing]
Let $\mathcal{S} = (\psi, G)$ be a partial strategy and $\mathcal{D}$ be a distribution over its question set.
A game $\game$ is a \emph{self-test for $\mathcal{S}$ over $\mathcal{D}$ with robustness $\delta(\epsilon)$}
if it satisfies the following two conditions.
\begin{itemize}
\itemsep -.5pt
\item[$\circ$] \textbf{Completeness:} There exists a (full) strategy $\mathcal{S}_{\mathrm{full}}$ consistent with $\mathcal{S}$ which passes~$\game$ with probability~$1$.
\item[$\circ$] \textbf{Soundness:} Let $\overline{\mathcal{S}} = (\overline{\psi}, \overline{M})$ be a strategy which passes~$\game$ with probability $1-\eps$.
	Then there exists a local isometry $\phi = \phi_{\mathrm{local}} \otimes \phi_{\mathrm{local}}$ and a state $\ket{\mathrm{aux}}$ such that
	\begin{equation*}
		\Vert \phi \ket{\overline{\psi}} - \ket{\psi}\ket{\mathrm{aux}} \Vert^2 \leq \delta(\epsilon).
	\end{equation*}
	\ignore{and
	\begin{equation*}
		\E_{(\bx, \bx') \sim\game} \sum_{a, a'} \Vert \phi (M_a^{\bx} \otimes M_{a'}^{\bx'} \ket{\psi'})
				- (M^x_a \otimes M^x_{a'} \ket{\psi}) \ket{\mathrm{aux}} \Vert^2 \leq \delta(\epsilon).
	\end{equation*}}
	Furthermore, if we define the new matrices $M_a^x := \phi_{\mathrm{local}} \cdot \overline{M}_a^x \cdot (\phi_{\mathrm{local}})^\dagger$, then
	\begin{equation}\label{eq:cite-once-then-never-again}
	M_a^x \otimes I_{\reg{Bob}} \approx_{\delta(\epsilon)} (G_a^x \otimes I_{\mathrm{aux}}) \otimes I_{\reg{Bob}},
	\end{equation}
	on states $\ket{\psi}\ket{\mathrm{aux}}$ and $\ket{\psi'}$ and distribution $\bx \sim \calD$.
\end{itemize}
\end{definition}

We note that this definition of self-testing differs in several key places from the one given in \cite[Definition~$2.5$]{NV18a}.
We will explain these differences in more detail when we cite the quantum low-degree test in \Cref{sec:self-test-pauli}.


\part{Implementing the registers}

\label{part:stack}


\section{Register overview}\label{sec:register-overview}

In this part, we implement the quantum registers.
Our goal is force Alice and Bob to share a state of the following form:
\newcommand{\attemptedwidth}{1cm}
\begin{center}
\begin{tabular}{|K{\attemptedwidth}|K{\attemptedwidth}|K{\attemptedwidth}|K{\attemptedwidth}|K{\attemptedwidth}|K{.5cm}|K{\attemptedwidth}|}
\cline{1-2}\cline{4-5}\cline{7-7}
$r_1$ & $r_2$ & $\cdots$ & $r_{k-1}$ & $r_k$ & $\otimes$ & $\mathrm{aux}$\\
\cline{1-2}\cline{4-5}\cline{7-7}
\end{tabular}~,
\end{center}
in which each register $r_i$ contains an EPR state,
and $\mathrm{aux}$ is a symmetric auxiliary state.
In addition, we want the verifier to be able to (i) force
the provers to perform Pauli basis queries on some of these registers and report back the outcomes
and (ii) ``hide" the remaining registers from the provers so that they do not measure them at all.

\subsection{Definitions}\label{sec:protocol-defs}

In this section,
we will begin by defining quantum registers
for \emph{nonuniform} games.
Defining registers for uniform games~$\game$ is a little more
complicated because we allow the number and size of registers for $\game(\mathsf{input})$
to depend on $\mathsf{input}$.
We detail this below in \Cref{sec:uniform-registers}.

\begin{definition}\label{def:register}
Let $k \geq 0$ be an integer, and let $n = (n_1, \ldots, n_k)$ and $q = (q_1, \ldots, q_k)$ be $k$-tuples of integers.
A \emph{$(k, n, q)$-register game} $\game$ is defined as follows.
\begin{itemize}
\itemsep -.5pt
\item[$\circ$] Questions~$x$ are formatted into two blocks $x = (x_1, x_2)$. The first block contains a list of~$k$ Pauli basis queries
			$x_1 = (W_1, \ldots, W_k)$, where each $W_i \in \{X, Z, \hideq, \noop\}$.
\item[$\circ$] Answers~$a$ are formatted into two blocks $a = (a_1, a_2)$. The first block contains a list of answers to the Pauli basis queries $a_1 = (u_1, \ldots, u_k)$.
			Here each $u_i  \in \F_{q_i}^{n_i} \cup \{\varnothing\}$.
\end{itemize}
An \emph{$(k, n, q)$-register strategy} $\calS$ is defined as follows.
\begin{itemize}
\itemsep -.5pt
\item[$\circ$] Alice and Bob share a state
\begin{equation*}
\ket{\psi} = \ket{r_1} \otimes \cdots \otimes \ket{r_k} \otimes \ket{\mathrm{aux}}.
\end{equation*}
Here, $\ket{r_i} = \ket{\mathrm{EPR}_{q_i}^{n_i}}$ for each~$i$, and $\ket{\mathrm{aux}}$ is an arbitrary symmetric shared state.
\item[$\circ$] Given a question $x = (x_1, x_2)$ with first block $x_1 = (W_1, \ldots, W_k)$, Alice and Bob act as follows.
	Let $i \in [k]$.
	\begin{itemize}
	\item If $W_i \in \{X, Z\}$, they measure $\tau^{W}$ on the $i$-th EPR register and set $u_i$ to be the outcome.
	\item If $W_i \in \{\hideq, \noop\}$, they set $u_i = \varnothing$.
	\end{itemize}
	Introduce the notation $\tau_{\varnothing}^W = I$ for $W \in \{\hideq, \noop\}$.
	We can write their measurement as
	\begin{equation}\label{eq:in-math}
		M^{x_1, x_2}_{a_1} = \tau^{W_1}_{u_1} \otimes \cdots \otimes \tau^{W_k}_{u_k} \otimes I_{\reg{aux}}.
	\end{equation}
	To produce the second part of their answer~$a_2$, Alice and Bob can measure any part of their state except the EPR registers which have been ``hidden".	
	This entails the following: let $S = \{i \mid W_i = \hideq\}$.  Then for any answer~$a$, the corresponding POVM acts as follows:
	\begin{equation}\label{eq:hide-coords-in-S}
	M_a^x = M_{\overline{S}} \otimes I_{S}.
	\end{equation}
	Here, $I_{S}$ is the identity matrix on the EPR registers in~$S$,
	whereas $M_{\overline{S}}$ is a POVM acting on the EPR registers in~$\overline{S}$ as well as the state~$\ket{\mathrm{aux}}$.
\end{itemize}
We define $\valreg{k,n,q}{\game}$ to be the maximum over $\valstrat{\game}{\calS}$, where $\calS$ is any $(k, n, q)$-register strategy.
\end{definition}

The~$X$ and~$Z$ questions specify the corresponding Pauli basis measurement,
and the~$\hideq$ question specifies that the register is to be hidden.
The~$\noop$ question is a ``no-op" and does not restrict Alice and Bob at all, other than making them respond with the ``no-op" answer $\varnothing$.
Thus, unlike with the data hiding question, they are allowed to measure the register as they see fit.
This will be useful later when we want Alice and Bob to measure both~$X$ \emph{and}~$Z$ observables on the same register.

In designing our compiler,
it will be convenient to define a set of strategies called ``semiregister strategies".
These will be strategies which are intermediate between $(k-1)$-register strategies and $k$-register strategies
in the sense that they have Pauli basis queries implemented on the final ($k$-th) register but not data hiding queries.
These are defined as follows.

\begin{definition}
A \emph{$(k, n, q)$-semiregister strategy} is defined just as a $(k, n, q)$-register strategy, with the following modification:
the set~$S$ used in \Cref{eq:hide-coords-in-S} is changed to be $S = \{i \neq k \mid W_i = \hideq\}$.
We define $\valsemi{k,n,q}{\game}$ to be the maximum over $\valstrat{\game}{\calS}$, where $\calS$ is any $(k, n, q)$-semiregister strategy.
\end{definition}
\noindent
Thus, querying the $k$-th register of a semiregister strategy with
a~$\hideq$ is the same as querying it with a~$\noop$.

The following lemma shows that we can restrict to \emph{projective}
register strategies without loss of generality.
\begin{lemma}\label{lem:proj-suffices}
  Let $\calS$ be a $(k, n,q)$-register strategy. Then there exists a
  $(k,n,q)$-register strategy $\calS'$ in which all measurements are
  projective, and which produces the same bipartite correlation as $\calS$.
\end{lemma}
\begin{proof}
  Start with the strategy $\calS$, and let the measurements be denoted
  $M^{x_1, x_2}_{a_1, a_2}$. From the definition of register
  strategies, we know that for every set of questions $x_1, x_2$, the
  corresponding measurement can be written as a product
  \[ M^{x_1, x_2}_{a_1,a_2} = (\tau^{W_{x_1}}_{a_1})_{S} \ot (A^{x_1, x_2, a_1}_{a_2})_{\ol{S}}, \]
  where $S$ is the set of registers which receive a Pauli basis query
  in the set $X, Z, H$, $\ol{S}$ is its complement, and the operators
  $\{A^{x_1, x_2, a_1}_{a_2}\}$  form valid POVMs with outcomes $a_2$
  for every choice of $x_1, x_2, a_1$.
  We will apply Naimark's theorem~\Cref{thm:naimark} using the
  universal auxiliary state $\ket{\rmaux}$ to the $A$ operator to produce projectors
  $A'^{x_1, x_2, a_1}_{a_2}$. Using these, we define a projective measurement
  \[ M'^{x_1, x_2}_{a_1, a_2} = (\tau^{W_{x_1}}_{a_1})_{S} \ot
    (A'^{x_1, x_2, a_1}_{a_2})_{\ol{S}}. \]
  It is not hard to see that $M'$ and $\ket{\rmaux}$ form a valid
  Naimark dilation of $M$.
  Let $\calS'$ be the strategy $\calS$ with the shared state
  $\ket{\psi}$ replaced by $\ket{\psi} \ot \ket{\rmaux_A} \ot
  \ket{\rmaux_B}$ and the measurements $M$ replaced by $M'$. By
  construction, $\calS'$ is a projective strategy. Further,
  from~\Cref{cor:bipartite-naimark}, it follows that the bipartite
  correlations produced by the strategies $\calS'$ and $\calS$ are the same.
\end{proof}
\subsection{Results}\label{sec:register-results}

The key elements of our compiler are two new nonlocal games
called the Pauli basis test and the data hiding game.
The Pauli basis test ensures that the provers share an EPR state and honestly answer Pauli basis queries to this state.
The data hiding game allows us to ``hide" this state from the provers, ensuring that they do not use this register unless we ask them to.

Our compiler operates a register at a time and involves two subroutines, $\calC_{k \rightarrow \mathrm{semi}}$ and $\calC_{\mathrm{semi}\rightarrow k-1}$.
Given a $k$-register game, $\calC_{k\rightarrow \mathrm{semi}}$ produces a $k$-semiregister game.
To do so, it removes the guarantee that the provers data hide the $k$-th register and replaces it by playing the data hiding game on this register.
Thus, although the provers are no longer forced to hide the $k$-th register, they will have to do so anyway if they want to pass the data hiding game.
Similarly, given a $k$-semiregister game, $\calC_{\mathrm{semi}\rightarrow k-1}$ produces a $(k-1)$-register game.
To do so, it removes the guarantee that the provers have a $k$-th EPR register and replaces it by playing the Pauli basis test.
Thus, by alternating these two subroutines, we can compile a $k$-register game into a $0$-register game, i.e.\ a general game.

Before giving the properties of the Pauli basis compiler, we will need two definitions.

\begin{definition}
Given a string $x = (x_1, \ldots, x_k)$ and an integer $0 \leq \ell \leq k$, write $\shorten{x}{\ell} := (x_1, \ldots, x_\ell)$.
We extend this to register parameters $\tau = (k, n, q)$ by setting 
$\shorten{\tau}{\ell} := (\ell, \shorten{n}{\ell}, \shorten{q}{\ell})$.
Thus, $\shorten{\tau}{\ell}$ is the register parameters for the first~$\ell$ registers of~$\tau$.
\end{definition}

\begin{definition}
Let~$n$ and~$q$ be integers and~$\eta$ be a real number. We say they \emph{satisfy the Pauli basis condition} if
\begin{equation*}
q = 2^t,\qquad
\frac{1}{\mathrm{\poly}(n)} \leq \eta \leq \frac{1}{2},
\qquad \frac{64 \log(n)^2}{\eta^2} \leq q \leq \poly(n).
\end{equation*}
\end{definition}

The following theorem describes the Pauli basis compiler.

\begin{theorem}\label{theorem:semi-to-k-1-layer}
Let $\register = (k, n, q)$, and let $n_k$, $q_k$, and $\eta$ satisfy the Pauli basis condition.
Suppose $\game_{\mathrm{semi}}$ is a $\register$-semiregister game,
and consider the $\shorten{\register}{k-1}$-register game $\game_{k-1} = \calC_{\mathrm{semi} \rightarrow (k-1)}(\game_{\mathrm{semi}})$.
\begin{itemize}
\item[$\circ$] \textbf{Completeness:} Suppose there is a value-$1$ $\register$-semiregister strategy for $\game_{\mathrm{semi}}$ which is also a real commuting EPR strategy.
	Then there is a value-$1$ $\shorten{\register}{k-1}$-register strategy for $\game_{k-1}$ which is also a real commuting EPR strategy.
\item[$\circ$] \textbf{Soundness:} If $\valreg{\shorten{\register}{k-1}}{\game_{k-1}}\geq 1-\eps$ then $\valsemi{\register}{\game_{\mathrm{semi}}} \geq 1 -\delta(\eps)$, where $\delta(\eps) = \mathrm{poly}(\eps, \eta)$.
\end{itemize}
Furthermore,
\begin{align*}
\qtime{\game_{k-1}} &= \qtime{\game_{\mathrm{semi}}} + O(\log(n_k)), \\
\qlength{\game_{k-1}} &= \qlength{\game_{\mathrm{semi}}} + O(\log(n_k)),  \\
\atime{\game_{k-1}} &= \atime{\game_{\mathrm{semi}}} + \poly(n_k),\\
\alength{\game_{k-1}} &= \alength{\game_{\mathrm{semi}}} + O(n_k \cdot \log \log(n_k)).
\end{align*}
\end{theorem}

The following theorem describes the data hiding compiler.

\begin{theorem}\label{theorem:k-to-semi-layer}
Suppose $\game_k$ is a $(k, n, q)$-register game,
and consider the $(k, n, q)$-semiregister game $\game_{\mathrm{semi}} = \calC_{k \rightarrow \mathrm{semi}}(\game_k)$.
\begin{itemize}
\item[$\circ$] \textbf{Completeness:} 
	Suppose there is a value-$1$ $(k,n,q)$-register strategy for $\game_{k}$ which is also a real commuting EPR strategy.
	Then there is a value-$1$ $(k,n,q)$-semiregister strategy for $\game_{\mathrm{semi}}$ which is also a real commuting EPR strategy.
\item[$\circ$] \textbf{Soundness:} If $\valsemi{k,n,q}{\game_{\mathrm{semi}}} \geq 1-\eps$ then $\valreg{k,n,q}{\game_{k}} \geq 1 -\delta(\eps)$, where $\delta(\eps) = \mathrm{poly}(\eps)$.
\end{itemize}
Furthermore,
\begin{equation*}
\qtime{\game_{\mathrm{semi}}} = O(\qtime{\game_k}), \quad
\atime{\game_{\mathrm{semi}}} = O(\atime{\game_k}),
\end{equation*}
\begin{equation*}
\qlength{\game_{\mathrm{semi}}} = O(\qlength{\game_k}), \quad
\alength{\game_{\mathrm{semi}}} = O(\alength{\game_k}).
\end{equation*}
\end{theorem}

Combining \Cref{theorem:semi-to-k-1-layer,theorem:k-to-semi-layer} gives us the main result of \Cref{part:stack},
a compiler~$\calC$ which compiles $k$-register games into general games.

\begin{theorem}\label{thm:registers}
Let $\game_k$ be a $(k, n, q)$-register game. 
Let $\eta = (\eta_1, \ldots, \eta_k)$, and suppose $n_i$, $q_i$, and $\eta_i$ pass the Pauli basis condition for all $i \in [k]$.
Write
\begin{equation*}
\game = \calC(\game_k) :=
\calC_{\mathrm{semi}\rightarrow 0}(\calC_{\mathrm{1} \rightarrow \mathrm{semi}}(~\cdots~
	\calC_{\mathrm{semi}\rightarrow k-1}(\calC_{\mathrm{k \rightarrow \mathrm{semi}}}(\game_k)))).
\end{equation*}
\begin{itemize}
\item[$\circ$] \textbf{Completeness:} Suppose there is a value-$1$ $(k,n,q)$-register strategy for $\game_{k}$ which is also a real commuting EPR strategy.
	Then there is a real commuting EPR strategy for~$\game$ with value~$1$.
\item[$\circ$] \textbf{Soundness:} If $\val{\game}\geq 1-\eps$ then $\valreg{(k,n,q)}{\game_{k}} \geq 1 -\delta(\eps)$, where $\delta(\eps) = \mathrm{poly}(\eps, \eta_1, \ldots, \eta_k)$.
\end{itemize}
Furthermore,
\begin{align*}
\qtime{\game} &= \qtime{\game_k} + O(\log(n_1))+ \cdots + O(\log(n_k)), \\
\qlength{\game} &= \qlength{\game_k} + O(\log(n_1)) + \cdots + O(\log(n_k)),  \\
\atime{\game} &= \atime{\game_k} + \poly(n_1) + \cdots + \poly(n_k),\\
\alength{\game} &= \alength{\game_k} + O(n_1 \cdot \log \log(n_1)) + \cdots O(n_k \cdot \log \log(n_k)).
\end{align*}
\end{theorem}

\subsection{Registers for uniform games}\label{sec:uniform-registers}

In this section, we generalize the notion of registers to the case of uniform games,
in which a different set of register parameters might be used for each input.
To compile these games, we will need for the register parameters themselves to be uniformly generated.

\begin{definition}\label{def:reg-params-generator}
Let $M_{\mathrm{Params}}$ be a Turing machine which,
given an input $\mathsf{input}$, outputs $\lambda = (k, n, q)$.
Let $\game$ be a (nonuniform) game.
Then we say $M_{\mathrm{Params}}$ \emph{outputs the register parameters of $\game$}
if for every $\mathsf{input}$, $\game(\mathsf{input})$ is a $M_{\mathrm{Params}}(\mathsf{input})$-register game.
\end{definition}

Given this, our compiler for uniform games is given as follows.

\begin{corollary}\label{thm:uniform-registers}
Let $\game(\cdot)$ be a (uniform) game,
and let $M_{\mathrm{Params}}$ be a Turing machine which outputs its register parameters.
Then there exists a (uniform) game $\game_{\mathrm{Compile}}(\cdot)$ with the following properties.
Given an input $\mathsf{input}$, write $\game:=\game(\mathsf{input})$,
$\game_{\mathrm{Compile}} := \game_{\mathrm{Compile}}(\mathsf{input})$,
and $\lambda = (k, n, q) := M_{\mathrm{Params}}(\mathsf{input})$.
\begin{itemize}
\item[$\circ$] \textbf{Completeness:} Suppose there is a value-$1$ $(k,n,q)$-register strategy for $\game$ which is also a real commuting EPR strategy.
	Then there is a real commuting EPR strategy for~$\game_{\mathrm{Compile}}$ with value~$1$.
\item[$\circ$] \textbf{Soundness:}
Let $\eta = (\eta_1, \ldots, \eta_k)$, and suppose $n_i$, $q_i$, and $\eta_i$ pass the Pauli basis condition for all $i \in [k]$.
If $\val{\game_{\mathrm{Compile}}}\geq 1-\eps$ then $\valreg{\lambda}{\game} \geq 1 -\delta(\eps)$, where $\delta(\eps) = \mathrm{poly}(\eps, \eta_1, \ldots, \eta_k)$.
\end{itemize}
Furthermore,
\begin{align*}
\qtime{\game} &= \qtime{\game_k} + O(\log(n_1))+ \cdots + O(\log(n_k)) + \mathsf{time}(M_{\mathrm{Params}}(\mathsf{input})), \\
\qlength{\game} &= \qlength{\game_k} + O(\log(n_1)) + \cdots + O(\log(n_k)),  \\
\atime{\game} &= \atime{\game_k} + \poly(n_1) + \cdots + \poly(n_k)+ \mathsf{time}(M_{\mathrm{Params}}(\mathsf{input})),\\
\alength{\game} &= \alength{\game_k} + O(n_1 \cdot \log \log(n_1)) + \cdots O(n_k \cdot \log \log(n_k)).
\end{align*}
\end{corollary}

\begin{proof}
We first compute $\lambda = M_{\mathrm{Params}}(\mathsf{input})$
in time $\mathsf{time}(M_{\mathrm{Params}}(\mathsf{input}))$.
Then it can be checked that the compiled game $\calC(\game(\mathsf{input}))$ from \Cref{thm:registers}
can be efficiently simulated given the register parameters~$\lambda$.
\end{proof}

\subsection{Organization}

The remainder of \Cref{part:stack} is organized as follows.
\begin{itemize}
\item In \Cref{sec:self-test-pauli}, we introduce the Pauli basis self-test and prove its correctness.
\item \Cref{sec:pauli-basis-compiler} implements the Pauli basis compiler.
\item In \Cref{sec:data-hiding-layer}, we introduce the data hiding game.
\item \Cref{sec:compile-sec} implements the data hiding compiler.
\item \Cref{sec:rotated-data-hiding} contains a generalization of the data hiding game which allows us to hide more general sets of Pauli observables.
		This is not needed to implement the quantum registers, but it will be needed in \Cref{part:neexp} when designing the $\neexp$ protocol.
\end{itemize}



\section{A self test for the Pauli basis}\label{sec:self-test-pauli}

In this section, we give a self test for the Pauli basis measurement.
Given~$W \in \{X, Z\}$, this test compels the prover
to measure an EPR register in the~$W$ basis
and return the outcome to the verifier.

\begin{definition}
The \emph{Pauli basis strategy with parameters $n$ and $q$} (a prime power),
denoted $\paulistrat{n}{q}$, is the partial strategy
with the state $\ket{\mathrm{EPR}_q^n}$ and measurement matrices $\tau^W_u$ for each $W \in \{X,Z\}, u \in \F_q^n$.
\end{definition}

The main result of this section is the following self-test for the case when~$q$ is a power of~$2$.

\begin{theorem}\label{thm:basis-test}
Let $\bW \sim \{X, Z\}$ uniformly at random.  Let $n$, $q$, $\eta$ satisfy the Pauli basis condition.
Then there is a self-test $\game_{\mathrm{basis}} := \game_{\mathrm{basis}}(n,q)$ for $\paulistrat{n}{q}$ over~$\bW$
with robustness $\delta(\eps) = \poly(\eps, \eta)$.
Moreover, there is a value-$1$ real commuting EPR strategy with auxiliary state $\ket{\mathrm{EPR}_2}$.
Finally,
\begin{equation*}
\qlength{\game_{\mathrm{basis}}} = O\left(\log(n)\right),
\quad
\alength{\game_{\mathrm{basis}}} = \poly(n),
\end{equation*}
\begin{equation*}
\qtime{\game_{\mathrm{basis}}} = O\left(\log(n)\right),
\quad
\atime{\game_{\mathrm{basis}}} = \poly(n).
\end{equation*}
\end{theorem}

We prove this by a straightforward reduction to the quantum low-degree test of~\cite{NV18a}.

\subsection{The quantum low-degree test}

The goal of the quantum low-degree test of~\cite{NV18a}
is to force the provers to use a ``compressed" version of the Pauli basis strategy.
Given~$\bW$, they should measure their register in the $\bW$ basis, receiving $\bu \in \F_q^n$.
However, $\bu$, a length-$n$ string, might be prohibitively expensive to communicate to the verifier,
so they should instead compute the low degree encoding $g_{\bu}$
and return its evaluation at a single point $\bw \in \F_q^m$ of the verifier's choosing.
(The point of this section is to ``uncompress" their protocol.)

\begin{definition}
Fix parameters for the low-degree encoding $\params := (q = p^t, h, H, m, n, \pi)$
satisfying the ``low-degree conditions"
$h \leq q$, and $n \leq h^m$.
For any string $u \in \F_q^n$,
these parameters give a low-degree encoding $g_u:\F_q^m \rightarrow \F_q$.

The \emph{low-degree Pauli strategy with parameters~$\params$},
denoted $\mathcal{LD}(\params)$, is the partial strategy
with state $\ket{\mathrm{EPR}^n_q}$ and measurement matrices
\begin{equation*}
\tau^{W, w}_a :=
\tau^{W}_{[g_u(w) = a]} =
\sum_{u : g_u(w) = a} \tau^W_u
\end{equation*}
for each $W \in \{X, Z\}, w \in \F_q^m, a \in \F_q$.
Equivalently, this is the strategy where we perform the Pauli $W$-basis measurement
and output the low-degree encoding of the outcome~$u$ evaluated at the point~$w$, i.e.\ the value $g_u(w)$.
\end{definition}

The main result of~\cite{NV18a} is the following.

\begin{theorem}[{\cite[Theorem 3.2]{NV18a}}]\label{thm:anand-thomas-low-degree}
Fix low-degree parameters~$\params$ with $p = 2$ (so that $q = 2^t$) and $m \geq 2$, 
and let $\calD$ be the uniform distribution over $(W, w)$ with $W \in \{X, Z\}, w \in \F_q^m$.
Then there is a self-test $\game_{\mathrm{Qlowdeg}} := \game_{\mathrm{Qlowdeg}}(\params)$
for $\mathcal{LD}(\params)$ over~$\calD$
with robustness
$
\delta(\epsilon) = \mathrm{poly}(\epsilon,md/q^c),
$
with $c > 0$.
Moreover, there is a value-$1$ real commuting EPR strategy with auxiliary state $\ket{\mathrm{EPR}_2}$.
Finally,
\begin{equation*}
\qlength{\game_{\mathrm{Qlowdeg}}} = O(m \log q),
\quad
\alength{\game_{\mathrm{Qlowdeg}}} = O(d^2 \log (q)),
\end{equation*}
\begin{equation*}
\qtime{\game_{\mathrm{Qlowdeg}}} = O(m \log q),
\quad
\atime{\game_{\mathrm{Qlowdeg}}} = \poly(m, d, \log q).
\end{equation*}
\end{theorem}

(We note that this result is stated in \cite{NV18a} for general primes~$p$.
However, the $p \neq 2$ case relied on a self-testing result for a generalization of the Magic Square game
which was recently discovered to contain a bug. 
Fortunately, the $p=2$ case needs only a self-testing result for the
``traditional" binary Magic Square game, and this follows
from~\cite{WBMS16}.)

\begin{remark}
We note that the quantum low-degree test, as stated in \cite{NV18a}, does \emph{not} have value-$1$ real commuting EPR strategies.
This is because it uses as a subroutine the standard magic square game,
and the magic square game does not have value-$1$ real commuting EPR strategies.
Its value-$1$ strategies \emph{are} EPR strategies,
and they \emph{are} real (all observables are either $X$ or $Z$, with the sole exception of the $Y \otimes Y$ observable,
which can be rewritten as $Y \otimes Y = -(X \otimes X) \cdot (Z \otimes Z)$, manifestly real).
But they are not commuting, because each row and column have at least one pair of noncommuting observables.

Consider instead the following ``oracularized" version of the magic square game:
one player is given a random row or column (and is expected to play as in the normal magic square game),
and the other player is given a random cell in that row or column,
and the verifier simply checks that they agree on that cell.
In addition, with some constant probability, both players are given the same cell and their answers are checked against each other.
In this case, all observables measured are commuting, and so this variant has a value-$1$ real commuting EPR strategy.
In addition, it certifies the same state and measurements as the normal magic square game,
and so we can use it as a subroutine in the quantum low-degree test instead.
\end{remark}

\begin{remark}
We note again that the soundness case in our definition of a self-test is quite different from the one given in \cite[Definition~$2.5$]{NV18a},
and it is not clear that a self-test in their sense implies a self-test in our sense.
However, for the quantum low-degree test, their soundness case \emph{does} match ours.
By~\cite[Lemma~$4.1$]{NV18a}, there is a local isometry~$\phi = \phi_1 \otimes \phi_2$ such that
\begin{equation}\label{eq:anand-testing-state}
\Vert \phi \ket{\overline{\psi}} - \ket{\psi}\ket{\mathrm{aux}} \Vert^2 \leq \delta(\epsilon).
\end{equation}
and
\begin{equation}\label{eq:anand-testing-state}
\E_{(\bW, \bw)} \sum_a \Vert \phi \cdot (\overline{M}^{\bW, \bw}_{a} \otimes I_{\reg{Bob}}) \ket{\overline{\psi}}
	 - (\tau_{a}^{\bW, \bw} \otimes I_{\reg{aux}})  \otimes I_{\reg{Bob}} \ket{\psi} \ket{\mathrm{aux}} \Vert^2
\leq \delta(\eps).
\end{equation}
The key difference from our self-test definition is that, as stated, their local isometry need not be symmetric (i.e.\ $\phi_1 \neq \phi_2$),
but their construction actually \emph{does} give a symmetric isometry with $\phi_1 = \phi_2$.
Then, from \Cref{eq:anand-testing-state} it is easy to derive \Cref{eq:cite-once-then-never-again} using \Cref{eq:anand-testing-state} and the triangle inequality (\Cref{fact:triangle}).
\end{remark}

\subsection{Proof of \Cref{thm:basis-test}: the Pauli basis test}

We now state the Pauli basis test.

\begin{definition}
Let $n,q,\eta$ be as in \Cref{thm:basis-test}.
Fix the remaining low-degree parameters~$\params$ as follows:
\begin{equation*}
h = \lceil q^{1/2}\rceil, \qquad m = 2\cdot \left\lceil \frac{\log(n)}{\log(q)}\right\rceil,
\qquad d = m \cdot (h-1).
\end{equation*}
Then the \emph{Pauli basis game}~$\game_{\mathrm{basis}}(n,q)$ is given by~\Cref{fig:pauli-basis}.
\end{definition}

{
\floatstyle{boxed} 
\restylefloat{figure}
\begin{figure}
With probability~$\tfrac{1}{2}$ each, perform one of the following two tests.
\begin{enumerate}
\item \textbf{Low-degree:} Perform $\game_{\mathrm{Qlowdeg}}(\params)$. \label{item:qlowdeg-test}
\item \textbf{Cross-check:}  Draw $\bW \sim \{X, Z\}$, $\bw \sim \F_q^m$. Flip an unbiased coin $\bb \sim \{0, 1\}$.
	Distribute the questions as follows:
	\begin{itemize}
	\item[$\circ$] Player~$\bb$: Give $\bW$; receive $\bu \in \F_q^n$.
	\item[$\circ$] Player~$\overline{\bb}$: give $(\bW,\bw)$; receive $\ba$.
	\end{itemize}
	Accept if $g_{\bu}(\bw) = \ba$.
\end{enumerate} \label{item:check-with-qlowdeg}
\caption{The game $\game_{\mathrm{basis}}(n,q)$.\label{fig:pauli-basis}}
\end{figure}
}

These parameters are chosen so that they are valid low-degree parameters
(guaranteeing the existence of the low-degree code),
which is necessary for the quantum low-degree test.
In particular, these satisfy (i) $h \leq q$ and (ii) $n \leq h^m$.
The first of these is immediate; as for the second,
\begin{equation*}
h^m \geq (q^{1/2})^{2 \cdot \log(n)/\log(q)} = q^{\log(n)/\log(q)} = n.
\end{equation*}
In addition, the code has relative distance $d/q \leq mh/q \leq \eta$.
\begin{equation*}
\frac{d}{q} = \frac{m \cdot(h-1)}{q} \leq \frac{m h}{q} \leq 8\cdot \frac{\log(n)}{\log(q)} \cdot \frac{q^{1/2}}{q} \leq \frac{8 \log(n)}{q^{1/2}} \leq \eta,
\end{equation*}
where the final step is because $n, q, \eta$ satisfy the Pauli basis condition.
Finally, we note that even if~$q$ is a large polynomial of~$n$, $m$ is always at least~$2$, which permits us to use the quantum low-degree test.
We now prove \Cref{thm:basis-test}.

\begin{proof}[Proof of \Cref{thm:basis-test}]
The question lengths and times of both the quantum low-degree test and the cross-check are given by
\begin{equation*}
m \log(q) = 2\cdot \left\lceil \frac{\log(n)}{\log(q)}\right\rceil  \cdot \log(q) = O(\log(n)).
\end{equation*}
As for the answer lengths and times, these are bounded by $\poly(n)$ for both the quantum low-degree test and the cross-check.
We now consider the completeness and soundness cases separately.
\paragraph{Completeness.}
Let~$(\psi, M)$ be the value-$1$ commuting EPR strategy for the quantum low-degree test guaranteed by \Cref{thm:anand-thomas-low-degree}.
This has state $\ket{\psi} = \ket{\mathrm{EPR}_q^n} \ket{\mathrm{EPR}_2}$ and measurement matrices $M^{W,w}_a = \tau^{W,w}_a \otimes I_{\reg{aux}}$.
If we add in the measurement matrices $M^W_u = \tau^W_u \otimes I_{\reg{aux}}$,
then this strategy passes the cross-check with probability~$1$.
This is because after Player~$\bb$ measures~$\bu$,
the state collapses to $\ket{\tau^{\bW}_{\bu}}\ket{\tau^{\bW}_{\bu}} \ket{\mathrm{EPR}_2}$,
and so Player~$\overline{\bb}$ will measure $\ba = g_{\bu}(\bw)$.
As a result, this is a value-$1$ strategy.
Furthermore, it is a commuting EPR strategy because the cross-check measurements $M^W$ and $M^{W,w}$ commute.
Finally, this strategy extends the Pauli basis strategy.
This proves the completeness case.

\paragraph{Soundness.}
Throughout the soundness, we will use $\delta(\eps)$ to denote a function of the form $\poly(\eps, \eta)$ which may change from use to use.
The $\delta(\eps)$ in \Cref{thm:anand-thomas-low-degree} is of this form because $d/q\leq \eta$.

Let $\overline{\calS} = (\overline{\psi}, \overline{M})$ be a strategy with $\valstrat{\game_{\mathrm{basis}}}{\overline{\calS}} = 1-\eps$.
Then this strategy must pass $\game_{\mathrm{Qlowdeg}}$ with probability at least $1-2\eps$.
By \Cref{thm:anand-thomas-low-degree} this gives us a local isometry
$\phi = \phi_{\mathrm{local}}\otimes \phi_{\mathrm{local}}$ and a state $\ket{\mathrm{aux}}$
with the following properties:
if we define the new strategy $\calS$
in which $\ket{\psi} = \phi \ket{\overline{\psi}}$
and $M_a^x = \phi_{\mathrm{local}}\cdot \overline{M}_a^x \cdot \phi_{\mathrm{local}}^\dagger$, then
\begin{equation}\label{eq:cant-think-of-a-good-name}
\Vert \ket{\psi} - \ket{\mathrm{EPR}_q^n} \ket{\reg{aux}} \Vert^2 \leq \delta(\eps),
\qquad
M_a^{W,w} \otimes I_{\reg{Bob}}
\approx_{\delta(\eps)}
(\tau_a^{W,w} \otimes I_{\reg{aux}}) \otimes I_{\reg{Bob}},
\end{equation}
on state $\ket{\psi}$ and distribution $\calD$.
Because $\calS$ is just a rotated version of $\overline{\calS}$,
it also passes $\game_{\mathrm{basis}}$ with probability $1-\eps$.
As a result, $\calS$ passes the cross-check in \Cref{item:check-with-qlowdeg} with probability at least $1-2\eps$.
By \Cref{fact:agreement}, we conclude that
\begin{equation}\label{eq:worth-studying}
M^W_{[g_u(w)=a]} \otimes I_{\reg{Bob}}
\approx_{\eps}
I_{\reg{Alice}} \otimes M^{W,w}_a
\approx_{\delta(\eps)} I_{\reg{Alice}} \otimes (\tau^{W,w}_a \otimes I_{\reg{aux}})
= I_{\reg{Alice}} \otimes (\tau^{W}_{[g_v(w)=a]} \otimes I_{\reg{aux}})
\end{equation}
on state $\ket{\psi}$.
By \Cref{fact:almost-agreement} and the fact that the~$\tau$ measurements are projective, this implies that
\begin{equation*}
M^W_{[g_u(w)=a]} \otimes I_{\reg{Bob}}
\consistency_{\delta(\eps)} I_{\reg{Alice}} \otimes (\tau^{W}_{[g_v(w)=a]} \otimes I_{\reg{aux}})
\end{equation*}
Now by
\Cref{prop:same-on-point-same-on-subspace} (where we let~$\bs$ be the singleton distribution on the ``trivial" subspace $\bs = \F_q^m$)
and the fact that $d/q \leq \eta$, we can conclude that
\begin{equation*}
M^W_{u} \otimes I_{\reg{Bob}}
\consistency_{\delta(\eps)} I_{\reg{Alice}} \otimes (\tau^{W}_{u} \otimes I_{\reg{aux}}).
\end{equation*}
Applying \Cref{fact:agreement} again, this yields
\begin{equation}\label{eq:almost-there-almost-there}
M_{u}^{W} \otimes I_{\reg{Bob}}
\approx_{\delta(\eps)} I_{\reg{Alice}} \otimes (\tau_{u}^{W} \otimes I_{\reg{aux}})
\approx_{\delta(\eps)} (\tau_{u}^W \otimes I_{\reg{aux}}) \otimes I_{\reg{Bob}}
\end{equation}
on state $\ket{\psi}$, where the last step uses \Cref{fact:the-ol-state-y-swaperoonie}
to combine \Cref{fact:the-ol-pauli-swaperoonie} with \Cref{eq:cant-think-of-a-good-name}.
The analogous statement for the state $\ket{\mathrm{EPR}_q^n}$ follows from~\Cref{fact:the-ol-state-y-swaperoonie}.
This establishes the theorem.
\end{proof}



\section{Compiling games with the Pauli basis test}\label{sec:pauli-basis-compiler}

In this section, we show how to use the Pauli basis test to implement the compiler $\calC_{\mathrm{semi}\rightarrow (k-1)}$.
Our construction is given in the following definition.

\begin{definition}
Let $\game_{\mathrm{semi}}$ be a $(k,n,q)$-semiregister game.
Then its compiled version is the game $\calC_{\mathrm{semi} \rightarrow (k-1)}(\game_{\mathrm{semi}})$ defined in \Cref{fig:one-to-zero}.
\end{definition}
{
\floatstyle{boxed} 
\restylefloat{figure}
\begin{figure}
Flip an unbiased coin $\bb \sim \{0, 1\}$. 
With probability~$\tfrac{1}{4}$ each, perform one of the following four tests.
\begin{enumerate}
	\item \textbf{Pauli basis:} Draw $(\bx, \bx') \sim \game_{\mathrm{basis}}(n_k, q_k, \eta)$.\label{item:first-layer-pauli-basis}
			Distribute the questions as follows:
				\begin{itemize}
				\item[$\circ$] Player~$\bb$: give $(\hideq^{k-1}, \bx)$; receive $\ba = (\ba_1, \ba_2)$.
				\item[$\circ$] Player~$\overline{\bb}$: give $(\hideq^{k-1}, \bx')$; receive $\ba' = (\ba_1', \ba_2')$.
				\end{itemize}
			Accept if $\ba_2$ and $\ba_2'$ are accepting answers to the Pauli basis test.
	\item \textbf{Cross-check:} Draw $(\bx, \bx') \sim \game_{\mathrm{semi}}$. Write $\bx = (\bx_1, \bx_2)$
			with $\bx_1 = (\bW_1, \ldots, \bW_k)$.
			\label{item:first-layer-cross-check}
			Distribute the questions as follows:
				\begin{itemize}
				\item[$\circ$] Player~$\bb$: give $\bx$; receive $\ba = (\ba_1, \ba_2)$,
							where $\ba_1 = (\bu_1, \ldots, \bu_k)$.
				\item[$\circ$] Player~$\overline{\bb}$: give $(\hideq^{k-1}, \bW_k)$;
							receive strings $\ba_1' = (\bu_1', \ldots, \bu_k')$, $\bu_i' \in \F_q^n$.
				\end{itemize}
				If $\bW_k \in \{X, Z\}$, accept if $\bu_k = \bu_k'$.  Otherwise, accept if $\bu_k = \varnothing$.
	\item \textbf{Consistency check:} Draw $(\bx, \bx') \sim \game_{\mathrm{semi}}$.
			Distribute the questions as follows:
				\begin{itemize}
				\item[$\circ$] Player~$\bb$: give $\bx$; receive $\ba$
				\item[$\circ$] Player~$\overline{\bb}$: give $\bx$; receive $\ba'$.
				\end{itemize}
				Accept if $\ba = \ba'$.
	\item \textbf{Play game:} Perform $\game_{\mathrm{semi}}$.
	\end{enumerate}
	\caption{The game $\calC_{\mathrm{semi} \rightarrow (k-1)}(\game_{\mathrm{semi}})$.\label{fig:one-to-zero}}
\end{figure}
}

In words, the provers might try to ``trick" the verifier by using one of their $(k-1)$ existing EPR registers to answer queries meant for the new $k$-th register.
To prevent this, the verifier performs the Pauli basis test with the first $k-1$ registers hidden, forcing the provers to introduce a new EPR register.
It then cross-checks the provers' answers in the Pauli basis test with their answers in the game $\game_{\mathrm{semi}}$.
The performance of the compiler is given by the following theorem.

\begin{theorem}\label{theorem:one-to-zero-compiler}
Let $\register = (k, n, q)$, and let $n_k$, $q_k$, and $\eta$ satisfy the Pauli basis condition.
Suppose $\game_{\mathrm{semi}}$ is a $\register$-semiregister game,
and consider the $\shorten{\register}{k-1}$-register game $\game_{k-1} = \calC_{\mathrm{semi} \rightarrow (k-1)}(\game_{\mathrm{semi}})$.
\begin{itemize}
\item[$\circ$] \textbf{Completeness:} Suppose there is a value-$1$ $\register$-semiregister strategy for $\game_{\mathrm{semi}}$ which is also a real commuting EPR strategy.
	Then there is a value-$1$ $\shorten{\register}{k-1}$-register strategy for $\game_{k-1}$ which is also a real commuting EPR strategy.
\item[$\circ$] \textbf{Soundness:} If $\valreg{\shorten{\register}{k-1}}{\game_{k-1}}\geq 1-\eps$ then $\valsemi{\register}{\game_{\mathrm{semi}}} \geq 1 -\delta(\eps)$, where $\delta(\eps) = \mathrm{poly}(\eps, \eta)$.
\end{itemize}
Furthermore,
\begin{align*}
\qlength{\game_{k-1}} &= \qlength{\game_{\mathrm{semi}}} + O(\log(n)),  \\
\qtime{\game_{k-1}} &= \qtime{\game_{\mathrm{semi}}} + O(\log(n)), \\
\alength{\game_{k-1}} &= \alength{\game_{\mathrm{semi}}} + \poly(n),\\
\atime{\game_{k-1}} &= \atime{\game_{\mathrm{semi}}} + \poly(n).
\end{align*}
\end{theorem}

\begin{proof}
The communication and time complexities are the result of combining
the communication and time complexities from $\game_{\mathrm{semi}}$ with the values for the Pauli basis test from \Cref{thm:basis-test}.

\paragraph{Completeness.}
Let $(\psi, M)$ be a value-$1$ $\register$-semiregister strategy for $\game_{\mathrm{semi}}$ which is also a real commuting EPR strategy. 
Then $\ket{\psi} = \ket{r_1} \cdots \ket{r_k} \ket{\mathrm{aux}}$, where each $\ket{r_i} = \ket{\mathrm{EPR}_{q_i}^{n_i}}$ and $\ket{\mathrm{aux}}$ is an EPR state.
In addition, let $(\psi', M')$ be the value-$1$ real commuting EPR strategy for $\game_{\mathrm{basis}}$ guaranteed by \Cref{thm:basis-test}.
Then $\ket{\psi'} = \ket{\mathrm{EPR}_{q_k}^{n_k}} \ket{\mathrm{aux}'}$, where $\ket{\mathrm{aux}'}$ is an EPR state.

Consider the following strategy for $\game_{k-1}$.
For its state, it uses $\ket{r_1}\cdots \ket{r_k}\ket{\mathrm{aux}}\ket{\mathrm{aux}'}$.
For inputs drawn from $\game_{\mathrm{semi}}$, it uses the matrices in~$M$ applied to all but the $\ket{\mathrm{aux}'}$ register.
For inputs of the form $(H^{k-1}, x)$, where $x$ is sampled from $\game_{\mathrm{basis}}$, it outputs $\varnothing^{k-1}$ along with the result of applying $M'$ to $\ket{r_k}$ and $\ket{\mathrm{aux}'}$.
Finally, for inputs of the form $(H^{k-1}, \hideq)$ and $(H^{k-1}, \noop)$, it outputs $\varnothing^k$.
This forms a valid $\shorten{\register}{k-1}$-register strategy for $\game_{k-1}$.
In addition, its ``auxiliary register" is $\ket{r_k} \ket{\mathrm{aux}}\ket{\mathrm{aux}'}$, which is an EPR state.
Now we show that it has value~$1$.

By construction, this strategy passes the Pauli basis test and $\game_{\mathrm{semi}}$ with probability~$1$.
As for the cross-check, when $\bW_k \in \{\hideq, \noop\}$, the strategy always succeeds because $(\psi, M)$ is a $\register$-semiregister strategy.
On the other hand, when $\bW_k \in \{X, Z\}$, this implies that~$\bu_k$ is the result of applying the $\tau^{\bW_k}$ measurement to $\ket{r_k}$,
putting it in state $\ket{\tau^{\bW_k}_{\bu_k}}\ket{\tau^{\bW_k}_{\bu_k}}$.
But then because $(\psi', M')$ implements the Pauli basis strategy on $\ket{r_k}$, the outcome $\bu_k'$ is also the result of applying the $\tau^{\bW_k}$ measurement to $\ket{r_k}$.
As a result, $\bu_k' = \bu_k$.

Finally, it is clear that this forms an EPR strategy.
As a result, by \Cref{fact:heh-heh-heh-gonna-make-anand-prove-this-so-i-can-take-the-day-off},
the consistency check passes with probability~$1$.
Thus, the strategy passes the overall test with probability~$1$.
Next, we show that this gives a \emph{commuting} EPR strategy.
For the questions that arise in the Pauli basis test, the consistency check, and~$\game_{\mathrm{semi}}$,
commutation follows because~$M$ and~$M'$ are commuting.
As for the cross-check, consider the case when $\bW_k \in \{X, Z\}$.
Then  the first (i.e.\ Player~$\bb$'s) measurement is given by
\begin{equation*}
(M^{x_1, x_2}_{a_1, a_2})_{\reg{1, \ldots, k, aux}} \otimes I_{\reg{aux'}}
= \tau^{W_k}_{u_k} \otimes (A^{x_1, x_2}_{u_1, \ldots, u_{k-1}, a_2})_{\reg{1, \ldots, k-1, aux, aux'}},
\end{equation*}
where~$A$ is some measurement.
This follows because~$M$ is a $\register$-semiregister strategy.
Similarly, the second (i.e., Player~$\overline{\bb}$'s) measurement is given by
\begin{equation*}
(\tau_{\varnothing}^H \otimes \cdots \otimes \tau^H_{\varnothing})_{\reg{1, \ldots, k-1}}
\otimes (M'^{W_k}_{u_k'})_{\reg{k, aux'}} \otimes I_{\reg{aux}}
= \tau^{W_k}_{u_k'} \otimes I_{\reg{1, \ldots, k-1, aux, aux'}}.
\end{equation*}
By inspection, these two commute.
On the other hand, when $\bW_k \in \{H, \bot\}$,
then Player~$\overline{\bb}$ always outputs $\varnothing^k$.
Their measurement for this outcome is the matrix $I_{\reg{1, \ldots, k, aux, aux'}}$,
and is the zero matrix for every other outcome.
These clearly commute with any strategy for Player~$\bb$.

Finally, because~$M$ and~$M'$ are real strategies,
this strategy is also real.
As a result, this gives a value-$1$ real commuting EPR strategy.

\paragraph{Soundness.}

Suppose $\calS_{\mathrm{reg}} = (\psi_{\mathrm{reg}}, M_{\mathrm{reg}})$ is a  $\shorten{\register}{k-1}$-register strategy for $\game_{k-1}$ with value $1-\epsilon$.
By \Cref{lem:proj-suffices}, we can assume without loss of generality that~$M$ is projective.
For $1 \leq i \leq k$, write $\ket{r_i} := \ket{\mathrm{EPR}_{q_i}^{n_i}}$.
By definition, $\ket{\psi_{\mathrm{reg}}} = \ket{r_1} \otimes \cdots \otimes \ket{r_{k-1}} \otimes \ket{\mathrm{aux}_{\mathrm{reg}}}$.
Our goal will be to decode $\calS_{\mathrm{reg}}$ into a $\register$-semiregister strategy $\calS_{\mathrm{semi}}$ for $\game_{\mathrm{semi}}$ with nearly the same value.

\paragraph{Using the Pauli basis test.}
Passing the overall test with probability $1-\epsilon$
means that $\calS_{\mathrm{reg}}$ must pass the test in \Cref{item:first-layer-pauli-basis} with probability $1-4\epsilon$.
This test only involves measurements of the form $\{(M_{\mathrm{reg}})^{\hideq,\ldots, \hideq, x}_{a_1, a_2}\}_{a_1, a_2}$.
Because the first $k-1$ coordinates are hidden, \Cref{eq:hide-coords-in-S} allows us to write
\begin{equation*}
(M_{\mathrm{reg}})^{\hideq, \ldots, \hideq, x}_{a_2} = I_{\reg{1, \ldots, k-1}} \otimes (A^x_{a_2})_{\reg{aux}},
\end{equation*}
where $\{A^x_a\}_x$ is some set of measurements.
As a result, the state $\ket{\mathrm{aux}_{\mathrm{reg}}}$ and measurements $\{A^x_a\}_x$ form a strategy 
for the game $\game_{\mathrm{basis}}(n_k, q_k, \eta)$ which succeeds with probability $1-4\epsilon$.
By \Cref{thm:basis-test} this gives us a local isometry
$\phi = \phi_{\mathrm{local}}\otimes \phi_{\mathrm{local}}$ and a state $\ket{\mathrm{aux}}$ such that
\begin{equation}\label{eq:applied-that-fact}
\Vert \phi \ket{\mathrm{aux}_{\mathrm{reg}}} - \ket{r_k} \ket{\mathrm{aux}} \Vert^2 \leq \delta(\eps),
\end{equation}
\begin{equation}\label{eq:applied-that-fact-dos}
(\phi_{\mathrm{local}}\cdot A_{u}^{W} \cdot \phi_{\mathrm{local}}^\dagger)_{\reg{Alice}}\otimes I_{\reg{Bob}}
\approx_{\delta(\eps)}
(\tau_{u}^{W} \otimes I_{\reg{aux}})_{\reg{Alice}} \otimes I_{\reg{Bob}},
\end{equation}
on state $\ket{r_k} \ket{\mathrm{aux}}$ and the uniform distribution on $\{X,Z\}$.

Define the new strategy $\calS$ in which $\ket{\psi} = \ket{r_1} \otimes \cdots \otimes \ket{r_{k-1}} \otimes (\phi \ket{\mathrm{aux}_{\mathrm{reg}}})$
and 
\begin{equation*}
M^x_a = (I_{\reg{1, \ldots, k-1}} \otimes (\phi_{\mathrm{local}})_{\reg{aux}})
	\cdot (M_{\mathrm{reg}})^x_a
		\cdot (I_{\reg{1, \ldots, k-1}} \otimes (\phi^\dagger_{\mathrm{local}})_{\reg{aux}}).
\end{equation*}
Then \Cref{eq:applied-that-fact,eq:applied-that-fact-dos} implies that
\begin{equation}\label{eq:applied-the-fact-again}
\Vert  \ket{\psi} - \ket{r_1}\otimes \cdots \otimes \ket{r_k} \ket{\mathrm{aux}} \Vert^2 \leq \delta(\eps),
\end{equation}
\begin{equation}\label{eq:applied-the-fact-again-dos}
(M^{H, \ldots, H, W}_u)_{\reg{Alice}} \otimes I_{\reg{Bob}}
\approx_{\delta(\eps)}
(I_{\reg{1, \ldots, k-1}} \otimes \tau_{u}^{W} \otimes I_{\reg{aux}})_{\reg{Alice}} \otimes I_{\reg{Bob}},
\end{equation}
on state $\ket{\psi}$ and the uniform distribution on $\{X,Z\}$.
Because $\calS$ is just a rotated version of $\calS_{\mathrm{reg}}$,
it also passes  $\game_{k-1}$ with probability $1-\eps$.
In addition, it is also a $\shorten{\register}{k-1}$-register strategy.

\paragraph{Performing the cross-check.}
To analyze the cross-check, we begin with a definition.
Given $W \in \{X, Z, \hideq, \noop\}$
and $u \in \F_{q_k}^{n_k} \cup \{\varnothing\}$, define
$\mathrm{null}_{W}(u) = u$ if $W \in \{X, Z\}$ and $\varnothing$ otherwise.
The cross-check in \Cref{item:first-layer-cross-check} checks equality between
$\bu_k$ and $\mathrm{null}_{\bW_k}(\bu_k')$.
As a result,
\begin{equation*}
(M^x_{u_k})_{\reg{Alice}} \otimes I_{\reg{Bob}}
\approx_{\eps} I_{\reg{Alice}} \otimes (M^{\hideq, \ldots, \hideq, W_k}_{[\mathrm{null}_{W_k}(u_k') = u_k]})_{\reg{Bob}}.
\end{equation*}
Next, we note that when $W_k \in \{\hideq, \noop\}$,
\begin{equation*}
M^{\hideq, \ldots, \hideq, W_k}_{[\mathrm{null}_{W_k}(u_k') = u_k]}
= I_{\reg{1, \ldots, k-1}} \otimes \tau^{W_k}_{u_k} \otimes I_{\reg{aux}},
\end{equation*}
because both sides are the identity when $u_k = \varnothing$ and zero otherwise.
On the other hand, when $W_k \in \{X, Z\}$, these two are close due to \Cref{eq:applied-the-fact-again-dos}.
Applying~\Cref{fact:average-over-dists} and \Cref{fact:triangle}, we get
\begin{equation}\label{eq:looks-like-tau}
(M^x_{u_k})_{\reg{Alice}} \otimes I _{\reg{Bob}}
\consistency_{\delta(\eps)} I_{\reg{Alice}} \otimes (I_{\reg{1, \ldots, k-1}} \otimes \tau^{W_k}_{u_k} \otimes I_{\reg{aux}})_{\reg{Bob}},
\end{equation}
where we have also applied \Cref{fact:agreement} to switch to the ``$\consistency_{\delta(\eps)}$" notation.

\paragraph{Extracting a strategy.}

Now we use this to define a $\register$-semiregister strategy $\calS_{\mathrm{semi}}$ for $\game_{\mathrm{semi}}$.
This strategy will have state $\ket{\psi_{\mathrm{semi}}} = \ket{r_1}\cdots\ket{r_k} \ket{\mathrm{aux}}$.
In addition, for each input $x = (x_1, x_2)$ and output $a = (a_1, a_2)$, it will have a matrix
\begin{equation*}
\Lambda_{a_1, a_2}^{x_1, x_2} := (I_{\reg{1, \ldots, k-1}} \otimes \tau^{W_k}_{u_k} \otimes I_{\reg{aux}})\cdot
M_{(u_1, \ldots, u_{k-1}),a_2}^{x_1, x_2} \cdot (I_{\reg{1, \ldots, k-1}} \otimes \tau^{W_k}_{u_k} \otimes I_{\reg{aux}}).
\end{equation*}
First, it follows from $M$ being a $\register|_{k-1}$-strategy that this is indeed a $\register$-semiregister strategy.
This is because
\begin{align*}
\Lambda_{a_1}^{x_1, x_2}
&= (I_{\reg{1, \ldots, k-1}} \otimes \tau^{W_k}_{u_k} \otimes I_{\reg{aux}})\cdot
	M_{u_1, \ldots, u_{k-1}}^{x_1, x_2} \cdot (I_{\reg{1, \ldots, k-1}} \otimes \tau^{W_k}_{u_k} \otimes I_{\reg{aux}})\\
&= (I_{\reg{1, \ldots, k-1}} \otimes \tau^{W_k}_{u_k} \otimes I_{\reg{aux}})\cdot
	(\tau^{W_1}_{u_1} \otimes \cdots \otimes \tau^{W_{k-1}}_{u_{k-1}} \otimes I_{\reg{k, aux}})
	\cdot (I_{\reg{1, \ldots, k-1}} \otimes \tau^{W_k}_{u_k} \otimes I_{\reg{aux}})\\
&= \tau^{W_1}_{u_1} \otimes \cdots \otimes \tau^{W_k}_{u_k} \otimes I_{\reg{aux}}.
\end{align*}
In addition, if $S = \{i \neq k \mid W_i = H\}$, then
\begin{equation*}
\Lambda_{a_1, a_2}^{x_1, x_2}
= (I_{\reg{1, \ldots, k-1}} \otimes \tau^{W_k}_{u_k} \otimes I_{\reg{aux}})\cdot
(I_S \otimes A_{\overline{S}}) \cdot (I_{\reg{1, \ldots, k-1}} \otimes \tau^{W_k}_{u_k} \otimes I_{\reg{aux}})
= I_S \otimes A'_{\overline{S}},
\end{equation*}
where~$A$ and~$A'$ are matrices acting on the registers not in~$S$ and on the auxiliary register.

Next, we show that this has good value.
Write $\calD$ for the marginal distribution of questions given to player~$1$ in $\game_{\mathrm{semi}}$.
By the consistency check,
\begin{equation*}
(M_{(u_1, \ldots, u_{k-1}),a_2}^{x_1, x_2})_{\reg{Alice}} \otimes I_{\reg{Bob}}
\consistency_{\delta(\eps)} I_{\reg{Alice}} \otimes (M_{(u_1, \ldots, u_{k-1}),a_2}^{x_1, x_2})_{\reg{Bob}}
\end{equation*}
with respect to~$\calD$.
As a result, \Cref{eq:looks-like-tau} and \Cref{fact:sandwich} imply that
\begin{equation*}
(\Lambda_a^x)_{\reg{Alice}} \otimes I_{\reg{Bob}}
\approx_{\delta(\eps)} I_{\reg{Alice}} \otimes (M_a^x)_{\reg{Bob}}
\approx_{\delta(\eps)} (M_a^x)_{\reg{Alice}} \otimes I_{\reg{Bob}},
\end{equation*}
where the last step uses the self-consistency of~$M$.
Applying \Cref{fact:approx-delta-game-value}, $\calS_{\mathrm{semi}}$ passes $\game_{\mathrm{semi}}$ with probability at least $\valstrat{\game_{k-1}}{\calS} - \delta(\eps)$.  
Thus, $\valsemi{\register}{\game_{\mathrm{semi}}} \geq 1 - \delta(\eps)$, and we are done.
\end{proof}



\def\bpm#1\epm{\begin{pmatrix}#1\end{pmatrix}}

\section{The data hiding game}\label{sec:data-hiding-layer}

In this section, we introduce a new, simple game called the \emph{data hiding game}.
This game assumes two $(k,n,q)$-semiregister provers with a shared state $\ket{r_1} \cdots \ket{r_k} \ket{\mathrm{aux}}$.
The goal is to test that a given measurement $\{M_a^x\}_a$ acts as the identity on the $k$-th register.

\begin{definition}\label{def:data-hide-def}
Let $x = (x_1, x_2)$ with $x_1 = (W_1, \ldots, W_k)$, and suppose $W_k = \hideq$.
Then the \emph{data hiding game}~$\game_{\mathrm{hide}} := \game_{\mathrm{hide}}(x)$ is given by~\Cref{fig:hide-one}.
It has the following parameters:
\begin{equation*}
\qtime{\game_{\mathrm{hide}}},
\qlength{\game_{\mathrm{hide}}} = O(|x|),
\quad
\atime{\game_{\mathrm{hide}}},
\alength{\game_{\mathrm{hide}}} = O(\textstyle{\sum_i} n_i \log(q_i) + \ell).
\end{equation*}
Here $\ell$ is the maximum of $|a_2|, |a_2'|$ over all answers $a_2$ and $a_2'$ given by the provers.
\end{definition}

{
\floatstyle{boxed} 
\restylefloat{figure}
\begin{figure}[htbp]
Draw $\bW \sim \{X, Z\}$. Set $\bx' = (\bx_1', x_2)$, where $\bx_1' = (W_1, \ldots, W_{k-1}, \bW)$.
Flip an unbiased coin $\bb \sim \{0, 1\}$. Distribute the questions as follows:
\begin{itemize}
\item[$\circ$] Player~$\bb$: give $x$; receive $(\ba_1, \ba_2)$.
\item[$\circ$] Player~$\overline{\bb}$: give $\bx'$; receive $(\ba_1', \ba_2')$.
\end{itemize}
Accept if $\ba_2 = \ba_2'$.
\caption{The game $\game_{\mathrm{hide}}(x)$, with input $x = (x_1, x_2)$ \label{fig:hide-one}}
\end{figure}
}

For a measurement $\{M_a\}_a$ which operates on multiple subsystems,
it will be convenient to define a version of the measurement in which one of the subsystems is ``hidden".

\begin{notation}
Let $M$ be a matrix which operates on $\mathcal{H}_1 \otimes \cdots \otimes \mathcal{H}_k \otimes \mathcal{H}_{\mathrm{aux}}$, and let $i \in [k]$.
Define the notation
\begin{equation*}
\mathrm{hide}_i(M)
:= \frac{1}{\tr(I_i)} \cdot I_i \otimes \tr_i(M).
\end{equation*}
\label{not:hide}
\end{notation}
If $\{M_a\}_a$ is a measurement, then so is $\{\hide{k}{M_a}\}_a$ (though it may not be projective, even if $\{M_a\}_a$ is).
Our main result regarding the data hiding game is that passing it with high probability certifies that $\{M_a\}_a$ is close to $\{\hide{k}{M_a}\}$.

\begin{theorem}\label{thm:data-hiding-game}
Let $x$ be as in \Cref{def:data-hide-def}.
\begin{itemize}
\item[$\circ$] \textbf{Completeness:}
Let $\calS_{\mathrm{partial}} = (\psi, M^x)$ be a partial $(k,n,q)$-register strategy
			which is also a real commuting EPR strategy.
			Then there is a $(k,n,q)$-register
                        strategy~$\calS$ extending
                        $\calS_{\mathrm{partial}}$ which is also a
                        real commuting EPR strategy
			such that $\valstrat{\game_{\mathrm{hide}}}{\calS} = 1$.
\item[$\circ$] \textbf{Soundness:} 
Let $\calS = (\psi, M)$ be a projective $(k, n, q)$-semiregister strategy such that $\valstrat{\game_{\mathrm{hide}}}{\calS} \geq 1-\epsilon$. Then
\begin{equation*}
(M_a^x)_{\reg{Alice}} \otimes I_{\reg{Bob}} \approx_\eps (\hide{k}{M_a^x})_{\reg{Alice}} \otimes I_{\reg{Bob}}
\end{equation*}
on the singleton distribution on input~$x$.
\end{itemize}
\end{theorem}

This section is organized as follows:
in \Cref{sec:pauli-twirl} we introduce the \emph{Pauli twirl},
and in \Cref{sec:data-hiding-game} we use it to prove \Cref{thm:data-hiding-game}.
Finally, in \Cref{sec:compile-sec}, we design our compiler from layer-two to layer-one.
This last step is essentially standard and is included for completeness.

\subsection{Some facts about the Pauli twirl}\label{sec:pauli-twirl}

\newcommand{\push}{\nu}
\newcommand{\bpush}{\bnu}
\newcommand{\twirlset}{S_{\mathrm{twirl}}}
\newcommand{\twirldist}{\calD_{\mathrm{twirl}}}

\begin{definition}
  The \emph{Pauli twirl} $\twirl: \calB((\mathbb{C}^q)^{\ot n})
  \to \calB((\mathbb{C}^q)^{\ot n})$ is the linear operator
  \[ \twirl(A) := \E_{\bu, \bu' \sim \F_q^n}\left[
      X(\bu) Z(\bu') \cdot A \cdot Z(-\bu') X(-\bu) \right]. \]
\end{definition}

\begin{proposition}\label{claim:twirled_pauli}
  Let $P$ be a Pauli matrix on $n$ qudits of dimension $q$. Then
$
\twirl(P) =  P
$
if $P$ is a multiple of the identity,
and otherwise $\twirl(P) = 0$.
\end{proposition}

\begin{proof}
The case when~$P$ is a multiple of the identity follows from the definition.
Otherwise, we can write $P = \omega^z X(a) Z(b)$, where at least one of~$a$ and~$b$ is nonzero. Then
\begin{align*}
     \twirl(P)
&=  \E_{\bu, \bu'} \left[ X(\bu) Z(\bu') \cdot P \cdot
		Z(-\bu') X(-\bu) \right]  \\
&= \omega^z \E_{\bu, \bu'} \left[ X(\bu) Z(\bu') \cdot X(a) Z(b) \cdot
		Z(-\bu') X(-\bu) \right].
\end{align*}
By the Pauli~$X$ and~$Z$ commutation relations (\Cref{eq:commutation-relations}),
this rearranges to
\begin{equation*}
 \omega^z \E_{\bu, \bu'} \left[\omega^{\tr[\bu' \cdot a - \bu \cdot b]}\right] \cdot X(a) Z(b)
=  \E_{\bu, \bu'} \left[ \omega^{\tr[\bu' \cdot a - \bu \cdot b]} \right] \cdot P
    =  \E_{\bu'} \big[ \omega^{\tr[\bu' \cdot a]}\big] \cdot \E_{\bu}\big[\omega^{- \tr[\bu \cdot b]} \big] \cdot P
    = 0.
\end{equation*}
Here the last step uses \Cref{fact:averages-to-zero} and the fact that at least one of~$a$ or~$b$ is nonzero.
\end{proof}

In the next couple of sections,
we will consider the effects of applying the Pauli twirl to our measurements.
For convenience, we will ``group" our state into two parts: $\ket{\psi_1} = \ket{r_k}$ 
is the subsystem we want to hide,
and $\ket{\psi_2} = \ket{r_1} \cdots \ket{r_{k-1}} \ket{\mathrm{aux}}$ is the remaining part of this state.
In this way, we can consider our measurements as operating on the bipartite state $\ket{\psi_1} \ket{\psi_2}$.

\begin{proposition}\label{prop:hidden-twirl}
Let $\{M_a\}$ be a measurement on the state $\ket{\psi} = \ket{\psi_1} \ket{\psi_2}$.
Then
\begin{equation*}
(\twirl_1\otimes \id_{2})[M_a]
=
\hide{1}{M_a},
\end{equation*}
where $\id_{2}$ is the identity superoperator applied to the second register.
\end{proposition}
\begin{proof}
Let $P_J$ be the elements of the Pauli group on $n$ qudits of dimension~$q$, with $P_0 = I$.
Because these form a basis for the set of matrices, we can write
\begin{equation*}
M_a = \sum_{J} P_J \otimes M_{a, J},
\end{equation*}
where the $M_{a, J}$'s are matrices acting on the auxiliary register. 
Using \Cref{claim:twirled_pauli},
\begin{equation*}
(\twirl_1\otimes \id_{2})[M_a]
= \sum_{J}  \twirl(P_J) \otimes M_{a,J}
= P_0 \otimes M_{a,0}
= I \otimes M_{a,0}.
\end{equation*}
On the other hand, because $P_J$ is traceless unless $J=0$ (i.e.\ $P_J$ is the identity),
\begin{equation*}
\hide{1}{M_a}
= \sum_{J} \hide{1}{P_J \otimes M_{a,J}} 
= \sum_{J} \frac{1}{q^n} \cdot I \otimes \tr_1(P_J \otimes M_{a,J})
= I \otimes M_{a, 0}.
\end{equation*}
These two are equal, completing the proof.
\end{proof}

\subsection{Hiding a single coordinate}\label{sec:data-hiding-game}

In this section, we prove \Cref{thm:data-hiding-game}. Prior to doing so, we prove a couple of technical lemmas.
The first shows that a measurement which approximately commutes with the Pauli measurements
also approximately commutes with the Pauli observables.

\begin{lemma}\label{lem:proj-to-obs}
Let $W \in \{X, Z\}$.
Suppose $\{M_{a}\}$ is a measurement on the state $\ket{\psi} = \ket{\mathrm{EPR}_q^n} \ket{\psi_2}$ for which
\begin{equation*}
(M_{a} \cdot (\tau^{W}_{u} \otimes I_{\reg{2}}))_{\reg{Alice}} \otimes I_{\reg{Bob}}
\approx_{\delta} ((\tau^{W}_{u} \otimes I_{\reg{2}}) \cdot M_{a})_{\reg{Alice}} \otimes I_{\reg{Bob}}.
\end{equation*}
Then the statement
\begin{equation*}
(M_{a} \cdot (W(u) \otimes I_{\reg{2}}))_{\reg{Alice}} \otimes I_{\reg{Bob}}
\approx_{\delta} ((W(u) \otimes I_{\reg{2}}) \cdot M_{a})_{\reg{Alice}} \otimes I_{\reg{Bob}}
\end{equation*}
holds with respect to the uniform distribution on $\bu \in \F_q^n$.
\end{lemma}
\begin{proof}
Our goal is to bound
\begin{equation}\label{eq:gonna-bound-this-foist}
\E_{\bu} \sum_{a} \Vert (M_{a} \cdot (W(\bu) \otimes I_{\reg{2}}) - (W(\bu) \otimes I_{\reg{2}}) \cdot M_{a}) \otimes I \ket{\psi} \Vert^2.
\end{equation}
by $\delta$.
To do so, for a fixed~$u$ we introduce the notation
\begin{align}
\Delta_{a}^{u}
&:= M_{a} \cdot (W(u) \otimes I_{\reg{2}}) - (W(u) \otimes I_{\reg{2}}) \cdot M_{a}\nonumber\\
&= \sum_{v \in \F_q^n} \omega^{\tr[u \cdot v]} (\underbrace{M_{a}
	\cdot (\tau^{W}_{v} \otimes I_{\reg{2}})- (\tau^{W}_v \otimes I_{\reg{2}}) \cdot M_{a}}_{\Delta_{a,v}}).\label{eq:pauli-obs-commutator-foist}
\end{align}
We record the following identity, which follows from~\Cref{eq:pauli-obs-commutator-foist}:
\begin{equation*}
\E_{\bu} (\Delta_{a}^{\bu})^\dagger\cdot \Delta_{a}^{\bu}
= \E_{\bu} \sum_{v, v' \in \F_q^n} \omega^{\tr[\bu \cdot (v' - v)]} (\Delta_{a,v})^\dagger \Delta_{a,v'}
= \sum_{v \in \F_q^n} (\Delta_{a,v})^\dagger \Delta_{a,v}.
\end{equation*}
As a result,
\begin{align*}
\eqref{eq:gonna-bound-this-foist}=
\E_{\bu} \sum_{a}\Vert (\Delta_{a}^{\bu} \otimes I) \ket{\psi}\Vert^2
&= \E_{\bu} \sum_{a} \bra{\psi} (\Delta_{a}^{\bu})^\dagger \Delta_{a}^{\bu} \otimes I \ket{\psi}\\
&=  \sum_{a} \sum_{v \in \F_q^n} \bra{\psi}(\Delta_{a,v})^\dagger \Delta_{a,v} \otimes I \ket{\psi}\\
&= \sum_{a, v} \Vert \Delta_{a,v} \otimes I \ket{\psi} \Vert^2.
\end{align*}
But this is at most $O(\delta)$, by assumption. This completes the proof.
\end{proof}

The next technical lemma shows that a measurement which approximately commutes with products of $X$ and $Z$ observables
is approximately equal to its own Pauli twirl.

\begin{lemma}\label{lem:commute-to-twirl}
Consider the distribution~$\calD$ on pairs $(\bu, \bu')$, where $\bu, \bu' \sim \F_q^n$.
Suppose $\{M_{a}\}$ is a measurement on the state $\ket{\psi} = \ket{\mathrm{EPR}_q^n} \ket{\psi_2}$ for which
\begin{equation*}
((Z(u') X(u)\otimes I_{\reg{2}}) \cdot M_{a}) \otimes I_{\reg{Bob}}
\approx_{\delta} (M_{a} \cdot (Z(u') X(u)\otimes I_{\reg{2}})) \otimes I_{\reg{Bob}}.
\end{equation*}
on distribution~$\calD$. Then
\begin{equation*}
M_{a} \otimes I_{\reg{Bob}} \approx_{\delta} (\twirl_1 \otimes \mathrm{id}_{\mathrm{2}})[M_{a}] \otimes I_{\reg{Bob}}.
\end{equation*}
\end{lemma}
\begin{proof}
By definition,
\begin{equation*}
(\twirl_1 \otimes \mathrm{id}_{\mathrm{2}})[M_{a}]
= \E_{\bu, \bu'} [(X(\bu) Z(\bu') \otimes I_{\reg{2}}) \cdot M_a \cdot (Z(-\bu') X(-\bu) \otimes I_{\reg{2}})].
\end{equation*}
Similarly,
\begin{equation*}
M_{a}
= \E_{\bu, \bu'} [(X(\bu) Z(\bu') \otimes I_{\reg{2}}) \cdot (Z(-\bu') X(-\bu) \otimes I_{\reg{2}}) \cdot M_a].
\end{equation*}
As a result, if we set
$
A^{\bu, \bu'} = X(\bu) Z(\bu') \otimes I_{\reg{2}},
$
and
\begin{equation*}
B^{\bu, \bu'}_a = (Z(-\bu') X(-\bu) \otimes I_{\reg{2}}) \cdot M_a - M_a \cdot (Z(-\bu') X(-\bu) \otimes I_{\reg{2}}),
\end{equation*}
then
\begin{equation*}
\Delta_a:=
M_a - (\twirl_1 \otimes \mathrm{id}_{\mathrm{2}})[M_a]
= \E_{\bu, \bu'} [A^{\bu, \bu'} \cdot B^{\bu, \bu'}_a].
\end{equation*}
We can therefore establish the lemma as follows:
\begin{align*}
\sum_a \Vert (\Delta_a)_{\reg{Alice}} \otimes I_{\reg{Bob}} \ket{\psi} \Vert^2
& = \sum_a \Vert \E_{\bu, \bu'} [A^{\bu, \bu'} \cdot B^{\bu, \bu'}_a] \otimes I_{\reg{Bob}} \ket{\psi} \Vert^2\\
& \leq \E_{\bu, \bu'} \sum_a \Vert (A^{\bu, \bu'} \cdot B^{\bu, \bu'}_a) \otimes I_{\reg{Bob}} \ket{\psi} \Vert^2\tag{Jensen's inequality}\\
& = \E_{\bu, \bu'} \sum_a \Vert (B^{\bu, \bu'}_a) \otimes I_{\reg{Bob}} \ket{\psi} \Vert^2 \tag{$A^{\bu, \bu'}$ is unitary}
\end{align*}
By assumption, this quantity is $O(\delta)$. This concludes the proof.
\end{proof}

Now we prove \Cref{thm:data-hiding-game}.
\begin{proof}[Proof of \Cref{thm:data-hiding-game}]
We consider the completeness and soundness cases separately.

\paragraph{Completeness.}
Let $\calS_{\mathrm{partial}} = (\psi, M^x)$ be a partial
$(k,n,q)$-register strategy which is also a real commuting EPR strategy.
To this strategy we will add matrices for the questions $x' = (x_1', x_2)$ with $x_1' = (W_1, \ldots, W_{k-1}, W)$.

Let $a_1 = (u_1, \ldots, u_k)$, where $u_i = \varnothing$ if $W_i \in \{\hideq, \noop\}$.
Let~$S = \{i \mid W_i \neq \noop\}$.
By definition of a $(k, n, q)$-register strategy,
\begin{equation*}
M^x_a = \bigotimes_{i \in S} \tau^{W_i}_{u_i} \otimes M^{x, a_1}_{a_2},
\end{equation*}
where $M^{x,a_1}_{a_2}$ acts on the auxiliary registers and the registers not in~$S$.
Next, set $a' = (a_1', a_2)$ where $a_1' = (u_1, \ldots, u_{k-1}, u_k')$ and $u_k' \in \F_{q_k}^{n_k}$.  Then we set
\begin{equation*}
(M')^{x_1', x_2}_{a'_1, a_2} = \bigotimes_{i \in S\setminus k} \tau^{W_i}_{u_i} \otimes \tau^{W}_{u_k'} \otimes M^{x, a_1}_{a_2}.
\end{equation*}
This is a $(k,n,q)$-register strategy for~$\game_{\mathrm{hide}}$ by design.
To see that it is value~$1$, suppose on question~$x$ Player~$\bb$ measures~$\ba_1$.
Then by~\Cref{fact:the-ol-pauli-swaperoonie}, Player~$\overline{\bb}$ will measure~$\ba_1'$ in which $\bu_i' = \bu_i$ for all $i < k$.
As a result, to measure~$\ba_2$, Player~$\bb$ will measure $M^{x, \ba_1}$
and Player~$\overline{\bb}$ will measure $M^{x, \ba_1}$, both on state $\ket{r_{\overline{S}}}\ket{\mathrm{aux}}$.
As this is an EPR state,
by~\Cref{fact:heh-heh-heh-gonna-make-anand-prove-this-so-i-can-take-the-day-off}
the outcomes will always be the same, and so this strategy has
value~$1$. The fact that this a real strategy follows from the
assumption that the matrices $M^x_a$ are real, and the fact that for
$W \in \{X, Z\}$, $\tau^W_u$ is a real matrix.
Finally, the fact that this is a commuting strategy follows from the fact that $M^x_a$ and $(M')^{x'}_{a'}$ are commuting.

\paragraph{Soundness.}
We write $x$ and $x' = (x_1', x_2)$ with $x_1' = (W_1, \ldots, W_{k-1}, W)$ as in the test.
Because the test passes with probability $1-\eps$, \Cref{fact:agreement} implies that
\begin{equation*}
(M_{a_2}^{x})_{\reg{Alice}} \otimes I_{\reg{Bob}}
\approx_{\eps} I_{\reg{Alice}} \otimes (M_{a_2}^{x'})_{\reg{Bob}}.
\end{equation*}
Because $\calS$ is a $(k,n,q)$-semiregister strategy, 
\Cref{eq:in-math} implies that $M_{u_k}^{x'} = \tau^{W}_{u_k} \otimes I_{\overline{k}}$,
where we write $I_{\overline{k}} := I_{\reg{1, \ldots, k-1, aux}}$.
Our next step is to show that the measurements approximately commute.
This follows the analysis of the commutation test (cf.~\cite[Lemma 28]{CGJV18}).
\begin{align*}
M_{u_k}^{x'} M_{a_2}^{x} \otimes I_{\reg{Bob}}
&\approx_{\eps} M_{u_k}^{x'} \otimes M_{a_2}^{x'}\tag{\Cref{fact:add-a-proj}}\\
&\approx_{0} I_{\reg{Alice}} \otimes M_{a_2}^{x'}M_{u_k}^{x'} \tag{\Cref{fact:the-ol-pauli-swaperoonie}}\\
&= I_{\reg{Alice}} \otimes M_{u_k}^{x'}M_{a_2}^{x'}\\
&\approx_{\eps} M_{a_2}^{x} \otimes M_{u_k}^{x'}\tag{\Cref{fact:add-a-proj}}\\
&\approx_{0} M_{a_2}^x M_{u_k}^{x'} \otimes I_{\reg{Bob}}.\tag{\Cref{fact:the-ol-pauli-swaperoonie}}
\end{align*}
In summary,
\begin{equation*}
((\tau^{W}_{u_k} \otimes I_{\overline{k}})\cdot M_{a_2}^{x})_{\reg{Alice}} \otimes I_{\reg{Bob}}
\approx_{\eps} (M_{a_2}^x \cdot (\tau^{W}_{u_k} \otimes I_{\overline{k}}))_{\reg{Alice}}\otimes I_{\reg{Bob}}.
\end{equation*}
Recall this is with respect to the distribution $\bW$ where $\bW \sim \{X,Z\}$ is uniform.
Therefore, it also holds with respect to the distribution where $\bW$ is fixed to either~$X$ or~$Z$.
As a result, for a fixed $W \in \{X,Z\}$, by \Cref{lem:proj-to-obs}, 
\begin{equation*}
(M_{a_2}^x \cdot ( W(u)\otimes I_{\overline{k}}))_{\reg{Alice}} \otimes I_{\reg{Bob}}
\approx_{\eps} (( W(u)\otimes I_{\overline{k}}) \cdot M_{a_2}^x)_{\reg{Alice}} \otimes I_{\reg{Bob}}.
\end{equation*}
on distribution $\bu \sim \F_q^n$.
As a result, by~\Cref{fact:the-ol-pauli-swaperoonie} and \Cref{fact:add-a-proj},
\begin{align*}
(M_{a_2}^x \cdot (Z(u') X(u)\otimes I_{\overline{k}}))_{\reg{Alice}} \otimes I_{\reg{Bob}}
& \approx_{0}(M_{a_2}^x \cdot (Z(u')\otimes I_{\overline{k}}))_{\reg{Alice}} \otimes (X(-u)\otimes I_{\overline{k}})_{\reg{Bob}}\\
& \approx_{\eps}((Z(u')\otimes I_{\overline{k}}) \cdot M_{a_2}^x)_{\reg{Alice}} \otimes (X(-u)\otimes I_{\overline{k}})_{\reg{Bob}}\\
& \approx_{0}((Z(u')\otimes I_{\overline{k}}) \cdot M_{a_2}^x \cdot (X(u) \otimes I_{\overline{k}}))_{\reg{Alice}} \otimes I_{\reg{Bob}}\\
& \approx_{\eps}((Z(u') X(u)\otimes I_{\overline{k}}) \cdot M_{a_2}^x)_{\reg{Alice}} \otimes I_{\reg{Bob}},
\end{align*}
on distribution $\bu, \bu' \sim \F_q^n$.
Applying \Cref{lem:commute-to-twirl} and \Cref{prop:hidden-twirl},
we can therefore conclude
\begin{equation*}
M_{a_2}^x \otimes I_{\reg{Bob}}
\approx_{\eps} (\twirl_k \otimes \mathrm{id}_{\overline{k}})[M_{a_2}^x] \otimes I_{\reg{Bob}}
= (\hide{k}{M_{a_2}^x}) \otimes I_{\reg{Bob}}.\qedhere
\end{equation*}
\end{proof}



\section{Compiling games with the data hiding test}\label{sec:compile-sec}

Now we can show how to compile games from the second layer to the first layer.
Our construction is given in the following definition.

\begin{definition}
Let $\game_k$ be a $(k, n, q)$-register game.
Then its compiled version is the game $\calC_{k \rightarrow \mathrm{semi}}(\game_k)$ defined in \Cref{fig:two-to-one}.
\end{definition}
{
\floatstyle{boxed} 
\restylefloat{figure}
\begin{figure}
With probability~$\tfrac{1}{2}$ each, perform one of the following three tests.
\begin{enumerate}
	\item \textbf{Data hiding:} Draw $(\bx, \bx', \bC) \sim \game_k$, where $\bx = (\bx_1, \bx_2)$ and $\bx_1 = (\bW_1, \ldots, \bW_k)$.
			 If $\bW_k = \hideq$, play $\game_{\mathrm{hide}}$ with question~$\bx$.\label{item:data-hiding}
	\item \textbf{Play game:} Perform $\game_{k}$.\label{item:play-ball!}
	\end{enumerate}
	\caption{The game $\calC_{k \rightarrow \mathrm{semi}}(\game_k)$.\label{fig:two-to-one}}
\end{figure}
}

\begin{theorem}\label{theorem:two-to-one-compiler}
Suppose $\game_k$ is a $(k, n, q)$-register game,
and consider the $(k, n, q)$-semiregister game $\game_{\mathrm{semi}} = \calC_{k \rightarrow \mathrm{semi}}(\game_k)$.
\begin{itemize}
\item[$\circ$] \textbf{Completeness:} 
	Suppose there is a value-$1$ $(k,n,q)$-register strategy for $\game_{k}$ which is also a real commuting EPR strategy.
	Then there is a value-$1$ $(k,n,q)$-semiregister strategy for $\game_{\mathrm{semi}}$ which is also a real commuting EPR strategy.
\item[$\circ$] \textbf{Soundness:} If $\valsemi{k,n,q}{\game_{\mathrm{semi}}} \geq 1-\eps$ then $\valreg{k,n,q}{\game_{k}} \geq 1 -\delta(\eps)$, where $\delta(\eps) = \mathrm{poly}(\eps)$.
\end{itemize}
Furthermore,
\begin{equation*}
\qlength{\game_{\mathrm{semi}}} = O(\qlength{\game_k}), \quad
\alength{\game_{\mathrm{semi}}} = O(\alength{\game_k}),
\end{equation*}
\begin{equation*}
\qtime{\game_{\mathrm{semi}}} = O(\qtime{\game_k}), \quad
\atime{\game_{\mathrm{semi}}} = O(\atime{\game_k}).
\end{equation*}
\end{theorem}

\newcommand{\phoenicianindahizzous}{\textphnc{h}}

\begin{proof}[Proof of \Cref{theorem:two-to-one-compiler}]
The communication and time complexities are the result of combining
the communication and time complexities from $\game_{\mathrm{k}}$ with the values for the data hiding game from \Cref{def:data-hide-def}.

\paragraph{Completeness.}
Let $(\psi, M)$ be a value-$1$ $(k,n,q)$-register strategy for $\game_{k}$ which is also a commuting EPR strategy.
Then for every $x = (x_1, x_2)$ where $x_1 = (W_1, \ldots, W_k)$ with $W_k = \hideq$,
by \Cref{thm:data-hiding-game}
we can extend this strategy to one that passes the data hiding game with question~$x$ with probability~$1$.
Thus, this strategy has value~$1$ overall.
In addition, \Cref{thm:data-hiding-game} implies this strategy is a real commuting EPR strategy as well.

\paragraph{Soundness.}
Suppose $\calS = (\psi, M)$ is a $(k, n, q)$-semiregister strategy for $\game_{\mathrm{semi}}$ with value $1-\eps$.
By \Cref{lem:proj-suffices}, we can assume without loss of generality that~$M$ is projective.
Our goal will be to decode $\calS$ into a $(k, n, q)$-register strategy $\calS_{k}$ with nearly the same value.

\paragraph{Using the data hiding test.}
For a fixed question~$x$, write $\nu_x$ for the probability that~$\calS$ passes the test in \Cref{item:data-hiding}.
Then on average, the probability that~$\calS$ passes this test is $\E_{\bx} \nu_{\bx}$,
which is at least $1-2\eps$ because the overall test passes with probability at least $1-\eps$.
This implies that $\nu_{\bx} \geq 1-\eps^{1/2}$ with probability at least $1-2\eps^{1/2}$.
Given a matrix~$M$ and a $W \in \{X, Z, \hideq, \noop\}$, let us write $\hide{W}{M}:= \hide{k}{M}$ if $W = \hideq$
and $\hide{W}{M}:=M$ otherwise.
For a question~$x$, if $W_k \neq \hideq$, then $\hide{W_k}{M^x_a} = M^x_a$ trivially.
On the other hand, suppose $W_k = \hideq$.
Then either $\nu_x \geq 1-\eps^{1/2}$,
in which case $M^x_a \otimes I_{\reg{Bob}} \approx_{\delta(\eps)} \hide{W_k}{M^x_a}\otimes I_{\reg{Bob}}$ by \Cref{thm:data-hiding-game},
or $\nu_x < 1-\eps^{1/2}$,
in which case we have the trivial bound $M^x_a \otimes I_{\reg{Bob}} \approx_{1} \hide{W_k}{M^x_a}\otimes I_{\reg{Bob}}$ from~\Cref{fact:trivial-upper-bound-approx-delta}.
Since this latter case happens with probability at most $2\eps^{1/2}$,
averaging over all~$x$ gives us
\begin{equation}\label{eq:dr-jekyll-mr-hide}
M^x_a \otimes I_{\reg{Bob}} \approx_{\delta(\eps)} \hide{W_k}{M^x_a}\otimes I_{\reg{Bob}},
\end{equation}
on the distribution $\calD$.

\paragraph{Extracting a strategy.}
Define the strategy $\calS_k = (\psi, \Lambda)$, in which
$\Lambda_a^x := \hide{W_k}{M^x_a}$.
First, we show that $\calS_k$ is a $(k, n, q)$-register strategy.
To do so, fix $x= (x_1, x_2)$ with $x_1 = (W_1, \ldots, W_k)$ and $a = (a_1, a_2)$ with $a_1 = (u_1, \ldots, u_k)$.
Then
\begin{equation*}
\Lambda_{a_1}^x
= \mathrm{hide}_{W_k}\big(\tau^{W_1}_{u_1} \otimes \cdots \otimes \tau^{W_k}_{u_k} \otimes I_{\reg{aux}}\big)\\
	= \tau^{W_1}_{u_1} \otimes \cdots \otimes \tau^{W_k}_{u_k} \otimes I_{\reg{aux}}.
\end{equation*}
The first equality is by definition of~$\Lambda$ and the fact that $\calS$ is a $(k, n, q)$-quasiregister strategy.
The second equality is trivial when $W_k \neq \hideq$ and follows from the fact that $\tau^{W_k}_{\varnothing} = I$ when $W_k = \hideq$.
Next, define $S = \{i\neq k  \mid W_i = \hideq\}$.
If $W_k \neq \hideq$ then $\Lambda_a^x = M^x_a = M_{\overline{S}} \otimes I_S$ for some matrix~$M$.
Otherwise, if $W_k = \hideq$, set $\phoenicianindahizzous = \tr(I_k)$. Then
\begin{equation*}
\Lambda_a^x = \hide{k}{M^x_a} = \hide{k}{M_{\overline{S}} \otimes I_S}
= \frac{1}{\phoenicianindahizzous} \cdot I_k \otimes \tr_k(M_{\overline{S}} \otimes I_S)
= \frac{1}{\phoenicianindahizzous} \cdot I_{S \cup k} \otimes \tr_k(M_{\overline{S}}).
\end{equation*}
The matrix $\tr_k(M_{\overline{S}}) \cdot \phoenicianindahizzous^{-1}$
only acts on the registers in $\overline{S \cup k}$ and the auxiliary register,
and as a result, this strategy satisfies data hiding.
Thus, $\calS_k$ is a $(k, n, q)$-register strategy.

It remains to show that $\calS_k$ has good value.
This follows by combining~\Cref{eq:dr-jekyll-mr-hide} with~\Cref{fact:approx-delta-game-value}:
$\valstrat{\game_{k}}{\calS_{k}} \geq \valstrat{\game_{\mathrm{semi}}}{\calS} - \delta(\eps)$,
and so $\valreg{k,n,q}{\game_{k}} \geq 1 - \delta(\eps)$.
\end{proof}



\section{Partial data hiding}\label{sec:rotated-data-hiding}

The data-hiding game presented above was used to show that the
provers' measurement acts as identity on a subset of the provers'
qudits, and thus the prover learns no information from those
qudits. In particular, the measurement outcome of any $X$- or
$Z$-observable measurement on the qubits in the subset is hidden from
the prover. In this subsection, we generalize this idea to show how
to certify that certain \emph{partial} information about a register is hidden
from a prover. This test is a crucial component in our technique of
introspection, wherein two provers measure a shared EPR state to
sample from the joint distribution over questions of a classical
game. The partial data hiding test will prevent one prover from
learning the question sampled by the other prover.

\begin{notation}
  Given a set $v =\{v_1, \dots, v_k\}$ of $k$
  vectors in $\F_q^n$, denote their span by $V = \mathrm{span}(\{v_1,
  \dots, v_k\})$. The orthogonal complement of their span is the
  subspace $V^\perp = \{a: \forall i \in \{1, \dots, k\}, \langle a,
  v_i \rangle = 0 \}$. We denote by $\surfaces{v}$ the set of all
  affine subspaces parallel to $V$, i.e.\ sets of the form:
  \[ s = \{u + \lambda_1 v_1 + \dots + \lambda_k v_k: \lambda_1,
    \dots, \lambda_k \in \F_q\}.\]
  For a subspace $s \in
  \surfaces{v}$, the subspace projector $\Pi^{v}_s$ is the projector
  \[ \Pi^{v}_s = \sum_{w \in s} \proj{w}. \]
\end{notation}
\begin{lemma}
  \label{lem:subspace-x-z-commute}
  Given a set of vectors $\{v_1, \dots, v_k\}$, let
  \[ \tau^{X}_{[\forall i,  u \cdot v_i = a_i]} = \sum_{u: \forall i,
       u \cdot v_i = a_i} \tau^X_{u}. \]
  Then $\tau^{X}_{[\forall i,  u \cdot v_i = a_i]}$ commutes with
  $\Pi^{v}_s$ for all $s \in \surfaces{v}$.
\end{lemma}
\begin{proof}
  The proof is by calculation. 
  \begin{align*}
    \Pi^{v}_s \tau^{X}_{[\forall i,  u \cdot v_i = a_i]}
    &= \sum_{w \in s}\proj{w} \sum_{u: \forall i, u \cdot v_i = a_i}
      \tau^X_u \\
    &= \sum_{w \in s} \sum_{u: \forall i, u \cdot v_i = a_i} 
      \E_{\bb} \omega^{-\tr[\bb \cdot u]} \proj{w}  X(\bb). \\
    \intertext{We note an important fact: for any two
  outcomes $u, u'$ satisfying $u\cdot v_i = u'\cdot v_i = a_i$ for all
  $i$, the difference $u - u'$ must lie in $V^\perp$. Fixing some
    appropriate outcome vector $u_0$, we can then express the
    summation variable $u$ as $u_0 + x$ where $x$ runs over $V^\perp$: }
    &= \sum_{w \in s} \sum_{x \in V^\perp} \E_{\bb} \omega^{-\tr[\bb \cdot
      (u_0 + x)]} \proj{w} X(\bb) \\
    &=\sum_{w \in s} \sum_{x \in V^\perp} \E_{\bb} \omega^{-\tr[\bb \cdot
      (u_0 + x)]} \ket{w}\bra{w-\bb}.\\
    \intertext{Now, the summation over $x$ vanishes unless $\bb \in
    (V^\perp)^\perp = V$, by
    \Cref{fact:averages-to-zero-subspace}. This happens with
    probability $q^{k-n}$ which cancels out the factor of $q^{n-k}$
    from evaluating the sum over $x \in V^\perp$, yielding:}
    &= \sum_{w \in s} \E_{\bb \in V} \omega^{-\tr[\bb \cdot u_0]} \ket{w}
      \bra{w-\bb}. \\
    \intertext{Now, since $\bb \in V$, and the summation variable $w$ runs over an affine
    subspace parallel to $V$, we can shift it from $w$ to $w+\bb$, yielding}
    &= \sum_{w \in s} \E_{\bb \in V} \omega^{-\tr[\bb \cdot u_0]} \ket{w+\bb}
      \bra{w}. \\
    \intertext{Finally, we can perform the same manipulations in reverse:}
    &= \dots \\
    &= \tau^{X}_{[\forall i, u \cdot v_i = a_i]} \Pi^{v}_s. \qedhere
  \end{align*}
\end{proof}


\begin{definition}\label{def:data-hide-def}
The \emph{partial data-hiding game} is given by~\Cref{fig:hide-observable}.
\end{definition}

{
\floatstyle{boxed} 
\restylefloat{figure}
\begin{figure}[htpb]
  Given a set $S$ of $k$-tuples of linearly independent set of vectors $v_1, \dots, v_k \in
  \F_q^{n}$ and a query string $x$. Sample $v = \{v_1, \dots, v_k\}$ uniformly from $S$.
  Flip an unbiased coin $\bb \sim \{0, 1\}$. Perform one of the following three tests with
  probability $1/3$ each.
  \begin{enumerate}
    \item Distribute the questions as follows:
      \begin{itemize}
      \item[$\circ$] Player~$\bb$: Give $(\bot, x,v)$; receive $(\varnothing, \ba_2)$.
      \item[$\circ$] Player~$\overline{\bb}$: give $(Z, x,v)$; receive
        $(\ba'_1, \ba'_2)$.
      \end{itemize}
      Accept if $\ba_2 = \ba_2'$.
    \item Distribute the questions as follows:
      \begin{itemize}
        \item[$\circ$] Player~$\bb$: Give $(\bot, x,v)$; receive $(\varnothing, \ba_2)$.
        \item[$\circ$] Player~$\overline{\bb}$: give $(\bot, x, \{X, v\})$; receive
        $(\varnothing, \ba'_2, \{\ba_{1,1}', \dots \ba_{1,k}'\})$. 
        \end{itemize}
        Accept if $\ba_2 = \ba_2'$.
      \item Distribute the questions as follows:
        \begin{itemize}
        \item[$\circ$] Player~$\bb$: Give $(X, \cdot)$; receive $(\ba_1, \cdot)$. (Here, ``$\cdot$" is the empty string.)
        \item[$\circ$] Player~$\overline{\bb}$: give $(\bot, \bot, \{X, v\})$; receive
        $(\varnothing,\varnothing, \{\ba_{1,1}', \dots \ba_{1,k}'\})$. 
        \end{itemize}
        Accept if $\ba'_{1,i} = v_i \cdot \ba_1$ for all $i \in \{1,
        \dots, k\}$.
      \item 
        Distribute the questions as follows:
        \begin{itemize}
          \item[$\circ$] Player~$\bb$: Give $(\bot, x, \{X, v\})$;
            receive $(\varnothing, \ba_2, \{\ba_{1,1}, \dots,
            \ba_{1,k}\})$.
          \item[$\circ$] Player~$\overline{\bb}$: give $(\bot,
            \bot, \{X, v\})$; receive $(\varnothing,
            \varnothing, \{\ba'_{1,1}, \dots, \ba'_{1,k}\})$.
          \end{itemize}
          Accept if $\ba_{1,i}  = \ba_{1,i}'$ for all $i \in \{1,
          \dots, k\}$.
      \end{enumerate}

\caption{The partial data-hiding game $\game_{\mathrm{hide}}(S, x)$.\label{fig:hide-observable}}
\end{figure}
}

\begin{theorem}  \label{thm:partial-data-hiding-game}
  Let $S$ be any set of $k$-tuples of vectors in
  $\F_q^n$, and let $x$ be an arbitrary query.
  \begin{itemize}
  \item[$\circ$]\textbf{Completeness:}
    Let $\calS_{\mathrm{partial}} = (\psi, M^{\bot, x,v})$ be a partial
    $(1,n,q)$-register strategy for~$\game_{\mathrm{hide}}(S,x)$,
    which is also a real commuting EPR strategy, and for which
    \[ M^{\bot, x,v}_{a_2} = \sum_{s \in \surfaces{v}} \Pi^{v}_s \ot A^{x,v,s}_{a_2}, \]
    for some measurement $A^{x,v,s}_{a_2}$ acting only on the
    $\reg{aux}$ register. Then
    there is a $(1,n,q)$-register strategy~$\calS$ extending
    $\calS_{\mathrm{partial}}$ 
    for which
    $\valstrat{\game_{\mathrm{hide}}}{\calS} = 1$.
    \item[$\circ$]\textbf{Soundness:}
      Let $\calS = (\psi, M)$ be a projective
      $(1,n,q)$-register strategy such that
      $\valstrat{\game_{\mathrm{hide}}(S,x)}{\calS}
      \geq 1 - \eps$. Then there exists an ideal measurement $M'^{\bot, x,
        v}_a$ with
      the property that
      \[ M'^{\bot, x,v }_a = \sum_{s \in \surfaces{v}} \Pi^{v}_s \ot
          M^{s,x, v}_a, \]
      such that  the measurement $M^{\bot, x,v}_a$ used by strategy $\calS$ in
      response to the query $x$  is close to $M'^{\bot, x,v}_a$:
      \[ (M_a^{\bot, x,v})_{\reg{Alice}} \otimes I_{\reg{Bob}} \approx_\eps (M'^{\bot, x,v}_a)_{\reg{Alice}} \otimes I_{\reg{Bob}}. \]
    \end{itemize}
\end{theorem}

To prove this theorem,  we will start with some basic facts about the subspace projector
  measurements. Let us denote the linear subspace spanned by the
  vectors $v_1, \dots, v_k$ by $V$.
  
  \begin{definition}
    For any distribution $\mathcal{U}$ over unitary matrices,
    the \emph{twirl by
      $\mathcal{U}$} is the linear operator $\twirl_{\mathcal{U}} :
    \calB((\mathbb{C}^q)^{\ot n}) \to \calB((\mathbb{C}^q)^{\ot n})$
    defined by
    \[ \twirl_{\mathcal{U}}(A) := \E_{\bU \sim \mathcal{U}} \left[ \bU A
        \bU^\dagger \right]. \]
  \end{definition}
  \begin{definition}
    Let $v = \{v_1, \dots, v_k\}$ be a set of linearly independent
    vectors over $\F_q$. Further let $\mathcal{V}$ be the uniform
    distribution over the set $\{X(a): a \in V \}$, $\mathcal{Z}$ be
    the uniform distribution over the set of all Pauli $Z$ operators
    $\{Z(a): a \in \F_q^n\}$, and $\mathcal{S}$ be the distribution
    over products $\bM\bN$ where $\bM$ is drawn from $\mathcal{V}$ and $\bN$
    from $\calZ$. Then the \emph{$v$-subspace twirl} is the
    twirl over $\calS$:
    \[ \twirl_{\calS} = \twirl_{\mathcal{V}} \circ \twirl_{\mathcal{Z}} \]
  \end{definition}
  \begin{proposition}
    \label{prop:hidden-subspace-twirl}
    Let $A$ be a Hermitian matrix and $v$ a set of $k$ vectors over $\F_q$. Then the $v$-subspace twirl of
    $A$ is a linear combination of projectors onto affine subspaces
    along $v$:
    \[ (\twirl_{\calS} \otimes \mathrm{id}_{\reg{aux}})(A) = \sum_{s \in \surfaces{v}}
      \Pi^{v}_{s} \ot (M_s)_{\reg{aux}} , \]
    for some choice of Hermitian matrices $M_s$ indexed by subspaces $s$.
  \end{proposition}
  \begin{proof}
    Start by decomposing $A$ into a linear combination of Pauli
    matrices:
    \[ A = \sum_{u,u'}  X(u) Z(u') \ot (A_{u,u'})_{\reg{aux}}. \]
    After the twirl over $\mathcal{Z}$, the only terms that survive
    are those with no $X$ part, i.e.
    \[ A' = (\twirl_{\mathcal{Z}} \otimes \mathrm{id}_{\mathrm{aux}})(A) = \sum_{u} Z(u) \ot (A_{0,u})_{\reg{aux}} \]
    Now if we perform the twirl over $\mathcal{V}$, we get
    \begin{align}
      (\twirl_{\mathcal{V}} \otimes \mathrm{id}_{\mathrm{aux}})(A')
      				&= \sum_u \E_{\ba \in V} X(\ba) Z(u) X(\ba)^\dagger \ot (A_{0,u})_{\reg{aux}}\nonumber\\
      				&= \sum_u  \E_{\ba \in V} \omega^{\tr[\langle
                                 \ba, u \rangle]} Z(u) \ot (A_{0,u})_{\reg{aux}}\nonumber \\
                               &= \sum_{u \in V^\perp}  Z(u) \ot (A_{0,u})_{\reg{aux}}\tag{\Cref{fact:averages-to-zero-subspace}}\nonumber\\
                               &= \sum_{u \in V^\perp}  \sum_{w} \omega^{\tr[\langle w, u\rangle]} \ket{w}\bra{w}  \ot (A_{0,u})_{\reg{aux}} \nonumber\\
                               &= \sum_{w}\ket{w}\bra{w}\ot \sum_{u \in V^\perp}  \omega^{\tr[\langle w, u\rangle]}    (A_{0,u})_{\reg{aux}}.\label{eq:let's-use-this-in-a-sec}
    \end{align}
    Now, consider a surface $s \in \surfaces{v}$.  For some~$x \in \F_q^n$,
    $s$ is the set of points written $w = x + v$, where $v \in V$.
    Then for any $u \in V^\perp$, $\langle w, u\rangle = \langle x + v, u \rangle = \langle x, u\rangle$,
    a quantity which depends only on the subspace and not on the point~$w$.
    Call this quantity $c_{s, u}$.
    As a result,
    \begin{equation*}
    \eqref{eq:let's-use-this-in-a-sec}
    = \sum_{s \in \surfaces{v}} \sum_{w \in s} \ket{w}\bra{w}\ot  \sum_{u \in V^\perp}c_{s, u} (A_{0,u})_{\reg{aux}}
    = \sum_{s \in \surfaces{v}} \Pi^{v}_{s} \ot (\hat{A}_s)_{\reg{aux}},
    \end{equation*}
    where $\hat{A}_s = \sum_{u \in V^\perp}c_{s, u} A_{0,u}$.
  \end{proof}
  
\begin{lemma}\label{lem:proj-to-obs-subspace}
Let $W \in \{X, Z\}$, and let $v = \{v_1, \dots, v_k\}$ be a set of
$k$ linearly independent vectors in $\F_q^n$ and $V$ be their span.
Suppose $\{M_{a_2}\}$ is a measurement for which
\begin{equation}\label{eq:m-commutes-subspace-proj}
(M_{a_2} \cdot (\tau^{W}_{[\forall i, v_i \cdot a_1 = a_{1,i}]} \otimes I_{\reg{aux}})) \otimes I_{\reg{Bob}}
\approx_{\delta} ((\tau^{W}_{[\forall i, v_i \cdot a_1 = a_{1,i}]} \otimes I_{\reg{aux}}) \cdot M_{a_2}) \otimes I_{\reg{Bob}},
\end{equation}
where
\[ \tau^{W}_{[\forall i, v_i \cdot a_1 = a_{1,i}]} = \sum_{a_1 :
    \forall i, v_i \cdot a_1 = a_{1,i}} \tau^W_{a_1}. \]
Then
\begin{equation*}
(M_{a_2} \cdot (W(u)\otimes I_{\reg{aux}})) \otimes I_{\reg{Bob}}
\approx_{\delta} ((W(u) \otimes I_{\reg{aux}}) \cdot M_{a_2}) \otimes I_{\reg{Bob}},
\end{equation*}
for a uniformly random $\bu$ drawn from $V$.
\end{lemma}
\begin{proof}
To start, given a set of outcomes $a_{1,1}, \dots, a_{1,k}$, suppose
$u$ and $u'$ are outcomes for a full $W$-basis measurement consistent with
these outcomes, i.e.\ $u$ and $u'$ are vectors such that for all $i$, $u \cdot
v_i = a_{1,i}$. Then it must hold that $u - u' \in V^\perp$. Using
this, the bound in~\Cref{eq:m-commutes-subspace-proj} becomes
\begin{equation}\label{eq:m-commutes-subspace-proj2}
  \sum_{a_2} \frac{1}{|V^\perp|} \sum_{u} \Big\| \sum_{w \in V^\perp} (M_{a_2} \cdot
  (\tau^W_{u+w} \otimes I_{\reg{aux}}) - (\tau^{W}_{u+w} \otimes I_{\reg{aux}}) \cdot M_{a_2}) \otimes I_{\reg{Bob}} \ket{\psi}\Big\|^2 \leq \delta,
\end{equation}
where the factor of $1/|V^\perp|$ is because each outcome $a_{1,1},
\cdots, a_{1,k}$ corresponds to $|V^\perp|$ different choices of $u$.
  
Our goal is to bound
\begin{equation}\label{eq:gonna-bound-this}
\E_{\bu\sim V} \sum_{a_2} \Vert (M_{a_2} \cdot (W(\bu)\otimes I_{\reg{aux}}) - (W(\bu)\otimes I_{\reg{aux}}) \cdot M_{a_2}) \otimes I_{\reg{Bob}} \ket{\psi} \Vert^2.
\end{equation}
by $\delta$.
To do so, for a fixed~$u$ we introduce the notation
\begin{align}
\Delta_{a_2}^{u}
&:= M_{a_2} \cdot (W(u) \otimes I_{\reg{aux}})- (W(u)\otimes I_{\reg{aux}}) \cdot M_{a_2}\\
&= \sum_{x \in \F_q^n} \omega^{\tr[u \cdot x]} (\underbrace{M_{a_2} \cdot (\tau^{W}_{x} \otimes I_{\reg{aux}})
	- (\tau^{W}_{x} \otimes I_{\reg{aux}}) \cdot M_{a_2}}_{\Delta_{a_2,x}}).\label{eq:pauli-obs-commutator}
\end{align}
We record the following identity, which follows
from~\Cref{eq:pauli-obs-commutator} and~\Cref{fact:averages-to-zero-subspace}:
\begin{equation*}
\E_{\bu \sim V } (\Delta_{a_2}^{\bu})^\dagger\cdot \Delta_{a_2}^{\bu}
= \E_{\bu \sim V} \sum_{x, x' \in \F_q^n} \omega^{\tr[\bu \cdot (x' - x)]} (\Delta_{a_2,x})^\dagger \Delta_{a_2,x'}
= \sum_{x \in \F_q^n} \sum_{w \in V^\perp} (\Delta_{a_2,x})^\dagger \Delta_{a_2,x+w}.
\end{equation*}
As a result,
\begin{align*}
\eqref{eq:gonna-bound-this}=
\E_{\bu} \sum_{a_2}\Vert (\Delta_{a_2}^{\bu} \otimes I_{\reg{Bob}}) \ket{\psi}\Vert^2
&= \E_{\bu} \sum_{a_2} \bra{\psi} (\Delta_{a_2}^{\bu})^\dagger \Delta_{a_2}^{\bu} \otimes I_{\reg{Bob}} \ket{\psi}\\
&=  \sum_{a_2} \sum_{x \in \F_q^n} \sum_{ w\in V^\perp} \bra{\psi}(\Delta_{a_2,x})^\dagger \Delta_{a_2,x+w} \otimes I_{\reg{Bob}} \ket{\psi}\\
&= \sum_{a_2, x} \frac{1}{|V^\perp|} \Big\| \sum_{w \in V^\perp} \Delta_{a_2,x+w} \otimes I_{\reg{Bob}} \ket{\psi} \Big\|^2,
\end{align*}
where the factor of $1/|V^\perp|$ is again to deal with overcounting. But this is at most $O(\delta)$, by~\Cref{eq:m-commutes-subspace-proj2}. This completes the proof.
\end{proof}

  \begin{lemma}\label{lem:commute-to-twirlU}
    Let $\{M_{a}\}$ be a measurement and $\calU$ be a distribution over
    unitaries,
    and suppose that for $U$ drawn uniformly from $\calU$,
    \begin{equation*}
      ((U^\dagger \otimes I_{\reg{aux}}) \cdot M_{a}) \otimes I_{\reg{Bob}}
      \approx_{\delta} (M_{a} \cdot (U^\dagger \otimes I_{\reg{aux}})) \otimes I_{\reg{Bob}},
    \end{equation*}
    where the distribution inherent in the $\approx_{\delta}$ notation
    is the uniform distribution over $\calU$.
    Then
    \begin{equation*}
      M_{a} \otimes I_{\reg{Bob}} \approx_{\delta} (\twirl_{\calU} \otimes I_{\reg{aux}})[M_{a}] \otimes I_{\reg{Bob}}.
    \end{equation*}
  \end{lemma}
  
  \begin{proof}
    By definition,
\begin{equation*}
(\twirl_1 \otimes I_{\reg{aux}})[M_{a}]
= \E_{\bU \sim \calU} [(\bU \otimes I_{\reg{aux}}) \cdot M_a \cdot (\bU^\dagger \otimes I_{\reg{aux}})].
\end{equation*}
Similarly,
\begin{equation*}
M_{a}
= \E_{\bU \sim \calU} [( \bU \otimes I_{\reg{aux}}) \cdot
(\bU^\dagger \otimes I_{\reg{aux}})  \cdot M_a].
\end{equation*}
As a result, if we set
\begin{equation*}
B(\bU)_a = (\bU^\dagger \otimes I_{\reg{aux}}) \cdot M_a - M_a \cdot (\bU^\dagger \ot I_{\reg{aux}}),
\end{equation*}
then
\begin{equation*}
\Delta_a:=
M_a - (\twirl_1 \otimes I_{\reg{aux}})[M_a]
= \E_{\bU \sim \calU} [\bU \cdot B(\bU)_a ].
\end{equation*}
We can therefore establish the lemma as follows:
\begin{align*}
\sum_a \Vert \Delta_a \otimes I_{\reg{Bob}} \ket{\psi} \Vert^2
& = \sum_a \Vert \E_{\bU \sim \calU} [ \bU \cdot B(\bU)_a ] \otimes I_{\reg{Bob}} \ket{\psi} \Vert^2\\
& \leq \E_{\bU} \sum_a \Vert (\bU \cdot B(\bU)_a ) \otimes I_{\reg{Bob}} \ket{\psi} \Vert^2\tag{Jensen's inequality}\\
& = \E_{\bU} \sum_a \Vert B(\bU)_a \otimes I_{\reg{Bob}} \ket{\psi} \Vert^2 \tag{$\bU$ is unitary}
\end{align*}
By assumption, this quantity is $O(\delta)$. This concludes the proof.
\end{proof}

\begin{proof}[Proof of \Cref{thm:partial-data-hiding-game}]
  We consider the completeness and soundness cases separately.
  \paragraph{Completeness}
  Let $\calS_{\mathrm{partial}} = (\psi, M^{\bot, x,v})$ be a partial
  $(k,n,q)$-register strategy for~$\game_{\mathrm{hide}}(S,x)$ which is also a real commuting EPR strategy, and for which
  the measurement $M^{\bot, x,v}_{a_2}$ has the form
  \[ M^{\bot, x,v}_{a_2} = \sum_{s \in \surfaces{v}} \Pi^{v}_s \ot A^{s,x,v}_{a_2} . \]
  To this strategy we will add matrices for the remaining questions.
  \begin{itemize}
  \item[$\circ$] Question $(Z, x)$: the measurement is
    \[ M^{(Z,x,v)}_{a_1, a_2} = \sum_{s \in \surfaces{v}} \Pi^{v}_s \cdot
      \tau^{Z}_{a_1} \ot A^{s,x,v}_{a_2}. \]
    This is a well-defined measurement as $\Pi^{v}_s$ is diagonal in the
    $Z$ basis and thus commutes with $\tau^Z_{a_1}$.
  \item[$\circ$] Question $(\bot, x, \{X, v_1, \dots, v_k\})$: the
    measurement is
    \[ M^{(\bot, x, \{X, v\})}_{\{a_{1,1}, \dots,
        a_{1,k}\}, a_2} = \sum_{s \in \surfaces{v}} \Pi^{v}_s \cdot
      \tau^{X}_{[\forall i, a_1 \cdot v_i = a_{1,i}]}  \ot
      A^{s,x,v}_{a_2}.\]
    This is a well-defined measurement as $\Pi^{v}_s$ commutes with
    $\tau^{X}_{[\forall i, a_1 \cdot v_i = a_{1,i}]}$
    by~\Cref{lem:subspace-x-z-commute}.
  \item[$\circ$] Question ($\bot, \bot, \{X, v_1, \dots, v_k\})$: the
    measurement is
    \[ M^{(\bot, \bot, \{X, v\})}_{\{a_{1,1}, \dots, a_{1,k},
        \varnothing} = \tau^X_{[\forall i, a_1 \cdot v_i = a_{1,i}]}
      \ot I. \]
  \item[$\circ$] Question $(X, \bot)$: the measurement is
    \[ M^{(X, \cdot)}_{a_1} = \tau^{X}_{a_1} \otimes I_{\reg{aux}}. \]
  \end{itemize}
This is a $(1,n,q)$-register strategy for~$\game_{\mathrm{hide}}$ by
design, and it is not hard to see that it achieves value $1$ on the
game. Assuming that the partial strategy $\cal{S}_{\mathrm{partial}}$
is a real commuting EPR strategy, it is not hard to see that the full
strategy above is also real (this is because if $M^{\bot, x,v}$ is real and
of the given form, then the matrices $A^{s,x,v}_{a_2}$ must also be
real). That the strategy is also commuting follows from the
description of $\game_{\mathrm{hide}}$. In particular, note that while
$M^{Z,x,v}_a$ does \emph{not} commute with $M^{\bot, x \{X, v\}}$ or
with $M^{X, \bot}$, the test never requires these measurements to be measured at the same time.

  \paragraph{Soundness}
   Recall that a strategy $\calS$ for this game
  consists of a state $\ket{\psi} = \ket{\epr^n_q} \ot \ket{\rmaux}$ and measurement operators of three types, corresponding to
  the four types of queries: $M^{(\bot,
    x,v)}_{\varnothing, a_2}$, $M^{(Z, x, v)}_{a_1, a_2}$, $M^{\bot, x, \{X, v\}, }_{\{a_{1,1}',
    \dots, a_{1,k}'\}, a_2'}$, and $M^{(X, \cdot)}_{a_1}$. 

  We start by analyzing the third and fourth parts of the test. The goal of these
  parts of the test is to certify that when given the query $(\{X, v\}, x)$, the prover returns $k$ answers $\ba_{1,1}',
  \dots, \ba_{1,k}'$ that are consistent with measuring the $X(v_1),
  \dots, X(v_k)$ observables on the state. We certify this in two
  stages. In part three of the test, we ask
  the first prover to do a complete measurement in the $X$ basis, and send
  the second prover the query $\{X, v\}$ indicating that it is to
  perform a partial $X$ measurement, and
  check consistency of outcomes. Importantly, in this part of the
  test, we \emph{cannot} send the second prover the query $x$, since
  the corresponding $\Pi^v_s$ measurement does not commute with the
  complete $X$ measurement performed by the first prover. Thus, in
  part four of the test, we send one prover the query $\{X, v\}$ and
  the other $(x, \{X, v\})$, and check consistency of their
  outcomes. 
  
  Since $\calS$ is a
  $(k,n,q)$-register strategy, \Cref{eq:in-math} implies that
  $M_{a_1}^{(X,\cdot)} = \tau^{X}_{a_1} \otimes I_{\reg{Bob}}$. We thus have
  from the third part of the test that
  \[ M^{\bot, \bot, \{X, v\}}_{\{a_{1,1}', \dots,
      a_{1,k}'\}} \ot I_{\reg{Bob}} \approx_\eps \tau^{X}_{[\forall i,
      u \cdot v_i = a'_{1,i}]} \ot I_{\reg{Bob}}.\]
  From the above equation and the fourth part of the test, we have
  \[ M^{\bot, x, \{X, v\}}_{\{a_{1,1}', \dots, a_{1,k}'\}} \ot
    I_{\reg{Bob}} \approx_\eps M^{\bot, \bot, \{X,v\}}_{\{a_{1,1}',
      \dots, a_{1,k}'\}} \ot I_{\reg{Bob}} \approx_\eps \tau^X_{[\forall i, u \cdot v_i =
      a'_{1,i}]} \ot I_{\reg{Bob}}. \]
  
  Next, we look at the first and second parts of $\game_{\mathrm{hide}}(S,x)$. These are
  essentially two instances of the commutation test. The first part of
  $\game_{\mathrm{hide}}$ certifies that the second outcome of $M^{Z,
    x,v}_{a_1, a_2}$ is consistent with $M^{\bot, x,v}_{\varnothing,
    a_2}$, and the hypothesis that the strategy $\calS$ is a
  $(1,n,q)$-register strategy tells us that the first outcome of
  $M^{Z,x,v}_{a_1, a_2}$ is consistent with $\tau^Z_{a_1} \ot I_{\reg{Bob}}$. Thus,
  applying the analysis of the commutation, it follows that
  \[ (M^{\bot, x,v }_{\varnothing, a_2} \cdot (\tau^Z_{a_1} \ot I_{\reg{aux}})) \otimes I_{\reg{Bob}} \approx_{\eps}
    ((\tau^Z_{a_1} \ot I_{\reg{Bob}})\cdot M^{\bot, x,v}_{\varnothing, a_2}) \otimes I_{\reg{Bob}}. \]
  A similar analysis for the second part of $\game_{\mathrm{hide}}(S)$
  certifies that
\[
((\tau^{X}_{[\forall i, a_1 \cdot v_i = a'_{1,i}]} \otimes
I_{\reg{aux}})\cdot M_{\varnothing, a_2}^{\bot,x, v}) \otimes I_{\reg{Bob}}
\approx_{\eps} (M_{\varnothing, a_2}^{\bot, x, v} \cdot(\tau^{X}_{[\forall i, a_1 \cdot v_i = a'_{1,i}]} \otimes I_{\reg{aux}})) \otimes I_{\reg{Bob}}.
\]

As a result, it holds that $W \in \{X,Z\}$, by \Cref{lem:proj-to-obs-subspace}, 
\begin{equation*}
(M_{\varnothing, a_2}^{\bot, x , v}\cdot (W(u)\otimes I_{\reg{aux}})) \otimes I_{\reg{Bob}}
\approx_{\eps} ((W(u)\otimes I_{\reg{aux}}) \cdot M_{\varnothing,
  a_2}^{\bot, x, v}) \otimes I_{\reg{Bob}}.
\end{equation*}
where if $W = X$, then $u$ is chosen uniformly over $V$, and
if $W = Z$, then $u$ is drawn uniformly from $\F_q^n$.
As a result, by~\Cref{fact:the-ol-pauli-swaperoonie} and \Cref{fact:add-a-proj},
\begin{align*}
(M_{\varnothing, a_2}^{\bot, x, v} \cdot (Z(u') X(u)\otimes I_{\reg{aux}})) \otimes I_{\reg{Bob}}
& \approx_{0}(M_{\varnothing, a_2}^{\bot,x,v} \cdot (Z(u')\otimes I_{\reg{aux}})) \otimes (X(-u)\otimes I_{\reg{aux}})\\
& \approx_{\eps}((Z(u')\otimes I_{\reg{aux}}) \cdot M_{\varnothing, a_2}^{\bot,x,v}) \otimes (X(-u)\otimes I_{\reg{aux}})\\
& \approx_{0}((Z(u')\otimes I_{\reg{aux}} )\cdot M_{\varnothing,a_2}^{\bot,x,v} \cdot (X(u) \otimes I_{\reg{aux}})) \otimes I_{\reg{Bob}}\\
& \approx_{\eps}((Z(u') X(u)\otimes I_{\reg{aux}}) \cdot
                                                                                                                                    M_{\varnothing,
                                                                                                                                    a_2}^{\bot,x,
                                                                                                                                    v})
                                                                                                                                    \otimes I_{\reg{Bob}},
\end{align*}
on the distribution where $v$ is chosen from $S$, and then $\bu \sim V$, $\bu' \sim \F_q^n$.
Applying \Cref{lem:commute-to-twirlU} and \Cref{prop:hidden-subspace-twirl},
we can therefore conclude that
\begin{equation*}
M_{\varnothing, a_2}^{\bot, x,v} \otimes I_{\reg{Bob}}
\approx_{\eps} (\twirl_{\calS} \otimes I_{\reg{aux}})[M_{\varnothing, a_2}^x]
\otimes I_{\reg{Bob}}
= \left(\sum_{s \in \surfaces{v}} \Pi^{v}_s \ot (A^{x,v,s}_{a_2})_{\reg{aux}}\right)
\otimes I_{\reg{Bob}}, 
\end{equation*}
for some choice of measurements $A^{x,v,s}_{a_2}$ on the aux register.
\end{proof}


\part{$\neexp$ protocol}

\label{part:neexp}


\section{A review of a classical PCP theorem}\label{sec:classical-pcp}

We begin \cref{part:neexp} by reviewing the following classical PCP theorem:
\begin{equation}
\succinct \in \mathrm{PCP}[n, \poly(n)].
\end{equation}
This implies, by standard reduction, that $\succinct \in \mathrm{MIP}$,
which is the main result of~\cite{BFL91}.
Reviewing this serves two purposes:
(i) our $\mip^*$ protocols are inspired by this PCP construction,
and (ii) their correctness is actually shown by reduction to the correctness of this PCP construction (\Cref{lem:polynomial-means-good} below).
This section closely follows the excellent course notes of Harsha~\cite{Har10}.

\subsection{The instance}

The input to the verifier is an instance of the $\succinct$ problem,
i.e.\ a circuit~$\calC$  of size~$s$ with $3n+3$ inputs.
We apply the Tseitin transformation to it to produce a formula~$\calF$ with $n' = 3n+3+s$ inputs.
It encodes the $\sat$ formula $\psi := \psi_{\calF}$ on~$N = 2^n$ variables
which contains $(x_i^{b_1} \lor x_j^{b_2} \lor x_k^{b_3})$ 
as a clause
if and only if~$\calF(i, j, k, b_1, b_2, b_3, w) = 1$ for some $w \in \{0, 1\}^s$.

\subsection{Encoding assignments}

Writing $\calS = \{0, 1\}^n$, an assignment to the variables of~$\psi$ is a string $a \in \{0,1\}^{\calS}$, or equivalently a string in $\{0, 1\}^N$.
The first step of the PCP theorem is to take the low-degree encoding of~$a$.
We begin by choosing parameters.

\begin{definition}\label{def:admissible}
Recall that
$N = 2^n$,
$h= 2^{t_1}$,
$q = 2^{t_2}$,
and $m$ are admissible parameters if $t_1 \leq t_2$ and $h^m \geq N$.
We call them \emph{exactly} admissible if the stronger condition $h^m = N = 2^n$ holds.
\end{definition}

We select $n$, $h = 2^{t_1}$, $q = 2^{t_2}$, and $m$ to satisfy our ``rule of thumb" parameter settings (\Cref{eq:parameters}):
\begin{equation*}
	h = \Theta(n), \quad m = \Theta\left(\frac{n}{\log(n)}\right), \quad q = \Theta((n')^{10}).
\end{equation*}
Note that $q$ depends on~$n'$ rather than just~$n$.

It remains to choose~$H$ and~$\pi$.
Our construction requires that the permutation~$\pi$ be efficiently computable,
and so we pick these according to the canonical low-degree encoding (\Cref{def:canonical-low-degree}).
This entails setting $H= H_{t_1, t_2}$.
As for $\pi$, we modify the construction slightly.
This is because the canonical low-degree encoding is designed for strings whose coordinates are indexed by an integer $i \in[n]$,
which must first be converted to binary when computing~$\pi$.
However, the coordinates of our strings $a \in \{0, 1\}^{\calS}$ are indexed by elements of $\calS = \{0, 1\}^n$,
which are already written in binary, allowing us to skip the conversion.
Hence, within this section, we define $\pi:=\pi_{n, t_1, t_2} :\calS \rightarrow H^m$ by setting
\begin{equation*}
\pi(b_1, \ldots, b_n) := \sigma_{n, t_1, t_2}(b_1, \ldots, b_n)= (\sigma(b_1, \ldots, b_{t_1}), \sigma(b_{t_1+1}, \ldots, b_{2t}), \ldots, \sigma(b_{n-t_1+1}, \ldots, b_n)),
\end{equation*}
where $\sigma:=\sigma_{t_1, t_2}$.
Given these parameters,
an assignment~$a$ is encoded as a degree-$O(mh)$ polynomial $g_a:\F_q^m \rightarrow \F_q$.

The crucial property of~$\pi$ that we will need later is that it has an efficiently-computable, low-degree inverse.
We will show this here.
To do so, we begin by recalling the notation $\indicator{H}{x}{y}$ for the indicator function of~$x \in H$ over~$H$:
\begin{equation*}
\indicator{H}{x}{y} = \frac{\prod_{b \neq x}(y - b)}{\prod_{b \neq x}(x-b)},
\end{equation*}
where~$b$ ranges over~$H$.
This is a degree-$h$ polynomial which can be computed in time $\poly(h, q)$.
\begin{definition}
Let $N= 2^n$, $h = 2^{t_1}$, $q= 2^{t_2}$, and $m$ be exactly admissible parameters.
Set $H = H_{t_1, t_2}$, $\sigma = \sigma_{t_1, t_2}$, and $\pi = \pi_{n, t_1, t_2}$.
Consider the function $\mu:= \mu_{t_1, t_2}:\F_q \rightarrow \F_q^{t_1}$ whose $i$-th component is defined as
\begin{equation*}
\mu_i(y) = \sum_{x \in H : \tr[e_i \cdot x] = 1} \indicator{H}{x}{y}.
\end{equation*}
Let $y = b_1 \cdot e_1 + \cdots + b_{t_1} \cdot e_{t_1}$ be an element of $H$.
Then $\mu_i(y) = b_i$, and so $\mu(y) = (b_1, \ldots, b_{t_1})$.
This means that $\mu(\sigma(b_1, \ldots, b_{t_1})) = (b_1, \ldots, b_{t_1})$.
As a result, if we define the function $\nu := \nu_{n, t_1, t_2} : \F_q^m \rightarrow \calS_n$ to be
\begin{equation*}
\nu(x_1, \ldots, x_m) := (\mu(x_1), \ldots, \mu(x_m))
\end{equation*}
then $\nu(\pi(x)) = x$ for any $x \in \calS_n$.
Each component of $\nu$ is the sum of $\frac{h}{2}$ indicator functions,
and is therefore degree-$h$ and computable in time $\poly(h,q)$.
As a result, $\nu$ is computable in time $\poly(n,h,q)$.
\end{definition}

\subsection{Encoding the formula}

Our next step is to produce a similar ``low-degree encoding" for the formula~$\psi$.
This will be a function $g_\psi: \F_q^{m'} \rightarrow \F_q$, for $m' = 3m + 3+s$, with the property that for all $i, j, k \in \calS$, $b_1, b_2, b_3 \in \{0, 1\}$, and $w \in \{0, 1\}^s$,
\begin{equation*}
g_\psi(\pi(i), \pi(j), \pi(k), b_1, b_2, b_3, w) = \calF(i, j, k, b_1, b_2, b_3, w).
\end{equation*}
This can be accomplished by setting $\calS' := \{0, 1\}^{n'}$, viewing $\calF$ as computing a string $a_{\psi} \in \{0, 1\}^{\calS'}$, and setting $g_{\psi}$ to be its low-degree encoding.
However, the verifier in our protocol will be required to evaluate $g_\psi$ on a particular input,
and this seems challenging given that this $g_{\psi}$ is computed by interpolating over an exponential number of points.
Instead, we will produce a $g_\psi$ which we can efficiently evaluate at any point using the fact that we have a succinct formula~$\calF$ representing~$\psi$.

To begin, we convert~$\calF$ into an algebraic formula which operates over $\F_q$-valued inputs using arithmetization (cf.\ \Cref{def:arithmetization}).
Set $\calF_{\mathrm{arith}} := \arith{q}{\calF}$.
This is a  function $\calF_{\mathrm{arith}}:\F_q^{n'}\rightarrow \F_q$
with the property that for any $x \in \{0, 1\}^{n'}$, $\calF_{\mathrm{arith}}(x) = \calF(x)$.
Furthermore, $\calF_{\mathrm{arith}}$ has degree~$O(n')$ and is computable in time $\poly(n', q)$.
We can now define the function $g_\psi$ as follows.
\begin{definition}\label{def:encoded-function}
Let $N=2^n$, $h = 2^{t_1}$, $q= 2^{t_2}$, and $m$ be exactly admissible parameters.
Set $\nu:=\nu_{n, t_1, t_2}$.
Let $\calC$ be a $\succinct$ instance whose Tseitin transformation~$\calF$
has $n' = 3n + 3 + s$ inputs and encodes the formula $\psi := \psi_{\calF}$,
and let $\calF_{\mathrm{arith}} = \arith{q}{\calF}$.
Write $m' = 3m + 3 + s$.
Then we define $g_\psi := g_{\psi, n, t_1, t_2} : \F_q^{m'} \rightarrow \F_q$ to be the function
\begin{equation*}
g_\psi(x_1, x_2, x_3, b_1, b_2, b_3, w) := \calF_{\mathrm{arith}}(\nu(x_1), \nu(x_2), \nu(x_3), b_1, b_2, b_3, w).
\end{equation*}
This is degree $h \cdot O(n')$ and can be computed in time $\poly(n',h,q)$.
\end{definition}
For shorthand, we will often write inputs to $g_\psi$ as tuples $(x, b, w) \in \F_q^{3m+3+s}$, where $x = (x_1, x_2, x_3)$ contains three strings in $\F_q^m$
and $b = (b_1, b_2, b_3)$ contains three numbers in $\F_q$.

\subsection{Zero on subcube}

Given a function $g:\F_q^m \rightarrow \F_q$,
we would like to check that it is the low-degree encoding of an assignment which satisfies~$\psi$.
To do so, we define the following function.
\begin{definition}\label{def:formula-function}
Let $N=2^n$, $h = 2^{t_1}$, $q= 2^{t_2}$, and $m$ be exactly admissible parameters.
Let $\calC$ be a $\succinct$ instance whose Tseitin transformation~$\calF$
has $n' = 3n + 3 + s$ inputs and encodes the formula $\psi := \psi_{\calF}$,
and let $g_\psi := g_{\psi, n, t_1, t_2}$.
Set $m' = 3m + 3 + s$.
Then given a function $g:\F_q^m \rightarrow \F_q$, 
we define $\mathrm{sat}_{\psi,g} := \mathrm{sat}_{\psi,g, n, t_1, t_2} : \F_q^{m'} \rightarrow \F_q$ to be the function
\begin{equation*}
\mathrm{sat}_{\psi,g}(x, b,w) := g_{\psi}(x, b, w) \cdot (g(x_1)-b_1) (g(x_2)-b_2) (g(x_3)-b_3).
\end{equation*}
\end{definition}
The crucial property we would like to check is that $\mathrm{sat}_{\psi, g}$ is \emph{zero on the subcube $H_{\mathrm{zero}} := H^{3m} \otimes \{0, 1\}^{3+s}$}.
\begin{proposition}\label{prop:zero-on-subcube}
The function $\mathrm{sat}_{\psi,g}$ is zero on the subcube~$H_{\mathrm{zero}}$ for some $g:\F_q^m\rightarrow \F_q$ if and only if~$\psi$ is satisfiable.
If it \emph{is} satisfiable, $g$ may be taken to be degree-$O(mh)$, in which case $\mathrm{sat}_{\psi,g}$ is  degree-$O(mh + h n')$.
\end{proposition}
Note what \Cref{prop:zero-on-subcube} does \emph{not} say.
It does not say that if $\mathrm{sat}_{\psi,g}$ is zero on the subcube, then~$g$ is the low degree encoding of a satisfying assignment of~$\psi$.
It does not even say that~$g$ must be low-degree.
(Indeed, $g$ might have high degree, as $\mathrm{sat}_{\psi,g}$ only checks~$g$ on those numbers in the range of~$\pi$.)
What it \emph{does} say is that \emph{if} $\psi$ is satisfiable, \emph{then} there exists such a $g$ which is low-degree: the low-degree encoding of a satisfying assignment.
Our strategy, then, will be to verify that that $g$ is low-degree and then use this fact to verify that $\mathrm{sat}_{\psi,g}$ is zero on the subcube.
(We can then ``forget" that~$g$ is low-degree, as it is no longer required for the analysis.)

To verify this that $\mathrm{sat}_{\psi,g}$ is zero on~$H_{\mathrm{zero}}$, we would like it to be encoded so that this is self-evidently true.
This entails expanding $\mathrm{sat}_{\psi,g}$ in a ``basis" of simple polynomials which are zero on the subcube.
To begin, given a subset $S \subseteq \F_q$, define
\begin{equation*}
\mathrm{zero}_S(x) := \prod_{b \in S}(x-b).
\end{equation*}
The following proposition shows how to expand into this ``zero" basis.
\begin{proposition}\label{prop:coefficient-polys}
Let $f:\F_q^{n} \rightarrow \F_q$ be a degree-$d$ polynomial which is zero on the subcube $H = H_1 \otimes \cdots \otimes H_n$. Then there exist degree-$(d-h)$ ``coefficient polynomials" $c_1, \ldots, c_n$ such that
\begin{equation*}
f(x) = \mathrm{zero}_{H,c}(x):= \sum_{i=1}^{n} \mathrm{zero}_H(x_i) \cdot c_i(x).
\end{equation*}
\end{proposition}

For simplicity, we will write $\mathrm{zero}_{H,c}$ instead of $\mathrm{zero}_{H_{\mathrm{zero}},c}$.
We would like our proof to consist of the function~$g$ and the coefficient polynomials $c_1,\ldots, c_{m'}$ so that we may check the equality
$\mathrm{sat}_{\psi,g} \equiv \mathrm{zero}_{H,c}$.
The following lemma shows so long as these functions are low-degree,
we can verify that they are equal, and therefore show that $\psi$ is satisfiable.

\begin{lemma}\label{lem:the-grand-finale}
Let $N=2^n$, $h = 2^{t_1}$, $q= 2^{t_2}$, and $m$ be exactly admissible parameters.
Let $\calC$ be a $\succinct$ instance whose Tseitin transformation~$\calF$
has $n' = 3n + 3 + s$ inputs and encodes the formula $\psi := \psi_{\calF}$,.
Set $m' = 3m + 3 + s$.
Let $g:\F_q^m \rightarrow \F_q$, and set $\mathrm{sat}_{\psi,g} := \mathrm{sat}_{\psi,g, n, t_1, t_2}$.
Let $c_1, \ldots, c_{m'}: \F_q^{m'} \rightarrow \F_q$, set $H_{\mathrm{zero}} = H^{3m} \otimes \{0, 1\}^{3 + s}$, and write $\mathrm{zero}_{H_c} := \mathrm{zero}_{H_{\mathrm{zero}},  c}$.
Suppose that $g$ is degree-$d_1$, and suppose that $c_1, \ldots, c_{m'}$ are degree-$d_2$.
Suppose
\begin{equation*}
\Pr_{\bx \sim \F_q^{m'}}[\mathrm{sat}_{\psi,g}(\bx) =  \mathrm{zero}_{H,c}(\bx)] > \frac{\max\{O(h n') + 3 d_1, h + d_2\}}{q}.
\end{equation*}
Then $\psi$ is satisfiable.
\end{lemma}
\begin{proof}
By \Cref{def:encoded-function},
$\mathrm{sat}_{\psi, g}$ has degree $h \cdot O(n') + 3 d_1$.
In addition, $\mathrm{zero}_{H,c}$ has degree $h + d_2$.
Define $f = \mathrm{sat}_{\psi,g} -  \mathrm{zero}_{H,c}$.
Then~$f$ has degree $\max\{O(h n') + 3 d_1, h + d_2\}$.
By assumption, $f(\bx) = 0$ with probability larger than $\deg(f)/q$ over a uniformly random $\bx \sim \F_q^{m'}$.
By Schwartz-Zippel (\Cref{lem:schwartz-zippel}), this means that $f \equiv 0$, which implies that $\mathrm{sat}_{\psi, g} \equiv \mathrm{zero}_{H,c}$.
But $\mathrm{zero}_{H,c}$ is self-evidently zero on the subcube $H_{\mathrm{zero}}$,
meaning that $\mathrm{sat}_{\psi,g}$ is as well.
By \Cref{prop:zero-on-subcube}, $\psi$ is satisfiable.
\end{proof}

Ensuring that $\mathrm{sat}_{\psi,g}$ is low-degree requires ensuring that~$g$ is low-degree, and this can be accomplished with the low-degree test.
Arguing similarly for $\mathrm{zero}_{H,c}$ requires ensuring that each~$c_i$ is low-degree, and this can be done with the simultaneous plane-versus-point low-degree test (\Cref{thm:simultaneous-raz-safra}).

\subsection{The PCP}\label{sec:the-pcp}

We can now state the contents of our probabilistically checkable proof for the satisfiability of~$\psi$.
It consists of the following four tables.
\begin{enumerate}
\itemsep -.5pt
\item A claimed low-degree polynomial $g:\F_q^m \rightarrow \F_q$.
\item A set of claimed low-degree polynomials $c_1, \ldots, c_{m'} : \F_q^{m'} \rightarrow \F_q$.
\item A ``planes table", containing for each plane~$s$ in $\F_q^m$ a degree-$d$ bivariate polynomial.
\item Another planes table, containing for each plane~$s$ in $\F_q^{m'}$ an $m'$-tuple of degree-$d$ bivariate polynomials.
\end{enumerate}
The verifier works as follows:
first, it performs the low-degree test between $g$ and its planes table.
Second, it performs the simultaneous low-degree test between the $c_i$'s and their plane table.
Both of these use the degree parameter $d = \Theta((n')^2)$,
which is chosen to upper-bound both $\Theta(mh)$ and $\Theta(mh + hn')$.
Finally, it picks a uniformly random $(\bx, \bb, \bw) \in \F_q^{m'}$
and checks the equality $\mathrm{sat}_{\psi,g}(\bx, \bb, \bw) = \mathrm{zero}_{H,c}(\bx, \bb, \bw)$.
It accepts if all the tests accept individually.

When $\psi$ is satisfiable, there is always a proof that makes the verifier accept with probability~$1$.
This entails setting $g$ to be the low-degree encoding of a satisfying assignment,
and setting $c_1, \ldots, c_{m'}$ to be the coefficient polynomials of $\mathrm{sat}_{\psi,g}$.
The following proposition shows that when~$\psi$ is not satisfiable, the verifier always rejects with probability at least~$\tfrac{1}{10}$.
\begin{proposition}
If the verifier accepts with probability at least $9/10$, then $\psi$ is satisfiable.
\end{proposition}
\begin{proof}
If the verifier accepts with probability at least $9/10$,
then each individual test accepts with probability at least $9/10$.
Applying \Cref{thm:raz-safra,thm:simultaneous-raz-safra},
we get degree-$d$ functions $\overline{g}:\F_q^n \rightarrow \F_q$ and $\overline{c}_1, \ldots, \overline{c}_{m'}:\F_q^{m'} \rightarrow \F_q$
such that
\begin{equation*}
\mathrm{dist}(g, \overline{g}) \leq \tfrac{2}{10}, \quad
\mathrm{dist}(c, \overline{c}) \leq \tfrac{2}{10}, \quad
\mathrm{dist}(\mathrm{sat}_{\psi,g}, \mathrm{zero}_{H,c}) \leq \tfrac{1}{10}.
\end{equation*}
(Here, we are assuming that $q$ is a sufficiently large function of~$m$ and~$h$.)
By the union bound, $\mathrm{dist}(\mathrm{sat}_{\psi,g}, \mathrm{sat}_{\psi,\overline{g}}) \leq 3\mathrm{dist}(g, \overline{g})$.
As a result, by the triangle inequality
\begin{align*}
\mathrm{dist}(\mathrm{sat}_{\psi,\overline{g}}, \mathrm{zero}_{H,\overline{c}})
&\leq \mathrm{dist}(\mathrm{sat}_{\psi,\overline{g}}, \mathrm{sat}_{\psi, g}) +
	\mathrm{dist}(\mathrm{sat}_{\psi, g}, \mathrm{zero}_{H,c})+
	\mathrm{dist}(\mathrm{zero}_{H,c} + \mathrm{zero}_{H,\overline{c}})\\
&\leq 3 \mathrm{dist}(g, \overline{g})+
	\mathrm{dist}(\mathrm{sat}_{\psi, g}, \mathrm{zero}_{H,c})+
	\mathrm{dist}(c, \overline{c}) \leq 3 \cdot \tfrac{2}{10} + \tfrac{2}{10} + \tfrac{1}{10} = \tfrac{9}{10}.
\end{align*}
By \Cref{lem:the-grand-finale}, $\psi$ is therefore satisfiable.
\end{proof}

\paragraph{Time and communication complexity.}
\begin{itemize}
\itemsep -.5pt
\item[$\circ$] \textbf{Question length:} The verifier performs two low-degree tests and draws a random point in $\F_q^{m'}$.
			These are of size $\Theta(m \log(q))$, $\Theta(m'\log(q))$, and $\Theta(m'\log(q))$, respectively,
			all of which are $O(n')$ bits.
\item[$\circ$] \textbf{Answer length:} The verifier performs one normal low-degree test, and then a second low-degree test with answer complexity $m'$ times the normal answer complexity. These are of total length $(m' + 1) \cdot d^2 \log(q) = O((n')^9)$. Finally, in the last test, it queries each of $g$ and $c_1, \ldots, c_{m'}$ for a point in $\F_q$, a total communication cost of $(m'+1) \log(q) = O(n')$. In total, the answer length is $\poly(n')$.
\item[$\circ$] \textbf{Runtime:} The verifier runs in time $\poly(n')$. This includes computing $\mathrm{sat}_{\psi,g}(\bx, \bb, \bw)$,
						which requires computing $g_{\psi}(\bx, \bb, \bw)$, taking time $\poly(n', h, q) = \poly(n')$.
\end{itemize}



\section{$\neexp$ preliminaries}

\subsection{Introspection games}

In this section, we introduce introspection games.
These are games in which, rather than the verifier sampling the questions,
the provers sample them instead.

\begin{definition}[Introspection games]
An \emph{introspection game} is played between two provers Alice and Bob
in which Alice returns two strings~$x_A$ and~$a$
and Bob returns two strings~$x_B$ and~$a'$ (the verifier does not specify a question).
Here, $x_A$ and~$x_B$ are interpreted as Alice and Bob's ``share" of the jointly sampled ``question" $x = (x_A, x_B)$,
and~$a$ and~$a'$ are interpreted as their ``answers".
The verifier then applies its \emph{evaluation function} $V$ to the answers
and accepts if $V(x_A, x_B, a, b)=1$.
\end{definition}

The following three facts show that we can convert between strategies for a ``normal" game and strategies for an introspective version of the game.
This allows us to prove soundness results for the ``normal" game and ``bootstrap" them up to the introspective game as well.

\begin{fact}\label{fact:introspective-outrospective}
This fact concerns two games and two strategies.
\begin{enumerate}
\item Let $\game_{\mathrm{intro}}$ be the introspective game with evaluation function $V$.
Consider a strategy $\calS_{\mathrm{intro}}$ for Alice and Bob with shared state $\ket{\mathrm{intro}} = \ket{\mathrm{question}}\otimes \ket{\mathrm{answer}}$
in which Alice and Bob's measurements are given by
\begin{equation*}
\{\mathsf{P}_{x_A} \otimes A^{x_A}_a\}_{x_A, a}, \quad \{\mathsf{Q}_{x_B} \otimes B^{x_B}_{a'}\}_{x_B, a'},
\end{equation*} 
respectively.
Write $\calD$ for the distribution on outcomes $(x_A, x_B)$
when the measurement $\{\mathsf{P}_{x_A} \otimes \mathsf{Q}_{x_B}\}_{x_A, x_B}$ is performed on $\ket{\mathrm{question}}$.
\item Let $\game$ be the ``normal" game played as follows: sample $\bx = (\bx_A, \bx_B) \sim \calD$. Distribute the questions as follows:
	\begin{itemize}
	\item[$\circ$] Alice: give $\bx_A$; receive~$\ba$.
	\item[$\circ$] Bob: give $\bx_B$; receive~$\bb$.
	\end{itemize}
	Accept if $V(\bx_A, \bx_B, \ba, \bb) = 1$.
	Write $\calS$ for the strategy with shared state $\ket{\mathrm{answer}}$ in which Alice's strategy is $\{A^{x_A}_a\}_a$
	and Bob's strategy is $\{B^{x_B}_{a'}\}_{a'}$.
\end{enumerate}
Then $\valstrat{\game}{\calS} = \valstrat{\game_{\mathrm{intro}}}{\calS_{\mathrm{intro}}}$.
\end{fact}

\begin{fact}\label{fact:introspectivize-it}
Let $\{A^x_a\}_x$ and $\{B^x_a\}_x$ be measurements such that
\begin{equation}\label{eq:outro-approx}
A^x_a \otimes I_{\reg{Bob}} \approx_\delta B^x_a \otimes I_{\reg{Bob}}
\end{equation}
on state $\ket{\mathrm{answer}}$ and distribution~$\calD$. Next, let $\{\mathsf{Q}_x\}_x$ be a measurement and $\ket{\mathrm{question}}$ be a bipartite state
such that the distribution of measurement outcomes produced by measuring $\{\mathsf{Q}_x \otimes I_{\reg{Bob}}\}_x$ on $\ket{\mathrm{question}}$ is $\calD$.
Then
\begin{equation}\label{eq:intro-approx}
(\mathsf{Q}_x \otimes A^x_a)_{\reg{Alice}} \otimes I_{\reg{Bob}} \approx_{\delta} (\mathsf{Q}_x \otimes B^x_a)_{\reg{Alice}} \otimes I_{\reg{Bob}}
\end{equation}
on state $\ket{\mathrm{question}}\otimes
\ket{\mathrm{answer}}$. Moreover, if $\mathsf{Q}_x$ is a projective
measurement, then the reverse implication holds:
if~\Cref{eq:intro-approx} holds on $\ket{\mathrm{question}} \otimes
\ket{\mathrm{answer}}$, then~\Cref{eq:outro-approx} holds on the state $\ket{\mathrm{answer}}$.
\end{fact}

\begin{proof}
First, we show the forward implication. By definition, we want to bound
\begin{align*}
&\sum_{x, a} \Vert (\mathsf{Q}_x \otimes A^x_a \otimes I_{\reg{Bob}} - \mathsf{Q}_x \otimes B^x_a \otimes I_{\reg{Bob}}) \ket{\mathrm{question}}\otimes \ket{\mathrm{answer}}\Vert^2\\
=& \sum_{x, a} \Vert (\mathsf{Q}_x \otimes I)_{\reg{question}} \otimes (A^x_a \otimes I -   B^x_a \otimes I)_{\reg{answer}} \ket{\mathrm{question}}\otimes \ket{\mathrm{answer}}\Vert^2\\
=& \sum_{x, a} \bra{\mathrm{question}}\otimes \bra{\mathrm{answer}}
	(\mathsf{Q}_x \otimes I)^2_{\reg{question}} \otimes (A^x_a \otimes I -   B^x_a \otimes I)_{\reg{answer}}^2 \ket{\mathrm{question}}\otimes \ket{\mathrm{answer}}\\
\leq& \sum_{x, a} \bra{\mathrm{question}}\otimes \bra{\mathrm{answer}}
	(\mathsf{Q}_x \otimes I)_{\reg{question}} \otimes (A^x_a \otimes I -   B^x_a \otimes I)_{\reg{answer}}^2 \ket{\mathrm{question}}\otimes \ket{\mathrm{answer}}\\
=& \E_{\bx} \sum_{a} \bra{\mathrm{answer}}
	(A^{\bx}_a \otimes I -   B^{\bx}_a \otimes I)^2 \ket{\mathrm{answer}}\\
=& \E_{\bx} \sum_{a} \Vert (A^{\bx}_a \otimes I -   B^{\bx}_a \otimes I)^2\ket{\mathrm{answer}}\Vert^2.
\end{align*}
But this is at most $\delta$ by assumption. Now, for the reverse
implication, note that if $\mathsf{Q}_x$ is projective, then the
inequality above becomes an equality.
\end{proof}

\begin{fact}\label{fact:introspective-consistency}
Let $\{A^x_a\}_x$ and $\{B^x_a\}_x$ be measurements such that
\begin{equation}\label{eq:outro-approx}
(A^x_a)_{\reg{Alice}} \otimes I_{\reg{Bob}} \consistency_\delta I_{\reg{Alice}} \otimes (B^x_a)_{\reg{Bob}}
\end{equation}
on state $\ket{\mathrm{answer}}$ and distribution~$\calD$. Next, let $\{\mathsf{Q}_x\}_x$ be a measurement and $\ket{\mathrm{question}}$ be a bipartite state
such that
\begin{equation*}
(\mathsf{Q}_x)_{\reg{Alice}} \otimes I_{\reg{Bob}} \consistency_\delta I_{\reg{Alice}} \otimes (\mathsf{Q}_x)_{\reg{Bob}}
\end{equation*}
on $\ket{\mathrm{question}}$.
Furthermore, suppose that the distribution of measurement outcomes produced by measuring $\{(\mathsf{Q}_x)_{\reg{Alice}} \otimes I_{\reg{Bob}}\}_x$ on $\ket{\mathrm{question}}$ is $\calD$.
Then
\begin{equation}\label{eq:intro-approx}
(\mathsf{Q}_x \otimes A^x_a)_{\reg{Alice}} \otimes I_{\reg{Bob}} \consistency_{\delta} I_{\reg{Alice}} \otimes (\mathsf{Q}_x \otimes B^x_a)_{\reg{Bob}}
\end{equation}
on state $\ket{\mathrm{question}}\otimes
\ket{\mathrm{answer}}$.
\end{fact}

\subsection{Subroutines and superregisters}

In the next few sections, we will design a set of protocols
to be used as a subroutine for our main $\neexp$ protocol.
In doing so, we will encounter the following notational difficulty:
a subroutine~$\game$ might be a $\lambda = (k, n, q)$-register game,
whereas the overall protocol which calls it might be a $\mu = (\ell, m, q)$-register game.
When~$\lambda$ is not equal to~$\mu$, how can we use~$\game$?
We will consider two answers to this question.
In both of them, we will consider the case when all the register field sizes are the same value ``$q$",
as this is the case relevant to our application.

\begin{notation}
First, the registers in~$\lambda$ might appear as a subset of the registers in $\mu$.
In this case, we will specify an injection $\kappa : [k] \rightarrow [\ell]$
such that $n_i = m_{\kappa(i)}$.
Given a string $W = (W_1, \ldots, W_k)$, we write $\kappa(W)$
for the length-$\ell$ string with $W_{\kappa(i)}$ in coordinate~$i$, for each~$i$,
and the ``hide" symbol $H$ in the remaining coordinates.
Similarly, given string $a = (a_1, \ldots, a_\ell)$,
we write $\kappa^{-1}(a)$ for the length-$k$ string with $a_{\kappa(i)}$ in its $i$-th coordinate.
Then \emph{playing $\game$ on registers $\kappa(1), \ldots, \kappa(k)$} means to do the following.
\begin{enumerate}
\item Draw $(\bx, \bx')$ from $\game$.
\item Send $(\kappa(\bx_1), \bx_2)$ to Alice and $(\kappa(\bx_1'), \bx_2')$ to Bob.
\item Receive $\ba, \ba'$. Accept if $\game$ accepts on the answers $(\kappa^{-1}(\ba_1), \ba_2)$ and $(\kappa^{-1}(\ba_1'), \ba_2')$.
\end{enumerate}
\end{notation}

\begin{notation}
Second, the registers in~$\lambda$ might appear as the concatenation of register in $\mu$.
In this case, we will specify concatenation lengths $c(1) + \cdots + c(k) = \ell$ such that
$n_1 = m_1 + \cdots + m_{c(1)}$, $n_2 = m_{c(1) + 1} + \cdots + m_{c(1) + c(2)}$.
Pictorially, the first register in $\lambda$ will be created as the following concatenation:
\begin{equation*}
\underbrace{
\ket{\mathrm{EPR}_q^{n_1}} \otimes \ket{\mathrm{EPR}_q^{n_2}} \otimes \cdots
\otimes \ket{\mathrm{EPR}_q^{n_{c(1)-1}}} \otimes \ket{\mathrm{EPR}_q^{n_{c(1)}}}. 
}_{\ket{\mathrm{EPR}_q^{n_1 + \cdots + n_{c(1)}}}}
\end{equation*}
We refer to these concatenations of registers as \emph{superregisters}.
A Pauli basis query $W \in \{X, Z, \hideq, \noop\}$ to a given superregister~$R$ can be simulated as follows:
\begin{enumerate}
\item  Implement each Pauli basis query $W$ by sending~$W$ to
  each register $r_{i_1}, \ldots, r_{i_{c}}$ in the superregister.
\item If $W \in \{X, Z\}$, the prover measures $\tau^W_{u_i}$ on each
  register $r_i$ in $R$, and the verifier receives the outcomes
$
\bu_{i_1} \in \F_q^{m_{i_1}}, \ldots, \bu_{i_{c}} \in \F_q^{m_{i_{c}}},
$
concatenated as $\bu = (\bu_{i_1}, \ldots, \bu_{i_c})$.
\item If $W = \hideq $, the prover performs $\oktimes{I}{c}$ on
  the registers in the superregister, and verifier receives~$c$ consecutive $\varnothing$'s, interpreted as a single $\varnothing$.
\item If $W = \bot$, the prover may perform any measurement it likes
  on the registers in the superregister.
  The verifier receives~$c$ consecutive $\varnothing$'s, interpreted as a single $\varnothing$.
\end{enumerate}
\end{notation}

The game~$\game$ will usually be proven sound against $\lambda$-register strategies,
but in our cases it will be straightforward to extend this soundness to $\mu$-register strategies
in the case when~$\game$ is applied as a subroutine as detailed above.
For example, suppose we know that a strategy~$A$ which passes~$\game$
with probability $1-\eps$ must satisfy
$(A_a)_{\reg{Alice}} \otimes I_{\reg{Bob}} \approx_{\delta} (B_a)_{\reg{Alice}} \otimes I_{\reg{Bob}}$.
Then it is straightforward to derive that if~$\game$ is played as a subroutine on the second register of~$\mu$
(this is the case when $k = 1$ and $n_1 = m_2$), and if~$A$ passes the subroutine with probability~$1-\eps$, then
\begin{equation*}
(A_a)_{\reg{Alice}} \otimes I_{\reg{Bob}} \approx_{\delta} (I_{\reg{1}} \otimes (B_a)_{\reg{2}} \otimes I_{\reg{3, \ldots, \ell}})_{\reg{Alice}}.\otimes I_{\reg{Bob}}
\end{equation*}
Likewise, suppose $\game$ is played as a subroutine on one superregister consisting of all the registers of~$\mu$
(this is the case when $k = 1$ and $n_1 = m_1 + \cdots + m_\ell$).
If $A$ passes the subroutine with probability~$1-\eps$, then
\begin{equation*}
(A_a)_{\reg{Alice}} \otimes I_{\reg{Bob}} \approx_{\delta} ((B_a)_{1, \ldots,\ell} )_{\reg{Alice}} \otimes I_{\reg{Bob}}.
\end{equation*}
For our applications, it will be straightforward to
extend the soundness of our games to the case when they are played as subroutines,
and we will leave this step implicit in our proofs.



\section{The introspective low-degree test}\label{sec:big-degree}

In this section, we give the introspective low-degree test.
This is an introspection game which simulates the classic surface-versus-point test,
but is able to reduce the question complexity by making the provers sample the questions themselves.
We allow for a fully general $k$-dimensional surface,
though in our application we will only use $k = 1$ (lines) and $k=2$ (planes).

Given an integer $n > 0$ and a power of two~$q$,
the introspective low-degree test is a $(k+1, n, q)$-register game.
In other words, the provers share a state of the following form:
\begin{equation*}
\ket{\psi}
=
\ket{\mathrm{EPR}_q^n}_{\reg{0}}
\otimes \ket{\mathrm{EPR}_q^n}_{\reg{1}}
\otimes
\cdots
\otimes
\ket{\mathrm{EPR}_q^n}_{\reg{k}}
\otimes
\ket{\mathrm{aux}}_{\reg{aux}}.
\end{equation*}
The intended behavior is this: the ``points" prover should measure the point $\bu \in \F_q^n$ from register~$0$.
The ``surface"  prover should measure directions $\bv = \{\bv_1, \ldots, \bv_k\}$ from registers~$1$ through~$k$
and then should receive their surface~$\bs$ from the surface measurement $\{\Pi^{\bv}_s\}_{s\in \surfaces{\bv}}$ on register~$0$.
If the provers act honestly, then~$\bu$ will be a uniformly random point in~$\bs$,
generating the same distribution as the questions in the surface-versus-point low-degree test.

The key difficulty is preventing the surface prover from fully
measuring the register~$0$ and thus learning the value of the point~$\bu$.
In this section, we design a test to enforce this behavior on the
surface prover, using an introspected version of the partial data-hiding game developed in \Cref{sec:rotated-data-hiding}.
This game lets us command the
surface prover to erase all information about $\bu$ except its value
modulo linear combinations of the directions $\bv_1, \dots, \bv_k$;
we give it in \Cref{sec:ihide} below.
We use this test in  \Cref{sec:big-low-degree-subsection} to design the introspective
low-degree test and prove its soundness.

\subsection{Introspected partial data-hiding}
\label{sec:ihide}

In this section, we give an introspected version of the partial
data-hiding game, which will be used to implement the surface and
intercept-scrambling measurements described above.

\begin{definition}
Let $k, n >0$ be integers, let~$q$ be a power of~$2$,
and let $\lambda = (k+1, n, q)$ be register parameters.
Let~$x$ be an arbitrary query.
Then the \emph{introspected partial data-hiding game}~$\game_{\mathrm{IntroHide}}(\lambda, x)$ is given in~\Cref{fig:ihide}.
\end{definition}

{
\floatstyle{boxed} 
\restylefloat{figure}
\begin{figure}[htbp]
  Flip an unbiased coin $\bb \sim \{0, 1\}$, and perform one of the following three tests with
  probability $1/3$ each.
  \begin{enumerate}
    \item Distribute the questions as follows:\label{item:did-you-give-me-the-right-surface-or-not}
      \begin{itemize}
      \item[$\circ$] Player~$\bb$: Give $(\bot, \ktimes{Z}{k}, x)$; receive
        $(\varnothing, \bv_1, \dots, \bv_k, \ba_2)$.
      \item[$\circ$] Player~$\ol{\bb}$: Give $(Z, \ktimes{Z}{k}, x)$; receive
        $(\ba'_1, \bv_1, \dots, \bv_k, \ba'_2)$.
      \end{itemize}
      Accept if $\ba_2 = \ba'_2$
    \item Distribute the questions as follows:
      \begin{itemize}
        \item[$\circ$] Player~$\bb$: Give $(\bot, \ktimes{Z}{k}, x)$; receive
          $(\varnothing,\bv_1, \dots, \bv_k, \ba_2)$.
        \item[$\circ$] Player~$\ol{\bb}$: Give $(\bot, \ktimes{Z}{k}, x,
          \{X\})$; receive $(\varnothing, \bv_1, \dots, \bv_k, \ba'_2, \{\ba_{1,1}', \dots, \ba_{1,k}'\})$.
        \end{itemize}
        Accept if $\ba_2 = \ba_2'$.
      \item Distribute the questions as follows:
        \begin{itemize}
        \item[$\circ$] Player~$\bb$: Give $(X, \ktimes{Z}{k}, \varnothing)$;
          receive $(\ba_1, \bv_1, \dots, \bv_k, \varnothing)$.
        \item[$\circ$] Player~$\overline{\bb}$: Give $(\bot, \ktimes{Z}{k}, \bot, \{X\})$; receive
        $(\varnothing, \bv'_1, \dots, \bv'_k, \varnothing, \{\ba_{1,1}', \dots \ba_{1,k}'\})$. 
        \end{itemize}
        Accept if $\bv_i = \bv'_i$ and $\ba'_{1,i} = \bv_i \cdot \ba_1$ for all $i \in \{1,
        \dots, k\}$.
      \item 
        Distribute the questions as follows:
        \begin{itemize}
          \item[$\circ$] Player~$\bb$: Give $(\bot, x, \ktimes{Z}{k}, \{X\})$;
            receive $(\varnothing, \bv_1, \dots, \bv_k, \ba_2, \{\ba_{1,1}, \dots,
            \ba_{1,k}\})$.
          \item[$\circ$] Player~$\overline{\bb}$: give $(\bot,
            \bot, \ktimes{Z}{k}, \{X\})$; receive $(\varnothing,
            \bv'_1, \dots, \bv'_k
            \varnothing, \{\ba'_{1,1}, \dots, \ba'_{1,k}\})$.
          \end{itemize}
          Accept if $\ba_{1,i}  = \ba_{1,i}'$ for all $i \in \{1,
          \dots, k\}$.
    \end{enumerate}
\caption{The introspected partial data-hiding game $\game_{\mathrm{IntroHide}}(\lambda, x)$.\label{fig:ihide}}
\end{figure}
}

The performance of the introspected partial data-hiding game is given in the following theorem.

\begin{theorem}  \label{thm:ihide}
Let $k, n >0$ be integers, let~$q$ be a power of~$2$,
and let $\lambda = (k+1, n, q)$ be register parameters.
Let~$x$ be an arbitrary query.
Then $\game_{\mathrm{IntroHide}} := \game_{\mathrm{IntroHide}}(\lambda, x)$
satisfies the following two properties.
  \begin{itemize}
  \item[$\circ$]\textbf{Completeness:}
    Let $\calS_{\mathrm{partial}} = (\psi, M^{\bot, Z, \dots, Z, x})$ be a partial
    $\lambda$-register strategy for~$\game_{\mathrm{IntroHide}}$, which
    is also a real commuting EPR strategy, and for which
    \[ M^{\bot, Z, \dots, Z, x}_{\varnothing, v_1, \dots, v_k, a_2} =
      \sum_{s \in \surfaces{v}} \Pi^{v}_s \ot \tau^Z_{v_1} \ot \dots \ot
      \tau^Z_{v_k} \ot A^{x,s,v}_{a_2}, \]
    for some measurement $A^{x,s,v}_{a_2}$ acting only on the
    $\reg{aux}$ register. Then there is a value-$1$ $\lambda$-register strategy for $\game_{\mathrm{IntroHide}}$
    extending $\calS_{\mathrm{partial}}$ which is also a real commuting EPR strategy.
    \ignore{
    there is a $\lambda$-register strategy~$\calS$ extending
    $\calS_{\mathrm{partial}}$ which is a real commuting EPR strategy
    for which
    $\valstrat{\game_{\mathrm{ihide}}}{\calS} = 1$.}
    \item[$\circ$]\textbf{Soundness:}
      Let $\calS = (\psi, M)$ be a projective
      $\lambda$-register strategy
      passing $\game_{\mathrm{IntroHide}}$ with probability at least $1-\eps$.
      Then there exists an ideal measurement $M'^x_{v_1, \ldots, v_k, a_2}$ with
      the property that
      \[ M'^{ x}_{v_1, \dots, v_k, a_2} = 
        \tau^Z_{v_1} \ot \dots \ot \tau^Z_{v_k} \ot \left( \sum_{s \in \surfaces{v}} \Pi^{v}_s \ot
          (M^{x, s, v}_{a_2})_{\reg{aux}}\right), \]
      such that  the measurement $M^{\bot, Z\dots, Z, x}_{\varnothing,
        v_1, \dots, v_k, a_2}$ used by strategy $\calS$ in
      response to the query $(\bot, \ktimes{Z}{k}, x)$  is close to
      $M'^x_{v_1, \dots, v_k, a_2}$:
      \[ (M_{a_2}^{\bot, Z, \dots, Z, x})_{\reg{Alice}} \otimes I_{\reg{Bob}} \approx_\eps (M'^x_{a_2})_{\reg{Alice}} \otimes I_{\reg{Bob}}. \]
    \end{itemize}
  \end{theorem}
  \begin{proof}
    We first show completeness. Very similarly to the non-introspected
    partial data-hiding game, we introduce measurements for the
    remaining questions as follows:
    \begin{align*}
      M^{Z, Z, \dots, Z, x}_{a_1, v_1, \dots, v_k, a_2} &= \sum_{s \in
                                                          \surfaces{v}}
                                                          (\Pi^{v}_s \cdot
                                                          \tau^Z_{a_1})
                                                          \ot
                                                          \tau^Z_{v_1}
                                                          \ot \dots
                                                          \ot
                                                          \tau^Z_{v_k}
                                                          \ot
                                                          A^{x,s,v}_{a_2},
      \\
      M^{X, Z, \dots, Z}_{a_1, v_1, \dots, v_k} &= \tau^X_{a_1} \ot
                                                  \tau^Z_{v_1} \ot
                                                  \dots \ot
                                                  \tau^Z_{v_k} \ot I_{\reg{aux}},\\
      M^{\bot, Z, \dots, Z, x, \{X\}}_{\varnothing, v_1, \dots, v_k,
      a_2, \{a_{1,1}, \dots, a_{1,k}\}} &= \sum_{s \in \surfaces{v}}
                                         ( \Pi^{v}_s \cdot \tau^X_{[\forall
                                          i,  a_1 \cdot v_i = a_{1, i}]})
                                          \ot \tau^Z_{v_1} \ot \dots
                                          \ot \tau^Z_{v_k} \ot
                                          A^{x,s,v}_{a_2,} \\
      M^{\bot, Z, \dots, Z, \{X\}}_{\varnothing, v_1, \dots, v_k,
      \varnothing, \{a_{1,1}, \dots, a_{1,k}\}} &=\tau^X_{[\forall i,
                                                  a_1 \cdot v_i = a_{1,i}]} \ot \tau^Z_{v_1}
                                                  \ot \dots
                                                  \tau^Z_{v_k} .
    \end{align*}
    By essentially the same arguments as in the proof of \Cref{thm:partial-data-hiding-game}, it
    follows that these measurements define a value-$1$ real commuting
    EPR strategy for $\game_{\mathrm{IntroHide}}$.

    We now show soundness. Suppose that the provers succeed in the game
    with probability $1 - \eps$ using a $\lambda$-register
    strategy. From the definition of a register strategy, we know that the measurement operators
    used by the provers have the following form.
    \begin{align*}
      M^{\bot, Z, \dots, Z, x}_{\varnothing, v_1, \dots, v_k, a_2} &=
      (A^{x, v_1, \dots, v_k}_{a_2})_{\reg{1, aux}} \ot \tau^Z_{v_1}
                                                                     \ot \cdots \ot \tau^Z_{v_k}, \\
      M^{\bot, Z, \dots, Z, x, \{X\}}_{\varnothing, v_1, \dots, v_k,
      a_2, \{a'_{1,1}, \dots, a'_{1,k}\}} &= (B^{x, v_1, \dots,
      v_k}_{a_2, \{a'_{1,1}, \dots, a'_{1,k}\}})_{\reg{1,aux}}
                                            \ot \tau^Z_{v_1} \ot \cdots  \ot \tau^Z_{v_k}, \\
      M^{Z, Z \dots, Z, x}_{a_1, v_1, \ldots, v_k, a_2} &= \tau^Z_{a_1}
                                                         \ot
                                                         \tau^Z_{v_1} \ot
                                                         \cdots \ot
                                                         \tau^Z_{v_k}
                                                         \ot (C^{x,
                                                         a_1, v_1,
                                                          \dots,
                                                          v_k}_{a_2})_{\reg{aux}},\\
      M^{\bot, Z, \dots, Z, \bot, \{X\}}_{\varnothing, v_1, \dots,
      v_k, \varnothing, \{a_{1,1}, \dots, a_{1,k}\}} &=  (D^{v_1,
                                                       \dots, v_k}_{a_{1,1},
                                                       \dotsm
                                                       a_{1,k}})_{\reg{1,aux}}
                                                       \ot
                                                       \tau^Z_{v_1}
                                                       \ot \dots \ot \tau^Z_{v_k}.
    \end{align*}
    where the operators $\{A^{x, v_1, \dots, v_k}_{a_2}\}$,
    $\{B^{x, v_1, \dots, v_k}_{a_2, \{a'_{1,1}, \dots, a'_{1,k}\}}\}$,
    $\{C^{x, a_1, v_1, \ldots, v_k}_{a_2}\}$, and $\{D^{v_1, \dots,
      v_k}_{a_{1,1}, \dots, a_{1,k}}\}$ form
    valid POVMs for every choice of $x, a_1, v_1, \dots, v_k$. We further
    know that the shared state of the provers is of the form
    \[ \
\ket{\psi}
=
\ket{\mathrm{EPR}_q^n}_{\reg{0}}
\otimes \ket{\mathrm{EPR}_q^n}_{\reg{1}}
\otimes
\cdots
\otimes
\ket{\mathrm{EPR}_q^n}_{\reg{k}}
\otimes
\ket{\mathrm{aux}}_{\reg{aux}}. \]

    From success in the four parts of the test, we may also deduce the
    following conditions:
    \begin{align*}
      (M^{\bot, Z, \dots ,Z, x}_{a_2, v_1, \dots, v_k})_{\reg{Alice}} \otimes I_{\reg{Bob}} &\simeq_{\eps}
                                       I_{\reg{Alice}} \otimes (M^{Z,Z, \dots, Z, x}_{v_1,
                             \dots, v_k, 
                                a_2})_{\reg{Bob}}, \\
      (M^{\bot, Z,\dots, Z, x}_{a_2, v_1, \dots, v_k})_{\reg{Alice}} \otimes I_{\reg{Bob}} &\simeq_{\eps} 
 	                     I_{\reg{Alice}} \otimes (M^{\bot,Z ,\dots, Z,  x, \{X\}}_{a_2, v_1,
                             \dots, v_k}), \\
      (M^{\bot, Z,\dots, Z,  \bot, \{X\}}_{v_1, \dots v_k, \{a_{1,1}, \dots,
                                                                    a_{1, k}\}})_{\reg{Alice}} \otimes I_{\reg{Bob}} &\simeq_{\eps}
      			I_{\reg{Alice}} \otimes (\tau^X_{[\forall i, a_1 \cdot v_i = a_{1, i}]} \ot \tau^{Z}_{v_1} \ot \dots \ot \tau^{Z}_{v_k} \ot
                                                                                                                       I_{\reg{aux}})_{\reg{Bob}}, \\
      (M^{\bot, Z, \dots, Z, \bot, \{X\}}_{v_1, \dots, v_k, \{a_{1,1},
      \dots, a_{1,k}\}})_{\reg{Alice}} \ot I_{\reg{Bob}}
                                                                                            &\simeq_{\eps}
                                                                                              I_{\reg{Alice}}
                                                                                              \ot
                                                                                              (M^{\bot,
                                                                                              Z,
                                                                                              \dots,
                                                                                              Z,
                                                                                              x,
                                                                                              \{X\}}_{v_1,
                                                                                              \dots,
                                                                                              v_k, a_2,
                                                                                              \{a_{1,1},
                                                                                              \dots,
                                                                                              a_{1,k}\}} )_{\reg{Bob}}.
    \end{align*}
    
    We would now like to argue that the operators $A,
    B, C, D$ form a good strategy for the partial data-hiding game. By
    \Cref{fact:introspectivize-it}, it holds that under the uniform
    distribution over $v_1, \dots, v_k$,
    \begin{align*}
      (A^{x, v_1,\dots, v_k}_{a_2})_{\reg{Alice}} \otimes I_{\reg{Bob}}&\simeq_{\eps} I_{\reg{Alice}} \otimes \big(\sum_{a_1}
                                    \tau^Z_{a_1} \ot C^{x,
                                    a_1, v_1, \dots, v_k}_{a_2}\big)_{\reg{Bob}}, \\
      (A^{x, v_1, \dots, v_k}_{a_2})_{\reg{Alice}} \otimes I_{\reg{Bob}} &\simeq_{\eps} I_{\reg{Alice}} \otimes \big(
      				B^{x, v_1, \dots,
                                                                           v_k}_{a_2}\big)_{\reg{Bob}}, \\
      (B^{x, v_1, \dots, v_k}_{\{a_{1,1}, \dots,
      a_{1,k}\}})_{\reg{Alice}} \otimes I_{\reg{Bob}} &\simeq
                                                        I_{\reg{Alice}}
                                                        \ot (D^{v_1,
                                                        \dots,
                                                        v_k}_{\{a_{1,1},
                                                        \dots,
                                                        a_{1,k}\}})_{\reg{Bob}} \\
     (D^{ v_1, \dots, v_k}_{\{a'_{1,1}, \dots, a'_{1,k}\}})_{\reg{Alice}} \otimes I_{\reg{Bob}}
                                  &\simeq_{\eps} I_{\reg{Alice}} \otimes (\tau^X_{[a_1 \cdot
                                    v_1 = a'_{1,1}, \dots, a_1 \cdot
                                    v_k = a'_{1,k}]} \otimes
                                    I_{\reg{aux}})_{\reg{Bob}},
    \end{align*}
    as well as the same conditions with the Alice and Bob registers
    exchanged.
    
    These are precisely the conditions of winning the partial data-hiding game
    $\game_{\mathrm{hide}}(S,x)$, where $S$ is the set of all tuples
    $v_1, \dots, v_k$ in $\F_q^n$, with probability $1 -
    O(\eps)$. Hence, by~\Cref{thm:partial-data-hiding-game}, it follows
    that there exists a measurement $A'^{x, v_1, \dots, v_k}_{a_2}$
    such that
    \begin{align*}
      A'^{x,v _1, \dots, v_k}_{a_2} &= \sum_{s \in
                                                       \surfaces{v}}
                                                       \Pi^{v}_s \ot
                                                       A^{s,x,v}_a,\\
      (A^{x,v_1, \dots, v_k}_{a_2})_{\reg{Alice}} \ot I_{\reg{Bob}} &\approx_{\eps} (A'^{x, v_1,
                                                      \dots, v_k}_{a_2})_{\reg{Alice}} \ot I_{\reg{Bob}}.
    \end{align*}
    The operator $M'$ in the conclusion of the theorem can then be
    taken to be
    \[M'^{\bot, Z, \dots, Z, x}_{a_2, v_1, \dots, v_k}  =
    (A'^{x,v_1, \dots, v_k}_{a_2}) \ot \tau^Z_{v_1} \ot \dots \ot
    \tau^Z_{v_k}. \qedhere\]
  \end{proof}

\subsection{An introspective surface sampler}

In this section, we will use the introspective data hiding game
to implement the ``surface prover".
This is a prover who samples a surface~$\bs$ from register~$0$ using the $\Pi^{\bv}$ measurement
and then reports back~$\bs$ to the verifier, along with a degree-$d$ polynomial $\boldf:\bs \rightarrow \F_q$.
As above, the prover is expected \emph{not} to measure register~$0$ any further, so that~$\boldf$ depends only on~$\bs$ and~$\bv$.
We can enforce this by running the introspective data hiding game and interpreting the provers' answers as $\ba_2 = \{\bs, \boldf\}$.
However, we must also verify that~$\bs$ corresponds to the actual surface measured by the prover in the~$0$-th register
and not some other surface.
We do this with a slight modification to the introspective data hiding game we call the ``introspective surface sampling game".

\begin{definition}
Let $k, n, d>0$ be integers, let~$q$ be a power of~$2$,
and let $\lambda = (k+1, n, q)$ be register parameters.
Then the \emph{introspective surface sampling game}~$\game_{\mathrm{IntroSurfSamp}}(\lambda, d)$ is given in~\Cref{fig:big-plane}.
\end{definition}

{
\floatstyle{boxed} 
\restylefloat{figure}
\begin{figure}[htbp]
  \begin{itemize}
  \item[$\circ$] Play the game $\game_{\mathrm{IntroHide}}(\lambda, x)$ with $x =
    \text{``surface"}$, and with the answer $a_2$ taking the form
    $\{\bs,\boldf\}$, where $\bs$ is a surface and
                  $\boldf$ is a degree-$d$ function $\boldf:\bs\rightarrow \F_q$.
  \item[$\circ$] Consider the test in \Cref{item:did-you-give-me-the-right-surface-or-not} of $\game_{\mathrm{IntroHide}}(\lambda, x)$.
  	Here, Player~$\bb$ replies with the answer
		$(\varnothing, \bv_1, \ldots, \bv_k, \{\bs, \boldf\})$,
		and Player~$\overline{\bb}$ replies with
		$(\ba_1', \bv_1, \ldots, \bv_k, \{\bs', \boldf'\})$.
   In the case where this test is chosen, accept if $\game_{\mathrm{IntroHide}}(\lambda, x)$ accepts and also if
   $\bs$ is the surface $\{\ba_1' +
    \sum_{i=1}^{k} \lambda_i \bv_i : \lambda_1,
    \dots, \lambda_k \in \F_q\}$.
    (We call this additional check the ``Correct Surface Check".)
    If this query is not given to the
    provers, then accept if $\game_{\mathrm{IntroHide}}(\lambda, x)$ accepts.
  \end{itemize}
  \caption{The game $\game_{\mathrm{IntroSurfSamp}}(\lambda,d)$.\label{fig:big-plane}}
\end{figure}
}

\begin{notation}
In the case when a prover is given the question $(\bot, Z, \ldots, Z, \mathrm{``surface"})$,
we refer to it as the \emph{surface prover}.  
It has the following intended behavior.
\begin{enumerate}
\item \textbf{Surface prover:}  
  \begin{description}[align=left]
  \item [Input:] Pauli basis queries $(\bot, \ktimes{Z}{k})$ and
    an auxiliary query ``surface''.
  \item [Output:] Pauli basis answers $\varnothing$ and $v_1,
    \ldots, v_k \in \F_q^{n}$, a $k$-dimensional surface~$s$, and the coefficients of
    a degree-$d$ polynomial function $f: \F_q^k \rightarrow \F_q$, where
    the domain of $f$ is to interpreted as $s$.
  \item [Goal:] The prover measures $\Pi^v$ on register~$0$ and sets~$s$ to be its outcome.
  		They then set $f = g|_s$, where $g: \F_q^n \rightarrow \F_q$ is a global degree-$d$ polynomial
		selected independently of~$s$ or~$v$.
  \end{description}
\end{enumerate}
In the case when $k = 1$, we may also refer to it as the \emph{lines} prover,
and in the case when $k = 2$, we may also refer to it as the \emph{planes} prover.
We will also refer to the \emph{surface prover's measurement},
which refers to the measurement $\{A_{v_1, \ldots, v_k, s, f}\}$ given by
\[A_{v_1, \ldots, v_k, s, f} = M^{\bot, Z, \dots, Z,
    \text{``surface"}}_{\varnothing, v_1, \dots, v_k, s,
    f}.\]
\end{notation}

The following theorem shows that the introspective surface sampling game
forces the surface prover to output the correct surface~$s$. 

\begin{theorem}\label{thm:big-planes}
Let $k, n, d>0$ be integers, let~$q$ be a power of~$2$,
and let $\lambda = (k+1, n, q)$ be register parameters.
Write $\{A_{v_1, \ldots, v_k, s, f}\}$ for the surface prover's measurement.
Then $\game_{\mathrm{IntroSurfSamp}} := \game_{\mathrm{IntroSurfSamp}}(\lambda, d)$ has the following two properties.
  \begin{itemize}
  \item[$\circ$] \textbf{Completeness:}
    		Suppose there is a degree-$d$ polynomial $g:\F_q^{n}\rightarrow \F_q$ such that
		\begin{equation*}
			A_{v_1 \ldots, v_k,s,f} = \Pi^{v}_s \otimes \tau^Z_{v_1} \otimes
                        \cdots \otimes \tau^{Z}_{v_k} \otimes I_{\reg{aux}} \cdot \bone[f = g|_{s}].
		\end{equation*}
		Then there is a value-$1$ $\lambda$-register strategy
                for $\game_{\mathrm{IntroSurfSamp}}$ with~$A$ as the
                surface prover's measurement.
  \item[$\circ$]\textbf{Soundness:}
Let $\calS$ be a projective $\register$-register strategy which passes $\game_{\mathrm{IntroSurfSamp}}$ with probability at least $1-\eps$.
Then there exists an ideal measurement $A'$ of the form
\[ A'_{v, s, f} =\Pi^v_s \otimes  \tau^Z_{v_1} \ot
  \dots \ot \tau^Z_{v_k} \ot 
    (M^{s,v}_{f})_{\reg{aux}}, \]
with $M^{s, v}_{f}$ an arbitrary measurement on the
$\reg{aux}$ register, such that $A'$ is close to the surface provers'
measurement $A$ in $\calS$:
\begin{equation*}
(A_{v, s, f})_{\reg{Alice}} \otimes I_{\reg{Bob}}
\approx_{\mathrm{poly}(\eps)} (A'_{v,
  s, f})_{\reg{Alice}} \otimes I_{\reg{Bob}}.
\end{equation*}
In particular, the surface output by~$A'$ is the same surface measured by~$A'$ in register~$0$.
\end{itemize}
\end{theorem}

\begin{proof}
First, the completeness follows immediately from
the completeness guarantee of~\Cref{thm:ihide}.

Next, we show soundness.
Passing with probability~$1-\eps$
implies passing $\game_{\mathrm{IntroHide}}(\lambda, \mathrm{``surface"})$ with probability~$1-\eps$.
By~\Cref{thm:ihide}, this implies an ideal measurement $M'^x_{v_1, \ldots, v_k, s, f}$ with
      the property that
      \begin{equation*} M'_{v_1, \dots, v_k, s, f} = 
        \tau^Z_{v_1} \ot \cdots \ot \tau^Z_{v_k} \ot \left( \sum_{s' \in \surfaces{v}} \Pi^{v}_{s'} \ot
          (M^{s', v}_{s, f})_{\reg{aux}}\right), \end{equation*}
      \[ (A_{v_1, \ldots, v_k,s, f})_{\reg{Alice}}
      		\otimes I_{\reg{Bob}} \approx_\eps (M'_{v_1, \ldots, v_k, s, f})_{\reg{Alice}} \otimes I_{\reg{Bob}}. \]
(Note that the measured surface~$s'$ in $M'$ is allowed to be different than the output surface~$s$.)
Next, set
$
B_{v_1, \ldots, v_k, s} := \Pi^v_{s} \otimes \tau^Z_{v_1} \otimes \cdots \otimes \tau^Z_{v_k} \otimes I_{\reg{aux}}.
$
Then passing the Correct Surface Check with probability $1-O(\eps)$ implies that
\begin{equation*}
(A_{v, s})_{\reg{Alice}} \otimes I_{\reg{Bob}}
\approx_{\eps} I_{\reg{Alice}} \otimes (B_{v, s})_{\reg{Bob}}.
\end{equation*}
Note that $M'_{v, s, f}$ and $B_{v, s}$ commute.
Thus, if we define $C_{v, s, f} := M'_{v, s, f} \cdot B_{v, s} = B_{v, s} \cdot M'_{v, s, f}$, then by \Cref{fact:add-a-proj},
\begin{align*}
(A_{v, s, f})_{\reg{Alice}} \otimes I_{\reg{Bob}}
&= (A_{v, s, f} \cdot A_{v, s})_{\reg{Alice}} \otimes I_{\reg{Bob}}\\
&\approx_\eps (A_{v, s, f})_{\reg{Alice}} \otimes (B_{v, s})_{\reg{Bob}}\\
&\approx_\eps (M'_{v, s, f})_{\reg{Alice}} \otimes (B_{v, s})_{\reg{Bob}}\\
&\approx_\eps (M'_{v, s, f} \cdot B_{v, s})_{\reg{Alice}} \otimes I_{\reg{Bob}}\\
&= (C_{v, s, f})_{\reg{Alice}} \otimes I_{\reg{Bob}},
\end{align*}
where the second-to-last step is by \Cref{fact:the-ol-pauli-swaperoonie}.
Now, set $C^{s, v}_f := M^{s, v}_{s, f}$. Then we can write
\begin{equation*}
C_{v_1, \dots, v_k, s, f} = \Pi^v_s \otimes
	\tau^Z_{v_1} \ot \cdots \ot \tau^Z_{v_k} \ot (C^{s, v}_{f})_{\reg{aux}}.
\end{equation*}
These matrices are almost of the form guaranteed by the theorem,
except they do not necessarily form a measurement because the matrices
$C^{s, v}_f$ do not necessarily sum to the identity.
However, this is still sufficient to imply the theorem by~\Cref{fact:sub-zero}.
\end{proof}

\subsection{The introspective cross-check}
In this section, we introduce the other subroutine in the introspective low-degree test.
In this subroutine, known as the ``introspective cross-check",
we introduce a new prover known as the ``points prover".
This is a prover who samples a point~$\bu$ from register~$0$ 
and then reports back a value~$\bnu \in \F_q$
interpreted as their assignment to the point~$\bu$.
By data hiding, we can assume the points prover does not read registers~$1$ through~$k$.
Then the introspective cross-check queries the points prover and the surface prover
and checks that their outputs agree on the point~$\bu$.

\begin{definition}
Let $k, n, d>0$ be integers, let~$q$ be a power of~$2$,
and let $\lambda = (k+1, n, q)$ be register parameters.
The \emph{introspective cross-check}, denoted $\game_{\mathrm{IntroCross}}(\lambda, d)$, is defined in \Cref{fig:big-cross-check}.
\end{definition}
{
\floatstyle{boxed} 
\restylefloat{figure}
\begin{figure}
Flip an unbiased coin $\bb \sim \{0, 1\}$.
Distribute the questions as follows. \label{item:big-low-degree-test}
		\begin{itemize}
		\item[$\circ$] Player~$\bb$: Give $(\bot, \ktimes{Z}{k}, \mathrm{``surface"})$; receive
        $(\varnothing, \bv_1, \dots, \bv_k, \bs, \boldf)$.
		\item[$\circ$] Player~$\overline{\bb}$: Give $(Z, \ktimes{\hideq}{k}, \mathrm{``point"})$;
				receive $(\bu, \ktimes{\varnothing}{k}, \bnu)$, where $\bnu \in \F_q$.
		\end{itemize}
		Accept if $\boldf(\bu) = \bnu$.
	\caption{The game $\game_{\mathrm{IntroCross}}(\lambda, d)$.\label{fig:big-cross-check}}
\end{figure}
}

\begin{notation}
In the case when a prover is given the question $(Z, \hideq, \ldots, \hideq, \mathrm{``point"})$,
we refer to it as the \emph{points prover}.  
It has the following intended behavior.
\begin{enumerate}
\setcounter{enumi}{1}
\item \textbf{Points prover:}
	\begin{description}[align=left]
	\item [Input:] Pauli basis queries $(Z, \hideq, \ldots, \hideq)$
          and auxiliary query ``point''.
	\item [Output:] String $u \in \F_q^{n}$ and $\varnothing,\ldots,  \varnothing$. A number $\nu \in \F_q$.
	\item [Goal:] The prover sets $\nu = g(u)$, where $g: \F_q^n \rightarrow \F_q$ is a global degree-$d$ polynomial
		selected independently of~$u$.
	\end{description}
\end{enumerate}
We will also refer to the \emph{point prover's measurement},
which refers to the measurement $\{B_{u, \nu}\}$ given by
\[B_{u, \nu} = M^{Z, \hideq, \dots, \hideq,
    \text{``point"}}_{u, \varnothing, \dots, \varnothing, \nu}.\]
\end{notation}

Our next lemma shows that if the surface prover's measurement for~$\boldf$ depends only on the surface~$\bs$ and directions~$\bv$
and not on the point~$\bu$, then we can relate the value of the introspective
cross-check to its non-introspected variant, $\game_{\mathrm{surface}}$.
\begin{lemma}\label{lem:big-cross-check}
Let $k, n, d>0$ be integers, let~$q$ be a power of~$2$,
and let $\lambda = (k+1, n, q)$ be register parameters.
Let $\calS$ be a $\register$-register strategy for $\game_{\mathrm{IntroCross}} := \game_{\mathrm{IntroCross}}(\lambda, d)$.
Let $\{A_{v, s, f}\}$ be the surface prover's measurement
and $\{B_{u, \nu}\}$ be the point prover's measurement,
and write $\ket{\mathrm{aux}}$ for the auxiliary state.
Suppose
\begin{align*}
A_{ v,s, f} &= \Pi^{v}_s \otimes \tau_{v_1}^Z \otimes \dots \otimes
\tau_{v_k}^Z \otimes A^{s,v}_f,\\
B_{u, \nu} &= \tau_u^Z \otimes \oktimes{I}{k} \otimes B^u_\nu,
\end{align*}
where $\{A^{s, v}_f\}$ and $\{B^u_\nu\}$ are POVM measurements on the auxiliary register.
Consider the strategy $\calS_{\mathrm{surface}} = (\mathrm{aux}, \{A^{s, v}, B^u\})$ for the game $\game_{\mathrm{surface}}  := \game_{\mathrm{surface}}(n, q, k, d)$.
Then
\begin{equation*}
\valstrat{\game_{\mathrm{surface}}}{\calS_{\mathrm{surface}}}
= \valstrat{\game_{\mathrm{IntroCross}}}{\calS}.
\end{equation*}
\end{lemma}
\begin{proof}
  By the definition of $\Pi^{v}_s$ and $\tau^Z_u$, it follows that for
  any choice of $k$ vectors $v$, if
  Alice and Bob each measure their half of register~$0$
  using the measurements $\Pi^v_s$ and $\tau^Z_u$, respectively, then
  the measurement outcomes obtained will be pairs $(\bs, \bu)$ where $\bs$
  is a uniformly random surface in $\surfaces{\bv}$ and $\bu$ is a uniformly
  random point in $\bs$. Moreover, if Alice measures her half of registers~$1$ through~$k$,
  she will obtain a uniformly random $k$-tuple $\bv = \{\bv_1,
  \dots, \bv_k\} \subseteq \F_q^n$. Combining these facts, we see that
  if Alice measures her half of registers~$0$ through~$k$ with
  $\Pi^v_s \ot \tau^Z_{v_1} \ot \dots \ot \tau^Z_{v_k}$, and Bob
  measure his half with $\tau^Z_u \ot \oktimes{I}{k}$, they obtain a
  pair of outcomes $(\bx_A = (\bs, \bv), \bx_B = \bu)$ distributed exactly
  according to the question distribution in
  $\game_{\mathrm{surface}}$. Thus, applying~\Cref{fact:introspective-outrospective}, we conclude
that $\valstrat{\game_{\mathrm{surface}}}{\calS_{\mathrm{surface}}} = \valstrat{\game_{\mathrm{IntroCross}}}{\calS}$.
\end{proof}

\subsection{The introspective low-degree test}\label{sec:big-low-degree-subsection}

In this section, we state the completed introspective low-degree test.

\begin{definition}
Let $k, n, d>0$ be integers, let~$q$ be a power of~$2$,
and let $\lambda = (k+1, n, q)$ be register parameters.
The \emph{introspective surface-versus-point low-degree test},
denoted $\game:=\game_{\mathrm{IntroLowDeg}}(\lambda, d)$, is defined in \Cref{fig:big-low-degree}.
It has the following properties:
\begin{equation*}
\qlength{\game} = O(1), \quad
\alength{\game} = O(k n \log(q) + (d+k)^k \log(q)),
\end{equation*}
\begin{equation*}
\qtime{\game} = O(1), \quad
\atime{\game} = \poly(k n \log(q), (d+k)^k \log(q)).
\end{equation*}
\end{definition}
{
\floatstyle{boxed} 
\restylefloat{figure}
\begin{figure}
With probability~$\tfrac{1}{2}$ each, perform one of the following three tests.
\begin{enumerate}
	\item \textbf{Surface sampler test:} Play $\game_{\mathrm{IntroSurfSamp}}(\lambda, d)$.
	\item \textbf{Cross-check test:} Play $\game_{\mathrm{IntroCross}}(\lambda, d)$.
	\end{enumerate}
	\caption{The game $\game_{\mathrm{IntroLowDeg}}(\lambda, d)$.\label{fig:big-low-degree}}
\end{figure}
}

The question complexities are immediate.
As for the answer length, the provers return $(k+1)$ elements of $\F_q^{n}$ and degree-$d$ polynomials on $k$-surfaces,
encoded as $\F_q$-valued strings of length $d[k] \leq (d+k)^k$.
Finally, all operations made by the verifier, such as polynomial evaluation, are efficient, so the answer time complexity is polynomial in the answer length.

Naturally, we analyze the introspective low-degree test
via introspection.
This involves a reduction to the non-introspected version of the game,
i.e.\ the ``normal'' surface-versus-point low-degree test.
By \Cref{thm:anand-thomas-classical-low-degree} we know quantum soundness for this test in the $k = 2$ (i.e.\ planes) case.
As a result, we get soundness for the introspective low-degree test in this case as well.

\begin{theorem}\label{thm:big-planes-low-degree}
Fix $k = 2$.
Let $n, d>0$ be integers, let~$q$ be a power of~$2$,
and let $\lambda = (k+1, n, q)$ be register parameters.
Write $\game := \game_{\mathrm{IntroLowDeg}}(\lambda, d)$.
\begin{itemize}
\item[$\circ$] \textbf{Completeness:}
		Suppose there is a degree-$d$ polynomial $g:\F_q^{n}\rightarrow \F_q$ such that
		\begin{equation*}
			B_{u, \nu} = \tau_u^Z \otimes I_{\reg{1}} \otimes I_{\reg{2}} \otimes I_{\reg{aux}} \cdot \bone[\nu = g(u)].
		\end{equation*}
		Then there is a value-$1$ $\lambda$-register real
                commuting EPR strategy for $\game$ with~$B$ as the point prover's measurement.
\item[$\circ$] \textbf{Soundness:}
		There exists a constant $c > 0$ and a function $\delta(\eps) = \poly(\eps, dm/q^c)$ such that the following holds.
		Suppose $\calS$ is a projective $\lambda$-register strategy with value $1-\eps$.
		Write $\{B_{u, \nu}\}$ for the point prover's measurement.
		Then there exists a POVM $\{G_g\}$
                in~$\polymeas{n}{d}{q}$ such that
		\begin{equation*}
			(B_{u, \nu})_{\reg{Alice}} \otimes I_{\reg{Bob}}
				\approx_{\delta(\eps)} (\tau^Z_{u}\otimes I_{\reg{1}} \otimes I_{\reg{2}}
					\otimes (G_{[g(u) = \nu]})_{\reg{aux}})_{\reg{Alice}} \otimes I_{\reg{Bob}}.
		\end{equation*}
\end{itemize}
\end{theorem}

\begin{proof}
Throughout this proof, we will write $\{A_{v, s, f}\}$ for the surface prover's measurement.
We first show completeness.
Assign the surface prover's measurement as follows:
\begin{equation*}
A_{v,s, f}  := \Pi^{v}_s \otimes \tau^Z_{v_1} \otimes \tau^Z_{v_2} \otimes I_{\reg{aux}} \cdot \bone[f = g|_{s}].
\end{equation*}
This is clearly a $\lambda$-register strategy. By the completeness
case of \Cref{thm:big-planes}, this can be extended into a real
commuting EPR strategy that passes $\game_{\mathrm{IntroSurfSamp}}(\lambda, d)$ with
probability~$1$. By~\Cref{lem:big-cross-check}, the strategy using $A$
and $B$ passes the cross check with the same probability as the
honest classical strategy to $\game_{\mathrm{surface}}(n, q, k, d)$ answering according
to the low-degree polynomial $g$, which is $1$. Hence, this strategy
passes both parts of $\game$ with probability~$1$.

Now, we show soundness.
The strategy $\calS$ is a $\lambda$-register strategy, so we can write the points prover's measurement as
\begin{equation*}
B_{u, \nu} = \tau^Z_u \otimes I_{\reg{1}} \otimes I_{\reg{2}} \otimes (B^u_{\nu})_{\reg{aux}}.
\end{equation*}
Passing $\game_{\mathrm{IntroSurfSamp}}(\lambda, d)$ with probability $1-2\eps$
implies via \Cref{thm:big-planes} a measurement $A'_{u, v, f}$ such that
\begin{equation*}
A'_{ v,s, f} = \Pi^{v}_s \otimes \tau^Z_{v_1} \otimes \tau^Z_{v_2} \otimes ((A')^{s, v}_f)_{\reg{aux}},
\end{equation*}
\begin{equation*}
(A_{v, s, f})_{\reg{Alice}} \otimes I_{\reg{Bob}}
\approx_{\poly(\eps)}  (A'_{v, s, f})_{\reg{Alice}} \otimes I_{\reg{Bob}},
\end{equation*}
where $\{(A')^{s,v}_f\}_{f}$ is a measurement on the auxiliary
register. By assumption, the measurement $\{A_{v,s,f}\}$ is
projective, so
we can apply~\Cref{fact:approx-delta-generalized-game-value} to deduce
that replacing $A_{v,s,f}$ by $A'_{v,s,f}$ changes the game value by
at most $\poly(\eps)$. Moreover, by applying~\Cref{thm:partial-naimark}, we can, by
performing a dilation of the auxiliary space, simulate the
$A'_{v,s,f}$ measurements by a projective measurement of the form
\[ A''_{v,s,f} = \Pi^v_s \ot \tau^Z_{v_1} \ot \tau^Z_{v_2} \ot
  ((A'')^{s,v}_f)_{\reg{aux}}, \]
where $(A'')^{s,v}_f$ is a projective measurement on the (expanded)
aux register. (Note that a direct invocation of Naimark's
theorem~\Cref{thm:naimark} would not have sufficed as the dilated
measurement would not necessarily act as desired on the non-aux
registers.) Using the dilated $A''$ measurements instead of $A'$ does
not change the value of the game. Thus, we deduce that the projective strategy
using measurements $B_{u,\nu}$ and $A''_{v,s,f}$ passes
$\game_{\mathrm{IntroCross}}(\lambda, d)$ with probability
$1-\poly(\eps)$.

Now we are in a position to reduce to the soundness of the non-introspective game. Define the strategy $\calS_{\mathrm{Plane}} := (\mathrm{aux}, \{B^u, (A'')^{s, v}\})$.
Then by \Cref{lem:big-cross-check}, $\calS_{\mathrm{Plane}}$ also passes $\game_{\mathrm{Plane}}$ with probability $1-\poly(\eps)$.
Applying~\Cref{thm:anand-thomas-classical-low-degree}, we have that
there exists a measurement $\{G_g\}$ in $\polymeas{n}{d}{q}$ such that
\begin{equation*}
B^u_{\nu}\otimes I \approx_{\delta(\eps)} G_{[\nu = g(u)]}\otimes I
\end{equation*}
on state $\ket{\mathrm{aux}}$.

The theorem then follows from~\Cref{fact:introspectivize-it}.
\end{proof}

\subsection{The introspective simultaneous low-degree test}

In this section, we extend the introspective low-degree test to handle multiple functions at once.
This is the  introspective version of the simultaneous low-degree test from \Cref{def:simultaneous-plane-v-point}.

\begin{definition}
Let $m \geq 1$.
Let $k, n, d>0$ be integers, let~$q$ be a power of~$2$,
and let $\lambda = (k+1, n, q)$ be register parameters.
The \emph{introspective simultaneous low-degree test}, denoted $\game :=\game_{\mathrm{IntroLowDeg}}(\lambda, d, m)$,
is defined by the following modifications to the introspective low-degree test.
First, the prover roles are modified as follows.
\begin{itemize}
\item[$\circ$] \textbf{Surface prover:} Rather than returning a function $f:s\rightarrow \F_q$, it should return~$m$ functions $f_1, \ldots, f_m : s\rightarrow \F_q$.
			The intent is that $f_i = g_i |_{s}$ for each $i$, where each $g_i :\F_q^n \rightarrow \F_q$ is a global degree-$d$ polynomial
				selected independently of~$s$ or~$v$.
\item[$\circ$] \textbf{Points prover:} Rather than returning a single number $\nu \in \F_q$, it should return~$m$ numbers $\nu_1, \ldots, \nu_m \in \F_q$.
			The intent is that $\nu_i = g_i(u)$ for each $i$, where each~$g_i$ is selected independently of~$u$.
\end{itemize}
Next, the subroutines are modified as follows.
\begin{itemize}
\item[$\circ$] \textbf{Introspective surface sampling game:} The answer $\ba_2$ has the
  form $\bs, \boldf_1, \ldots, \boldf_m$ (rather than $\bs, \boldf$
  for a single function $\boldf$).
\item[$\circ$] \textbf{Introspective cross-check:} Receive $\boldf_1, \ldots, \boldf_m : \bs \rightarrow \F_q$ from the surface prover and $\bnu_1, \ldots, \bnu_m \in \F_q$ from the points prover
				(rather than a single~$\boldf$ and~$\bnu$). Check that $\boldf_i(\bu) = \bnu_i$ for all $i$.
\end{itemize}
It has the following properties:
\begin{equation*}
\qlength{\game} = O(1), \quad
\alength{\game} = O(k n \log(q) + m (d+k)^k \log(q)),
\end{equation*}
\begin{equation*}
\qtime{\game} = O(1), \quad
\atime{\game} = \poly(k n \log(q), m (d+k)^k \log(q)).
\end{equation*}
\end{definition}

The following theorem gives the performance of the introspective simultaneous low-degree test in the case of $k=2$ (i.e.\ planes).

\begin{theorem}\label{thm:big-simultaneous-planes-low-degree}
Fix $k = 2$.
Let $n, d, m>0$ be integers, let~$q$ be a power of~$2$,
and let $\lambda = (k+1, n, q)$ be register parameters.
Write $\game := \game_{\mathrm{IntroLowDeg}}(\lambda, d, m)$.
\begin{itemize}
\item[$\circ$] \textbf{Completeness:}
		Suppose there are degree-$d$ polynomials $g_1, \ldots, g_m:\F_q^{n}\rightarrow \F_q$ such that
		\begin{equation*}
			B_{u, \nu_1, \ldots, \nu_m} = \tau_u^Z \otimes I_{\reg{1}} \otimes I_{\reg{2}} \otimes I_{\reg{aux}} \cdot \bone[\forall i,~\nu_i = g_i(u)].
		\end{equation*}
		Then there is a value-$1$ $\lambda$-register real
                commuting EPR strategy for $\game$ with~$B$ as the point prover's measurement.
\item[$\circ$] \textbf{Soundness:}
		There exists a constant $c > 0$ and a function $\delta(\eps) = \poly(\eps, d(n+m)/q^c)$ such that the following holds.
		Suppose $\calS$ is a projective $\lambda$-register strategy with value $1-\eps$.
		Write $\{B_{u, \nu_1, \ldots, \nu_m}\}$ for the point prover's measurement.
		Then there exists a POVM $\{G_{g_1, \ldots, g_m}\}$
                in~$\simulpolymeas{n}{d}{q}{m}$ such that
		\begin{equation*}
			(B_{u, \nu_1, \ldots, \nu_m})_{\reg{Alice}} \otimes I_{\reg{Bob}}
				\approx_{\delta(\eps)} (\tau^Z_{u}\otimes I_{\reg{1}} \otimes I_{\reg{2}}
					\otimes (G_{[g_1(u), \ldots, g_m(u) = \nu_1, \ldots, \nu_m]})_{\reg{aux}})_{\reg{Alice}} \otimes I_{\reg{Bob}}.
		\end{equation*}
\end{itemize}
\end{theorem}

The proof, which we omit,
is analogous to the proof of \Cref{thm:big-planes-low-degree},
except rather than reducing to \Cref{thm:anand-thomas-classical-low-degree},
we reduce to the soundness of the non-introspective simultaneous low-degree test given by \Cref{thm:simultaneous-ldt}.



\section{The intersecting lines test}\label{sec:intersection}

The introspective low-degree test forces a prover to sample a point from a register and return the evaluation of a global function at that point.
In our eventual protocol, we will want the prover to use the same global function to answer point queries sampled from \emph{multiple} different registers.
In this section, we design a game which allows us to ``transfer" global functions used from one register  to another.
We keep in mind the following picture:
\begin{equation*}
\ket{\psi} = \ket{\mathrm{EPR}^n_q}_{\reg{1}} \otimes \ket{\mathrm{EPR}^n_q}_{\reg{2}} \otimes \ket{\mathrm{aux}}_{\reg{aux}}.
\end{equation*}
We view register~$1$ as the register with the global function
and register~$2$ as the register we would like to transfer this global function to.

To accomplish this, we introduce a new test called the ``intersecting lines test".
This involves performing two introspective line-versus-point low-degree tests.
The first uses register~$1$ as its point register and register~$2$ as its slope register.
This gives us a points prover who samples~$\bu$ from register~$1$ and returns a label on it
and a line prover who samples $\bv$ from register~$2$ and returns a function on the line $\{\bu + \lambda \bv\}$,
and we know that if the points prover labels their point using a low-degree polynomial~$g$, 
then the line prover must label their line with the same polynomial~$g$.
The second low-degree test uses register~$2$ as its point register and register~$1$ as its slope register.
This gives a second line prover who returns a function on the line $\{\bv + \lambda \bu\}$.
Noting that the point $\bu + \bv$ is contained in both line provers' lines,
we can check consistency between their functions
by comparing them on this point,
forcing the second line prover to label their line using~$g$ as well.
This then entails that the second line prover from the second low-degree test must also 
label their point~$\bv$ using~$g$.
Thus, we have successfully ``transferred" the function~$g$ from the first register to the second.

In \Cref{sec:intersecting-lines}, we first introduce the intersecting lines test and prove soundness.
Following that, in \Cref{sec:introspective-lines-transfer}
we introduce an introspective version of this test which will later be used in our $\neexp$ protocol.

\subsection{The intersecting lines test}\label{sec:intersecting-lines}

\begin{definition}[Intersecting lines test]
Let $n, d > 0$ be integers, and let~$q$ be a power of~$2$.
The \emph{intersecting lines test}, denoted $\game_{\mathrm{intersect}}(n, q, d)$, is defined as follows.
Sample $\bu, \bv$ uniformly at random from $\F_q^n$,
and let $\bell$ and $\bell'$ be the two lines
$\bell := \{\bu + \lambda \bv : \lambda \in \F_q\}$
and $\bell' := \{\bv + \lambda \bu : \lambda \in \F_q\}$.
The test is performed as follows.
\begin{itemize}
\itemsep -.5pt
\item[$\circ$] The line $\bell$ and~$\bv$ are given to Alice, who responds with a degree-$d$ polynomial $\boldf: \bell \rightarrow \F_q$.
\item[$\circ$] The line $\bell'$ and~$\bu$ are given to Bob, who responds with a degree-$d$ polynomial $\boldf':\bell' \rightarrow  \F_q$.
\end{itemize}
Alice and Bob pass the test if $\boldf(\bu + \bv) = \boldf'(\bu + \bv)$.
\end{definition}

We begin by showing that although Bob knows~$\bu$, since he doesn't know~$\bv$,
the point $\bu + \bv$ looks like a uniform point in $\bell'$ to him.

\begin{fact}\label{fact:easy-peasy-lemon-peasy}
Conditioned on $\bell'$ and~$\bu$, the point~$\bu+\bv$ is distributed as a uniformly random element in $\bell'$.
\end{fact}
\begin{proof}
Let $\bw$ be a point in $\F_q^n$ such that $\bell' = \{\bw + \lambda \bu : \lambda \in \F_q\}$. 
Then for any $c \in \F_q$, $\bell'$ is also equal to the set
$\{(\bw + c\bu) + (\lambda-c) \bu : \lambda \in \F_q\}$.
Hence, $\bv$ is equally likely to be any element in the set $\{\bw + c \bu : c \in \F_q\}$,
and therefore so is $\bu + \bv$.
Since this set is also equal to $\bell'$, this proves the fact.
\end{proof}

We will be interested in the case when Alice responds using a global function~$g:\F_q^n \rightarrow \F_q$,
always setting $\boldf = g|_{\bell}$.
In this case, the following lemma shows that to succeed with high probability,
Bob must usually play the same global function as Alice.

\begin{lemma}\label{lem:intersecting-lines}
Let $(\psi, M)$ be a POVM strategy for the intersecting lines game
with value $1-\eps$.
Suppose further that there is a measurement $\{G_g\}_g$ in $\mathrm{PolyMeas}(n,d,q)$ such that
$M^{\ell, v}_{f} = G_{[g|_{\ell}=f]}$. Then
\begin{equation*}
(M^{\ell', u}_{f})_{\reg{Alice}} \otimes I_{\reg{Bob}} \consistency_{\eps + d/q} I_{\reg{Alice}} \otimes (G_{[g|_{\ell'}=f]})_{\reg{Bob}}.
\end{equation*}
\end{lemma}
\begin{proof}
Success on the test implies that
\begin{equation*}
(G_{[g(u + v) = \nu]})_{\reg{Alice}} \otimes I_{\reg{Bob}}
\consistency_{\eps} I_{\reg{Alice}} \otimes (M^{\ell', u}_{[f(u + v) = \nu]})_{\reg{Bob}}.
\end{equation*}
But by \Cref{fact:easy-peasy-lemon-peasy}, conditioned on~$\bell'$ and~$\bu$, $\bu + \bv$
is distributed as a uniformly random point in~$\bell'$.
As a result, if we let~$\bw$ be a uniformly random point in~$\bell'$, then
\begin{equation*}
(G_{[g(w) = \nu]})_{\reg{Alice}} \otimes I_{\reg{Bob}}
\consistency_{\eps} I_{\reg{Alice}} \otimes (M^{\ell', u}_{[f(w) = \nu]})_{\reg{Bob}}.
\end{equation*}
The lemma then follows from \Cref{prop:same-on-point-same-on-subspace}.
\end{proof}

\subsection{The introspective intersecting lines test}\label{sec:introspective-lines-transfer}

Now we introduce the introspective intersecting lines test.
This will be an introspective version of the intersecting lines test.

\begin{definition}
Let $n, d> 0$ be integers, let~$q$ be a power of~$2$, and let $\lambda = (2, n, q)$ be register parameters.
The \emph{introspective intersecting lines test}, denoted $\game_{\mathrm{IntroIntersect}}(\lambda, d)$,
is a $\lambda$-register game involving two registers, named ``$1$" and ``$2$", and a possible third auxiliary register.
It involves two line-versus point low-degree tests, instantiated as follows.
\begin{itemize}
\item[$\circ$] Let $\calG_1$ be a copy of $\game_{\mathrm{IntroLowDeg}}(\lambda, d)$
		using register~$1$ as the point register and register~$2$ as the slope register.
	Write $\mathsf{Lines}_1$ for the surface prover in $\calG_1$ and write $\mathsf{Points}_1$ for the points prover.
\item[$\circ$] Let $\calG_2$ be a copy of $\game_{\mathrm{IntroLowDeg}}(\lambda, d)$
		using register~$2$ as the point register and register~$1$ as the slope register.
	Write $\mathsf{Lines}_2$ for the surface prover in $\calG_2$ and write $\mathsf{Points}_2$ for the points prover.
\end{itemize}
Then $\game_{\mathrm{IntroIntersect}}(\lambda, d)$ is defined in \Cref{fig:big-transfer}.
\end{definition}

{
\floatstyle{boxed} 
\restylefloat{figure}
\begin{figure}
With probability~$\tfrac{1}{4}$ each, perform one of the following four tests.
\begin{enumerate}
	\item \textbf{Low degree test 1:} Play $\game_1$.
	\item \textbf{Low degree test 2:} Play $\game_2$.
	\item \textbf{Intersecting lines test:} Flip an unbiased coin $\bb \sim \{0, 1\}$.
		Assign the first role to Player~$\bb$ and the second role to Player~$\overline{\bb}$.
		\begin{itemize}
		\item[$\circ$] $\mathsf{Lines}_1$: Receive $\bell$, $\bv$, $\boldf:\bell\rightarrow \F_q$.
		\item[$\circ$] $\mathsf{Lines}_2$: Receive $\bell'$, $\bu$, $\boldf':\bell'\rightarrow \F_q$.
		\end{itemize}
		Accept if $\bell$ and $\bell'$ both contain $\bu + \bv$ and $\boldf(\bu + \bv) = \boldf'(\bu + \bv)$.
	\item \textbf{Consistency test:} Assign the first role to Player~$1$ and the second role to Player~$2$.
		\begin{itemize}
		\item[$\circ$] $\mathsf{Points}_1$: Receive $\bnu$.
		\item[$\circ$] $\mathsf{Points}_1$: Receive $\bnu'$.
		\end{itemize}
		Accept if $\bnu = \bnu'$.
	\end{enumerate}
	\caption{The game $\game_{\mathrm{IntroIntersect}}(\lambda, d)$.\label{fig:big-transfer}}
\end{figure}
}

\begin{remark}
We remark that although the test runs two separate introspective low-degree tests,
we cannot from these alone conclude that either of the points provers answers according to a global function.
This is because we use the lines $(k=1)$ introspective low-degree test,
whereas from \Cref{thm:big-planes-low-degree} we only know soundness for the planes $(k=2)$ introspective low-degree test.
Hence, proving soundness for the introspective intersecting lines test will require an additional assumption,
i.e.\ that one of the two points provers already answers queries according to a global function.
\end{remark}

Our main result about the introspective intersecting lines test is the following theorem.

\begin{theorem}\label{thm:big-transfer}
Let $n, d> 0$ be integers, let~$q$ be a power of~$2$, and let $\lambda = (2, n, q)$ be register parameters.
Write $\game := \game_{\mathrm{IntroIntersect}}(\lambda, d)$.
Write $A$ for the point prover's measurement in~$\game_1$,
and write~$B$ for the point prover's measurement in~$\game_2$.
\begin{itemize}
\item[$\circ$] \textbf{Completeness:}
		Suppose there is a degree-$d$ polynomial $g:\F_q^{n}\rightarrow \F_q$ such that
		\begin{equation*}
			A_{u, \nu} = \tau_u^Z \otimes I_{\reg{2}} \otimes I_{\reg{aux}} \cdot \bone[\nu = g(u)],
			\quad
			B_{v, \nu} = I_{\reg{1}} \otimes \tau_v^Z \otimes I_{\reg{aux}} \cdot \bone[\nu = g(v)].
		\end{equation*}
		Then there is a value-$1$ $\lambda$-register real commuting EPR strategy strategy for $\game$
		extending~$A$ and~$B$.
\item[$\circ$] \textbf{Soundness:}
		There exists a function $\delta(\eps) = \mathrm{poly}(\eps, d/q)$ such that the following holds.
		Let $\calS$ be a projective $\lambda$-register strategy which passes~$\game$ with probability $1-\eps$.
		Further, suppose that there exists a projective measurement
		$\{G_g\}_g$ in $\mathrm{PolyMeas}(n, d, q)$ acting on the auxiliary register such that
		\begin{equation*}
		A_{u, \nu} = \tau_u^Z \otimes I_{\reg{2}} \otimes G_{[g(u)= \nu]}.
		\end{equation*}
		Then
		\begin{equation*}
		(B_{v, \nu})_{\reg{Alice}} \otimes I_{\reg{Bob}}
			\approx_{\delta(\eps)} (I_{\reg{1}} \otimes \tau_v^Z \otimes G_{[g(v)=\nu]} )_{\reg{Alice}}\otimes I_{\reg{Bob}}.
		\end{equation*}
\end{itemize}
Furthermore,
\begin{equation*}
\qlength{\game} = O(1), \quad
\alength{\game} = O(n \log(q) + d \log(q)),
\end{equation*}
\begin{equation*}
\qtime{\game} = O(1), \quad
\atime{\game} = \poly(n \log(q), d \log(q)).
\end{equation*}
\end{theorem}

\begin{proof}[Proof of \Cref{thm:big-transfer}.]
The runtime and communication complexities follows from the $k=1$ case of the low-degree test.
The completeness follows immediately from the completeness of the introspective low-degree test.

Now, we show soundness.
Write $C$ for the line prover's measurement in~$\game_1$,
and write~$E$ for the line prover's measurement in~$\game_2$.
We can write the point provers' measurements as
\begin{equation*}
A_{u, \nu} = \tau_u^Z \otimes I_{\reg{2}} \otimes G_{[g(u)=\nu]},
\qquad B_{v, \nu} = I_{\reg{1}} \otimes \tau_v^Z \otimes B^v_\nu.
\end{equation*}
The strategy passes the consistency test with probability $1-\delta(\eps)$. As a result,
\begin{equation*}
A_{u, \nu} \otimes I_{\reg{Bob}} \consistency_{\delta(\eps)} I_{\reg{Alice}} \otimes A_{u, \nu}.
\end{equation*}
By \Cref{fact:introspective-outrospective}, this implies that
\begin{equation}\label{eq:to-use-at-the-very-end-when-you-least-expect-it}
G_{[g(u)=\nu]} \otimes I_{\reg{Bob}} \consistency_{\delta(\eps)} I_{\reg{Alice}} \otimes G_{[g(u)=\nu]}
\end{equation}
on state~$\ket{\mathrm{aux}}$ and the uniform distribution on $\F_q^n$.
By \Cref{fact:specialize-the-simeq}, this implies that
\begin{equation}\label{eq:dunno-what-to-name-this-honestly}
G_{[g|_{\ell}=f]} \otimes I_{\reg{Bob}} \consistency_{\delta(\eps)} I_{\reg{Alice}} \otimes G_{[g|_{\ell}=f]},
\end{equation}
where $\ell$ is distributed as $\bell = \{\bu + \lambda \bv : \lambda \in \F_q\}$ for uniformly random $\bu, \bv \in \F_q^n$.

Next, the strategy passes both introspective low-degree tests $\game_1$ and $\game_2$ with probability~$1-\delta(\eps)$.
By~\Cref{thm:big-planes}, this implies measurements $\{C^{\ell,v}_f\}$ and $\{E^{\ell', u}_{f'}\}$ on the auxiliary register such that
\begin{equation}\label{eq:doggone-it-one}
(C_{\ell, v, f})_{\reg{Alice}} \otimes I_{\reg{Bob}}
\approx_{\delta(\eps)} \left(\Pi^v_\ell \otimes \tau_v^Z \otimes C^{\ell,v}_f\right)_{\reg{Alice}} \otimes I_{\reg{Bob}},
\end{equation}
\begin{equation}\label{eq:doggone-it-two}
(E_{\ell', u, f'})_{\reg{Alice}} \otimes I_{\reg{Bob}}
\approx_{\delta(\eps)} \left(\tau_u^Z \otimes \Pi_{\ell'}^u \otimes E^{\ell',u}_{f'}\right)_{\reg{Alice}} \otimes I_{\reg{Bob}}.
\end{equation}
By \Cref{fact:approx-delta-game-value}, we can assume \Cref{eq:doggone-it-one} holds with equality, incurring a loss of only~$\delta(\eps)$ in the game value. (We will do the same for \Cref{eq:doggone-it-two} later.)

The strategy is now in a form that allows us to apply \Cref{lem:big-cross-check}
to the introspective cross-check in $\game_1$.
This implies that the measurements $G_{[g(u) = \nu]}$ and $C^{\ell, v}_f$ give a good strategy for the line-versus point test. In other words,
\begin{equation*}
C^{\ell,v}_{[f(u) = \nu]} \otimes I_{\reg{Bob}} \consistency_{\delta(\eps)} I_{\reg{Alice}} \otimes G_{[g(u)=\nu]}.
\end{equation*}
on state $\ket{\mathrm{aux}}$. By \Cref{prop:same-on-point-same-on-subspace}, this implies that
\begin{equation*}
C^{\ell,v}_{f} \otimes I_{\reg{Bob}} \consistency_{\delta(\eps)} I_{\reg{Alice}} \otimes G_{[g|_{\ell} = f]}.
\end{equation*}
Via~\Cref{eq:dunno-what-to-name-this-honestly}, this implies
\begin{equation*}
C^{\ell,v}_{f} \otimes I_{\reg{Bob}}
\approx_{\delta(\eps)} I_{\reg{Alice}} \otimes G_{[g|_{\ell} = f]}
\approx_{\delta(\eps)} G_{[g|_{\ell} = f]} \otimes I_{\reg{Bob}}.
\end{equation*}
As a result, by \Cref{fact:introspectivize-it},
\begin{equation*}
(C_{\ell, v, f})_{\reg{Alice}} \otimes I_{\reg{Bob}} \approx_{\delta(\eps)} (\Pi_{\ell}^v \otimes \tau_v^Z \otimes G_{[g|_{\ell}=f]})_{\reg{Alice}} \otimes I_{\reg{Bob}}.
\end{equation*}
By assumption, the right-hand side is projective.
As a result, by \Cref{fact:approx-delta-game-value}, we can assume this expression holds with equality, incurring a loss of only~$\delta(\eps)$ in the game value.
Following this, we apply \Cref{fact:approx-delta-game-value} again to assume \Cref{eq:doggone-it-two} holds with equality.

The distribution given by on $(\ell, v)$ and $(\ell', u)$ when we measure with~$C$ and~$E$ is exactly
the question distribution of the (non-introspective) intersecting lines test.
As a result, \Cref{fact:introspective-outrospective} implies that
the measurements $G_{[g_{\ell} = f]}$ and $E^{\ell', u}_{f'}$ pass the intersecting lines test with probability $1-\delta(\eps)$.
In other words,
\begin{equation*}
E^{\ell',u}_{[f'(u+v) = \nu]} \otimes I_{\reg{Alice}}
\consistency_{\delta(\eps)} I_{\reg{Bob}} \otimes G_{[g(u+v) = \nu]}.
\end{equation*}
on state $\ket{\mathrm{aux}}$.
Then by~\Cref{lem:intersecting-lines}, 
\begin{equation*}
E^{\ell',u}_{f'} \otimes I_{\reg{Bob}}
\consistency_{\delta(\eps)} I_{\reg{Alice}} \otimes G_{[g|_{\ell'} = f']}.
\end{equation*}
Via~\Cref{eq:dunno-what-to-name-this-honestly}, this implies
\begin{equation*}
E^{\ell',u}_{f'} \otimes I_{\reg{Bob}}
\approx_{\delta(\eps)} I_{\reg{Alice}} \otimes G_{[g|_{\ell} = f']}
\approx_{\delta(\eps)} G_{[g|_{\ell} = f']} \otimes I_{\reg{Bob}}.
\end{equation*}
Thus, \Cref{fact:introspectivize-it} implies that
\begin{equation*}
(E_{\ell', u, f'})_{\reg{Alice}} \otimes I_{\reg{Bob}}
\approx_{\delta(\eps)} \left(\tau_u^Z \otimes \Pi_{\ell}^u \otimes G_{[g|_{\ell'} = f']}\right)_{\reg{Alice}} \otimes I_{\reg{Bob}}.
\end{equation*}
By assumption, the right-hand side is projective.
As a result, by \Cref{fact:approx-delta-game-value}, we can assume this expression holds with equality, incurring a loss of only~$\delta(\eps)$ in the game value.

The strategy is now in a form that allows us to apply \Cref{lem:big-cross-check}
to the introspective cross-check in $\game_2$.
This implies that the measurements $B^v_{\nu}$ and $G_{[g|_{\ell'}(v) = \nu]} = G_{[g(v) = \nu]}$  give a good strategy for the line-versus-point low-degree test. In other words,
\begin{equation*}
B^v_{\nu} \otimes I_{\reg{Bob}} \consistency_{\delta(\eps)} I_{\reg{Alice}} \otimes G_{[g(v)=\nu]}.
\end{equation*}
on state $\ket{\mathrm{aux}}$.
Via~\Cref{eq:to-use-at-the-very-end-when-you-least-expect-it}, this implies
\begin{equation*}
B^{v}_{\nu} \otimes I_{\reg{Bob}}
\approx_{\delta(\eps)} I_{\reg{Alice}} \otimes G_{[g(v) = \nu]}
\approx_{\delta(\eps)} G_{[g(v) = \nu]} \otimes I_{\reg{Bob}}.
\end{equation*}
As a result, by \Cref{fact:introspectivize-it},
\begin{equation*}
(B_{v, \nu})_{\reg{Alice}} \otimes I_{\reg{Bob}} \approx_{\delta(\eps)} (I_{\reg{1}} \otimes \tau_v^Z \otimes G_{[g(v) = \nu]})_{\reg{Alice}} \otimes I_{\reg{Bob}}.
\end{equation*}
This completes the proof of the theorem.
\end{proof}


\section{The introspective $\neexp$ protocol}\label{sec:big-neexp-protocol}

In this question, we give the complete short-question, introspective $\neexp$ protocol.
The goal is a protocol for $\succinctsquared$ instances of size~$s_{\mathrm{inst}}$
with $\poly(s_{\mathrm{inst}})$ question length and running time and $\poly(2^{s_{\mathrm{inst}}})$ answer length and running time.
Our construction will be an introspective version of the classical PCP construction from \Cref{sec:classical-pcp},
in which we replace the low-degree tests and simultaneous low-degree tests
with our introspective low-degree test and introspective simultaneous low-degree test.

We summarize the protocol here.
Given the $\succinctsquared$ instance~$\calC_{\mathrm{inst}}$
of size~$s_{\mathrm{inst}}$,
let $\calC$ be the size-$s$, $(3n+3)$-input $\succinct$ instance it succinctly represents,
where~$s$ and~$n$ are roughly exponential in~$s_{\mathrm{inst}}$.
Following \Cref{sec:classical-pcp},
we would like the introspective prover to sample strings $\bx_1, \bx_2, \bx_3 \in \F_q^m$ and $(\bb, \bw) \in \F_q^{3+s}$,
which they should return to the verifier.
In addition, they should return the evaluations $g(\bx_1), g(\bx_2), g(\bx_3)$ and $c_1(\bx, \bb, \bw), \ldots, c_{m'}(\bx, \bb, \bw)$,
where $g$ and the $c_i$'s are purported degree-$d$ polynomials.
This suggests using the following registers:
\begin{equation*}
\ket{\mathrm{EPR}_q^m}_{\reg{1}}
\otimes \ket{\mathrm{EPR}_q^m}_{\reg{2}}
\otimes \ket{\mathrm{EPR}_q^m}_{\reg{3}}
\otimes \ket{\mathrm{EPR}_q^{3+s}}_{\reg{4}}.
\end{equation*}
The difficulty in this protocol is ensuring that the polynomials involved are low-degree.
To begin, we can run the introspective low-degree test on the first register,
which guarantees that $g(\bx_1)$ corresponds to a low-degree polynomial. 
Doing the same on registers~$2$ and~$3$ would guarantee
the functions evaluated on~$\bx_2$ and~$\bx_3$ are also low-degree polynomials,
but it would not guarantee that the prover is using \emph{same} low-degree polynomial~$g$ on all three.
Instead, we run the introspective intersecting lines test twice,
ensuring that prover evaluates~$\bx_1$, $\bx_2$, and~$\bx_3$ using the same function~$g$.

Next, we consider the coefficient polynomials $c_1, \ldots, c_{m'}$.
They are evaluated on the concatenated outputs of the four registers, i.e.\ the string $(\bx, \bb, \bw)$.
As a result, we view the four registers as a single superregister of length $m' = 3 m + 3 + s$,
and we would like to perform the introspective simultaneous low-degree test on this superregister.
However, this test requires two additional superregisters of length $m'$ to serve as the direction registers.
As a result, the shared state between the two provers will be of the following form:
\begin{multline*}
\ket{\psi} = 
(\ket{\mathrm{EPR}_q^m}_{\reg{1}})
\otimes \ket{\mathrm{EPR}_q^m}_{\reg{2}}
\otimes \ket{\mathrm{EPR}_q^m}_{\reg{3}}
\otimes \ket{\mathrm{EPR}_q^{3+s}}_{\reg{4}})_{\reg{Super Reg 1}}\\
\otimes (\ket{\mathrm{EPR}_q^{m'}}_{\reg{5}})_{\reg{Super Reg 2}}
\otimes (\ket{\mathrm{EPR}_q^{m'}}_{\reg{6}})_{\reg{Super Reg 3}}
\otimes \ket{\mathrm{aux}}_{\reg{aux}}.
\end{multline*}

Having checked that the provers' functions are low-degree,
we conclude with a consistency check between~$g$ and the~$c_i$'s
to ensure that they encode a satisfying assignment to our $\succinct$.
In \Cref{sec:classical-pcp},
this was done by the ``formula test",
i.e.\ the check that $\mathrm{sat}_{\psi, g}(\bx, \bb, \bw) = \mathrm{zero}_{H, c}(\bx, \bb, \bw)$.
Here, this will be accomplished by an introspective version of this test,
in which the provers sample~$\bx$, $\bb$, and $\bw$ themselves.
Passing this test with high probability proves
that $\calC_{\mathrm{inst}}$ is a YES instance of the $\succinctsquared$ problem.

This section is organized as follows.
In \Cref{sec:computing-registers}, we will discuss the register parameters algorithm,
needed for the register compiler from \Cref{sec:register-overview}.
Next, \Cref{sec:introspective-formula} introduces the introspective formula game.
Finally, \Cref{sec:the-whole-banana} completes the construction and gives the introspective $\neexp$ game.

\subsection{Computing the register parameters}\label{sec:computing-registers}

Given the $\succinctsquared$ instance~$\calC_{\mathrm{inst}}$
of size~$s_{\mathrm{inst}}$,
let $\calC$ be the size-$s$, $(3n+3)$-input $\mathsf{Succinct}$-$\mathsf{3Sat}$ instance it succinctly represents.
To compile our protocol to one sound against general provers,
we need a register parameters algorithm which runs in time $\poly(s_{\mathrm{inst}})$ (\Cref{def:reg-params-generator}).
As described above, the register parameters will be simple functions of the numbers~$s$ and~$n$
(for example $m$, a simple function of~$n$ to be determined later).
However, $s$ and~$n$ themselves may not be easy to compute,
as the natural way of computing them involves first computing~$\calC$,
a time $2^{s_{\mathrm{inst}}}$ task.
We solve this by ``guessing" values for these numbers which are guaranteed to be larger than the actual values,
and then later ``fixing" the circuit~$\calC$ so that it actually has the guessed input length and size.
This is detailed in the following definition.

\begin{definition}\label{def:register-parameters}
Let $\calC_{\mathrm{inst}}$ be a size-$s_{\mathrm{inst}}$ instance of the $\succinctsquared$ problem.
\begin{enumerate}
\item Let $\calC$ be the size-$s$ $\mathsf{Succinct}$-$\mathsf{3Sat}$ instance it succinctly represents.
This circuit takes inputs $i, j, k$, each of some length~$n$, and bits $b_1, b_2, b_3$.
Then $s$ and $n$ can both be trivially upper-bounded by $N := 2^{s_{\mathrm{inst}}}$.
\item Consider a new circuit $\calC_{\mathrm{pad}}$ with inputs $i, j, k \in \{0, 1\}^{N}$ and $b \in \{0, 1\}^3$. We write $i = (i_1, i_2)$, where $i_1$ is of length $N - n$ and $i_2$ is of length $n$, and likewise for~$j$ and~$k$. Let this circuit act as follows:
\begin{itemize}
\item[$\circ$] Compute the $\lor$ of the bits in $i_1$, $j_1$, and~$k_1$. Output~$0$ if this is~$1$.
\item[$\circ$] Otherwise, output $\calC_{\mathrm{dec}}(i_2, j_2, k_2, b_1, b_2, b_3)$.
\end{itemize}
As defined, this circuit has size $s + 3 (N - n) + 2 \leq 4 N =: S$,
and we will pad it with additional gates in a trivial manner so that it has exactly~$S$ gates.
It can be checked that it succinctly represents the same $\sat$ formula as $\calC_{\mathrm{dec}}$.
\end{enumerate}
We set $\mathrm{PadC}(\calC_{\mathrm{inst}}) := \calC_{\mathrm{pad}}$, $\mathrm{PadN}(\calC_{\mathrm{inst}}) := N$,
and $\mathrm{PadS}(\calC_{\mathrm{inst}}):= 4 \cdot N$.
We note that given $\calC_{\mathrm{inst}}$, the value of~$N$ is efficiently computable.
\end{definition}

\subsection{An introspective formula game}\label{sec:introspective-formula}

In this section, we introduce the ``introspective formula game".
This game is the introspective version of the formula check in \Cref{sec:classical-pcp},
in which we check $\mathrm{sat}_{\psi, g}(\bx, \bb, \bw) = \mathrm{zero}_{H, c}(\bx, \bb, \bw)$
on a randomly chosen point $(\bx, \bb, \bw)$ in $\F_q^{m'}$.
Prior to stating the introspective formula game,
we will begin by recalling what this notation means.

Let $\calC_{\mathrm{inst}}$ be a size-$(s_{\mathrm{inst}})$ $\succinctsquared$ instance.
Let $\calC = \mathrm{PadC}(\calC_{\mathrm{inst}})$ be a $\succinct$ instance,
and let $n = \mathrm{PadN}(\calC_{\mathrm{inst}})$ and $s = \mathrm{PadS}(\calC_{\mathrm{inst}})$.
Then $\calC$ is a size-$s$, $(3n+3)$-variable circuit
which is a YES instance of the $\succinct$ problem if and only if $\calC_{\mathrm{inst}}$
is a YES instance of the $\succinctsquared$  problem.
Introduce $h = 2^{t_1}$, $q = 2^{t_2}$, and $m$ such that $N=2^n$, $h$, $q$, and $m$ are exactly admissible parameters (\Cref{def:admissible}).
Set $n' = n + 3 + s$ and $m' = m + 3 + s$.
We also recall the following pieces of notation.
\begin{itemize}
\item[$\circ$] (\Cref{def:canonical-low-degree}): Write $H := H_{t_1, t_2}$.
\item[$\circ$]  (\Cref{def:formula-function}): Given a function $g:\F_q^m \rightarrow \F_q$,
recall the notation $\mathrm{sat}_{\psi,g} := \mathrm{sat}_{\psi,g, n, t_1, t_2}$.
\item[$\circ$]  (\Cref{prop:coefficient-polys}): Writing $H_{\mathrm{zero}} = H^{3m} \otimes \{0, 1\}^{3 + s}$.
Given $c_1, \ldots, c_{m'} : \F_q^{m'} \rightarrow \F_q$, recall the notation $\mathrm{zero}_{H,c} = \mathrm{zero}_{H_{\mathrm{zero}}, c}$.
\end{itemize}

Before stating the introspective formula game, we must first dispense with the following annoying technicality.

\begin{notation}\label{not:annoying-technicality}
In the classical case (\Cref{sec:classical-pcp}),
we have a fixed proof which contains fixed functions which may or may not be low-degree.
In the quantum case, however, we are dealing not with a fixed proof but an interactive prover,
and the formula prover may not respond based on fixed functions (their responses might be randomized, for example).
To account for this, we modify the definitions of $\mathrm{sat}$ and $\mathrm{zero}$ as follows.
First, we recall the notation $g_\psi := g_{\psi, n, t_1, t_2}$ (\Cref{def:encoded-function}). 
\begin{itemize}
	\item[$\circ$] Given $\nu_1, \nu_2, \nu_3 \in \F_q$, define
		\begin{equation*}
			\mathrm{sat}_{\psi, \nu}(x, b, w) := g_{\psi}(x, b, w) \cdot (\nu_1 - b_1) (\nu_2 - b_2) (\nu_3 - b_3).
		\end{equation*}
	\item[$\circ$] Given $\mu_1, \ldots, \mu_{m'} \in \F_q$, define
		\begin{equation*}
			\mathrm{zero}_{H, \mu}(x) = \sum_{i=1}^{m'} \mathrm{zero}_{(H_{\mathrm{zero}})_i}(x_i) \cdot \mu_i,
		\end{equation*}
		where by definition $(H_{\mathrm{zero}})_i = H$ for $i \in [3m]$ and $(H_{\mathrm{zero}})_i = \{0, 1\}$ otherwise.
\end{itemize}
We note that if there is a function~$g$ such that $\nu_i = g(x_i)$, then $\mathrm{sat}_{\psi, \nu} = \mathrm{sat}_{\psi, g}$.
Similarly, if there are functions $c_1, \ldots, c_{m'}$ such that $\mu_i = c_i(x)$, then $\mathrm{zero}_{H, \mu} = \mathrm{zero}_{H, c}$.
\end{notation}

Now we state the introspective formula game.

\begin{definition}\label{def:formula-game}
Let $\calC_{\mathrm{inst}}$ be a size-$(s_{\mathrm{inst}})$ $\succinctsquared$ instance.
Let $n = \mathrm{PadN}(\calC_{\mathrm{inst}})$ and $s = \mathrm{PadS}(\calC_{\mathrm{inst}})$.
Suppose $n$, $h = 2^{t_1}$, $q = 2^{t_2}$, and $m$ are exactly admissible parameters.
The \emph{introspective formula game},
denoted $\game := \game_{\mathrm{IntroForm}}(\calC_{\mathrm{inst}}, h, q, m)$, is defined in \Cref{fig:formula-game}.
This is a $\lambda_{\calC_{\mathrm{inst}}, q} := (4, \ell, q)$-register game, for $\ell = (m, m, m, 3 + s)$.
Furthermore,
\begin{equation*}
\qlength{\game} = O(1), \quad
\alength{\game} = O(m' \log(q)),
\end{equation*}
\begin{equation*}
\qtime{\game} = O(1), \quad
\atime{\game} = \poly(s, n, n', h, q, m').
\end{equation*}
\end{definition}
{
\floatstyle{boxed} 
\restylefloat{figure}
\begin{figure}
Flip an unbiased coin $\bb \sim \{0, 1\}$.
	\begin{itemize}
	\item[$\circ$] Player~$\bb$: Give ($Z$, $Z$, $Z$, $Z$, ``formula"); 
				receive $\bu_1, \bu_2, \bu_3, (\bb, \bw)$ and $\bnu_1, \bnu_2, \bnu_3$ and $\bmu_1, \ldots, \bmu_{M'}$.
	\end{itemize}
	Compute~$\mathrm{sat}_{\psi, \bnu}(\bu, \bb, \bw)$ and~$\mathrm{zero}_{H, \bmu}(\bu, \bb, \bw)$.
	Accept if they are equal.
	\caption{The game $\game_{\mathrm{IntroForm}}(\calC_{\mathrm{inst}}, h, q, m)$.\label{fig:formula-game}}
\end{figure}
}

\begin{notation}
In the case when a prover is given the question $(Z, Z, Z, Z, \mathrm{``formula"})$, we refer to it as the \emph{formula prover}. 
It has the following intended behavior.
\begin{enumerate}
\setcounter{enumi}{2}
\item \textbf{Formula prover:}
	\begin{description}[align=left]
	\item [Input:] Pauli basis queries $(Z, Z, Z, Z)$ and auxiliary query ``formula".
	\item [Output:] Strings $u_1, u_2, u_3 \in\F_q^{m}$ and $(b, w) \in \F_q^{3+s}$.
			Three numbers $\nu_1,\nu_2, \nu_3 \in \F_q$
			and $m'$ numbers $\mu_1, \ldots, \mu_{m'} \in \F_q$.
	\item [Goal:] The prover should act as follows.
		\begin{itemize}
		\item[$\circ$] The prover sets $\nu_1 = g(u_1)$, $\nu_2 = g(u_2)$, $\nu_3 = g(u_3)$,
					where $g:\F_q^m \rightarrow \F_q$ is a global degree-$d_1$ polynomial selected independently of~$u$.
		\item[$\circ$] They then set $\mu_i = c_i(u_1, u_2, u_3, b, w)$,
					where for each $i$,
					$c_i:\F_q^{m'} \rightarrow \F_q$ is a global degree-$d_2$ polynomial selected independently of $(u, b, w)$.
		\end{itemize}
	\end{description}
\end{enumerate}
Here, $d_1$ and~$d_2$ are polynomial degrees which will be selected later.
We will also refer to the \emph{formula prover's measurement}, 
which refers to the measurement $\{F_{u, b, w, \nu_i,\mu_j}\}$ such that
\begin{equation*}
F_{u, b, w, \nu_1, \nu_2, \nu_3, \mu_1, \ldots, \mu_{m'}}
= M^{Z, Z, Z, Z, \mathrm{``formula"}}_{u, b, w, \nu_1, \nu_2, \nu_3, \mu_1, \ldots, \mu_{m'}}.
\end{equation*}
\end{notation}

We begin by showing the completeness case of the introspective formula game.

\begin{proposition}[Introspective formula game completeness]\label{prop:formula-game-completeness}
Suppose~$\calC_{\mathrm{inst}}$ is a YES instance of the $\succinctsquared$ problem.
Let $a:\{0,1\}^n \rightarrow \{0, 1\}$ be a satisfying assignment to the $\mathsf{3Sat}$ instance it encodes,
and let $g:=g_a:\F_q^m \rightarrow \F_q$ be its low-degree encoding.
Let $c_1, \ldots, c_{m'}:\F_q^{m'} \rightarrow \F_q$ be the coefficient polynomials guaranteed to make 
$\mathrm{sat}_{\psi, g} = \mathrm{zero}_{H_{\mathrm{zero}}, c}$ by \Cref{prop:coefficient-polys}.
Both $g$ and the $c_i$'s are degree-$O(hn')$ polynomials.
Consider the $\lambda_{\calC_{\mathrm{inst}},q}$-register strategy $(\psi, A)$ with no auxiliary register in which
\begin{equation*}
A_{u, b, w, \nu, \mu} =  \tau^Z_{u_1} \otimes \tau^Z_{u_2} \otimes \tau^Z_{u_3} \otimes \tau^Z_{b, w}  \cdot \bone[\nu_i = g(u_i), \mu_j = c_j(u, b, w)],
\end{equation*}
where the indices range over $i \in [3]$ and $j \in [m']$.
Then this strategy passes $\game_{\mathrm{IntroForm}}(\calC_{\mathrm{inst}}, h, q, m)$ with probability~$1$.
\end{proposition}
\begin{proof}
This game is simply the oracularized version of the formula check in the classical PCP.
The proposition follows from the discussion in \Cref{sec:the-pcp}.
\end{proof}

Our next lemma covers the soundness case of the introspective formula game.
It concerns provers of a particular form,
namely those whose measurements correspond to low-degree polynomials.
We show that if there exists such a prover with nonnegligible value, then the formula must be satisfiable.

\begin{lemma}[Formula game partial soundness]\label{lem:polynomial-means-good}
Let $\calC_{\mathrm{inst}}$ be a $\succinctsquared$ instance, and set $\game := \game_{\mathrm{IntroForm}}(\calC_{\mathrm{inst}}, h, q, m)$.
Let $\calS = (\psi, A)$ be a $\lambda_{\calC_{\mathrm{inst}}, q}$-register strategy.
Consider a measurement on the auxiliary register
\begin{equation*}
G = \{G_{g, c_1, \ldots, c_{m'}}\}
\end{equation*}
with outcomes degree-$d_1$ polynomials $g:\F_q^m \rightarrow \F_q$
and degree-$d_2$ polynomials $c_1, \ldots, c_{m'}: \F_Q^{m'} \rightarrow \F_q$.
Suppose~$A$ has the following form: for each $u$, $b$, $w$, $\nu$, and $\mu$,
\begin{equation}\label{eq:weird-form}
A_{u, b, w, \nu, \mu} = \tau^Z_{u_1} \otimes \tau^Z_{u_2} \otimes \tau^Z_{u_3} \otimes \tau^Z_{b, w} \otimes \left(G_{[g(u_i)= \nu_i,  c_j(u, b, w)= \mu_j]}\right)_{\reg{aux}},
\end{equation}
where the subscript of the~$G$ measurement ranges over all $i \in [3]$ and $j \in [m']$.
If the probability $\calS$ passes $\calG$ is at least
\begin{equation*}
\frac{\max\{O(h n') + 3 d_1, h + d_2\}}{q},
\end{equation*}
then $\psi$ is satisfiable.
\end{lemma}
\begin{proof}
Consider the following three step strategy:
\begin{enumerate}
\item Measure the auxiliary register with $\{G_{g, c_1, \ldots, c_{m'}}\}$, receiving functions $\bg$, $\bc_1$, \ldots, $\bc_{m'}$.
\item Measure the EPR registers in the~$Z$ basis, receiving $\bu$, $\bb$, and~$\bw$.
\item Output $\bu$, $\bb$, $\bw$, $\bg(\bu_1)$, $\bg(\bu_2)$, $\bg(\bu_3)$ and $\bc_1(\bu,\bb, \bw)$ through $\bc_{M'}(\bu, \bb, \bw)$.
\end{enumerate}
This passes the formula game with probability $\valstrat{\calG}{\calS}$.
Then there exists functions $g$, $c_1$, \ldots, $c_{m'}$ such that conditioned on measuring them in step one,
this strategy passes with probability at least $\valstrat{\calG}{\calS}$.
By the remark at the end of \Cref{not:annoying-technicality},
this is the probability that
\begin{equation*}
\mathrm{sat}_{\psi, g}(\bx, \bb, \bw) = \mathrm{zero}_{H, c}(\bx, \bb, \bw),
\end{equation*}
where $(\bx, \bb, \bw)$ is drawn from $\F_q^{m'}$ uniformly at random.
The lemma follows from \Cref{lem:the-grand-finale}.
\end{proof}

\subsection{The complete introspective protocol}\label{sec:the-whole-banana}

In this section, we introduce the introspective protocol for $\neexp$ and prove its correctness.
The introspective $\neexp$ protocol builds on top of the introspective formula game
by using a series of introspective low-degree tests to ensure that the formula prover satisfies the condition in \Cref{eq:weird-form}.
Having done this, we can then apply \Cref{lem:polynomial-means-good},
ensuring that if a strategy passes with high probability, then the instance is satisfiable.

\begin{definition}\label{def:the-protocol-to-end-all-protocols}
Let $\calC_{\mathrm{inst}}$ be a size-$(s_{\mathrm{inst}})$ $\succinctsquared$ instance.
Let $n = \mathrm{PadN}(\calC_{\mathrm{inst}})$ and $s = \mathrm{PadS}(\calC_{\mathrm{inst}})$.
The verifier chooses $h = 2^{t_1}$, $q = 2^{t_2}$, $m$, and $d$
such that $m$, $h$, $q$, and $m$ are exactly admissible parameters satisfying
\begin{equation*}
h = \Theta(n),
\quad
m = \Theta\left(\frac{n}{\log(n)}\right),
\quad
q = \poly(n),
\quad
d = O(hn') = O(n^2).
\end{equation*}
(We will choose the polynomial for~$q$ in \Cref{thm:main-result-of-this-part} below.)
The verifier sets $\lambda = (6, \ell, q)$,
where $\ell = (m,m,m, 3 + s, m', m')$.

We begin by instantiating the following list of subroutines.
\begin{itemize}
\item[$\circ$]
Let $\lambda_{\mathrm{LD}} = (3, m, q)$ be register parameters.
Let $\game_{\mathrm{LD}}$ be a copy of $\game_{\mathrm{IntroLowDeg}}(\lambda_{\mathrm{LD}}, d)$,
	using register~$1$ as the point register and registers~$2$ and~$3$ as the direction registers.
	Write $\mathsf{Points}_1$ for the points prover.
\item[$\circ$]
Let $\lambda_{\mathrm{IL}} = (2, m, q)$ be register parameters.
Let $\game_{\mathrm{IL1}}$ be a copy of $\game_{\mathrm{IntroIntersect}}(\lambda_{\mathrm{IL}}, d)$ on registers~$1$ and~$2$
	whose points prover for register~$1$ is $\mathsf{Points}_1$ from $\game_{\mathrm{LD}}$.
	Write $\mathsf{Points}_2$ for the points prover on register~$2$.
\item[$\circ$] Let $\game_{\mathrm{IL2}}$ be a copy of
	$\game_{\mathrm{IntroIntersect}}(\lambda_{\mathrm{IL}}, d)$ on registers~$1$ and~$3$
	whose points prover for register~$1$ is $\mathsf{Points}_1$ from $\game_{\mathrm{LD}}$.
	Write $\mathsf{Points}_3$ for the points prover on register~$3$.
\item[$\circ$] Let $\game_{\mathrm{F}}$ be a copy of $\game_{\mathrm{IntroForm}}(\calC_{\mathrm{inst}}, h, q, m)$ on registers~$1$, $2$, $3$, and~$4$.
	Write $\mathsf{Formula}$ for the formula prover.
\item[$\circ$]
Let $\lambda_{\mathrm{LDSUP}} = (3, m', q)$ be register parameters.
Let $\game_{\mathrm{LDSUP}}$ be a copy of $\game_{\mathrm{IntroLowDeg}}(\lambda_{\mathrm{LDSUP}}, d, 3+m')$,
applied to the following three superregisters:
registers~$1$ through~$4$ are combined into the point superregister,
register~$5$ is used as the first direction superregister,
and register~$6$ is used as the second direction superregister.
In addition, use $\mathsf{Formula}$ from $\game_{\mathrm{form}}$ as its points prover.
\end{itemize}
Then the \emph{introspective $\neexp$ game}, denoted $\game_{\mathrm{Intro}\neexp}(\calC_{\mathrm{inst}})$, is defined in \Cref{fig:big-neexp}.
\end{definition}

{
\floatstyle{boxed} 
\restylefloat{figure}
\begin{figure}
With probability~$\tfrac{1}{9}$ each, perform one of the following nine tests.
\begin{enumerate}
	\item \textbf{Low degree test:} Play $\game_{\mathrm{LD}}$.
	\item \textbf{Intersecting lines test 1:} Play $\game_{\mathrm{IL1}}$.
	\item \textbf{Intersecting lines test 2:} Play $\game_{\mathrm{IL2}}$.
	\item \textbf{Simultaneous low degree test:} Play $\game_{\mathrm{LDSUP}}$.
	\item \textbf{Formula test:} Player $\game_{\mathrm{F}}$.
\end{enumerate}
For the remaining tests, flip an unbiased coin $\bb \sim \{0, 1\}$.
Assign the first role to Player~$\bb$ and the second role to Player~$\overline{\bb}$.
\begin{enumerate}
\setcounter{enumi}{5}
	\item \textbf{Consistency test 1:} 
		\begin{itemize}
		\item[$\circ$] $\mathsf{Points}_1$: Receive $\bnu$.
		\item[$\circ$] $\mathsf{Formula}$: Receive $\bnu_1$.
		\end{itemize}
		Accept if $\bnu = \bnu_1$.
	\item \textbf{Consistency test 2:} 
		\begin{itemize}
		\item[$\circ$] $\mathsf{Points}_2$: Receive $\bnu$.
		\item[$\circ$] $\mathsf{Formula}$: Receive $\bnu_2$.
		\end{itemize}
		Accept if $\bnu = \bnu_2$.
	\item \textbf{Consistency test 3:} 
		\begin{itemize}
		\item[$\circ$] $\mathsf{Points}_3$: Receive $\bnu$.
		\item[$\circ$] $\mathsf{Formula}$: Receive $\bnu_3$.
		\end{itemize}
		Accept if $\bnu = \bnu_3$.
	\item \textbf{Consistency test 4:} 
		\begin{itemize}
		\item[$\circ$] $\mathsf{Formula}$: Receive $\bnu_1, \bnu_2, \bnu_3$ and $\bmu_1, \ldots, \bmu_{m'}$.
		\item[$\circ$] $\mathsf{Formula}$: Receive $\bnu_1', \bnu_2', \bnu_3'$ and $\bmu_1', \ldots, \bmu_{m'}'$.
		\end{itemize}
		Accept if $\bnu_i = \bnu_i'$ and $\bmu_j = \bmu_j'$ for all $i \in [3]$, $j \in [m']$.
	\end{enumerate}
	\caption{The game $\game_{\mathrm{Intro}\neexp}(\calC_{\mathrm{inst}})$.\label{fig:big-neexp}}
\end{figure}
}

The main result \Cref{part:neexp} is the following theorem.

\begin{theorem}\label{thm:main-result-of-this-part}
Let $\calC_{\mathrm{inst}}$ be a size-$(s_{\mathrm{inst}})$ $\succinctsquared$ instance.
Let $q$ be a sufficiently large $\poly(n)$ and $\eps > 0$ a sufficiently small constant such that \Cref{eq:formula-success-probability} is at least $\tfrac{1}{2}$
and \Cref{eq:other-expression-dont-even-know-what-to-call-this} is less than $\tfrac{1}{2}$.
Write $\game := \game_{\mathrm{Intro}\neexp}(\calC_{\mathrm{inst}})$.
\begin{itemize}
\item[$\circ$] \textbf{Completeness:}
		Suppose $\calC_{\mathrm{inst}}$ encodes a satisfiable formula.
		Then there is a value-$1$ $\lambda$-register strategy for $\game$ with no auxiliary register.
\item[$\circ$] \textbf{Soundness:}
		If there is a $\lambda$-register strategy for $\game$ with value at least $1-\eps$,
		then $\calC_{\mathrm{inst}}$ encodes a satisfiable formula.
\end{itemize}
Furthermore,
\begin{equation*}
\qlength{\game} = O(1), \quad
\alength{\game} = \poly(2^{s_\mathrm{inst}}),
\end{equation*}
\begin{equation*}
\qtime{\game} = O(1), \quad
\atime{\game} = \poly(2^{s_\mathrm{inst}}).
\end{equation*}
\end{theorem}

\begin{proof}
The question lengths and question runtimes are both $O(1)$ because all involved subtests have $O(1)$ question complexity.
The answer lengths and question runtimes are both $\poly(2^{s_{\mathrm{inst}}})$
because all our parameters are at most polynomial in $n = 2^{s_{\mathrm{inst}}}$,
and the question lengths and question runtimes of all involved subtests are polynomial in these parameters.

We name the measurements used by the provers as follows.
\begin{equation*}
\mathsf{Points}_1: A,\quad
\mathsf{Points}_2: B,\quad
\mathsf{Points}_3: C,\quad
\mathsf{Formula}: F.
\end{equation*}
We will write the identity matrix on registers~$5$ and~$6$ as $I_{\reg{5, 6}} := I_{\reg{5}} \otimes I_{\reg{6}}$.

\paragraph{Completeness.}
Suppose $\calC_{\mathrm{inst}}$ encodes a satisfiable formula.
By \Cref{prop:formula-game-completeness}, 
there are degree-$d$ polynomials $g:\F_q^m \rightarrow \F_q$ and $c_1, \ldots, c_{m'}:\F_q^{m'} \rightarrow \F_q$ such that if we define
\begin{equation*}
F_{u, b, w, \nu, \mu} =  \tau^Z_{u_1} \otimes \tau^Z_{u_2} \otimes \tau^Z_{u_3} \otimes \tau^Z_{b, w} \otimes I_{\reg{5,6}} \cdot \bone[\nu_i = g(u_i), \mu_j = c_j(u, b, w)],
\end{equation*}
then this strategy passes the formula test with probability~$1$.
We extend this strategy to the remaining measurements as follows.
\begin{equation*}
A_{u_1, \nu_1}
= \tau^Z_{u_1} \otimes I_{\reg{2}} \otimes I_{\reg{3}} \otimes I_{\reg{4}}\otimes I_{\reg{5,6}}\cdot \bone[\nu_1 = g(u_1)],
\end{equation*}
\begin{equation*}
B_{u_2, \nu_2}
= I_{\reg{1}} \otimes \tau^Z_{u_2} \otimes I_{\reg{3}} \otimes I_{\reg{4}} \otimes I_{\reg{5,6}} \cdot\bone[\nu_2 = g(u_2)],
\end{equation*}
\begin{equation*}
C_{u_3, \nu_3}
= I_{\reg{1}} \otimes I_{\reg{2}} \otimes \tau^Z_{u_3} \otimes I_{\reg{4}} \otimes I_{\reg{5, 6}}\cdot \bone[\nu_3 = g(u_3)].
\end{equation*}
By the completeness of the introspective low-degree and intersecting lines tests, these can be extended to a strategy which passes the whole test with probability~$1$.

\paragraph{Soundness.}
Throughout this proof, we use $\delta(\eps)$ to represent functions of the form
\begin{equation*}
\delta(\eps) = \poly(\eps, m \cdot d/q^e),
\end{equation*}
where $e > 0$ is an absolute constant.

\paragraph{Low-degree tests.}
The strategy passes the introspective low-degree test with probability $1-\delta(\eps)$.
Applying~\Cref{thm:big-planes-low-degree}, there is a measurement $G = \{G_g\}$ in $\mathrm{PolyMeas}(m, d, q)$ such that
\begin{equation*}
(A_{u_1, \nu_1})_{\reg{Alice}} \otimes I_{\reg{Bob}}
\approx_{\delta(\eps)}
\left(\tau^Z_{u_1} \otimes I_{\reg{2}} \otimes I_{\reg{3}} \otimes I_{\reg{4}} \otimes I_{\reg{5,6}} \otimes (G_{[g(u_1) = \nu_1]})_{\reg{aux}}\right)_{\reg{Alice}} \otimes I_{\reg{Bob}}.
\end{equation*}
By \Cref{fact:approx-delta-game-value}, we can assume this holds with equality
with a loss of only $\delta(\eps)$ in the game value.
In addition, by \Cref{thm:naimark}, we can assume that the~$G$ measurements are all projective.

Next, the strategy passes the two introspective intersecting lines tests with probability $1-\delta(\eps)$ each.
By~\Cref{thm:big-transfer}, this implies that
\begin{equation}\label{eq:whatevs}
(B_{u_2, \nu_2})_{\reg{Alice}} \otimes I_{\reg{Bob}}
\approx_{\delta(\eps)}
\left(I_{\reg{1}} \otimes \tau^Z_{u_2} \otimes I_{\reg{3}} \otimes I_{\reg{4}} \otimes I_{\reg{5,6}} \otimes (G_{[g(u_2) = \nu_2]})_{\reg{aux}}\right)_{\reg{Alice}} \otimes I_{\reg{Bob}},
\end{equation}
\begin{equation}\label{eq:whatevs-dos}
(C_{u_3, \nu_3})_{\reg{Alice}} \otimes I_{\reg{Bob}}
\approx_{\delta(\eps)}
\left(I_{\reg{1}} \otimes I_{\reg{2}} \otimes \tau^Z_{u_3} \otimes I_{\reg{4}} \otimes I_{\reg{5,6}} \otimes (G_{[g(u_3) = \nu_3]})_{\reg{aux}}\right)_{\reg{Alice}} \otimes I_{\reg{Bob}}.
\end{equation}

Similarly, the strategy passes the introspective simultaneous low-degree test with probability $1-\delta(\eps)$.
Applying~\Cref{thm:big-simultaneous-planes-low-degree}, there is a measurement $J = \{J_{f_1, f_2, f_3, c_1, \ldots, c_{m'}}\}$ in $\mathrm{PolyMeas}(m',d, q, 3+m')$ such that
\begin{multline}\label{eq:whatevs-tres}
(F_{u, b, w, \nu, \mu})_{\reg{Alice}} \otimes I_{\reg{Bob}}\\
\approx_{\delta(\eps)}
\left(\tau^Z_{u_1} \otimes \tau^Z_{u_2} \otimes \tau^Z_{u_3} \otimes \tau^Z_{b, w}
	\otimes I_{\reg{5,6}} \otimes (J_{[f_i(u, b, w) = \nu_i, c_j(u, b, w) = \mu_j]})_{\reg{aux}}\right)_{\reg{Alice}} \otimes I_{\reg{Bob}},
\end{multline}
where the subscript of the~$J$ measurement ranges over all $i \in [3]$ and $j \in [m']$.
By \Cref{fact:approx-delta-game-value}, we can assume \Cref{eq:whatevs,eq:whatevs-dos,eq:whatevs-tres} holds with equality
with a loss of only $\delta(\eps)$ in the game value.
In addition, by \Cref{thm:naimark}, we can assume that the~$J$ measurements are all projective.

\paragraph{Consistency tests.}
The strategy passes the four consistency tests with probability $1-\delta(\eps)$ each, implying
\begin{equation*}
(F_{u_1, \nu_1})_{\reg{Alice}} \otimes I_{\reg{Bob}} \simeq_{\delta(\eps)} I_{\reg{Alice}} \otimes (A_{u_1, \nu_1})_{\reg{Bob}},
\end{equation*}
\begin{equation*}
(F_{u_2, \nu_2})_{\reg{Alice}} \otimes I_{\reg{Bob}} \simeq_{\delta(\eps)} I_{\reg{Alice}} \otimes (B_{u_2, \nu_2})_{\reg{Bob}},
\end{equation*}
\begin{equation*}
(F_{u_3, \nu_3})_{\reg{Alice}} \otimes I_{\reg{Bob}} \simeq_{\delta(\eps)} I_{\reg{Alice}} \otimes (C_{u_3, \nu_3})_{\reg{Bob}},
\end{equation*}
\begin{equation*}
(F_{u, b, w, \nu, \mu})_{\reg{Alice}} \otimes I_{\reg{Bob}} \simeq_{\delta(\eps)} I_{\reg{Alice}} \otimes (F_{u, b, w, \nu, \mu})_{\reg{Bob}}.
\end{equation*}
By introspection~(\Cref{fact:introspective-consistency}), these imply the following statements:
\begin{equation*}
(J_{[f_1(u, b, w) = \nu_1]})_{\reg{Alice}} \otimes I_{\reg{Bob}} \simeq_{\delta(\eps)} I_{\reg{Alice}} \otimes (G_{[g(u_1) = \nu_1]}),
\end{equation*}
\begin{equation*}
(J_{[f_2(u, b, w) = \nu_2]})_{\reg{Alice}} \otimes I_{\reg{Bob}} \simeq_{\delta(\eps)} I_{\reg{Alice}} \otimes (G_{[g(u_2) = \nu_2]}),
\end{equation*}
\begin{equation*}
(J_{[f_3(u, b, w) = \nu_3]})_{\reg{Alice}} \otimes I_{\reg{Bob}} \simeq_{\delta(\eps)} I_{\reg{Alice}} \otimes (G_{[g(u_3) = \nu_3]}),
\end{equation*}
\begin{equation*}
(J_{[c_j(u, b, w) = \mu_j]})_{\reg{Alice}} \otimes I_{\reg{Bob}} \simeq_{\delta(\eps)} I_{\reg{Alice}} \otimes (J_{[c_j(u, b, w) = \mu_j]}),
\end{equation*}
where the subscript of the~$J$ measurement ranges over all $i \in [3]$ and $j \in [m']$.
Here, these statements are with respect to the strategy's auxiliary state and to the uniform distribution on $(\bu, \bb, \bw) \in \F_q^{m'}$.

Now we apply \Cref{fact:low-degree-sandwich}.
To do so, let us specify the sets $\calG_i$ and the distance parameter.
The three sets $\calG_2$, $\calG_3$, and $\calG_4$ will just contain all degree-$d$ polynomials $g:\F_q^{m'} \rightarrow \F_q$.
(Note that we can view the outputs of $G_g$ as degree-$d$ polynomials which disregard all of their input $(u, b, w)$ aside from one of the three strings $u_1$, $u_2$, or $u_3$.)
By Schwarz-Zippel, these have distance at least $1-d/q$.
The remaining set, $\calG_1$, is defined as follows:
for each tuple of degree-$d$ polynomials $c_1, \ldots, c_{m'}$,
it contains a function~$c$ defined as $c(u, b, w) = (c_1(u, b, w), \ldots, c_{m'}(u, b, w))$.
Any two nonequal $c, c' \in \calG_1$ have some coordinate~$i$ in which $c_i \neq c_i'$,
and on this coordinate alone they will have distance at least $1-d/q$ by Schwarz-Zippel.
Thus, $c$ and $c'$ have distance at least $1-d/q$.

Define the measurement $\{K_{g, c_1, \ldots, c_{m'}}\}$ as 
\begin{equation*}
K_{g, c_1, \ldots, c_{m'}} := G_g \cdot J_{c_1, \ldots, c_{m'}} \cdot G_g,
\end{equation*}
Then \Cref{fact:low-degree-sandwich} implies that
\begin{equation*}
(J_{[f_i(u, b, w) = \nu_i, c_j(u, b, w) = \mu_j]})_{\reg{Alice}} \otimes I_{\reg{Bob}} \simeq_{\delta(\eps)} I_{\reg{Alice}} \otimes (K_{[g(u_i) = \nu_i, c_j'(u, b, w) = \mu_j]}),
\end{equation*}
where the subscripts range over all $i \in [3]$ and $j \in [m']$.
By introspection~(\Cref{fact:introspective-consistency}), this implies that
\begin{multline}\label{eq:cant-believe-it-but-we-did-it}
(F_{u, b, w, \nu, \mu})_{\reg{Alice}} \otimes I_{\reg{Bob}}\\
\simeq_{\delta(\eps)}
\left(\tau^Z_{u_1} \otimes \tau^Z_{u_2} \otimes \tau^Z_{u_3} \otimes \tau^Z_{b, w}
	\otimes I_{\reg{5,6}} \otimes (K_{[g(u_i) = \nu_i, c_j(u, b, w) = \mu_j]})_{\reg{aux}}\right)_{\reg{Alice}} \otimes I_{\reg{Bob}},
\end{multline}
where the subscripts range over all $i \in [3]$ and $j \in [m']$.
By \Cref{fact:approx-delta-game-value}, we can assume \Cref{eq:cant-believe-it-but-we-did-it} holds with equality
with a loss of only $\delta(\eps)$ in the game value.

\paragraph{Formula test:}
At this point, the formula prover's strategy~$F$ satisfies the condition in \Cref{eq:weird-form} with $d_1 = d_2 = d$.
In addition, it passes the introspective formula test with probability
\begin{equation}\label{eq:formula-success-probability}
1-\poly(\eps, m \cdot d/q),
\end{equation}
which by our setting of parameters is at least $\tfrac{1}{2}$.
Finally, our setting of parameters also implies that
\begin{equation}\label{eq:other-expression-dont-even-know-what-to-call-this}
\frac{\max\{O(h n') + 3d, h + d\}}{q}
\end{equation}
is less than $\tfrac{1}{2}$.
As a result, we can apply \Cref{lem:polynomial-means-good} to conclude that~$\psi$ is satisfiable.
\end{proof}

\Cref{thm:main-result-of-this-part} only proves soundness of the introspective $\neexp$ protocol against $\lambda$-register strategies.
Our last step is to compile this protocol into one which is sound against \emph{general} strategies,
while only slightly increasing the question length.

\begin{corollary}\label{cor:succinct-sat-protocol-with-big-answer-size}
There is an absolute constant $\eps > 0$ such that the following is true.
Let $\calC_{\mathrm{inst}}$ be a size-$(s_{\mathrm{inst}})$ $\succinctsquared$ instance.
Then there exists a game $\game := \game_{\mathrm{Intro}\neexp}(\calC_{\mathrm{inst}})$ with the following properties.
\begin{itemize}
\item[$\circ$] \textbf{Completeness:}
		Suppose $\calC_{\mathrm{inst}}$ encodes a satisfiable formula.
		Then there is a value-$1$ real commuting EPR strategy for $\game$.
\item[$\circ$] \textbf{Soundness:}
		If there is a strategy for $\game$ with value at least $1-\eps$,
		then $\calC_{\mathrm{inst}}$ encodes a satisfiable formula.
\end{itemize}
Furthermore,
\begin{equation*}
\qlength{\game} = O(s_{\mathrm{inst}}), \quad
\alength{\game} = \poly(2^{s_{\mathrm{Inst}}}),
\end{equation*}
\begin{equation*}
\qtime{\game} = O(s_{\mathrm{inst}}), \quad
\atime{\game} = \poly(2^{s_{\mathrm{inst}}}).
\end{equation*}
\end{corollary}
\begin{proof}
Let $n = \mathrm{PadN}(\calC_{\mathrm{inst}})$ and $s = \mathrm{PadS}(\calC_{\mathrm{inst}})$.
Set $m = \Theta(n/\log(n))$ and $q = \poly(n)$,
as in \Cref{def:the-protocol-to-end-all-protocols}.
Set $\lambda = (6, \ell, q)$,
where $\ell = (m, m, m, 3+s, 3m + 3 + s, 3m + 3 + s)$.
Then $\game_{\mathrm{Intro}\neexp}(\calC_{\mathrm{inst}})$
is a $\lambda$-register game.
Furthermore, by \Cref{def:register-parameters},
the register parameters are computable in time $\poly(s_{\mathrm{inst}})$.

Let $\eps >0$ be as in \Cref{thm:main-result-of-this-part}, and select a constant $\eps' > 0$ and $\eta_1, \ldots, \eta_{6} = 1/\poly(n)$
such that $\delta(\eps') \leq \eps$, where $\delta(\eps') = \poly(\eps', \eta_1, \ldots, \eta_{6})$ is as in \Cref{thm:uniform-registers}.
Now, if we apply \Cref{thm:uniform-registers}, it gives us a game~$\game$ with the following properties.
\begin{itemize}
\item[$\circ$] If $\calC$ is a ``Yes" instance, then $\game_{\mathrm{Intro}\neexp}(\calC_{\mathrm{inst}})$ has a value-$1$ strategy with no auxiliary state, which implies that $\game$ has a value-$1$ commuting EPR strategy.
\item[$\circ$] If $\calC$ is a ``No" instance, then every $\lambda$-register strategy for
		$\game_{\mathrm{Intro}\neexp}(\calC_{\mathrm{inst}})$ has value less than $1-\eps$.
		By our choice of parameters, this is less than $1-\delta(\eps')$, which implies that $\game$ has no strategy with value $1-\eps'$.
\end{itemize}
Furthermore, $\log(m) = O(\log(n)) = O(s_{\mathrm{inst}})$ and $\log(s) = O(\log(n)) = O(s_{\mathrm{inst}})$, giving us our desired question complexities,
and $\poly(m) = \poly(n) = \poly(2^{s_{\mathrm{inst}}})$ and $\poly(s) = \poly(n) = \poly(2^{s_\mathrm{inst}})$, giving us our desired answer complexities.
\end{proof}


\part{Answer reduction}

\label{part:answer}


\section{Testing error-correcting codes}\label{sec:error}

In \Cref{sec:answer-reduction} below,
rather than the prover sending the verifier their entire ``large" answer~$a$,
they will instead encode it into $\codee(a)$ using an error correcting code
and allow the verifier to query individual bits of the encoding.
(The fact that the verifier is allowed to query bits of the encoding rather than the original string
stems from the PCPP technology we use. See \Cref{sec:composing-with-an-error-correcting-code} for more details.)
In this section, we develop the tests which verify that provers are performing this task honestly,
so that when we query a subset of the bits~$I$,
they respond based on the bits of a codeword which was sampled independently of~$I$.
We will develop such a test for the low-degree code (\Cref{sec:testing-low-deg}).

Our proofs are entirely standard: we start with the known property tester for this code (i.e..\ \Cref{thm:anand-thomas-classical-low-degree}),
which allows us to query the prover's codeword at a uniformly random location.
Then we use the local decodability properties of this code to allow us to query arbitrary subsets of coordinates.
We begin by stating a slightly nonstandard definition of error-correcting codes relevant to our application.

\begin{definition}[Error-correcting codes]\label{def:error-correcting-code}
Let $m$ and $q$ be integers,
and let $\eta \in [0, 1]$.
An \emph{$(n, m, q, \eta)$-error-correcting code} $\codec = (\codee,\coded,\codes)$ is defined as follows.
\begin{itemize}
\item[$\circ$] $\codes$ is a subset of $\F_q^m$ such that for each $x \neq y \in \codes$, $x$ and $y$ have normalized Hamming agreement at most~$\eta$
		(i.e.\ the probability, over a uniformly random $\bi \in [m]$, that $x_{\bi} = y_{\bi}$ is at most $\eta$).
\item[$\circ$] $\codee:\{0, 1\}^n \rightarrow \codes \subseteq \F_q^m$ is the \emph{encoding map}.
\item[$\circ$] $\coded: \F_q^m \rightarrow \{0, 1\}^n \cup \{\bot\}$ is the \emph{decoding map}.
		For each $x \in \{0, 1\}^n$, $\coded(\codee(x)) = x$.
		In addition, for every~$w$ not in the range of~$\codee$, $\coded(w) = \bot$.
\end{itemize}
\end{definition}

\begin{remark}
The purpose of the subset $\codes$ is this:
in this section, we are designing games which test that a prover responds according to an error-correcting code.
This means that the prover should respond based on the encoding $\codee(x)$ of some string~$x \in \{0, 1\}^n$.
However, the games we design may only be able to test if the prover responds based on a string in~$\codes$,
which contains the encodings $\codee(x)$ but may include other strings as well.
This definition ensures that these other strings are still far from each other in Hamming distance.
\end{remark}

The next definition defines a subset tester.

\begin{definition}\label{def:testing-a-code}
Let $\codec = (\codee, \coded, \codes)$ be an $(n, m, q, \eta)$-error-correcting code.
Let~$k$ be an integer.
Given a game~$\game(\cdot)$  whose inputs are from the set of subsets of~$[m]$ of size~$k$
and a probability distribution~$\calD$ over this set,
we write $\game(\calD)$ for the game in which we first sample $\bI \sim \calD$ and then play~$\game(\bI)$.
Then $\game$ is a \emph{$k$-subset tester with robustness $\delta(\eps)$} for $\codec$ if it satisfies the following two properties.
\begin{itemize}
\item[$\circ$] \textbf{Completeness:}
Let $(\psi, M)$ be an EPR strategy in which $\{M_w\}$ is a measurement with outcomes in~$\{0, 1\}^n$.
Consider the partial strategy $(\psi, G)$ in which
\begin{equation*}
	G^I_{a_1, \ldots, a_k} := M_{[\codee(w)|_I = a_1, \ldots, a_k]}.
\end{equation*}
Then this can be extended to a (full) real commuting EPR strategy which, for each~$I$, passes $\game(I)$ with probability~$1$.
\item[$\circ$] \textbf{Soundness:} For any distribution~$\calD$, consider a strategy $(\psi, M)$ which passes $\game(\calD)$ with probability $1-\eps$.
Then there exists a measurement $\{G_w\}_{w}$ with outcomes~$w$ in~$\codes$ such that
\begin{equation*}
	M^I_{a_1, \ldots, a_k} \otimes I_{\reg{Bob}} \consistency_{\delta(\eps)} I_{\reg{Alice}} \otimes G_{[w|_I = a_1, \ldots, a_k]}.
\end{equation*}
\end{itemize}
\end{definition}

\subsection{Testing the low-degree code}\label{sec:testing-low-deg}

In this section, we show how to test the low-degree code.
This is essentially an exercise in generalizing \Cref{thm:anand-thomas-classical-low-degree} to arbitrary subsets.
We begin with some notation.

\begin{notation}
We write $\mathcal{F}_{q, k}^{m}$ for the family
$\displaystyle
\mathcal{F}^m_{q, k} = \{F \subseteq \F_q^m \mid |F| \leq k\}
$.
\end{notation}

Now we define the low-degree code tester.

\begin{definition}
Let $m$, $q$, and~$d$ be integers.  Let~$k$ be an integer, and let~$F$ be an element of~$\mathcal{F}^m_{q, k}$.
Then $\game_{\mathrm{LDsubset}}(m, q, d, F)$ is the game defined in \Cref{fig:classical-subset}.
\end{definition}

{
\floatstyle{boxed} 
\restylefloat{figure}
\begin{figure}
With probability~$\tfrac{1}{2}$ each, perform one of the following two tests.
\begin{enumerate}
	\item \textbf{Low-degree:} Perform $\game_{\mathrm{Surface}}(m, d, q,2)$. \label{item:low-deg-point}
	\item \textbf{Cross-check:}
			Flip an unbiased coin $\bb \sim \{0, 1\}$.
			Let~$\bs$ be a uniformly random subspace of dimension $k+1$ containing the points in~$F$.
			 With probability $\tfrac{1}{2}$ each:
			\begin{enumerate}
			\item Let $\bw$ be a uniformly random point in~$\bs$. Distribute the questions as follows:\label{item:bootstrapping-to-subspaces}
				\begin{itemize}
				\item[$\circ$] Player~$\bb$: give $\bw$; receive a value $\by \in \F_q$.
				\item[$\circ$] Player~$\overline{\bb}$: give $\bs$; receive a degree-$d$ polynomial $\bg:\bs \rightarrow \F_q$.
				\end{itemize}
				Accept if $\bg(\bw) = \by$.
			\item Distribute the questions as follows: \label{item:boosting-to-subsets}
				\begin{itemize}
				\item[$\circ$] Player~$\bb$: give $\bs$; receive a degree-$d$ polynomial $\bg:\bs \rightarrow \F_q$.
				\item[$\circ$] Player~$\overline{\bb}$: give $F$; receive a function $\boldf:\bF \rightarrow \F_q$.
				\end{itemize}
				Accept if $\bg|_{F} = \boldf$.
			\end{enumerate}
	\end{enumerate}
	\caption{The game $\game_{\mathrm{LDsubset}}(m, q, d, F)$. \label{fig:classical-subset}}
\end{figure}
}

The performance of the low-degree subset game is given by the following theorem.

\begin{theorem}\label{thm:low-degree-code-tester}
Consider low-degree parameters $\params = (n, q, h, H, m, \calS, \pi)$. Set $d = m(h-1)$.
Set $m' = q^m$.  We will identify strings in $\F_q^{m'}$ with functions $g:\F_q^m \rightarrow \F_q$.
Given $a \in \{0, 1\}^n$, define $\codee(a) = g_a$ and $\coded(g_a) = a$.
For all other $g:\F_q^m \rightarrow \F_q$ (i.e.\ those which are not the low-degree encoding of a string~$a$),
define $\coded(g) = \bot$.
Finally, define $\codes$ to be the set of degree~$d$ polynomials~$g:\F_q^m \rightarrow\F_q$.
Then $\codec = (\codee, \coded, \codes)$ is an $(n, m', q, d/q)$-error-correcting code.

Furthermore, there exists a constant $c > 0$ and a function $\delta(\eps) = \poly(\eps, dm/q^c)$ such that the following holds.
Let~$k$ be an integer.
Then $\game_{\mathrm{LDsubset}} := \game_{\mathrm{LDsubset}}(m, q, d, \cdot)$ is a $k$-subset tester for $\mathrm{LDCode}$ with robustness $\delta(\eps)$.

Finally,
\begin{equation*}
\qtime{\game_{\mathrm{LDsubset}}} = \poly(m, k,\log q),
\quad
\atime{\game_{\mathrm{LDsubset}}} = \poly(m, d^k, \log q),
\end{equation*}
\begin{equation*}
\qlength{\game_{\mathrm{LDsubset}}} = O(k m \log q),
\quad
\alength{\game_{\mathrm{LDsubset}}} = O(d^k \log (q)).
\end{equation*}
\end{theorem}

Before proving this, we need the following proposition.

\begin{proposition}\label{prop:almost-uniform}
Let $F \subseteq \F_q^m$ be of size at most~$k$. Consider the distribution $\calD_{\mathrm{twostep}}$ on points~$x\in\F_q^m$ generated by the following two-step process:
(i) let $\bs$ be a uniformly random subspace of size~$k+1$ containing~$F$, and (ii) draw $\bx$ uniformly at random from~$\bs$.
Let $\calD_{\mathrm{unif}}$ be the uniform distribution on~$\F_q^m.$
Then $d_\mathrm{TV}(\calD_{\mathrm{twostep}}, \calD_{\mathrm{unif}}) \leq 1/q$.
\end{proposition}
\begin{proof}
Let $x_1, \ldots, x_\ell$ be a maximal set of linearly independent elements from~$F$.
A uniformly random subspace of size~$k+1$ containing~$F$ can be generated as follows:
first, pick a uniformly random nonzero vector $\by_{\ell+1}$ linearly independent from~$F$,
then pick a uniformly random nonzero vector $\by_{\ell+2}$ linearly independent from~$F\cup\{\by_{\ell+1}\}$,
and so forth. Set $\bs = \mathrm{span}\{x_1, \ldots, x_\ell, \by_{\ell+1}, \ldots, \by_{k+1}\}$.
A uniformly random point in~$\bs$ will be of the form 
\begin{equation*}
\bx = \bt_1 x_1 + \cdots + \bt_\ell x_\ell + \bt_{\ell+1} \by_{\ell+1} + \cdots + \bt_{k+1} \by_{k+1},
\end{equation*}
where each $\bt_{i}$ is a uniformly random element in $\F_q$.
Because all the $\by_i$'s are linearly independent,
the linear combination $\bt_{\ell+1} \by_{\ell+1} + \cdots + \bt_{k+1} \by_{k+1}$
is zero only when $\bt_{\ell+1} = \cdots = \bt_{k+1} = 0$.
Otherwise, this linear combination is distributed as a uniformly random nonzero vector linearly independent from~$F$.
Thus, with probability $(q^{k+1-\ell})^{-1}$, $\bx$ is distributed as a uniformly random vector in the span of~$F$,
and otherwise it is distributed as a uniformly random vector outside the span of~$F$.
Given that the span of~$F$ has $q^{\ell}$ points, the total variation distance is
\begin{equation*}
\frac{1}{2} \left| \frac{q^{\ell}}{q^m} - \frac{1}{q^{k+1-\ell}}\right|
+ \frac{1}{2} \left| \frac{q^m- q^{\ell}}{q^m} - \left(1- \frac{1}{q^{k+1-\ell}}\right)\right|
= q^{\ell}\left(\frac{1}{q^{k+1}} - \frac{1}{q^m}\right)
\leq \frac{1}{q}.\qedhere
\end{equation*}
\end{proof}

Now we prove \Cref{thm:low-degree-code-tester}.

\begin{proof}[Proof of \Cref{thm:low-degree-code-tester}]
The fact that $\codec$ is an $(n, m', q, d/q)$-error-correcting code follows from Schwartz-Zippel (\Cref{lem:schwartz-zippel}).

\paragraph{Completeness.}
Let $(\psi, M)$ be an EPR strategy in which $\{M_w\}$ is a measurement with outcomes in~$\{0, 1\}^n$.
Consider the strategy $(\psi, G)$ in which for any subset of points~$F = \{y_1, \ldots, y_\ell\}$,
\begin{equation*}
	G^I_{a_1, \ldots, a_\ell} := M_{[g_x(y_1), \ldots, g_x(y_\ell) = a_1, \ldots, a_\ell]}.
\end{equation*}
(This covers the case of points ($\ell = 1$) and subsets~$F$ ($\ell = k$).)
In addition, for any subspace~$s$,
\begin{equation*}
	G^s_{f} := M_{[g_x|_s = f]}.
\end{equation*}
(This covers the case of the $2$-dimensional subspaces used in $\game_{\mathrm{Point}}$
and the $(k+1)$-dimensional subspaces used for the local decoding.)
By construction, $(\psi, G)$ is an EPR strategy,
and it is easy to see that it is a \emph{commuting} one as well.

We claim that $(\psi, G)$ passes $\game_{\mathrm{LDsubset}}(m, q, d, \calD)$ with probability~$1$.
We begin with the low-degree test.
By \Cref{fact:heh-heh-heh-gonna-make-anand-prove-this-so-i-can-take-the-day-off},
$M_w \otimes I_{\reg{Bob}} \consistency_0 I_{\reg{Alice}} \otimes M_w$.
Then by \Cref{fact:specialize-the-simeq},
\begin{equation*}
M_{[g_x(w) = b]} \otimes I_{\reg{Bob}}
\consistency_{0} I_{\reg{Alice}} \otimes M_{[g_x(w) = b]}.
\end{equation*}
This implies passing the low-degree test with probability~$1$, because
\begin{equation*}
G^s_{[f(w) = b]} \otimes I_{\reg{Bob}}
= M_{[g_x(w) = b]} \otimes I_{\reg{Bob}}
\consistency_{0} I_{\reg{Alice}} \otimes M_{[g_x(w) = b]} = I_{\reg{Alice}} \otimes G^w_b.
\end{equation*}
A similar argument shows the other tests pass with probability~$1$ as well.

\paragraph{Soundness.}

Let~$\calD$ be a distribution,
and let $(\psi, M)$ be a strategy which passes $\game_{\mathrm{LDsubset}}(m, q, d, \calD)$ with probability $1-\eps$.
The outline of the proof is as follows:
first we will use the low degree test in \Cref{item:low-deg-point} to ensure the test correctly answers low-degree point queries.
\Cref{item:bootstrapping-to-subspaces} will then bootstrap this to subspaces,
and \Cref{item:boosting-to-subsets} will further bootstrap this to subsets, proving the theorem.

\paragraph{Using the low-degree test.}
Passing the test with probability $1-\epsilon$ means passing the low-degree test with probability at least $1-2\epsilon$.
By \Cref{thm:anand-thomas-classical-low-degree}, this means that there exists a POVM measurement $G \in \polymeas{m}{d}{q}$ such that
\begin{equation}\label{eq:just-applied-low-degree}
M^w_b \otimes I_{\reg{Bob}} \consistency_{\delta(\eps)} I_{\reg{Alice}} \otimes G_{[g(w) = b]},
\qquad G_g \otimes I_{\reg{Bob}} \consistency_{\delta(\eps)} I_{\reg{Alice}} \otimes G_g,
\end{equation}
where the first is on the uniform distribution over $\F_q^m$.

\paragraph{Bootstrapping to subspaces.}
Define $\calD_{\mathrm{twostep}}$ to be the two-step sampling process $(\bF, \bs, \bw)$ as in \Cref{item:bootstrapping-to-subspaces}.
By \Cref{prop:almost-uniform}, the marginal distribution on~$\bw$ has total variation distance at most~$1/q$ with $\calD_{\mathrm{unif}}$.
As a result, we can apply \Cref{fact:swap-dists} to \Cref{eq:just-applied-low-degree}, yielding
\begin{equation}\label{eq:gonna-triangle-the-shit-outta-this}
M^w_b \otimes I_{\reg{Bob}} \consistency_{\delta(\eps)} I_{\reg{Alice}} \otimes G_{[g(w) = b]}
\end{equation}
on distribution $\calD_{\mathrm{twostep}}$.
Similarly, by \Cref{fact:specialize-the-simeq},
\begin{equation}\label{eq:incoming-triangle}
G_{[g(w) = b]} \otimes I_{\reg{Bob}} \consistency_{\delta(\eps)} I_{\reg{Alice}} \otimes G_{[g(w) = b]}
\end{equation}
on distribution $\calD_{\mathrm{twostep}}$.

Next, because the strategy passes the test in \Cref{item:bootstrapping-to-subspaces} with probability at least $1-4\epsilon$,
\begin{equation}\label{eq:watch-out-for-the-triangle}
M^{w}_y \otimes I_{\reg{Bob}} \consistency_{\epsilon} I_{\reg{Alice}} \otimes M^s_{[g(w) = y]}.
\end{equation}
on distribution $\calD_{\mathrm{twostep}}$.
Combining \Cref{eq:gonna-triangle-the-shit-outta-this,eq:incoming-triangle,eq:watch-out-for-the-triangle} with our second triangle inequality (\Cref{fact:triangle-like}),
\begin{equation*}
M^s_{[g(w) = y]} \otimes I_{\reg{Bob}}
\consistency_{\delta(\eps)} I_{\reg{Alice}} \otimes G_{[g(w) = b]}.
\end{equation*}
By~\Cref{prop:same-on-point-same-on-subspace}, we conclude that
\begin{equation}\label{eq:finished-with-subspaces}
M^s_{f} \otimes I_{\reg{Bob}}
\consistency_{\delta(\eps)} I_{\reg{Alice}} \otimes G_{[g|_s = f]}.
\end{equation}

\paragraph{Concluding with subsets.}
The strategy passes the test in \Cref{item:boosting-to-subsets} with probability at least $1-4\epsilon$. As a result,
\begin{equation*}
M^F_f \otimes I_{\reg{Bob}}
\consistency_{\epsilon}
I_{\reg{Alice}} \otimes M^s_{[g|_F = f]}.
\end{equation*}
Applying \Cref{fact:specialize-the-simeq} to \Cref{eq:finished-with-subspaces},
\begin{equation*}
M^s_{[g|_F=f]} \otimes I_{\reg{Bob}}
\consistency_{\delta(\eps)} I_{\reg{Alice}} \otimes G_{[h|_F = h]}.
\end{equation*}
Similarly, applying \Cref{fact:specialize-the-simeq} to \Cref{eq:just-applied-low-degree},
\begin{equation*}
G_{[h|_F = f]} \otimes I_{\reg{Bob}} \consistency_{\delta(\eps)} I_{\reg{Alice}} \otimes G_{[h|_F=f]}
\end{equation*}
Applying the triangle inequality (\Cref{fact:triangle-like}) to these three equations,
we get
\begin{equation*}
M^F_f \otimes I_{\reg{Bob}}
\consistency_{\delta(\epsilon)}
I_{\reg{Alice}} \otimes G_{[g|_{F} = f]}
\end{equation*}
with respect to distribution $\calD_{\mathrm{twostep}}$,
and therefore, by \Cref{fact:close-on-marginal}, with respect to $\calD$.
\end{proof}

\subsection{Efficiently decodable codes}

Our application requires error-correcting codes with two further properties.
The first property is that the decoding map $\coded(\cdot)$ be efficiently computable.
(The encoding map, on the other hand, is allowed arbitrary complexity.
This is because we will leave the task of computing the encoding maps to the provers.) 
The second, more technical property is
we require that the code \emph{embed} the codeword, in the following sense:
the encoding $\codee(x)$ of a string~$x$ should actually \emph{contain} the string~$x$,
and the function for where to find each bit of~$x$ in $\codee(x)$ should be efficiently computable.

\begin{definition}[Efficiently-decodable error-correcting codes]\label{def:efficient-codes}
Let $m, q:\Z^+\rightarrow \Z^+$,
and let $\eta:\Z^+\rightarrow [0, 1]$.
Let $t_{\mathrm{Dec}}, t_{\mathrm{Emb}}:\Z^+ \rightarrow \Z^+$.
We say that $\codec_n=(\codee_n, \coded_n, \codes_n)$ is an \emph{$(n, m, q, \eta, t_{\mathrm{Dec}}, t_{\mathrm{Emb}})$-efficient code family}
if the following three conditions are true.
\begin{itemize}
\item[$\circ$] For each~$n$, $(\codee_n, \coded_n, \codes_n)$ is an $(n, m(n), q(n), \eta(n))$-error-correcting code.
\item[$\circ$] There exists an algorithm $\mathrm{Alg}_{\coded}$ which, on input $(n, w)$, outputs $\coded_n(w)$.
			Furthermore, $\mathrm{Alg}_{\coded}$ runs in time $t_{\mathrm{Dec}}(n)$.
\item[$\circ$] There exists an embedding $\mu_n:[n] \rightarrow [m(n)]$ such that for each $i \in [n]$, $x_i = (\codee_n(x))_{\mu_n(i)}$.
			Furthermore, there is an algorithm $\mathrm{Alg}_{\mathrm{Emb}}$ which, on input $(n, i)$, computes $\mu_n(i)$ in time
			 $t_{\mathrm{Emb}}(n)$.
\end{itemize}
\end{definition}

Now, we show that the low-degree code is efficiently-decodable.
The decoding algorithm follows a simple strategy:
assuming that the input is a proper encoding of a message,
they can directly read off the message from the input.
Then they compute the encoding of the purported message and check that it equals the input.

\begin{fact}\label{fact:low-degree-code-is-efficient}
There is a $(n, m', q, \eta, t_{\mathrm{Dec}}, t_{\mathrm{Emb}})$-error-correcting code $\codec$ with parameters set as follows:
\begin{equation*}
m'(n) = \poly(n), \quad q(n) = \mathrm{polylog}(n), \quad \eta(n) = \frac{1}{\mathrm{polylog}(n)},
\end{equation*}
\begin{equation*}
t_{\mathrm{Dec}}(n) = \poly(n), \qquad t_{\mathrm{Emb}}(n) = \mathrm{polylog}(n).
\end{equation*}
In addition, $\codec$ has a $k$-subset test~$\game$ with robustness $\delta(\eps) = \poly(\eps, 1/\log(n))$ such that
\begin{equation*}
\qtime{\game} = \poly(\log n, k),
\quad
\atime{\game} = \poly(\log(n)^k),
\end{equation*}
\begin{equation*}
\qlength{\game} = O(k \log n),
\quad
\alength{\game} = O(\log(n)^{2k}).
\end{equation*}
\end{fact}
\begin{proof}
We instantiate the canonical low-degree encoding from \Cref{def:canonical-low-degree} with the ``rule of thumb" parameters from \Cref{eq:parameters}:
\begin{equation*}
h(n) = \Theta(\log(n)),
\qquad
m(n) = \Theta\left(\frac{\log(n)}{\log\log(n)}\right),
\qquad
q(n) = \mathrm{polylog}(n).
\end{equation*}
If we set $d(n) = m(n)\cdot (h(n) - 1)$, then this is a code with distance $\eta(n) = 1-d(n)/q(n) = 1-1/\mathrm{polylog}(n)$.
In addition, it has length $m'(n) = q(n)^{m(n)} = \poly(n)$.
Finally, the canonical low-degree encoding gives us the embedding $\mu_{\mathrm{Emb}} := \sigma_{m, t_1, t_2}$.
By \Cref{prop:canonical-time}, it takes time $t_{\mathrm{Emb}}(n) = \mathrm{poly log}(n)$ to compute.

Now we design the decoding algorithm $\mathrm{Alg}_{\coded}$.
On input $(n, w)$, it rejects if~$w$ is not length~$m'$.
Otherwise, it interprets~$w$ as a function $f:\F_q^m\rightarrow \F_q$.
It queries~$g$ on the points $\pi(1), \ldots, \pi(n)$.
Let $a \in \{0, 1\}^n$ be the received answers.
If~$g$ is a codeword, it equals the low-degree function~$g_a$.
So the algorithm simply iterates over all~$x \in \F_q^m$ and checks that $f(x) = g_a(x)$.
By \Cref{prop:canonical-time}, computing $g_a(x)$ can be done in time $\poly(n)$,
and so this takes time $t_{\mathrm{Dec}}(n) = \poly(n)$ in total.

Finally, the performance of the subset tester follows from \Cref{thm:low-degree-code-tester} with our setting of parameters.
\end{proof}



\section{Answer reduction}\label{sec:answer-reduction}

In this section, we carry out the answer reduction.
Our main result will be to take the $\poly(n)$ question complexity, $O(2^n)$ answer complexity $\MIP^*$ protocol for $\succinctsquared$
given by \Cref{cor:succinct-sat-protocol-with-big-answer-size}
and convert it to one whose answer complexity is also $\poly(n)$;
this is \Cref{thm:main} below.

Our answer reduction will apply to any game with a value-$1$ real commuting EPR strategy.
We will require two properties of these strategies:
first, that they can be extended to strategies that pass subset tests with probability~$1$, as in \Cref{def:testing-a-code};
and second, that they are ``oracularizable".
We explain this second property in the next section. 

\subsection{Oracularization}

Our technique will not work for all entangled games but only for a
subset, for which a single prover can simulate both prover's actions
if required to. We call such games ``oracularizable'' games.

\begin{definition}
Given a two-player entangled game~$\game$,
its \emph{oracularization} is the game $C_{\mathrm{oracle}}(\game)$ given in \Cref{fig:oracle}.
If $\game$ is value-$1$,
then we call it \emph{oracularizable}, if $\val{C_{\mathrm{oracle}}(\game)} = 1$ as well. 
We also note that for \emph{any} game~$\game$, if $\val{\game} \leq 1 -
  \delta$, then $\val{C_{\mathrm{oracle}}(\game)} \leq 1 - O(\delta)$.
\end{definition}

{
  \floatstyle{boxed}
  \restylefloat{figure}
  \begin{figure}
    Given a game $\game$, sample a tuple $(\bx_0, \bx_1, \bC) \sim
    \game$, and flip two unbiased coins $\bb, \bc \sim \{0,1\}$.
    With probability $\tfrac{1}{2}$ each, perform one of the following two tests. 
    \begin{enumerate}
      \item \textbf{Verify:} Distribute the questions as follows:
      	\begin{itemize}
	\item[$\circ$] Player~$\bb$: send the pair $(\bx_0, \bx_1)$ and receive answers $(\ba_0, \ba_1)$.
	\item[$\circ$] Player~$\overline{\bb}$: send $\bx_{\bc}$ and receive an answer $\ba_2$.
	\end{itemize}
	\item \textbf{Consistency:} Play the consistency game with question $\bx_0, \bx_1$.
    \end{enumerate}
    Accept if $\ba_2 = \ba_{\bc}$ and $V(\bx_0, \bx_1, \ba_0, \ba_1) = 1$.
    \caption{The oracularized game $C_{\mathrm{oracle}}(\game)$.\label{fig:oracle}}
  \end{figure}
  }
  
\noindent
A real commuting EPR strategy allows ``Player~$\bb$" to sample both questions~$\bx_0$ and~$\bx_1$ simultaneously.
As a result, if a game~$\game$ has a value-$1$ real commuting EPR strategy, then it is oracularizable.

The value of oracularization is that when the verifier checks $V(\bx_0, \bx_1, \ba_0, \ba_1) = 1$,
both~$\ba_0$ and~$\ba_1$ come from the same prover rather than two different provers.
This seems like a minor change, but in fact it makes all the difference.
Our goal is to reduce the verifier's runtime by having the provers encode their answers using PCP technology.
When the answers come from \emph{both} provers,
the relevant piece of PCP technology is a \emph{distributed} PCP,
but it is known by a simple argument of Reingold that distributed PCPs do not exist
(see the discussion in~\cite{ARW17}).
The key difficulty comes from the fact that Alice needs to prepare her PCP proof without
knowing Bob's question and answer, and vice versa, and this turns out to be impossible in general.
On the other hand,
when the answers come from a single prover,
we can use traditional PCPs to implement the answer reduction, of which we have a variety of constructions.
We note that oracularized games \emph{do} still have checks between players,
but these are equality checks and will be easy to implement in the answer reduction regime.

\subsection{Probabilistically checkable proofs of proximity}

In this section, we introduce the main PCP technology we will use for our answer reduction.
In the oracularized game,
the provers want to convince us not just that $V(\cdot, \cdot, \cdot, \cdot)$
is satisfiable---which we already know to be true by construction---but that $(\bx_0, \bx_1, \ba_0, \ba_1)$
is a particular assignment which satisfies it.
For this, we need a stronger notion of a PCP called a \emph{probabilistically checkable proof of proximity (PCPP)}.
These allow one to check that an input~$x$ is close to a satisfying assignment of a circuit~$C$ (hence the ``proximity")
by making a small number of queries to~$x$.
These were originally introduced in the independent works of~\cite{BGH+06}
and~\cite{DR06} (where they were called \emph{assignment testers}).

In our case, we will need even stronger PCPPs in which the verifier is not only query-efficient but \emph{time}-efficient as well.
The history of these time-efficient PCPPs goes back to the original proof of $\MIP = \nexp$ and the various attempts to ``scale it down"~\cite{OD05}.
The most famous line of research considered proof systems in which the verifier's query complexity is restricted,
and this eventually led to the proof of the PCP theorem~\cite{AS98,ALM+98}.
A parallel line of research considered proof systems in which the verifier's runtime is restricted (so-called ``transparent" proofs)~\cite{BFLS91}.
The latter of these was revisited in the work of Ben-Sasson et al.~\cite{BGH+05},
who showed that both lines of research could be remerged in the ``scaled down" setting
by constructing a PCPP in which the verifier is both query-efficient \emph{and} time-efficient.
Though their main result is actually sufficient for our purposes,
we will cite the work of Mie~\cite{Mie09}, which improves on their result in the regime we care about.
Finally, we note the work of Meir~\cite{Mei14}, who reproves the bounds of Ben-Sasson et al.~\cite{BGH+05} using combinatorial methods.

To our knowledge, ours is the first use of a time-efficient PCPP specifically for its time-efficient properties in the quantum literature.
Natarajan and Vidick~\cite{NV18a} used the time-efficient PCPP of~\cite{BGH+05} to prove the quantum games PCP conjecture,
but the property they needed was not that it was time-efficient,
but that the bits of the proof are linear functions of the bits of the assignment.
We note that we do not need this property here.

In this literature, it is common to consider ``pair languages" consisting of strings $(x, y)$
in which~$x$ is small and given to the verifier and~$y$ is large and accessible only through query access.
This maps perfectly onto our scenario,
in which the verifier supplies the ``small" questions $\bx_0, \bx_1$
and the prover supplies the ``large" answers $\ba_0, \ba_1$.

\begin{definition}
A \emph{pair language} $L$ is a subset of $\{0, 1\}^* \times \{0, 1\}^*$.
Given $x \in \{0, 1\}^*$, we write $L_x = \{y \in \{0, 1\}^* \mid (x,y) \in L\}$.
\end{definition}

The next two definitions state the notion of an efficient PCPP verifier.

\begin{definition}[{\cite[Definition 2.1]{BGH+05}}]
Let $r, q:\Z^+ \rightarrow \Z^+$ and $t:\Z^+ \times \Z^+ \rightarrow \Z^+$.
An \emph{$(r, q, t)$-restricted PCPP verifier} is a probabilistic machine that, given a string~$x$ (called the \emph{explicit input})
and a number~$K$ (in binary) as well as oracle access to an \emph{implicit input} $y \in \{0, 1\}^{K}$ and to a \emph{proof oracle} $\pi \in \{0, 1\}^*$,
tosses $r(|x|+K)$ coins, queries the oracles $(y, \pi)$ for a total of $q(|x|+K)$ symbols,
runs in time $t(|x|, K)$, and outputs a Boolean verdict in $\{\mathrm{accept}, \mathrm{reject}\}$.
\end{definition}

\begin{definition}[{\cite[Definition 2.2]{BGH+05}}]
For functions $r, q:\Z^+ \rightarrow \Z^+$, $t:\Z^+\times \Z^+ \rightarrow \Z^+$,
and constants
$s, \gamma \in  [0, 1]$,
a pair language $L \subseteq \{0, 1\}^* \times \{0, 1\}^*$ is in $\mathrm{PCPP}_{s, \gamma}[r, q, t]$
if there exists an $(r, q, t)$-restricted PCPP verifier~$V$ with the following properties:
\begin{itemize}
\item[$\circ$] \textbf{Completeness:} If $(x, y) \in L$ then there exists a $\pi$ such that $\Pr_R[\text{$V^{y, \pi}(x, |y|; R)$ accepts}] = 1$,
	where $V^{y, \pi}(x, |y|; R)$ denotes the decision $V$ on input $(x, |y|)$,
	oracle access to $(y, \pi)$, and coin tosses $R$.
\item[$\circ$] \textbf{Soundness:} If $(x, y)$ is such that $y$ is $\gamma$-far from $L_x \cap \Sigma^{|y|}$,
	then for every $\pi$ it holds that $\Pr_R[\text{$V^{y, \pi}(x, |y|;R)$ accepts}] \leq s$.
\end{itemize}
\end{definition}

Mie's time-efficient PCPP is states as follows.

\begin{theorem}[{\cite[Theorem 1]{Mie09}}]\label{thm:pcpp-mie}
Suppose that $L$ is a pair language in $\ntime(T)$
for some non-decreasing function $T : \Z^+ \rightarrow \Z^+$.
Then, for every two constants $s, \gamma >0$,
we have $L \in \mathrm{PCPP}_{s, \gamma}[r, q, t]$, for
\begin{itemize}
\item[$\circ$] Randomness complexity $r(m) = \log_2 T(m) + O(\log\log T(m))$.
\item[$\circ$] Query complexity $q(m) = O(1)$,
\item[$\circ$] Verification time $t(n, K) = \poly(n, \log K, \log T(n + K))$.
\end{itemize}
\end{theorem}

We note that this is in fact a much stronger than what we will actually need.
In particular, we will only apply this to languages $L$ in \emph{deterministic} $\mathsf{TIME}(T)$,
which are trivially in $\ntime(T)$.

\subsection{Composing with an error-correcting code}\label{sec:composing-with-an-error-correcting-code}

The verifier in a PCPP rejects any input which is $\gamma$-far from an accepting input,
but of course we want our verifier to reject \emph{all} non-accepting inputs, no matter their distance.
To do this, we will (i) encode the verifier's inputs using an error-correcting code
and (ii) check that the inputs are properly encoded (using, for example, the low-degree test).
This approach of composing a PCPP with an error-correcting code is standard
and stretches back in spirit to the transparent proofs of~\cite{BFLS91} (see the discussion of this in~\cite{BGH+06}).

Now we show how to compose an $\MIP^*$ game with an error-correcting code.

\begin{definition}[Error-correcting the provers' answers]\label{def:error-correcting-answers}
Let $V = (\mathrm{Alg}_{\mathrm{Q}}, \mathrm{Alg}_{\mathrm{A}})$ be an $\MIP^*$ verifier (the language it verifies is not important).
Suppose on inputs of size~$n$ it has question length~$\ell_{\mathrm{Q}}(n)$ answer length~$\ell_{\mathrm{A}}(n)$.
Write $L_{\mathrm{A}}$ for the language decided by $\mathrm{Alg}_{\mathrm{A}}$.
Let $\codec_k = (\codee_k,\coded_k, \codes_k)$ be a $(k, m, q, \eta, t_{\mathrm{Dec}}, t_{\mathrm{Emb}})$-efficient code family with decoding algorithm $\mathrm{Alg}_{\mathrm{Dec}}$.
Then $L_{\mathrm{A}} \circ \codec$ is a new language defined as follows:
suppose $(\mathsf{input}, x_0, x_1, y_0, y_1) \in L_{\mathrm{A}}$.
Let $n$ be the length of $\mathsf{input}$ and $\ell = \ell_{\mathrm{A}}(n)$.
Then $(\mathsf{input}, x_0, x_1, \codee_{\ell}(y_0), \codee_{\ell}(y_1)) \in L_{\mathrm{A}} \circ \mathsf{Code}$.
\end{definition}

Now, we prove a couple of properties about the composed verifier.
First, we show that its runtime is not much slower than the original verifier's.

\begin{proposition}[Runtime of the composed verifier]\label{prop:composed-runtime}
Let~$V$ and~$\codec_k$ be as in \Cref{def:error-correcting-answers}.
Suppose $\mathrm{Alg}_{\mathrm{A}}$ runs in time $T(n)$.
Then there is an algorithm, which we denote $\mathrm{Alg}_{\mathrm{A}} \circ \mathsf{Code}$,
deciding the language $L_{\mathrm{A}} \circ \mathsf{Code}$.
In addition, on inputs $(\mathsf{input}, x_0, x_1, z_0, z_1)$
in which $|\mathsf{input}| = n$,
$|x_0| = |x_1| = \ell_{\mathrm{Q}}(n)$,
and $|z_0| = |z_1| = m(\ell_{\mathrm{A}}(n))$,
the algorithm runs
in time $T(n) + t_{\mathrm{Dec}}(\ell_{\mathrm{A}}(n))$.
\end{proposition}

\begin{proof}
On input $(\mathsf{input}, x_0, x_1, z_0, z_1)$, we define the action of $\mathrm{Alg}_{\mathrm{A}} \circ \mathsf{Code}$ as follows.
\begin{enumerate}
\item Compute $n$, the length of $\mathsf{input}$. Set $\ell := \ell_{\mathrm{A}}(n)$.
\item Check that~$z_0$ and~$z_1$ have length~$m(\ell)$. If they don't, reject.
\item Compute $y_0 = \mathrm{Alg}_{\coded}(\ell, z_0)$ and $y_1 = \mathrm{Alg}_{\coded}(\ell, z_1)$. If either $y_0$ or $y_1$ is $\bot$, reject.\label{item:decode-that-string}
\item Otherwise, we know that $y_0, y_1 \in \{0, 1\}^\ell$.  Run $\mathrm{Alg}_{\mathrm{A}}(\mathsf{input}, x_0, x_1, y_0, y_1)$. Accept if it accepts, and reject if it rejects.\label{item:run-that-algorithm}
\end{enumerate}
It is immediate that $\mathrm{Alg}_{\mathrm{A}} \circ \mathsf{Code}$ computes $L_{\mathrm{A}} \circ \mathsf{Code}$.
As for the time complexity, \Cref{item:decode-that-string} runs in time $t_{\mathrm{Dec}}(\ell_{\mathrm{A}}(n))$ 
and \Cref{item:run-that-algorithm} runs in time $T(n)$.
Combined, these two give the bound in the proposition statement.
\end{proof}

Next, we show that this construction solves the ``problem" discussed at the beginning of the section,
namely that if we perform answer reduction by replacing $\mathrm{Alg}_{\mathrm{A}} \circ \mathsf{Code}$ with a PCPP verifier,
rather than just $\mathrm{Alg}_{\mathrm{A}}$, then the verifier will reject \emph{all} inputs which are not in the language,
not just those which are $\delta$-far, provided that those inputs are encoded as per \Cref{def:error-correcting-answers}.

\begin{proposition}\label{prop:new-soundness}
Let~$V$ and~$\codec_k$ be as in \Cref{def:error-correcting-answers}.
Let $s, \gamma > 0$ be constants,
and let $V_{\mathrm{PCPP}}$ be the PCPP verifier for the language $L_{\mathrm{A}} \circ \mathsf{Code}$
guaranteed by \Cref{thm:pcpp-mie} with these parameters.
Suppose that $1-\eta(k) \geq 2\gamma$ for all~$k$.
Then we have the following soundness condition.
\begin{itemize}
\item[$\circ$] \textbf{Soundness:} Consider $(\mathsf{input}, x_0, x_1, z_0, z_1)$ for $\mathsf{input}$ of length~$n$, $x_0$ and~$x_1$ of length~$\ell_{\mathrm{Q}}(n)$,
	and $z_0, z_1 \in \codes_{\ell}$, for $\ell := \ell_{\mathrm{A}}(n)$.
	Suppose this does not correspond to the encoding of an accepting assignment in $L_{\mathrm{A}}$.
	In other words, suppose that there are no $y_0, y_1 \in \{0, 1\}^\ell$ such that $(\mathsf{input}, x_0, x_1, y_0, y_1)$ is in $L_{\mathrm{A}}$
	and $z_0 = \codee_{\ell}(y_0), z_1 = \codee_{\ell}(y_1)$.
	Then $V_{\mathrm{PCPP}}$ accepts $(\mathsf{input}, x_0, x_1, z_0, z_1)$ with probability at most~$s$.
	In math, for every $\pi$ it holds that
	\begin{equation*}
		\Pr_R[\text{$V_{\mathrm{PCPP}}^{z_0, z_1, \pi}(\mathsf{input}, x_0, x_1, |z_0| + |z_1|; R)$ accepts}] \leq s.
	\end{equation*}
\end{itemize}
\end{proposition}
\begin{proof}
Given $(\mathsf{input}, x_0, x_1, z_0, z_1)$, write $A:= (L_{\mathrm{A}} \circ \mathsf{Code})_{\mathsf{input}, x_0, x_1} \cap \F_{q(\ell)}^{|z_0| + |z_1|}$.
By assumption, $(z_0, z_1)$ is not in~$A$.
Using this, we would like to show that $(z_0, z_1)$ is in fact $\gamma$-far from~$A$, in which case the PCPP verifier accepts with probability at most~$s$.

To do this, suppose $(z_0', z_1') \in A$.
By design, there exists $y_0', y_1' \in \{0, 1\}^\ell$ such that $z_0' = \codee(y_0')$ and $z_1' = \codee(y_1')$.
This means that $z_0', z_1' \in \codes_{\ell}$.
On the other hand, since~$(z_0, z_1)$ is not in~$A$, we must have either $z_0' \neq z_0$ or $z_1' \neq z_1$ (or both).
We will assume the first without loss of generality.
Then by the distance property of the code, since $z_0, z_0' \in \codes_{\ell}$,
their normalized Hamming distance is at least $1-\eta(\ell) \geq 2\gamma$.
This immediately means that $(z_0,z_1)$ and $(z_0', z_1')$.
are at least $\gamma$-far apart, and we are done.
\end{proof}

\subsection{The answer reduction protocol}

We are almost ready to state the answer reduction protocol.
Before doing so, we discuss one final nuisance, which is that we will also need the prover to encode their \emph{proof} with an error-correcting code. 
The reason is that we would like to query the proof on a view~$\bJ$ sampled by the PCPP verifier.
However, the prover might cheat and respond based only on the view~$\bJ$ rather than a global proof~$\pi$.
To prevent this, we force them to commit to a global error-correcting encoding of their proof~$\pi$ using a tester as in \Cref{def:testing-a-code}.
Then, we use the fact that the error-correcting code \emph{embeds} their string to allow us to extract the view~$\bJ$ by asking for the coordinates in~$\mu(\bJ)$.

We now state the answer reduction protocol.

\begin{definition}\label{def:answer-reduction-game}
We instantiate the answer-reduced $\mip^*$ protocol with the following algorithms and parameters.
\begin{itemize}
\item[$\circ$] Let $V = (\mathrm{Alg}_{\mathrm{Q}}, \mathrm{Alg}_{\mathrm{A}})$ be an $\MIP^*$ verifier for a language~$L$.
		Write $L_{\mathrm{A}}$ for the language decided by $\mathrm{Alg}_{\mathrm{A}}$.
		Suppose on inputs of size~$n$, the verifier~$V$ has question length~$\ell_{V,\mathrm{Q}}(n)$, answer length~$\ell_{V,\mathrm{A}}(n)$,
		question time~$t_{V,\mathrm{Q}}(n)$, and answer time~$t_{V,\mathrm{A}}(n)$.
\item[$\circ$] Let $\codec_k = (\codee_k,\coded_k, \codes_k)$ be a $(k, m, q, \eta, t_{\mathrm{Dec}}, t_{\mathrm{Emb}})$-efficient code family
		with decoding algorithm $\mathrm{Alg}_{\mathrm{Dec}}$ and embedding $\mu_k$.
\item[$\circ$] Let $\game_k$ be a game which tests for $\codec_k$ with robustness $\chi_k(\eps)$.
		Suppose it has question length~$\ell_{\game, \mathrm{Q}}(k)$, answer length~$\ell_{\game,\mathrm{A}}(k)$,
		question time~$t_{\game,\mathrm{Q}}(k)$, and answer time~$t_{\game,\mathrm{A}}(k)$.
\item[$\circ$] Let $s, \delta > 0$ be constants, and let $V_{\mathrm{PCPP}}$ be the PCPP verifier for the language $L_{\mathrm{A}} \circ \mathsf{Code}$
			guaranteed by \Cref{thm:pcpp-mie} with these parameters.
			Suppose on inputs of size~$n$ it has proof length~$\ell_\pi(n)$.
			By \Cref{prop:composed-runtime}, $L_{\mathrm{A}} \circ \codec$ is in time
			$t_{\mathrm{compose}}(n) = t_{V, \mathrm{A}}(n) + t_{\mathrm{Dec}}(\ell_{V,\mathrm{A}}(n))$.
			We can therefore write $V_{\mathrm{PCPP}}$'s verification time as
			\begin{equation*}
				t_{\mathrm{PCPP}}(n) = \poly(n + \ell_{V, \mathrm{Q}}(n), \log(m(\ell_{V, \mathrm{A}}(n))), \log(t_{\mathrm{compose}}(n))).
			\end{equation*}
			Finally, $\ell_\pi(n) = t_{\mathrm{compose}}(n) \cdot \mathrm{poly log}(t_{\mathrm{compose}}(n))$.
\end{itemize}
Write $\ell_1 := \ell_{V, \mathrm{A}}(n)$ and $\ell_2 := \ell_\pi(n)$.
Then the \emph{answer reduction game} $\game_{\mathrm{answer}}(\mathsf{input};V, \codec, \game, s, \delta)$
is given in \Cref{fig:answer-reduction}.
We write $V_{\mathrm{answer}}$ for the corresponding verifier.
\end{definition}

{
	\floatstyle{boxed}
	\restylefloat{figure}
	\begin{figure}Flip two unbiased coins $\bb, \bc \sim \{0, 1\}$. 
    		Sample questions $(\bx_0, \bx_1) \sim \mathrm{Alg}_{\mathrm{Q}}(\mathsf{input})$.
		Sample a view $\bI_0, \bI_1, \bJ \sim V_{\mathrm{PCPP}}(\mathsf{input}, \bx_0, \bx_1)$.
		Set $\bJ' = \mu_{\ell_2}(\bJ)$.
		Select $\bi_0, \bi_1 \in [m(\ell_1)]$ and $\bj \in [m(\ell_\pi(n))]$ uniformly at random.
		Set $\bT_0 = \bI_0 \cup \{\bi_0\}$, $\bT_1 = \bI_1 \cup \{\bi_1\}$, and $\bU = \bJ' \cup \{\bj\}$.
	With probability $\frac{1}{8}$ each, perform one of the following eight tests.
	\begin{enumerate}
	\item \textbf{Verify:}
		Distribute the question as follows:
			\begin{itemize}
			\item[$\circ$] Player~$\bb$: give $(\bx_0, \bx_1)$, $\bT_0, \bT_1, \bU$; receive $\ba_0, \ba_1, \ba_2$.
			\end{itemize}
		Accept if $V_{\mathrm{PCPP}}(\mathsf{instance}, \bx_0, \bx_1)$ accepts on $\ba_0|_{\bI_0}, \ba_1|_{\bI_1}, \ba_2|_{\bJ'}$.
	\item \textbf{Cross checks:}
	\begin{enumerate}
	\item \textbf{Consistency test:} Distribute the questions as follows:
			\begin{itemize}
			\item[$\circ$] Player~$\bb$: give $(\bx_0, \bx_1)$, $\bT_0, \bT_1, \bU$; receive $\ba_0, \ba_1, \ba_2$.
			\item[$\circ$] Player~$\overline{\bb}$: give $(\bx_0, \bx_1)$, $\bT_0, \bT_1, \bU$; receive $\ba_0', \ba_1', \ba_2'$.
			\end{itemize}
		Accept if $\ba_0 = \ba_0'$, $\ba_1 = \ba_1'$, and $\ba_2 = \ba_2'$.
	\item \textbf{Answer cross-check:} Distribute the questions as follows:
			\begin{itemize}
			\item[$\circ$] Player~$\bb$: give $(\bx_0, \bx_1)$, $\bT_0, \bT_1, \bU$; receive $\ba_0, \ba_1, \ba_2$.
			\item[$\circ$] Player~$\overline{\bb}$: give $\bx_{\bc}, \bT_{\bc}'$; receive $\ba_{\bc}'$.
			\end{itemize}
		Accept if $\ba_{\bc} = \ba_{\bc}'$.
	\item \textbf{Proof cross-check:} Distribute the questions as follows:
			\begin{itemize}
			\item[$\circ$] Player~$\bb$: give $(\bx_0, \bx_1)$, $\bT_0, \bT_1, \bU$; receive $\ba_0, \ba_1, \ba_2$.
			\item[$\circ$] Player~$\overline{\bb}$: give $\bx_0, \bx_1, \bU$; receive $\ba_{2}'$.
			\end{itemize}
		Accept if $\ba_2 = \ba_2'$.
	\end{enumerate}
	\item \textbf{Code checks:}
	\begin{enumerate}
	\item \textbf{Answer code check:}
		Sample questions $(\bw_0, \bw_1) \sim \game_{\ell_1}(\bT_{\bc})$.
		Distribute the questions as follows:
			\begin{itemize}
			\item[$\circ$] Player~$\bb$: give $\bx_{\bc}, \bw_{0}$; receive $\ba_0$.
			\item[$\circ$] Player~$\overline{\bb}$: give $\bx_{\bc}, \bw_{1}$; receive $\ba_1$.
			\end{itemize}
		Accept if $\game_{\ell_1}(\bT_{\bc})$ accepts on $\ba_0, \ba_1$.
	\item \textbf{Proof code check:}
		Sample questions $(\bw_0, \bw_1) \sim \game_{\ell_2}(\bU)$.
		Distribute the questions as follows:
			\begin{itemize}
			\item[$\circ$] Player~$\bb$: give $\bx_0, \bx_1$, $\bw_{0}$; receive $\ba_0$.
			\item[$\circ$] Player~$\overline{\bb}$: give $\bx_0, \bx_1$, $\bw_{1}$; receive $\ba_1$.
			\end{itemize}
		Accept if $\game_{\ell_2}(\bU)$ accepts on $\ba_0, \ba_1$.
	\end{enumerate}
    \end{enumerate}
    \caption{The answer reduction game $\game_{\mathrm{answer}}(\mathsf{input};V, \codec, \game, s, \delta)$.\label{fig:answer-reduction}}
  \end{figure}
}

\begin{theorem}\label{thm:answer-reduction}
Suppose~$V$, $\codec$, $\game$, and $V_{\mathrm{PCPP}}$ are as in \Cref{def:answer-reduction-game}.
Suppose $s, \gamma$ are chosen to be constants such that $\eta(k) \geq 2\gamma$ for all~$k$.
Suppose further that~$V$ has the following property: for any input in~$L$,
the provers have a real commuting EPR strategy with value~$1$.
Then~$V_{\mathrm{answer}}$ is also an $\MIP^*$ verifier for~$L$ with the following two conditions:
\begin{itemize}
\item[$\circ$] (Completeness) If $\mathsf{input} \in L$, then there is a value-$1$ strategy.
\item[$\circ$] (Soundness) Given $\mathsf{input}$, suppose there is a strategy with value $1-\eps$.
			Then there is a strategy for~$V$ on input $\mathsf{input}$ with value $1-\delta(\eps)$, where $\delta(\eps)$ is given by
\begin{equation*}
\delta(\eps) := \poly(\chi_{\ell_1}(\poly(\eps)), \chi_{\ell_2}(\poly(\eps)), \eta(\ell_1), \eta(\ell_2)).
\end{equation*}
\end{itemize}
Hence, if we choose our parameters so that $1-\delta(\eps)$ is greater than the soundness of~$V$,
this implies that $V_{\mathrm{answer}}$ is an $\MIP^*$ verifier for~$L$ with soundness $1-\eps$.

Furthermore, the question and answer lengths and runtimes are dominated by two subroutines: the ``Verify" subroutine $S_1$ 
and the ``Code Check" subroutine $S_2$ (consisting of both the answer code check and the proof code check).
The complexity of the Verify subroutine is
\begin{align*}
\qlength{S_1} &= O(\ell_{V, \mathrm{Q}}(n) + \log(m(\ell_{V,\mathrm{A}}(n)))+ \log(m(\ell_{\pi}(n)))),\\
\alength{S_1} &= O(\log(q(\ell_{V,\mathrm{A}}(n))) + \log(q(\ell_\pi(n)))),\\
\qtime{S_1} &= O\left(t_{V,\mathrm{Q}}(n) + t_{\mathrm{PCPP}}(n) + t_{\mathrm{Emb}}(\ell_{\pi}(n))\right),\\
\atime{S_1} &= O(t_{\mathrm{PCPP}}(n)).
\end{align*}
In addition, the complexity of the Code Check subroutine is
\begin{align*}
\qlength{S_2} &= O(\ell_{\game, \mathrm{Q}}(\ell_{V, \mathrm{A}}(n)) + \ell_{\game, \mathrm{Q}}(\ell_\pi(n)) + \ell_{V, \mathrm{Q}}(n)),\\
\alength{S_2} &= O(\ell_{\game, \mathrm{A}}(\ell_{V, \mathrm{A}}(n)) + \ell_{\game, \mathrm{A}}(\ell_\pi(n))),\\
\qtime{S_2} &=  O(t_{\game, \mathrm{Q}}(\ell_{V, \mathrm{A}}(n)) + t_{\game, \mathrm{Q}}(\ell_\pi(n)) + t_{V, \mathrm{Q}}(n) + t_{\mathrm{Emb}}(\ell_{\pi}(n))),\\
\atime{S_2} &= O(t_{\game, \mathrm{A}}(\ell_{V, \mathrm{A}}(n)) + t_{\game, \mathrm{A}}(\ell_\pi(n))).
\end{align*}
Thus, the complexity of the overall protocol is the sum of these two.
\end{theorem}
\begin{proof}
The fact  that~$S_1$ and~$S_2$ dominate the lengths and runtimes of the protocol
is because~$S_1$ dominates the lengths and runtimes of the two cross-check subroutines,
whose questions and answers are subsets of those in~$S_1$.
Now we compute the complexity of~$S_1$.
\begin{itemize}
\item[$\circ$] \textbf{Question length:} The pair $(\bx_0, \bx_1)$ has total length $\ell_{V,\mathrm{Q}}(n)$ by definition.
		The pair $\bI_0, \bI_1$ are subsets of indices of constant size into each of the implicit inputs of $L_{\mathrm{A}} \circ \codec$,
		which are supposed to be encodings of strings of size $\ell_{V, \mathrm{A}}(n)$.
		Hence, the encodings have size $m(\ell_{V, \mathrm{A}}(n))$, and so each input is specified with the log of this many bits.
		Finally, $\bJ$ is a constant-sized set of indices into a proof of size $\ell_{\pi}(n)$, and $\mu(\bJ)$ converts these into indices into an encoding of of this proof.
		As the encoding has length $m(\ell_{\pi}(n))$, each index can be specified with the log of this many bits.
\item[$\circ$] \textbf{Answer length:} The strings $\ba_0, \ba_1$ contains values from an error-correcting code with alphabet $q(\ell_{V, \mathrm{A}}(n))$,
		and the string~$\ba_2$ contains values from an error-correcting code with alphabet $q(\ell_{\pi}(n))$.
\item[$\circ$] \textbf{Question time:} The running time of $\mathrm{Alg}_{\mathrm{Q}}$ is $t_{V,\mathrm{Q}}(n)$.
		The running time of $V_{\mathrm{PCPP}}$ is $t_{\mathrm{PCPP}}(n)$.
		Finally, the running time to compute $\mu(\bJ)$ given~$\bJ$ is $t_{\mathrm{Emb}}(\ell_{\pi}(n))$.
\item[$\circ$] \textbf{Answer time:} The running time is simply the running time of $V_{\mathrm{PCPP}}$, i.e.\ $t_{\mathrm{PCPP}}(n)$.
\end{itemize}
As for the complexity of~$S_2$, it just performs the code tester~$\game_k$ for message lengths $k = \ell_{V, \mathrm{A}}(n)$ and $\ell_\pi(n)$ and so inherits the lengths and runtimes of the tester for these two values of~$k$, except on top of that it also has to sample~$(\bx_0, \bx_1)$ and compute~$\bJ'$.
Sampling~$(\bx_0, \bx_1)$ takes time~$t_{V,\mathrm{Q}}(n)$ and contributes $O(\ell_{V, \mathrm{Q}}(n))$ to the question lengths,
and computing~$\bJ'$ takes time~$t_{\mathrm{Emb}}(\ell_{\pi}(n))$.

\paragraph{Completeness.}
Suppose $\mathsf{input}$ is in~$L$.
Then there is a real commuting EPR strategy $(\psi, M)$ with value~$1$ for~$V$ on $\mathsf{input}$.
We will use this to demonstrate a value-$1$ strategy for $V_{\mathrm{answer}}$.
This will be the strategy $(\psi, G)$ which uses the same EPR state $\ket{\psi}$ and has measurement matrices~$G$ defined as follows.

Fix questions $x_0, x_1$, $T_0, T_1$, and $U$.  We begin by defining the simplest measurement,
\begin{equation}\label{eq:i-dont-know-what-to-call-this}
G^{x_c, T_c}_{a_c} := M^{x_c}_{[\codee_{\ell_1}(z_c)|_{T_c} = a_c]}.
\end{equation}
If Alice and Bob measure with $M^{x_0}$ and $M^{x_1}$ and receive strings $z_0, z_1$,
then because this strategy is value~$1$,
we will always have $V(\mathsf{input}, x_0, x_1, z_0, z_1) = 1$.
As a result, there always exists \emph{some} proof for $V_{\mathrm{PCPP}}$
that $(\mathsf{input}, x_0, x_1, \codee_{\ell_1}(z_0), \codee_{\ell_1}(z_1))$ is in $L_{\mathrm{A}} \circ \mathsf{Code}$.
We denote this proof $\pi(x_0, x_1, z_0, z_1)$; if there are multiple such proofs, we pick one arbitrarily.
Then we define the measurement
\begin{equation}\label{eq:i-dont-know-what-to-call-this-part-two-return-of-the-bat}
G^{x_0, x_1, U}_{a_2} := (M^{x_0} \cdot M^{x_1})_{[\codee_{\ell_2}(\pi(x_0, x_1, z_0, z_1))|_{U} = a_2}.
\end{equation}
Next, we define the measurement
\begin{equation*}
G^{x_0, x_1, T_0, T_1, U}_{a_0, a_1, a_2}
:= (M^{x_0}\cdot M^{x_1})_{[\codee_{\ell_1}(z_0)|_{T_0}, \codee_{\ell_1}(z_1)|_{T_1}, \codee_{\ell_2}(\pi(x_0, x_1, z_0, z_1))|_{U} = a_0, a_1, a_2]}.
\end{equation*}
Now, via \Cref{eq:i-dont-know-what-to-call-this,eq:i-dont-know-what-to-call-this-part-two-return-of-the-bat},
the~$G$ measurement is exactly of the form required by \Cref{def:testing-a-code}.
As a result, it can be extended to a measurement which passes the answer and proof code checks with probability~$1$.
Performing this extension concludes the design of the strategy.

By construction, this strategy passes the answer and proof code checks with probability~$1$.
As for the remaining tests, let us begin with the answer cross-check in the case of $\bc = 0$, the other case being symmetric.
Because~$M$ is a real commuting EPR strategy,
by \Cref{fact:heh-heh-heh-gonna-make-anand-prove-this-so-i-can-take-the-day-off}
we have that $M^x_a \otimes I_{\reg{Bob}} \consistency_0 I_{\reg{Alice}} \otimes M^x_a$ for any distribution on~$x$.
If we consider the measurement $(M^{x_0} \cdot M^{x_1})_{z_0, z_1}$, then $(M^{x_0} \cdot M^{x_1})_{z_0} = M^{x_0}_{z_0}$. As a result,
\begin{equation*}
(M^{x_0} \cdot M^{x_1})_{z_0} \otimes I_{\reg{Bob}}
\consistency_0 I_{\reg{Alice}} \otimes M^{x_0}_{z_0}.
\end{equation*}
Finally, by data processing (\Cref{fact:specialize-the-simeq}), this implies that
\begin{equation*}
(M^{x_0} \cdot M^{x_1})_{[\codee_{\ell_1}(z_0)|_{T_0}=a_0]} \otimes I_{\reg{Bob}}
\consistency_0 I_{\reg{Alice}} \otimes M^{x_0}_{[\codee_{\ell_1}(z_0')|_{T_0} = a_0]}.
\end{equation*}
But this is equivalent to saying that $G^{x_0, x_1, T_0, T_1, U}_{a_0} \otimes I_{\reg{Bob}} \consistency_0 I_{\reg{Alice}} \otimes G^{x_0, T_0}_{a_0}$,
which implies passing the cross-check test with probability~$1$.
A similar argument holds for the other tests, with the exception of the verification step.

Consider the measurement $M^{x_0}_{z_0} \cdot M^{x_1}_{z_1}$.
By construction and the properties of the PCPP verifier, if this measurement always outputs~$z_0, z_1$ such that $V(\mathsf{input}, x_0, x_1, z_0, z_1) = 1$,
then the~$G$ strategy always passes the verify step.
But because~$M$ is a real commuting EPR strategy,
$M^x_z \otimes I_{\reg{Bob}} \consistency_0 I_{\reg{Alice}} \otimes M^x_z$,
which implies that
\begin{equation*}
M^{x_0}_{z_0} \otimes M^{x_1}_{z_1}
\approx_0 M^{x_0}_{z_0} \cdot M^{x_1}_{z_1} \otimes I_{\reg{Bob}}.
\end{equation*}
Thus, these two measurements have the same output distribution.
But the left-hand side always outputs~$z_0, z_1$ which satisfy the verifier, because this strategy passes the verifier with probability~$1$.
This concludes the completeness step.

\paragraph{Soundness.}
Suppose $\mathsf{input}$ is not in~$L$.
Let $(\psi, M)$ be a strategy that passes with probability $1-\eps$.

\paragraph{Code checks.}
Passing the overall test with probability $1-\eps$
means the strategy passes the answer code check with probability $1-8\eps$.
Given values $c $, $x_{c}$, write $1-\eps_{c, x_{c}}$ for the probability the code check passes conditioned on these values.
Then with probability at least $1-8\eps^{1/2}$, $\eps_{c, x_{c}} \leq \eps^{1/2}$.
When this occurs, \Cref{thm:low-degree-code-tester} implies that there exists a measurement $\{G^{x_c}_w\}_w$ with outcomes in $\codes_{\ell_1}$ such that
\begin{equation*}
M^{x_c, T_c}_a \otimes I_{\reg{Bob}}\consistency_{\delta(\eps)} I_{\reg{Alice}} \otimes G^{x_c}_{[w|_{T_c} = a]}
\end{equation*}
with respect to the distribution of~$\bT_c$ conditioned on~$c$ and~${x}_{c}$.
When this does \emph{not} occur, we still can assume such a measurement so that
\begin{equation*}
M^{x_c, T_c}_a \otimes I_{\reg{Bob}}\consistency_{1} I_{\reg{Alice}} \otimes G^{x_c}_{[w|_{T_c} = a]}
\end{equation*}
trivially, by \Cref{fact:trivial-upper-bound-approx-delta}. Thus, if we average over $\bc$ and $\bx_{\bc}$,
\begin{equation}\label{eq:almost-at-tweedledee}
M^{x_c, T_c}_a \otimes I_{\reg{Bob}}\consistency_{\delta(\eps)} I_{\reg{Alice}} \otimes G^{x_c}_{[w|_{T_c} = a]}
\end{equation}
with respect to the distribution on $\bc, \bx_{\bc}, \bT_{\bc}$.
A similar argument with respect to the consistency guarantee of \Cref{thm:low-degree-code-tester} implies that
\begin{equation}\label{eq:consistency-of-the-encoding}
G^{x_c}_w \otimes I_{\reg{Bob}} \consistency_{\delta(\eps)} I_{\reg{Alice}} \otimes G^{x_c}_w.
\end{equation}
By \Cref{fact:specialize-the-simeq}, this implies that
\begin{equation*}
G^{x_c}_{[w|_{T_c} = a]} \otimes I_{\reg{Bob}} \consistency I_{\reg{Alice}} \otimes G^{x_c}_{[w|_{T_c} = a]}.
\end{equation*}
As a result, if we apply \Cref{fact:agreement} to this and \Cref{eq:almost-at-tweedledee} and then use the triangle inequality (\Cref{fact:triangle}), we conclude
\begin{equation}\label{eq:tweedledee}
M^{x_c, I_c}_a \otimes I_{\reg{Bob}}\approx_{\delta(\eps)}G^{x_c}_{[w|_{I_c} = a]} \otimes  I_{\reg{Bob}}
\end{equation}
with respect to the distribution on $\bc, \bx_{\bc}, \bI_{\bc}$.

Applying a similar argument yet again, this time to the proof code check, implies that for every $x_0, x_1$,
there exists a measurement $\{H^{x_0, x_1}_w\}_w$ with outcomes in $\codes_{\ell_2}$ such that
\begin{equation}\label{eq:tweedledum}
M^{x_0, x_1, U}_a \otimes I_{\reg{Bob}} \approx_{\delta(\eps)} H^{x_0, x_1}_{[w|_{U} = a]} \otimes I_{\reg{Bob}}
\end{equation}
with respect to the distribution on~$\bx_0, \bx_1, \bU$.
Thus, by \Cref{fact:approx-delta-game-value}, we can assume that \Cref{eq:tweedledee,eq:tweedledum} hold with equality
with a loss of only $\delta(\eps)$ in the game value.
In addition, by \Cref{thm:naimark}, we can assume that the~$G$ and~$H$ measurements are all projective,
possibly replacing~$\psi$ with a different state.

\paragraph{Cross checks.}
Our next step is to apply the cross-checks. Passing these with probability $1-\delta(\eps)$ implies the bounds
\begin{equation}\label{eq:apply-the-cross-checks}
M^{x_0, x_1, T_0, T_1, U}_{a_0} \otimes I_{\reg{Bob}}\simeq_{\delta(\eps)} I_{\reg{Alice}} \otimes M^{x_0, T_0}_{a_0}
= I_{\reg{Alice}} \otimes G^{x_0}_{[w|_{T_0} = a_0]},
\end{equation}
\begin{equation}\label{eq:apply-the-cross-checks-2}
M^{x_0, x_1, T_0, T_1, U}_{a_1} \otimes I_{\reg{Bob}}\simeq_{\delta(\eps)} I_{\reg{Alice}} \otimes  M^{x_1, T_1}_{a_1}
= I_{\reg{Alice}} \otimes G^{x_1}_{[w|_{T_1} = a_1]},
\end{equation}
\begin{equation*}
M^{x_0, x_1, T_0, T_1, U}_{a_2} \otimes I_{\reg{Bob}}\simeq_{\delta(\eps)} I_{\reg{Alice}} \otimes  M^{x_0, x_1, U}_{a_2}
= I_{\reg{Alice}} \otimes H^{x_0, x_1}_{[w|_{U} = a_2]},
\end{equation*}
\begin{equation}\label{eq:Ms-self-consistency}
M^{x_0, x_1, T_0, T_1, U}_{a_0, a_1, a_2} \otimes I_{\reg{Bob}}\simeq_{\delta(\eps)} I_{\reg{Alice}} \otimes  M^{x_0, x_1, T_0, T_1, U}_{a_0, a_1, a_2}.
\end{equation}
At this point, we would like to apply \Cref{fact:low-degree-sandwich-on-steroids}.
To do so, we have to verify the distance property of our functions,
and this will follow from the fact that we augmented our index sets $\bI_0$, $\bI_1$, and $\bJ'$ with an additional uniformly random index.
To see this, consider two nonequal $w$ and $w'$ in $\codes_{\ell_1}$.
Then for them to agree on $\bT_0$, they must agree on $\bi_0$, and this happens only $\eta(\ell_1)$ fraction of the time.
The same holds for $\bU$, with the bound of $\eta(\ell_2)$.
As a result, \Cref{fact:low-degree-sandwich-on-steroids} implies the following:
consider the POVM measurement $\{\Lambda^{x_0, x_1}_{w_0, w_1, \pi}\}$
with outcomes~$w_0, w_1$ in $\codes_{\ell_1}$ and~$\pi$ in $\codes_{\ell_2}$
defined as
\begin{equation}\label{eq:dont-touch-me-bro}
\Lambda^{x_0, x_1}_{w_0, w_1, \pi}
:= G^{x_0}_{w_0} \cdot G^{x_1}_{w_1} \cdot H^{x_0, x_1}_{\pi} \cdot G^{x_1}_{w_1} \cdot G^{x_0}_{w_0}.
\end{equation}
Then
\begin{equation}\label{eq:gonna-use-this-once}
M^{x_0, x_1, T_0, T_1, U}_{a_0, a_1, a_2} \otimes I_{\reg{Bob}}\consistency_{\delta(\eps)}
I_{\reg{Alice}} \otimes \Lambda^{x_0, x_1}_{[w_0|_{T_0}, w_1|_{T_1}, \pi|_{U} = a_0, a_1, a_2]}.
\end{equation}
From this, \Cref{eq:Ms-self-consistency} implies
\begin{equation}\label{eq:gonna-use-this-once}
M^{x_0, x_1, T_0, T_1, U}_{a_0, a_1, a_2} \otimes I_{\reg{Bob}}\approx_{\delta(\eps)}
\Lambda^{x_0, x_1}_{[w_0|_{T_0}, w_1|_{T_1}, \pi|_{U} = a_0, a_1, a_2]} \otimes I_{\reg{Bob}}.
\end{equation}
Thus, by \Cref{fact:approx-delta-game-value}, we can assume that \Cref{eq:gonna-use-this-once} holds with equality by replacing~$M$ with~$G$,
incurring a loss of only $\delta(\eps)$ in the game value.
(Unlike before, here we do \emph{not} invoke \Cref{thm:naimark} on~$J^{x_0, x_1}$ to make it a projective measurement,
as that would likely change the structure in \Cref{eq:dont-touch-me-bro}, which we will need later.)

\paragraph{Verification.}
The strategy passes the verify check with probability $1-\delta(\eps)$.  
By \Cref{eq:gonna-use-this-once}
(which we now assume is equality),
this is the same probability as if we (i) sample~$\bx_0, \bx_1$,
(ii) use~$\Lambda$ to draw $\bw_0, \bw_1, \bpi$,
(iii) draw $\bI_0, \bI_1, \bJ$ conditioned on~$\bx_0, \bx_1$,
(iv) then draw $\bT_0, \bT_1, \bU$ conditioned on~$\bI_0, \bI_1, \bJ$,
(v) compute $\ba_0 = \bw_0|_{\bT_0}$, $\ba_1 = \bw_1|_{\bT_1}$, and $\ba_2 = \bw_2|_{\bU}$,
and (vi) give $\ba_0|_{\bI_0}, \ba_1|_{\bI_1}, \ba_2|_{\bJ}$ to $V_{\mathrm{PCPP}}$ and accept if it accepts.

Condition on a fixed choice of~$x_0, x_1$ and a draw for~$w_0, w_1, \pi$.
The PCPP verifier receives answers to its~$\bI_0$ and~$\bI_1$ queries based on~$w_0$ and~$w_1$, which are in~$\codes_{\ell_1}$.
In addition, although~$\pi$ is in~$\codes_{\ell_2}$ and may not correspond to the encoding of an actual proof string,
the verifier only queries it at points in the image of the embedding~$\mu_{\ell_2}$.
As a result, the answers $V_{\mathrm{PCPP}}$ receives to its~$J$ queries are consistent with \emph{some} fixed proof string.
Thus, by \Cref{prop:new-soundness}, since $1-\eta(k) \geq 2\gamma$ for all~$k$, if the probability
the verifier accepts is greater than~$s$, then there are strings $y_0, y_1 \in \{0, 1\}^{\ell_1}$
such that $w_0 = \codee_{\ell_1}(y_0)$, $w_1 = \codee_{\ell_1}(y_1)$
and $V(\mathsf{input}, x_0, x_1, y_0, y_1) = 1$.
Averaging over all~$\bx_0, \bx_1$ and~$\bw_0, \bw_1, \bpi$, we conclude that
\begin{equation}\label{eq:absolute-unit-of-a-strategy}
\Pr[V(\mathsf{input}, \bx_0, \bx_1, \coded_{\ell_1}(\bw_0), \coded_{\ell_1}(\bw_1)) = 1] \geq \frac{1-\delta(\eps)-s}{1-s} = 1-\delta(\eps).
\end{equation}
Recall that the decoding map is one-to-one except on those strings not in the range of the encoding map, which it maps to~$\bot$ instead.
As we can assume that the verifier~$V$ always rejects when it receives~$\bot$ for an answer,
this tells us that $\coded_{\ell_1}(\bw_0), \coded_{\ell_1}(\bw_1) \neq \bot$ with probability at least $1-\delta(\eps)$.

\paragraph{Wrapping it up.}
Now we give a strategy for causing the verifier~$V$ to accept with high probability on~$\mathsf{input}$.
It uses state~$\psi$, and given question~$x$ it applies the measurement $\{A^x_a\}_a$ defined as
\begin{equation*}
A^x_a := G^x_{[\coded_{\ell_1}(w) = a]}.
\end{equation*}
Consider the verifier~$V'$ which samples $(\bx_0, \bx_1)$, gives them to Alice and Bob, receives~$\bw_0, \bw_1$,
and accepts if $V(\mathsf{input}, \bx_0, \bx_1, \coded_{\ell_1}(\bw_0), \coded_{\ell_1}(\bw_1)) = 1$.
Then~$V$ accepts on strategy~$A$ with the same probability that~$V'$ accepts on strategy~$G$.
In other words, if we define $S(x_0, x_1)$ to be the set of $(w_0, w_1)$
such that
\begin{equation*}
V(\mathsf{input}, x_0, x_1, \coded_{\ell_1}(w_0), \coded_{\ell_1}(w_1))=1,
\end{equation*}
then the probability that~$V'$ accepts on strategy~$G$ is
\begin{equation}\label{eq:its-the-final-countdown}
\E_{(\bx_0, \bx_1)} \sum_{w_0, w_1 \in S(\bx_0, \bx_1)} \bra{\psi} G^{x_0}_{w_0} \otimes G^{x_1}_{w_1}\ket{\psi}.
\end{equation}

To show this is large, we begin by showing that the~$G$'s commute with each other.
To see this, note that \Cref{eq:apply-the-cross-checks,eq:apply-the-cross-checks-2} implies that for a fixed $c \in \{0, 1\}$,
\begin{equation*}
\Lambda^{x_0, x_1}_{[w_c|_{T_c} = a_c]} \otimes I_{\reg{Bob}}\simeq_{\delta(\eps)}
 I_{\reg{Alice}} \otimes G^{x_c}_{[w_c'|_{T_c} = a_c]}.
\end{equation*}
However, by~the distance properties of our code and the fact that $\bT_c$ contains a uniformly random index, this implies that
\begin{equation}\label{eq:dont-copy-that-floppy}
\Lambda^{x_0, x_1}_{w_c} \otimes I_{\reg{Bob}}\simeq_{\delta(\eps)}
 I_{\reg{Alice}} \otimes G^{x_c}_{w_c}.
\end{equation}
As a result,
\begin{align*}
G^{x_0}_{w_0} \cdot G^{x_1}_{w_1} \otimes I_{\reg{Bob}}
&\approx_{\delta(\eps)} G^{x_0}_{w_0} \otimes \Lambda^{x_0, x_1}_{w_1}\\
&\approx_{\delta(\eps)} I_{\reg{Alice}} \otimes \Lambda^{x_0, x_1}_{w_1}\cdot \Lambda^{x_0, x_1}_{w_0}\\
&\approx_{\delta(\eps)} I_{\reg{Alice}} \otimes \Lambda^{x_0, x_1}_{w_0}\cdot \Lambda^{x_0, x_1}_{w_1}\\
&\approx_{\delta(\eps)} G^{x_1}_{w_1} \otimes \Lambda^{x_0, x_1}_{w_0}\\
&\approx_{\delta(\eps)}  G^{x_1}_{w_1} \cdot G^{x_0}_{w_0} \otimes I_{\reg{Bob}}.
\end{align*}
A similar argument as the one establishing~\eqref{eq:dont-copy-that-floppy} implies that
\begin{equation*}
G^{x_c}_{w_c} \otimes I_{\reg{Bob}}\simeq_{\delta(\eps)}
 I_{\reg{Alice}} \otimes G^{x_c}_{w_c}.
\end{equation*}
Thus,
\begin{align*}
G^{x_0}_{w_0} \otimes G^{x_1}_{w_1}
&= G^{x_0}_{w_0} \cdot G^{x_0}_{w_0} \otimes G^{x_1}_{w_1}\\
&\approx_{\delta(\eps)} G^{x_0}_{w_0} \cdot G^{x_0}_{w_0} \cdot G^{x_1}_{w_1} \otimes I_{\reg{Bob}}\\
&\approx_{\delta(\eps)} G^{x_0}_{w_0} \cdot G^{x_1}_{w_1} \cdot G^{x_0}_{w_0} \otimes I_{\reg{Bob}}\\
&= \Lambda^{x_0, x_1}_{w_0, w_1} \otimes I_{\reg{Bob}}.
\end{align*}
As a result, by \Cref{fact:approx-delta-generalized-game-value},  \Cref{eq:its-the-final-countdown} is at least $1-\delta(\eps)$ by \Cref{eq:absolute-unit-of-a-strategy}.
This concludes the proof of the theorem.
\end{proof}

\subsection{Applying the answer reduction protocol}

In this section, we instantiate \Cref{thm:answer-reduction} with the low-degree code
and then apply it to our $\neexp$ protocol.

\begin{theorem}\label{thm:applied-with-low-degree-code}
Let $V = (\mathrm{Alg}_{\mathrm{Q}}, \mathrm{Alg}_{\mathrm{A}})$ be an $\MIP^*$ verifier for a language~$L$.
		Write $L_{\mathrm{A}}$ for the language decided by $\mathrm{Alg}_{\mathrm{A}}$.
		Suppose on inputs of size~$n$, the verifier~$V$ has question length~$\ell_{V,\mathrm{Q}}(n)$, answer length~$\ell_{V,\mathrm{A}}(n)$,
		question time~$t_{V,\mathrm{Q}}(n)$, and answer time~$t_{V,\mathrm{A}}(n)$.
		Then there exists another $\MIP^*$ verifier $V_{\mathrm{ans}}$ for~$L$ with the following parameters.
\begin{align*}
\qlength{V_{\mathrm{ans}}} &= O(\ell_{V, \mathrm{Q}}(n) + \log(\ell_{V,\mathrm{A}}(n))+ \log(t_{V, \mathrm{A}}(n))),\\
\alength{V_{\mathrm{ans}}} &= O(\mathrm{poly log}(\ell_{V, \mathrm{A}}(n)) + \mathrm{poly log}(t_{V, \mathrm{A}}(n))),\\
\qtime{V_{\mathrm{ans}}} &= O\left(t_{V,\mathrm{Q}}(n)) + \poly(n + \ell_{V, \mathrm{Q}}(n), \log(\ell_{V, \mathrm{A}}(n)), \log(t_{V, \mathrm{A}}(n))\right),\\
\atime{V_{\mathrm{ans}}} &= \poly(n + \ell_{V, \mathrm{Q}}(n), \log(\ell_{V, \mathrm{A}}(n)), \log(t_{V, \mathrm{A}}(n))).
\end{align*}
\end{theorem}
\begin{proof}
We instantiate the low-degree code in \Cref{fact:low-degree-code-is-efficient}.
It gives an error correcting code with parameters $(n, \poly(n), \mathrm{poly log}(n), \mathrm{poly log}(n)^{-1},\poly(n), \mathrm{polylog}(n))$
and a $c$-subset test~$\game_k$ with robustness $\chi_k(\eps) = \poly(\eps, \log(k)^{-1})$ such that
\begin{equation*}
\qtime{\game_k} = \poly(\log k, c),
\quad
\atime{\game_k} = \poly(\log(k)^c),
\end{equation*}
\begin{equation*}
\qlength{\game_k} = O(c \log k),
\quad
\alength{\game_k} = O(\log(k)^{2c}).
\end{equation*}
We then apply \Cref{thm:answer-reduction} with $s, \gamma = \tfrac{1}{10}$.
At this point, the theorem follows immediately,
but as deriving it can be cumbersome, we fill in the details.

By construction, $t_{\mathrm{Dec}}(n) = \poly(n)$.
As a result,
\begin{equation*}
t_{\mathrm{compose}}(n)
= t_{V, \mathrm{A}}(n) + t_{\mathrm{Dec}}(\ell_{V,\mathrm{A}}(n))
= t_{V, \mathrm{A}}(n) + \poly(\ell_{V,\mathrm{A}}(n)).
\end{equation*}
Thus,
\begin{equation*}
\ell_\pi(n)
= t_{\mathrm{compose}}(n) \cdot \mathrm{poly log}(t_{\mathrm{compose}}(n))
= \poly(t_{V, \mathrm{A}}(n), \ell_{V,\mathrm{A}}(n)).
\end{equation*}
Now, $m(n) = \poly(n)$. Thus,
\begin{align*}
t_{\mathrm{PCPP}}(n)
&= \poly(n + \ell_{V, \mathrm{Q}}(n), \log(m(\ell_{V, \mathrm{A}}(n))), \log(t_{\mathrm{compose}}(n)))\\
&= \poly(n + \ell_{V, \mathrm{Q}}(n), \log(\ell_{V, \mathrm{A}}(n)), \log(t_{V, \mathrm{A}}(n))).
\end{align*}
Furthermore, $q(n) = \mathrm{polylog}(n)$ and $t_{\mathrm{Emb}}(n) = \mathrm{poly log}(n)$.
As a result,
\begin{align*}
\log(m(\ell_\pi(n))) &= O(\log(\ell_{V,\mathrm{A}}(n))+ \log(t_{V, \mathrm{A}}(n))),\\
\log(q(\ell_\pi(n))) &= O(\log \log(\ell_{V,\mathrm{A}}(n))+ \log \log(t_{V, \mathrm{A}}(n))),\\
t_{\mathrm{Emb}}(\ell_\pi(n)) &= \poly(\log(\ell_{V, \mathrm{A}}(n)), \log(t_{V, \mathrm{A}}(n))).
\end{align*}
The theorem now follows from applying these bounds to \Cref{thm:answer-reduction}.
\end{proof}

Crucially, although polynomial factors of $t_{V,Q}(n)$ and $\ell_{V, Q}(n)$ appear in \Cref{thm:applied-with-low-degree-code},
only the \emph{logarithms} of $t_{V,A}(n)$ and $\ell_{V, A}(n)$ appear in this theorem.
As a result, if we apply this to \Cref{cor:succinct-sat-protocol-with-big-answer-size}, we arrive at our main result.

\begin{theorem}
  \label{thm:main}
There is an $\MIP^*$ verifier~$\game$ for $\succinctsquared$ with parameters
\begin{equation*}
\qlength{\game} =O(n), \quad
\alength{\game} = \poly(n),
\end{equation*}
\begin{equation*}
\qtime{\game} = \poly(n), \quad
\atime{\game} = \poly(n).
\end{equation*}
\end{theorem}


\bibliographystyle{myhalpha}
\bibliography{wright}

\end{document}